\documentclass{lmcs} %
\pdfoutput=1

\usepackage{lastpage}
\lmcsdoi{21}{4}{27}
\lmcsheading{}{\pageref{LastPage}}{}{}%
{Aug.~18,~2024}{Dec.~08,~2025}{}


\usepackage{hyperref}

\usepackage{cleveref}
\usepackage{ebproof}
\usepackage{stmaryrd} %
\usepackage{stackengine} %
\usepackage[utf8]{inputenc}

\usepackage{tikz}
\usetikzlibrary{shapes, arrows, arrows.spaced, arrows.meta, chains,positioning, cd, shapes.geometric,
    decorations, decorations.pathmorphing, decorations.markings, decorations.pathreplacing, automata, backgrounds, petri, matrix, fit, calc, graphs, quotes, bending}
\tikzset{apply style/.code={\tikzset{#1}}}

\usepackage[notion,paper]{knowledge}

\usepackage{amssymb}
\usepackage{mathtools}
\usepackage{thmtools, thm-restate}

\DeclarePairedDelimiterX\tog[1]{\langle\!\langle}{\rangle\!\rangle}{#1}

\makeatletter
\let\orgdescriptionlabel\descriptionlabel
\renewcommand*{\descriptionlabel}[1]{%
  \let\orglabel\label
  \let\label\@gobble
  \phantomsection
  \edef\@currentlabel{#1}%
  \let\label\orglabel
  \orgdescriptionlabel{#1}%
}
\makeatother

\newcounter{mylabelcounter}
\makeatletter
\newcommand{\labeltext}[2]{%
#1\refstepcounter{mylabelcounter}%
\immediate\write\@auxout{%
  \string\newlabel{#2}{{1}{\thepage}{{\unexpanded{#1}}}{mylabelcounter.\number\value{mylabelcounter}}{}}%
}%
}
\makeatother

\crefname{thm}{Theorem}{Theorems}
\crefname{thmC}{Theorem}{Theorems}
\crefname{prop}{Proposition}{Propositions}
\crefname{propC}{Proposition}{Propositions}
\crefname{property}{Property}{Properties}
\crefname{lem}{Lemma}{Lemmas}
\crefname{lemC}{Lemma}{Lemmas}
\crefname{cor}{Corollary}{Corollaries}
\crefname{corC}{Corollary}{Corollaries}
\crefname{defi}{Definition}{Definitions}
\crefname{defiC}{Definition}{Definitions}
\crefname{section}{Section}{Sections}
\crefname{figure}{Figure}{Figures}
\crefname{exa}{Example}{Examples}
\crefname{rem}{Remark}{Remarks}
\crefname{remC}{Remark}{Remarks}

\usepackage{centernot}

\usepackage{scalerel} %
\DeclarePairedDelimiter\set{\{}{\}}
\DeclarePairedDelimiter\tuple{\langle}{\rangle}
\DeclarePairedDelimiter\range{[}{]}

\newcommand{\card}{\#}
\newcommand{\const}[1]{\mathsf{#1}}
\newcommand{\bl}{\_}
\newcommand{\defeq}{\mathrel{\ensurestackMath{\stackon[1pt]{=}{\scriptscriptstyle\triangle}}}}
\newcommand{\defiff}{\mathrel{\ensurestackMath{\stackon[1pt]{\iff}{\scriptscriptstyle\triangle}}}}
\newcommand{\nat}{\mathbb{N}}

\DeclarePairedDelimiter\univ{|}{|}

\DeclarePairedDelimiter\ljump{[}{]}

\NewDocumentCommand\word{O{1}}{%
    \ifcase#1
        undefined
    \or w
    \or v
    \or u
    \else undefined

    \fi
}

\NewDocumentCommand\la{O{1}}{%
    \ifcase#1
        undefined
    \or \mathcal{L}
    \or \mathcal{K}
    \else undefined

    \fi
}

\NewDocumentCommand\rel{O{1}}{%
    \ifcase#1
        undefined
    \or R
    \or S
    \else undefined

    \fi
}

\NewDocumentCommand\struc{O{1}}{%
    \ifcase#1
        undefined
    \or \mathfrak{A}
    \or \mathfrak{B}
    \else undefined

    \fi
}

\NewDocumentCommand\strucclass{O{1}}{%
    \ifcase#1
        undefined
    \or \mathcal{C}
    \or \mathcal{D}
    \else undefined

    \fi
}

\NewDocumentCommand\graph{O{1}}{%
    \ifcase#1
        undefined
    \or G
    \or H
    \or J
    \else undefined

    \fi
}
\NewDocumentCommand\dgraph{O{1}}{\dot{\graph[#1]}}
\NewDocumentCommand\ddgraph{O{1}}{\ddot{\graph[#1]}}

\newcommand{\src}{\const{1}}
\newcommand{\tgt}{\const{2}}

\newcommand{\homo}{\longrightarrow}

\NewDocumentCommand\aterm{O{1}}{%
    \ifcase#1
        undefined
    \or a
    \or b
    \or c
    \or d
    \else undefined

    \fi
}
\NewDocumentCommand\term{O{1}}{%
    \ifcase#1
        undefined
    \or t
    \or s
    \or {u}
    \else undefined

    \fi
}

\newcommand{\eps}{\varepsilon}

\newcommand{\union}{\mathbin{+}}
\newcommand{\intersection}{\mathbin{\cap}}
\newcommand{\id}{\const{1}}
\newcommand{\emp}{\const{0}}
\newcommand{\compo}{\mathbin{;}}
\DeclareMathOperator*{\bigcompo}{\scalerel*{\compo}{\sum}}

\NewDocumentCommand\tset{O{1}}{%
    \ifcase#1
        undefined
    \or T
    \or S
    \or {U}
    \else undefined

    \fi
}

\newcommand{\Lab}{L}
\NewDocumentCommand\lab{O{1}}{%
    \ifcase#1
        undefined
    \or x
    \or y
    \or z
    \else undefined

    \fi
}

\NewDocumentCommand\glang{O{1}}{%
    \ifcase#1
        undefined
    \or \mathcal{G}
    \or \mathcal{H}
    \else undefined

    \fi
}

\newcommand{\REL}{\mathsf{REL}}
\newcommand{\LANG}{\mathsf{LANG}}

\NewDocumentCommand\run{O{1}}{%
    \ifcase#1
        undefined
    \or \mathcal{R}
    \else undefined

    \fi
}

\NewDocumentCommand\tr{O{1}}{%
    \ifcase#1
        undefined
    \or \mathrm{Tr}
    \else undefined

    \fi
}
\NewDocumentCommand\ptr{O{1}}{%
    \ifcase#1
        undefined
    \or \mathrm{PTr}
    \else undefined

    \fi
}
\NewDocumentCommand\ttr{O{1}}{%
    \ifcase#1
        undefined
    \or \mathrm{TTr}
    \else undefined

    \fi
}

\NewDocumentCommand\trace{O{1}}{%
    \ifcase#1
        undefined
    \or \tau
    \or \sigma
    \or \rho
    \else undefined
    \fi
}

\NewDocumentCommand\lterm{O{1}}{\tilde{\term[#1]}}
\NewDocumentCommand\ltset{O{1}}{\tilde{\tset[#1]}}

\tikzset{
    every edge/.append style = {
            line width = .3pt,
        },
    plab/.style={line width = 0.1pt, fill=#1, inner sep = .025cm, anchor=center, font = \fontsize{6pt}{0}},
    plab/.default= white,
    elab/.style={draw, rectangle, line width = 0.1pt, fill=#1, inner sep = .035cm, anchor=center, font = \footnotesize},
    elab/.default= white,
    tlab/.style={line width = 0.1pt, fill=#1, inner sep = .025cm, anchor=center, font = \fontsize{6pt}{0}\selectfont},
    tlab/.default= white
}
\tikzset{
    png export/.style={
            external/system call/.add={}{; convert -density 300 -transparent white "\image.pdf" "\image.png"},
            /pgf/images/external info,
            /pgf/images/include external/.code={
                    \includegraphics[width=\pgfexternalwidth,height=\pgfexternalheight]{##1.png}
                },
        }
}

\tikzstyle{mynode} = [inner sep = 1.5pt, fill= gray!20, font=\footnotesize]
\tikzstyle{mysmallnode} = [inner sep = 1.pt, fill= gray!20]

\tikzstyle{vert} = [draw, circle, mynode]

\tikzset{earrow/.style={>={{[flex] Latex[length=.1cm, width=2.5pt]}}}}

\tikzset{homoarrow/.style={earrow, line width = .3pt, color = olive, opacity=0.8}}

\tikzstyle{elabel} = [inner sep = 1.pt, font = \scriptsize, opacity = 1]

\tikzstyle{dvert} = [vert, color = gray!20]
\tikzstyle{dearrow} = [earrow, color = gray!20]

\makeatletter
\def\widebreve{\mathpalette\wide@breve}
\def\wide@breve#1#2{\sbox\z@{$#1#2$}%
    \mathop{\vbox{\m@th\ialign{##\crcr
                \kern0.08em\brevefill#1{0.8\wd\z@}\crcr\noalign{\nointerlineskip}%
                $\hss#1#2\hss$\crcr}}}\limits}
\def\brevefill#1#2{$\m@th\sbox\tw@{$#1($}%
        \hss\resizebox{#2}{\wd\tw@}{\rotatebox[origin=c]{90}{\upshape(}}\hss$}
\makeatletter

\newcommand{\jump}[1]{\llbracket #1 \rrbracket}

\newcommand\xrsquigarrow[1]{%
    \mathrel{%
        \begin{tikzpicture}[%
                baseline={(current bounding box.south)}
            ]
            \node[%
                ,inner sep=.44ex
                ,align=center
            ] (tmp) {$\scriptstyle #1$};
            \path[%
                ,draw,<-
                ,decorate,decoration={%
                        ,zigzag
                        ,amplitude=0.7pt
                        ,segment length=1.2mm,pre length=3.5pt
                    }
            ]
            (tmp.south east) -- (tmp.south west);
        \end{tikzpicture}
    }
}
\newcommand{\longleadsto}{\xrsquigarrow{\hspace{1.em}}}

\newcommand{\ty}{\mathrm{\mathop{ty}}}
\newcommand{\vsig}{\Sigma}
\newcommand{\fork}{\mathtt{f}}
\newcommand{\join}{\mathtt{j}}

\DeclareMathOperator{\EPS}{\mathsf{E}}
\DeclareMathOperator{\D}{\mathsf{D}}

\DeclareMathOperator{\cl}{\mathrm{cl}}

\NewDocumentCommand\automaton{O{1}}{%
    \ifcase#1
        undefined
    \or \mathcal{A}
    \or \mathcal{B}
    \or \mathcal{C}
    \else undefined
    \fi
}

\DeclareMathOperator{\iw}{\mathrm{iw}}
\DeclareMathOperator{\pw}{\mathrm{pw}}
\DeclareMathOperator{\tw}{\mathrm{tw}}

\newcommand{\series}{\diamond}

\newcommand{\LT}{\tilde{\mathbf{T}}}

\newcommand{\lop}{\circlearrowleft}

\newcommand{\sub}{\mathsf{sub}}
\newcommand{\subl}{\mathsf{subl}}
\newcommand{\subr}{\mathsf{subr}}

\newcommand{\SPR}{\mathsf{SPR}}
\newcommand{\subSPR}{\mathsf{subSPR}}
\newcommand{\subSPRl}{\mathsf{subSPRl}}
\newcommand{\subSPRr}{\mathsf{subSPRr}}

\newcommand{\lSPR}{\mathsf{lSPR}}
\newcommand{\lsubSPR}{\mathsf{lsubSPR}}
\newcommand{\lsubSPRl}{\mathsf{lsubSPRl}}
\newcommand{\lsubSPRr}{\mathsf{lsubSPRr}}

\DeclarePairedDelimiter\quo{[}{]}

\newcommand{\STR}{\REL}

\newcommand{\isoc}{\mathsf{I}} %

\knowledge{notion}
| word
| words

\knowledge{notion}
| empty word

\knowledge{notion}
| letter
| letters

\knowledge{notion}
| concatenation

\knowledge{notion}
| Kleene star

\knowledge{notion}
| Kleene plus

\knowledge{notion}
| word language
| word languages
| language
| languages

\knowledge{notion}
| 1NFA
| 1NFAs
| nondeterministic word automaton
| nondeterministic word automata

\knowledge{notion}
| 2AFA
| 2AFAs
| two-way alternating word automaton
| two-way alternating word automata
| two-way alternating finite word automaton
| two-way alternating finite word automata

\knowledge{notion}
| inclusion problem

\knowledge{ignore}
| tree
| trees

\knowledge{ignore}
| forest
| forests

\knowledge{notion}
| binary relations

\knowledge{notion}
| converse

\knowledge{notion}
| composition

\knowledge{notion}
| $n$-th iteration

\knowledge{notion, scope = lang}
| $n$-th iteration

\knowledge{notion}
| transitive closure

\knowledge{notion}
| reflexive transitive closure

\knowledge{notion}
| graph
| graphs

\knowledge{notion}
| source
| sources

\knowledge{notion}
| target
| targets

\knowledge{notion}
| empty

\knowledge{notion}
| null

\knowledge{notion}
| quotient graph
| quotient graphs

\knowledge{notion}
| graph homomorphism
| graph homomorphisms

\knowledge{notion}
| graph isomorphism
| graph isomorphisms

\knowledge{notion}
| path decomposition
| path decompositions

\knowledge{pathwidth-width}[width | widths]{notion, scope= pathwidth}
\knowledge{pathwidth-bag}[bag | bags]{notion, scope= pathwidth}

\knowledge{notion}
| pathwidth
| pathwidths

\knowledge{notion}
| tree decomposition
| tree decompositions

\knowledge{treewidth-width}[width | widths]{notion, scope= treewidth}
\knowledge{treewidth-bag}[bag | bags]{notion, scope= treewidth}

\knowledge{notion}
| treewidth 
| treewidths

\knowledge{notion}
| atomic graph

\knowledge{notion}
| atomic path decomposition

\knowledge{notion}
| variable
| variables

\knowledge{notion}
| PCoR*
| PCoR* term
| PCoR* terms

\knowledge{notion}
| atomic test
| atomic tests

\knowledge{notion}
| test
| tests

\knowledge{notion}
| nominal
| nominals

\knowledge{notion}
| PCoR* term with tests
| PCoR* terms with tests
| PCoR* with tests

\knowledge{notion}
| PCoR* term with tests and nominals
| PCoR* terms with tests and nominals
| PCoR* with tests and nominals

\knowledge{notion}
| Kleene algebra term
| Kleene algebra terms
| KA term
| KA terms

\knowledge{notion}
| equation
| equations

\knowledge{notion}
| inequation
| inequations

\knowledge{notion}
| size
| sizes

\knowledge{notion}
| intersection width
| intersection widths

\knowledge{notion}
| semantics

\knowledge{notion}
| structure
| structures
| (relational) structure

\knowledge{notion}
| vertex
| vertices

\knowledge{notion}
| equational theory
| equational theories

\knowledge{notion}
| converse normal form

\knowledge{notion}
| graph language
| graph languages

\knowledge{notion}
| KL term
| KL terms
| Kleene lattice term
| Kleene lattice terms
| Kleene lattice
| Kleene lattices
| KL

\knowledge{notion}
| KL term with top
| KL terms with top
| Kleene lattice term with top
| Kleene lattice terms with top

\knowledge{notion}
| identity-free Kleene lattices
\knowledge{notion}
| lKL term
| lKL terms
| labeled Kleene lattice term
| labeled Kleene lattice terms
| labeled KL term
| labeled KL terms
| lKL

\knowledge{notion}
| graph language bounded model property

\knowledge{notion}
| linearly bounded pathwidth model property

\knowledge{notion}
| treewidth at most $2$ model property

\knowledge{notion}
| closure

\knowledge{notion}
| locally split

\knowledge{notion}
| one-way deterministic finite word automaton
| one-way deterministic finite word automata
| 1DFA
| 1DFAs

\knowledge{notion}
| derivative
| derivatives

\knowledge{notion}
| run
| runs
| run with $\tuple{n,m}$-interface
| runs with $\tuple{n,m}$-interface

\knowledge{notion}
| run language
| run languages

\knowledge{notion}
| $\struc$-run
| $\struc$-runs

\knowledge{notion}
| $\struc$-run language
| $\struc$-run languages

\knowledge{notion}
| label
| labels

\knowledge{notion}
| series composition
| series compositions

\knowledge{notion}
| parallel composition
| parallel compositions

\knowledge{notion}
| atomic run
| atomic runs

\knowledge{notion}
| atomic $\struc$-run
| atomic $\struc$-runs

\knowledge{notion}
| left quotient
| left quotients

\knowledge{notion}
| emptiness property

\knowledge{notion}
| truth value

\knowledge{notion}
| $\struc$-run isomorphism

\knowledge{series composition-srun}[series composition | series compositions]{notion, scope= srun}

\knowledge{parallel composition-srun}[parallel composition | parallel compositions]{notion, scope= srun}

\knowledge{left-quotient-srun}[left quotient | left quotients]{notion, scope= srun}

\knowledge{emptiness property-srun}[emptiness property]{notion, scope= srun}

\knowledge{semantics-lterm}[semantics]{notion, scope= lterm}

\knowledge{notion}
| sub-run
| sub-runs

\knowledge{notion}
| series-parallel run
| series-parallel runs
| SP run
| SP runs
| SP

\knowledge{notion}
| sub-series-parallel run
| sub-series-parallel runs
| sub-SP run
| sub-SP runs
| sub-SP

\knowledge{notion}
| sub-SP-l run
| sub-SP-l runs
| sub-SP-l

\knowledge{notion}
| sub-SP-r run
| sub-SP-r runs
| sub-SP-r

\knowledge{notion}
| local

\knowledge{notion}
| $i$-split

\knowledge{notion}
| connected

\knowledge{notion}
| disjoint union

\knowledge{2AFA size}[size | sizes]{notion, scope= 2AFA}

\knowledge{notion}
| isomorphic closure

\knowledge{notion}
| Petri automata

\knowledge{notion}
| branching automata
\theoremstyle{plain} %

\def\eg{{\em e.g.}}
\def\cf{{\em cf.}}

\begin{document}

\title[Derivatives on Graphs for PCoR*]{Derivatives on Graphs for the Positive Calculus of Relations with Transitive Closure}
\titlecomment{This paper is a revised and extended version of \cite{nakamuraPartialDerivativesGraphs2017}.}
\author[Y.~Nakamura]{Yoshiki Nakamura\lmcsorcid{0000-0003-4106-0408}}

\address{Institute of Science Tokyo, Japan}	%
\email{nakamura.yoshiki.ny@gmail.com}  %

\begin{abstract}
We prove that the equational theory of the positive calculus of relations with transitive closure (PCoR*) is EXPSPACE-complete.
Here, PCoR* terms consist of the following standard operators on binary relations: identity, empty, universality, union, intersection, composition, converse, and reflexive transitive closure (so, PCoR* terms subsume Kleene algebra and allegory terms as fragments).
Additionally, we show that the equational theory of PCoR* extended with tests and nominals (in hybrid logic) is still EXPSPACE-complete;
moreover, it is PSPACE-complete for its intersection-free fragment.

To this end, we design derivatives on \emph{graphs} by extending derivatives on words for regular expressions.
The derivatives give a finite automata construction on path decompositions, like those on words.
Because the equational theory has a linearly bounded pathwidth model property, we can decide the equational theory of PCoR* using these automata.
 \end{abstract}

\maketitle

\section{Introduction} \label{section: introduction}
We consider \emph{the positive calculus of relations with transitive closure} (PCoR*) \cite{pousPositiveCalculusRelations2018}:
the algebraic system on binary relations,
with the operators $\set{\id, \emp, \top, \compo, \union, \intersection, \bl^{\smile}, \bl^{*}}$ of identity ($\id$), empty ($\emp$), universality ($\top$), composition ($\compo$), union ($\union$), intersection ($\intersection$), converse ($\bl^{\smile}$), and reflexive transitive closure ($\bl^{*}$);
so, PCoR* consist of the operators of Kleene algebra $\set{\id, \emp, \compo, \union, \bl^{*}}$ \cite{kleeneRepresentationEventsNerve1956,conwayRegularAlgebraFinite1971,kozenCompletenessTheoremKleene1991}, those of allegory $\set{\id, \compo, \intersection, \bl^{\smile}}$ \cite{freydCategoriesAllegories1990},
and $\top$.
PCoR* terms without $\top$ are sometimes called \emph{Kleene allegory} terms \cite{brunetPetriAutomataKleene2015,brunetPetriAutomata2017,nakamuraPartialDerivativesGraphs2017}.

If we added the complement operator ($\bl^{-}$) to PCoR*, then we would obtain the calculus of relations \cite{tarskiCalculusRelations1941,tarskiFormalizationSetTheory1987} with transitive closure \cite{ngRelationAlgebrasTransitive1984}.
However, the \kl{equational theory} of the calculus of relations is well-known undecidable \cite{tarskiCalculusRelations1941,tarskiFormalizationSetTheory1987}.
The undecidability result holds even for the terms with the three operators $\set{\compo, \union, \bl^{-}}$ \cite{hirschDecidabilityEquationalTheories2018} and moreover even with one variable \cite{nakamuraUndecidabilityFO3Calculus2019}.
Also, in PCoR* with $\bl^{-}$, even if the complement only applies to atomic terms (variables or constants), the \kl{equational theory} is still undecidable \cite{nakamuraExistentialCalculiRelations2023}.
In this paper, by excluding complement, we focus on \emph{positive} fragments.
For Kleene algebras (with respect to binary relations), the \kl{equational theory} coincides with the language equivalence problem of regular expressions (see, \eg, \cite[Theorem 4]{pousPositiveCalculusRelations2018}), and thus the \kl{equational theory} of Kleene algebras is decidable and PSPACE-complete, by the results on regular expressions \cite{meyerEquivalenceProblemRegular1972}.
For representable allegories (with respect to binary relations), the \kl{equational theory} coincides with the graph homomorphic equivalence problem via an encoding of terms as graphs \cite{freydCategoriesAllegories1990,andrekaEquationalTheoryUnionfree1995,chandraOptimalImplementationConjunctive1977},
and thus the \kl{equational theory} of representable allegories is also decidable.
For the \kl{equational theory} of PCoR*, Brunet and Pous extended the \kl{graph homomorphism} characterizations above using graph languages \cite{brunetPetriAutomataKleene2015,brunetPetriAutomata2017}.
By introducing \kl{Petri automata}, which is an automata model for expressing \kl{graph languages},
they have shown that the equational theory of \intro*\kl{identity-free Kleene lattices} $\set{\emp, \compo, \union, \intersection, \bl^{+}}$ with respect to binary relations \cite{andrekaEquationalTheoryKleene2011} is decidable and EXPSPACE-complete \cite{brunetPetriAutomataKleene2015,brunetPetriAutomata2017} where $(\bl^{+})$ denotes the transitive closure operator.
However, the decidability and complexity for Kleene allegories (and PCoR*) were left open at LICS 2015 \cite{brunetPetriAutomataKleene2015}.

Our main contribution is to prove that the \kl{equational theory} of PCoR* is still decidable and EXPSPACE-complete.
We first note that the \kl{equational theory} has the \kl{treewidth at most $2$ model property} (\Cref{proposition: tw 2 property}) and the \kl{linearly bounded pathwidth model property} (\Cref{proposition: bounded pw property}).
Thus, the decidability can be derived from the decidability of MSO over bounded treewidth structures \cite{courcelleMonadicSecondorderLogic1988,courcelleMonadicSecondorderLogic1990,courcelleGraphStructureMonadic2012}.
However, the naive algorithm obtained from this fact has a non-elementary complexity (see \Cref{section: conclusion}, for more comparisons to other systems related to PCoR*).
The EXPSPACE-hardness is shown by giving a reduction from the universality problem of regular expressions with intersection, which is EXPSPACE-complete \cite[Theorem 2]{furerComplexityInequivalenceProblem1980}.
(This was independently shown in \cite{nakamuraPartialDerivativesGraphs2017, brunetPetriAutomata2017}.)

To obtain the EXPSPACE upper bound of the \kl{equational theory} of PCoR*,
we use the idea of derivatives on words for regular expressions, \eg, Brzozowski's derivatives \cite{brzozowskiDerivativesRegularExpressions1964} and Antimirov's (partial) derivatives \cite{antimirovPartialDerivativesRegular1996,bastosStateComplexityPartial2016},
that are tools to obtain automata from regular expressions syntactically (we refer to the book \cite{sakarovitchElementsAutomataTheory2009}, for more details of \kl{derivatives} on words).
In this paper, we extend the derivatives from words to graphs.
To this end, we consider \kl{Kleene lattices} $\set{\id, \emp, \compo, \union, \intersection, \bl^{*}}$, which is a core fragment of PCoR*,
and we extend the syntax of terms with \emph{\kl{labels}} for pointing multiple \kl{vertices} on \kl{graphs} (\Cref{section: lKL}).
Thanks to this extension, similar to those on words, we can give \kl{derivatives} on \kl{graphs}/\kl{structures} (\Cref{section: derivative}).
\kl[derivatives]{Derivatives} of \kl{labeled Kleene lattice terms} can simulate \kl{left quotients} of \kl[run languages]{languages of runs (having a special form of DAG)} (\Cref{section: runs}),
as with that derivatives of regular expressions can simulate left quotients of word languages.
Moreover, we show that our \kl{derivatives} are decomposable on \kl{path decompositions} (\Cref{theorem: decomposition}).
Namely, from derivatives on (small) pre-glued graphs, we can compute the derivative on the (large) glued graph (in this sense, our approach can be viewed as a variant of the Feferman--Vaught decomposition theorem \cite{makowskyAlgorithmicUsesFeferman2004,courcelleGraphStructureMonadic2012}).
Thanks to this decomposition with the \kl{linearly bounded pathwidth model property} (\Cref{proposition: bounded pw property}),
we can obtain an automata construction for \kl{Kleene lattice terms} (\Cref{section: automata construction}).
Using this construction, we have that the \kl{equational theory} of \kl{Kleene lattices} with respect to relations is EXPSPACE-complete (\Cref{corollary: KL EXPSPACE-complete}).
Moreover, we show that on our automata construction, we can give encodings of the following operators: universality ($\top$), converse ($\bl^{\smile}$), \kl{tests} in Kleene algebra with tests (KAT), and \kl{nominals} in hybrid logic.
Thus, the \kl{equational theory} of \kl{PCoR*} is EXPSPACE-complete (\Cref{corollary: PCoR* EXPSPACE-complete}), and moreover,
the \kl{equational theory} of \kl{PCoR* with tests and nominals} is EXPSPACE-complete (\Cref{corollary: PCoR* with tests EXPSPACE-complete}).
Additionally, if the \kl{intersection width} \cite{gollerPDLIntersectionConverse2009} is fixed (particularly, if $\cap$ does not occur), 
the \kl{equational theory} of \kl{PCoR* with tests and nominals} is PSPACE-complete (\Cref{corollary: PCoR* with tests v iw fixed PSPACE-complete}).

\subsection*{Differences with the version at LICS 2017}
This is a revised and extended version of the paper \cite{nakamuraPartialDerivativesGraphs2017}, presented at the 32nd Annual ACM/IEEE Symposium on Logic in Computer Science (LICS 2017).
The differences from the conference version are as follows.
\begin{itemize}
    \item 
    To express bounded \kl{pathwidth} \kl{structures},
    we use a graph gluing operator ($\odot$) \cite{bojanczykDefinabilityEqualsRecognizability2016}, instead of ``Sequential Graph Constructing Procedures'' (SGCPs) \cite{nakamuraPartialDerivativesGraphs2017}.
    Roughly speaking,
    SGCPs enumerate ``normalized'' (but complete to generate all \kl{structures} of \kl{pathwidth} at most $k$) \kl{path decompositions} of width at most $k$,
    but combinations of the gluing operator enumerate \emph{all} \kl{path decompositions} of width at most $k$.
    Due to this,
    the alphabet size, here, the number of \kl{structures} of size at most $k+1$, is large, naively, but we can reduce it (see \Cref{section: alphabet size}).
    Thanks to this separation of concerns, our formalizations are significantly simplified without changing the essence.

    \item Thanks to the above, we can easily encode some additional operators, \eg, top ($\top$), converse ($\bl^{\smile}$), tests in Kleene algebra with tests, and nominals in hybrid logic (\Cref{section: encoding PCoR*,section: encoding more PCoR*}).
    In the body of this paper, relying on these encodings, instead of \kl{PCoR* terms},
    we mainly consider \kl{Kleene lattice terms}---\kl{PCoR* terms} with neither $\top$ nor $\bl^{\smile}$.
    This separation also slightly simplifies our formalizations.

    \item We use \kl{runs} (defined in \Cref{section: runs}) instead of ``simple graphs'' in \cite{nakamuraPartialDerivativesGraphs2017}.
    Because each \kl{vertex} of \kl{runs} precisely has one fun-in and one fun-out,
    we can more easily define \kl{left quotients}, \cf\ \cite{nakamuraPartialDerivativesGraphs2017}.
    Our \kl{runs} are essentially those of \kl{branching automata} by Lodaya and Weil \cite{lodayaSeriesparallelPosetsAlgebra1998,lodayaSeriesParallelLanguages2000,lodayaRationalityAlgebrasSeries2001} without minor changes.

    \item We point out an error of \cite[Section IV]{nakamuraPartialDerivativesGraphs2017} in \Cref{section: incomplete}.
    We give a counterexample that
    the decomposition rules corresponding to the argument in \cite[Lemma IV.10]{nakamuraPartialDerivativesGraphs2017} are not complete (\cf\ \Cref{lemma: decomposition lsubSPR}), due to one-way reading.
    To avoid this problem, we use two-way alternating finite automata (\kl{2AFAs}), not one-way.
\end{itemize}\noindent 

\subsection*{Outline}
In \Cref{section: preliminaries},  we give basic definitions, including \kl{PCoR* terms} and \kl{Kleene lattice terms} (\kl{KL terms}).
In \Cref{section: lKL}, we extend \kl{KL terms} with \kl{labels}.
In \Cref{section: derivative}, we define \kl{derivatives} on \kl{structures} (\kl{graphs} without interfaces) for \kl{labeled KL terms}.
In \Cref{section: automata construction}, we give an automata construction using the \kl{derivatives}.
We then have that the \kl{equational theory} of \kl{KL terms} is EXPSPACE-complete. 
In \Cref{section: encoding variants}, we give encodings of additional operators in our automata construction.
Consequently, the \kl{equational theory} of \kl{PCoR* terms with tests and nominals} is also EXPSPACE-complete. 
In \Cref{section: theorem: decomposition}, we prove the decomposition theorem for the \kl{derivatives}, which is the key to our automata construction in \Cref{section: automata construction}.
In \Cref{section: conclusion}, we conclude this paper with related and future work. 

\section{Preliminaries}\label{section: preliminaries}
We write $\nat$ for the set of non-negative integers.
For $l, r \in \nat$, we write $\range{l, r}$ for the set $\{i \in \nat \mid l \le i \le r\}$.
For $n \in \nat$, we abbreviate $[1, n]$ to $[n]$.
For a set $X$, we write
$\card X$ for the cardinality of $X$ and $\wp(X)$ for the power set of $X$.
We use $\uplus$ to denote that the union $\cup$ is disjoint.

\subsection*{Word languages}
We write $\word[1] \word[2]$ for the \kl[word]{concatenation} of \kl{words} $\word[1]$ and $\word[2]$.
We write $\eps$ for the \intro*\kl{empty word}.
For a set $X$ of \intro*\kl{letters}, we write $X^{*}$ for the set of \intro*\kl{words} over $X$.
A \intro*\kl{language} over $X$ is a subset of $X^{*}$.
We use $\word[1], \word[2]$ to denote \kl{words} and use $\la[1], \la[2]$ to denote \kl{languages}.
For \kl{languages} $\la[1], \la[2] \subseteq X^{*}$, the \intro*\kl{concatenation} $\la[1] \compo \la[2]$,
the \intro*\kl(lang){$n$-th iteration} $\la[1]^n$ (where $n \in \nat$),
the \intro*\kl{Kleene plus} $\la[1]^{+}$, and
the \intro*\kl{Kleene star} $\la[1]^{*}$ are defined by:
\begin{align*}
    \la[1] \compo \la[2] & \defeq \set{\word[1] \word[2] \mid \word[1] \in \la[1] \land \word[2] \in \la[2]},\hspace{-.1em}                        &
    \la[1]^{n} & \defeq \begin{cases}
        \la[1] \compo \la[1]^{n-1} & (n \ge 1)\\
        \set{\eps} & (n = 0)
    \end{cases},\hspace{-.1em} & 
    \la[1]^{+}  & \defeq \bigcup_{n \ge 1} \la[1]^{n},\hspace{-.1em}&
    \la[1]^{*}  & \defeq \bigcup_{n \ge 0} \la[1]^{n}.
\end{align*}

\subsection*{Binary relations}
We write $\triangle_{A}$ for the identity relation on a set $A$: $\triangle_{A} \defeq \set{\tuple{x, x} \mid x \in A}$.
For \kl{binary relations} $\rel[1], \rel[2]$ on a set $B$, the \intro*\kl{composition} $\rel[1] \compo \rel[2]$, the \intro*\kl{converse} $\rel[1]^{\smile}$, the \intro*\kl{$n$-th iteration} $\rel[1]^{n}$ (where $n \in \nat$), the \intro*\kl{transitive closure} $\rel[1]^{+}$, and the \intro*\kl{reflexive transitive closure} $\rel[1]^*$ are defined by:
\begin{align*}
    \rel[1] \compo \rel[2] & \defeq \set{\tuple{x, z} \mid \exists y,\ \tuple{x, y} \in \rel[1] \land \tuple{y, z} \in \rel[2]}, & \rel[1]^{\smile} & \defeq \set{\tuple{x, y} \mid \tuple{y, x} \in \rel[1]}, \\
    \rel[1]^n             & \defeq \begin{cases}
                                       \rel[1] \compo \rel[1]^{n-1} & (n \ge 1) \\
                                       \triangle_{B}               & (n = 0)
                                   \end{cases}, & \rel[1]^{+}        & \defeq \bigcup_{n \ge 1} \rel[1]^n,  \qquad \rel[1]^*  \defeq \bigcup_{n \ge 0} \rel[1]^n.
\end{align*}

\subsection*{Graphs with bi-interface}
We consider \kl{graphs} with bi-interface, inspired by \cite[bi-interface graphs]{bojanczykDefinabilityEqualsRecognizability2016}\cite[symmetric monoidal structures]{bonchiRewritingModuloSymmetric2016}.
Interfaces are used to define the \kl{series composition} (\Cref{definition: series composition and parallel composition}).
Let $A$ be a set with a map $\ty \colon A \to \nat^2$; we also let $\ty_{\src}$ and $\ty_{\tgt}$ be such that $\ty(a) = \tuple{\ty_{\src}(a), \ty_{\tgt}(a)}$ for each $a \in A$.
For $n, m \in \nat$, a \intro*\kl{graph} $\graph$ over $A$ (with $\tuple{n, m}$-interface) is a tuple
$\tuple{\univ{\graph}, \set{a^{\graph}}_{a \in A},
\src_{1}^{\graph}, \dots, \src_{n}^{\graph},
\tgt_{1}^{\graph}, \dots, \tgt_{m}^{\graph}}$, where
\begin{itemize}
    \item $\univ{\graph}$ is a (possibly empty) set;
    \item $a^{\graph} \subseteq \univ{\graph}^{\ty_{1}(a) + \ty_{2}(a)}$ is a $(\ty_{1}(a) + \ty_{2}(a))$-ary relation for each $a \in A$;
    \item $\src_{l}^{\graph} \in \univ{\graph}$ is the $l$-th \intro*\kl{source} \kl{vertex} for each $l \in \range{n}$;
    \item $\tgt_{r}^{\graph} \in \univ{\graph}$ is the $r$-th \intro*\kl{target} \kl{vertex} for each $r \in \range{m}$.
\end{itemize}\noindent 
Particularly when $n = m = 0$, we say that $\graph$ is a \kl{graph} without interfaces or $\graph$ has no interfaces.
Note that, if $\univ{\graph} = \emptyset$, then $n = m = 0$.
We let $\ty(\graph) \defeq \tuple{\ty_{1}(\graph), \ty_{2}(\graph)} \defeq \tuple{n, m}$.
We often abbreviate $\src_{1}^{\graph}$ to $\src^{\graph}$ and $\tgt_{1}^{\graph}$ to $\tgt^{\graph}$, respectively.

Let $\graph[1]$ and $\graph[2]$ be \kl{graphs} over a set $A$ with $\tuple{n, m}$-interface.
We say that a map $f \colon \univ{\graph[1]} \to \univ{\graph[2]}$ is a \intro*\kl{graph homomorphism} from $\graph[1]$ to $\graph[2]$, written $f \colon \graph[1] \homo \graph[2]$, if
(1) for each $a \in A$ and $x_1, \dots, x_{p} \in \univ{\graph[1]}$,
if $\tuple{x_1, \dots, x_p} \in a^{\graph[1]}$, then $\tuple{f(x_1), \dots, f(x_p)} \in a^{\graph[2]}$,
(2) for each $l \in \range{n}$, $f(\src_{l}^{\graph[1]}) = \src_{l}^{\graph[2]}$, and
(3) for each $r \in \range{m}$, $f(\tgt_{r}^{\graph[1]}) = \tgt_{r}^{\graph[2]}$.
Particularly, we say that $f$ is a \intro*\kl{graph isomorphism} from $\graph[1]$ to $\graph[2]$, written $f \colon \graph[1] \cong \graph[2]$, if
$f$ is a bijective \kl{graph homomorphism} from $\graph[1]$ to $\graph[2]$, and 
for each $a \in A$ and $x_1, \dots, x_{p} \in \univ{\graph[1]}$, 
$\tuple{x_1, \dots, x_{p}} \in a^{\graph[1]}$ iff $\tuple{f(x_1), \dots, f(x_{p})} \in a^{\graph[2]}$.
We write $\graph[1] \homo \graph[2]$ (resp.\ $\graph[1] \cong \graph[2]$) if there is some $f \colon \graph[1] \homo \graph[2]$
(resp.\ $f \colon \graph[1] \cong \graph[2]$).
In the sequel, we identify two \kl{graphs} if there is a \kl{graph isomorphism} between them (except for \kl{structures}; see \Cref{section: lKL semantics});
we write $\graph[1] = \graph[2]$ when $\graph[1] \cong \graph[2]$.

For instance, when $A$ is the set $\set{a, \fork, \join}$ with $\ty(a) = \tuple{1, 1}$, $\ty(\fork) = \tuple{1, 2}$, and $\ty(\join) = \tuple{2, 1}$
and the \kl{graph} $\graph$ is given by $\univ{\graph} = \set{1, 2, 3, 4}$, $a^{\graph} = \set{\tuple{1, 2}}$, $\fork^{\graph} = \set{\tuple{1, 2, 3}}$, $\join^{\graph} = \set{\tuple{ 2, 3, 4}}$, $\ty(\graph) = \tuple{1, 1}$,
$\src^{\graph}_{1} = 1$, and $\tgt^{\graph}_{1} = 4$, 
we depict $\graph$ as the figure on the left-hand side.
Here, we will use the symbols $\fork$ and $\join$ as special constants (\Cref{section: runs}).
For notational simplicity, we often omit \kl{vertex} labels and some ``$1$'' labels as the right-hand side.
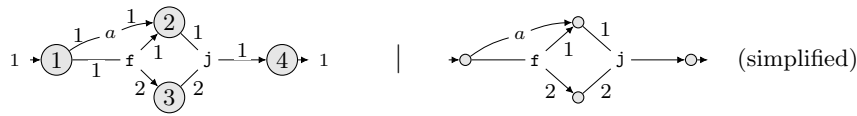
\begin{center}
    \begin{tabular}{ccc}
        \begin{tikzpicture}[baseline = -3.5ex]
            \graph[grow right = 1.5cm, branch down = .5cm, nodes={}]{
            {/, 1/{$1$}[mynode,draw,circle]} -!-
            {2/{$2$}[mynode,draw,circle],/, 3/{$3$}[mynode,draw,circle]} -!-
            {/, 4/{$4$}[mynode,draw,circle]}
            };
            \node[left = 4pt of 1](s){\tiny $1$};
            \node[right = 4pt of 4](t){\tiny $1$};
            \path (1) edge[pos =.5, opacity = 0, bend left = 10] node[opacity = 1, inner sep = 1pt, outer sep = 1pt](a){\tiny  $a$} (2);
            \path (1) edge[pos =.3, opacity = 0] node[opacity = 1, inner sep = 1pt, outer sep = 1pt](f){\tiny  $\fork$} (4);
            \path (1) edge[pos =.7, opacity = 0] node[opacity = 1, inner sep = 1pt, outer sep = 1pt](j){\tiny  $\join$} (4);
            \graph[use existing nodes, edges={color=black, pos = .5, earrow}, edge quotes={fill=white, inner sep=1pt,font= \scriptsize}]{
            s -> 1; 4 -> t;
            1 --["$1$"{auto, above= 1.pt, inner sep = 0}, pos = .3, bend left = 5] a; a->["$1$"{auto, above = 1.pt, inner sep = 0}, pos = .3, bend left = 5] 2;
            1 --["$1$"{auto,below}] f; f->["$1$"{auto, below right = .3pt}] 2; f->["$2$"{auto, below left = .3pt}] 3;
            2 --["$1$"{auto}] j; 3--["$2$"{auto, below right = .3pt}] j; j->["$1$"{auto}] 4;
            };
        \end{tikzpicture}
        & \quad$\mid$\quad &
        \begin{tikzpicture}[baseline = -3.5ex]
            \graph[grow right = 1.5cm, branch down = .5cm, nodes={}]{
            {/, 1/{}[mynode,draw,circle]} -!-
            {2/{}[mynode,draw,circle],/, 3/{}[mynode,draw,circle]} -!-
            {/, 4/{}[mynode,draw,circle]}
            };
            \node[left = 4pt of 1](s){};
            \node[right = 4pt of 4](t){};
            \path (1) edge[pos =.5, opacity = 0, bend left = 10] node[opacity = 1, inner sep = 1pt, outer sep = 1pt](a){\tiny  $a$} (2);
            \path (1) edge[pos =.3, opacity = 0] node[opacity = 1, inner sep = 1pt, outer sep = 1pt](f){\tiny  $\fork$} (4);
            \path (1) edge[pos =.7, opacity = 0] node[opacity = 1, inner sep = 1pt, outer sep = 1pt](j){\tiny  $\join$} (4);
            \graph[use existing nodes, edges={color=black, pos = .5, earrow}, edge quotes={fill=white, inner sep=1pt,font= \scriptsize}]{
            s -> 1; 4 -> t;
            1 --[bend left = 5] a; a->[bend left = 5] 2;
            1 -- f; f->["$1$"{auto, below right = .3pt}] 2; f->["$2$"{auto, below left = .3pt}] 3;
            2 --["$1$"{auto}] j; 3--["$2$"{auto, below right = .3pt}] j; j-> 4;
            };
        \end{tikzpicture} {\footnotesize (simplified)}
    \end{tabular}
\end{center}
We may depict \kl{graphs} $\graph$ with $\ty(\graph) = \tuple{n, m}$ as one edge:
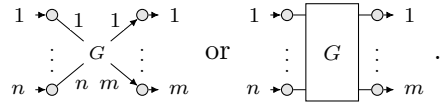
\begin{tikzpicture}[baseline = -3.5ex]
    \graph[grow right = 1.2cm, branch down = 3.ex, nodes={}]{
    {L1/{}[mynode,draw,circle], L12/{\tiny $\vdots$}, L2/{}[mynode,draw,circle]} -!-
    {R1/{}[mynode,draw,circle], R12/{\tiny $\vdots$}, R2/{}[mynode,draw,circle]}
    };
    \path (L12) edge[pos =.5, opacity = 0] node[opacity = 1, inner sep = 1pt, outer sep = 1pt](G){\tiny  $\graph$} (R12);         
    \node[left = 4pt of L1](s1){\scriptsize $1$}; \path (s1) edge[earrow, ->] (L1);
    \node[left = 4pt of L2](s2){\scriptsize $n$}; \path (s2) edge[earrow, ->] (L2);
    \node[right = 4pt of R1](t1){\scriptsize $1$}; \path (R1) edge[earrow, ->] (t1);
    \node[right = 4pt of R2](t2){\scriptsize $m$}; \path (R2) edge[earrow, ->] (t2);
    \graph[use existing nodes, edges={color=black, pos = .5, earrow}, edge quotes={fill=white, inner sep=1pt,font= \scriptsize}]{
        L1 --["$1$"{auto}] G;
        L2 --["$n$"{auto, below right = 1pt}] G;
        G ->["$1$"{auto}] R1;
        G ->["$m$"{auto, below left = 1pt}] R2;
        };
\end{tikzpicture} or \begin{tikzpicture}[baseline = -3.5ex]
    \graph[grow right = 1.2cm, branch down = 3.ex, nodes={}]{
    {L1/{}[mynode,draw,circle], L12/{\tiny $\vdots$}, L2/{}[mynode,draw,circle]} -!-
    {R1/{}[mynode,draw,circle], R12/{\tiny $\vdots$}, R2/{}[mynode,draw,circle]} -!-
    };
    \node[left = 4pt of L1](s1){\scriptsize $1$}; \path (s1) edge[earrow, ->] (L1);
    \node[left = 4pt of L2](s2){\scriptsize $n$}; \path (s2) edge[earrow, ->] (L2);
    \node[right = 4pt of R1](t1){\scriptsize $1$}; \path (R1) edge[earrow, ->] (t1);
    \node[right = 4pt of R2](t2){\scriptsize $m$}; \path (R2) edge[earrow, ->] (t2);
    \path (L1) edge[earrow, =] (R1);
    \path (L2) edge[earrow, =] (R2);
    \draw[fill = white] ($(L1)+(.25,.15)$) -- ($(L2)+(.25,-.15)$) -- ($(R2)+(-.25,-.15)$) -- ($(R1)+(-.25,.15)$) -- cycle;
    \path (L12) edge[pos =.5, opacity = 0] node[opacity = 1]{\scriptsize  $\graph[1]$} (R12);
\end{tikzpicture}.

\vskip\baselineskip 
Additionally, for an equivalence relation $E$ on $\univ{\graph[1]}$,
the \intro*\kl{quotient graph} of $\graph[1]$ with respect to $E$ is defined as the \kl{graph} $\graph[1]/E \defeq \tuple{\univ{\graph[1]}/E, \set{\tuple{X, Y} \mid \exists x \in X, y \in Y, \tuple{x, y} \in a^{\graph[1]}}_{a \in A}, \quo{\src_1^{\graph}}_{E}, \dots, \quo{\src_n^{\graph}}_{E}, \quo{\tgt_1^{\graph}}_{E}, \dots, \quo{\tgt_m^{\graph}}_{E}}$ where 
$\univ{\graph[1]}/E$ denotes the set of equivalence classes of $\univ{\graph[1]}$ by $E$ and $\quo{x}_{E}$ denotes the equivalence class of $x$.

\subsubsection*{Series and parallel composition}
For \kl{graphs}, we use the \kl{series composition} and the \kl{parallel composition}, defined as follows.
\begin{defi}\label{definition: series composition and parallel composition}
    For \kl{graphs}, $\graph[1] =
        \left(\hspace{-.3em}\begin{tikzpicture}[baseline = -3.5ex]
            \graph[grow right = 1.2cm, branch down = 3.ex, nodes={}]{
            {L1/{}[mynode,draw,circle], L12/{\tiny $\vdots$}, L2/{}[mynode,draw,circle]} -!-
            {R1/{}[mynode,draw,circle], R12/{\tiny $\vdots$}, R2/{}[mynode,draw,circle]} -!-
            };
            \node[left = 4pt of L1](s1){\scriptsize $1$}; \path (s1) edge[earrow, ->] (L1);
            \node[left = 4pt of L2](s2){\scriptsize $n$}; \path (s2) edge[earrow, ->] (L2);
            \node[right = 4pt of R1](t1){\scriptsize $1$}; \path (R1) edge[earrow, ->] (t1);
            \node[right = 4pt of R2](t2){\scriptsize $m$}; \path (R2) edge[earrow, ->] (t2);
            \path (L1) edge[earrow, =] (R1);
            \path (L2) edge[earrow, =] (R2);
            \draw[fill = white] ($(L1)+(.25,.15)$) -- ($(L2)+(.25,-.15)$) -- ($(R2)+(-.25,-.15)$) -- ($(R1)+(-.25,.15)$) -- cycle;
            \path (L12) edge[pos =.5, opacity = 0] node[opacity = 1]{\scriptsize  $\graph[1]$} (R12);
        \end{tikzpicture}\hspace{-.3em}\right)$ and $\graph[2] =
        \left(\hspace{-.3em}\begin{tikzpicture}[baseline = -3.5ex]
            \graph[grow right = 1.2cm, branch down = 3.ex, nodes={}]{
            {L1/{}[mynode,draw,circle], L12/{\tiny $\vdots$}, L2/{}[mynode,draw,circle]} -!-
            {R1/{}[mynode,draw,circle], R12/{\tiny $\vdots$}, R2/{}[mynode,draw,circle]} -!-
            };
            \node[left = 4pt of L1](s1){\scriptsize $1$}; \path (s1) edge[earrow, ->] (L1);
            \node[left = 4pt of L2](s2){\scriptsize $n'$}; \path (s2) edge[earrow, ->] (L2);
            \node[right = 4pt of R1](t1){\scriptsize $1$}; \path (R1) edge[earrow, ->] (t1);
            \node[right = 4pt of R2](t2){\scriptsize $m'$}; \path (R2) edge[earrow, ->] (t2);
            \path (L1) edge[earrow, =] (R1);
            \path (L2) edge[earrow, =] (R2);
            \draw[fill = white] ($(L1)+(.25,.15)$) -- ($(L2)+(.25,-.15)$) -- ($(R2)+(-.25,-.15)$) -- ($(R1)+(-.25,.15)$) -- cycle;
            \path (L12) edge[pos =.5, opacity = 0] node[opacity = 1]{\scriptsize  $\graph[2]$} (R12);
        \end{tikzpicture}\hspace{-.3em}\right)$,
    the \intro*\kl{series composition} $\graph[1] \series \graph[2]$
    and the \intro*\kl{parallel composition} $\graph[1] \parallel \graph[2]$ are defined as follows, respectively:
    \begin{align*}
        \graph[1] \series \graph[2] &\defeq \begin{cases}
            \left(\hspace{-.3em}\begin{tikzpicture}[baseline = -3.5ex]
                \graph[grow right = 1.2cm, branch down = 3.ex, nodes={}]{
                {L1/{}[mynode,draw,circle], L12/{\tiny $\vdots$}, L2/{}[mynode,draw,circle]} -!-
                {C1/{}[mynode,draw,circle], C12/{\tiny $\vdots$}, C2/{}[mynode,draw,circle]} -!-
                {R1/{}[mynode,draw,circle], R12/{\tiny $\vdots$}, R2/{}[mynode,draw,circle]} -!-
                };
                \node[left = 4pt of L1](s1){\scriptsize $1$}; \path (s1) edge[earrow, ->] (L1);
                \node[left = 4pt of L2](s2){\scriptsize $n$}; \path (s2) edge[earrow, ->] (L2);
                \node[right = 4pt of R1](t1){\scriptsize $1$}; \path (R1) edge[earrow, ->] (t1);
                \node[right = 4pt of R2](t2){\scriptsize $m'$}; \path (R2) edge[earrow, ->] (t2);
                \path (L1) edge[earrow, =] (C1);
                \path (L2) edge[earrow, =] (C2);
                \path (C1) edge[earrow, =] (R1);
                \path (C2) edge[earrow, =] (R2);
                \draw[fill = white] ($(L1)+(.25,.15)$) -- ($(L2)+(.25,-.15)$) -- ($(C2)+(-.25,-.15)$) -- ($(C1)+(-.25,.15)$) -- cycle;
                \path (L12) edge[pos =.5, opacity = 0] node[opacity = 1]{\scriptsize  $\graph[1]$} (C12);
                \draw[fill = white] ($(C1)+(.25,.15)$) -- ($(C2)+(.25,-.15)$) -- ($(R2)+(-.25,-.15)$) -- ($(R1)+(-.25,.15)$) -- cycle;
                \path (C12) edge[pos =.5, opacity = 0] node[opacity = 1]{\scriptsize  $\graph[2]$} (R12);
            \end{tikzpicture}\hspace{-.3em}\right)
        & (m  = n')\\
        \mbox{undefined} & (\mbox{otherwise})
    \end{cases},&
    \graph[1] \parallel \graph[2] &\defeq \left(\hspace{-.3em}\begin{tikzpicture}[baseline = -8.5ex]
        \graph[grow right = 1.2cm, branch down = 3.ex, nodes={}]{
        {L1/{}[mynode,draw,circle], L12/{\tiny $\vdots$}, L2/{}[mynode,draw,circle], L'1/{}[mynode,draw,circle], L'12/{\tiny $\vdots$}, L'2/{}[mynode,draw,circle]} -!-
        {R1/{}[mynode,draw,circle], R12/{\tiny $\vdots$}, R2/{}[mynode,draw,circle], R'1/{}[mynode,draw,circle], R'12/{\tiny $\vdots$}, R'2/{}[mynode,draw,circle]} -!-
        };
        \node[left = 4pt of L1](s1){\scriptsize $1$}; \path (s1) edge[earrow, ->] (L1);
        \node[left = 4pt of L2](s2){\scriptsize $n$}; \path (s2) edge[earrow, ->] (L2);
        \node[left = 4pt of L'1](s'1){\scriptsize $n+1$}; \path (s'1) edge[earrow, ->] (L'1);
        \node[left = 4pt of L'2](s'2){\scriptsize $n+n'$}; \path (s'2) edge[earrow, ->] (L'2);
        \node[right = 4pt of R1](t1){\scriptsize $1$}; \path (R1) edge[earrow, ->] (t1);
        \node[right = 4pt of R2](t2){\scriptsize $m$}; \path (R2) edge[earrow, ->] (t2);
        \node[right = 4pt of R'1](t'1){\scriptsize $m+1$}; \path (R'1) edge[earrow, ->] (t'1);
        \node[right = 4pt of R'2](t'2){\scriptsize $m+m'$}; \path (R'2) edge[earrow, ->] (t'2);
        \path (L1)  edge[earrow, =] (R1);
        \path (L2)  edge[earrow, =] (R2);
        \path (L'1) edge[earrow, =] (R'1);
        \path (L'2) edge[earrow, =] (R'2);
        \draw[fill = white] ($(L1)+(.25,.15)$) -- ($(L2)+(.25,-.15)$) -- ($(R2)+(-.25,-.15)$) -- ($(R1)+(-.25,.15)$) -- cycle;
        \path (L12) edge[pos =.5, opacity = 0] node[opacity = 1]{\scriptsize  $\graph[1]$} (R12);
        \draw[fill = white] ($(L'1)+(.25,.15)$) -- ($(L'2)+(.25,-.15)$) -- ($(R'2)+(-.25,-.15)$) -- ($(R'1)+(-.25,.15)$) -- cycle;
        \path (L'12) edge[pos =.5, opacity = 0] node[opacity = 1]{\scriptsize  $\graph[2]$} (R'12);
    \end{tikzpicture}\hspace{-.3em}\right).
    \end{align*}
    ($\graph[1] \series \graph[2]$ is the \kl{graph} obtained from $\graph[1]$ and $\graph[2]$ by merging $\tgt^{\graph[1]}_{i}$ and $\src^{\graph[2]}_{i}$ for $i \in \range{m}$.)
\end{defi}
Additionally, we may denote $\graph[1] \parallel \graph[2]$ vertically as $\substack{\graph[1]\\\parallel\\ \graph[2]}$.
For instance,
\begin{align*}
    \left(\hspace{-.3em}\begin{tikzpicture}[baseline = -.5ex]
        \graph[grow right = 1.cm, branch down = 6ex, nodes={mynode}]{
        {0/{}[draw, circle]}
        };
        \node[left = .5em of 0](l){\tiny };
        \node[above right = .5em of 0](r1){\tiny $1$};
        \node[below right = .5em of 0](r2){\tiny $2$};
        \graph[use existing nodes, edges={color=black, pos = .5, earrow}, edge quotes={fill=white, inner sep=1pt,font= \scriptsize}]{
            l -> 0; 0 -> {r1, r2};
        };
    \end{tikzpicture}\hspace{-.3em}\right) \series \left(\hspace{-.3em}\begin{tikzpicture}[baseline = -2.5ex]
        \graph[grow right = 1.cm, branch down = 3ex, nodes={mynode}]{
        {0/{}[draw, circle],0d/[draw, circle]} -!- {1/{}[draw, circle],1d/[draw, circle]}
        };
        \node[left = .5em of 0](l1){\tiny $1$};
        \node[left = .5em of 0d](l2){\tiny $2$};
        \node[right = .5em of 1](r1){\tiny $1$};
        \node[right = .5em of 1d](r2){\tiny $2$};
        \graph[use existing nodes, edges={color=black, pos = .5, earrow}, edge quotes={fill=white, inner sep=1pt,font= \scriptsize}]{
            l1 -> 0 -> ["$a$"] 1-> r1;
            l2 -> 0d -> ["$b$"] 1d -> r2;
        };
    \end{tikzpicture}\hspace{-.3em}\right) \series \left(\hspace{-.3em}\begin{tikzpicture}[baseline = -2.5ex]
        \graph[grow right = 1.cm, branch down = 3ex, nodes={}]{
        {0/{}[mynode, draw, circle],0d/[mynode, draw, circle]} -!- {/, 1d/[mynode, draw, circle]}
        };
        \node[left = .5em of 0](l1){\tiny $1$};
        \node[left = .5em of 0d](l2){\tiny $2$};
        \node[right = .5em of 0](r1){\tiny $1$};
        \node[right = .5em of 1d](r2){\tiny $2$};
        \graph[use existing nodes, edges={color=black, pos = .5, earrow}, edge quotes={fill=white, inner sep=1pt,font= \scriptsize}]{
            l1 -> 0 -> r1;
            l2 -> 0d <- ["$c$"] 1d -> r2;
        };
    \end{tikzpicture}\hspace{-.3em}\right) \series \left(\hspace{-.3em}\begin{tikzpicture}[baseline = -.5ex]
        \graph[grow right = 1.cm, branch down = 2ex, nodes={mynode}]{
        {0/{}[draw, circle]}
        };
        \node[above left = .5em of 0](l1){\tiny $1$};
        \node[below left = .5em of 0](l2){\tiny $2$};
        \node[right = .5em of 0](r){\tiny };
        \graph[use existing nodes, edges={color=black, pos = .5, earrow}, edge quotes={fill=white, inner sep=1pt,font= \scriptsize}]{
            {l1,l2} -> 0; 0 -> {r};
        };
    \end{tikzpicture}\hspace{-.3em}\right) &\quad=\quad \left(\hspace{-.3em}\begin{tikzpicture}[baseline = -2.5ex]
    \graph[grow right = 1.cm, branch down = 2ex]{
    {/, 0/{}[mynode, draw, circle]} -!- {/, /, 1/{}[mynode, draw, circle]} -!- {/, 2/{}[mynode, draw, circle]}
    };
    \node[left = .5em of 0](l){\tiny };
    \node[right = .5em of 2](r){\tiny };
    \graph[use existing nodes, edges={color=black, pos = .5, earrow}, edge quotes={fill=white, inner sep=1pt,font= \scriptsize}]{
        l -> 0; 2 -> r;
        0 ->["$a$", bend left = 10] 2;
        0 ->["$b$", bend right = 10] 1 <- ["$c$", bend right = 10] 2;
    };
    \end{tikzpicture}\hspace{-.3em}\right),\\
    \left(\hspace{-.3em}\begin{tikzpicture}[baseline = -.5ex]
        \graph[grow right = 1.cm, branch down = 6ex, nodes={mynode}]{
        {0/{}[draw, circle]}
        };
        \node[left = .5em of 0](l){\tiny };
        \node[above right = .5em of 0](r1){\tiny $1$};
        \node[below right = .5em of 0](r2){\tiny $2$};
        \graph[use existing nodes, edges={color=black, pos = .5, earrow}, edge quotes={fill=white, inner sep=1pt,font= \scriptsize}]{
            l -> 0; 0 -> {r1, r2};
        };
    \end{tikzpicture}\hspace{-.3em}\right) \series \left(
    \begin{gathered}
        \left(\hspace{-.3em}\begin{tikzpicture}[baseline = -.5ex]
            \graph[grow right = 1.cm, branch down = 2ex, nodes={mynode}]{
            {0/{}[draw, circle]} -!- {1/{}[draw, circle]}
            };
            \node[left = .5em of 0](l){\tiny };
            \node[right = .5em of 1](r){\tiny };
            \graph[use existing nodes, edges={color=black, pos = .5, earrow}, edge quotes={fill=white, inner sep=1pt,font= \scriptsize}]{
                l -> 0 ->["$a$"] 1 -> r;
            };
        \end{tikzpicture}\hspace{-.3em}\right)\\
        \parallel\\
        \left(\hspace{-.3em}\begin{tikzpicture}[baseline = -.5ex]
            \graph[grow right = 1.cm, branch down = 2ex, nodes={mynode}]{
            {0/{}[draw, circle]} -!- {1/{}[draw, circle]}
            };
            \node[left = .5em of 0](l){\tiny };
            \node[right = .5em of 1](r){\tiny };
            \graph[use existing nodes, edges={color=black, pos = .5, earrow}, edge quotes={fill=white, inner sep=1pt,font= \scriptsize}]{
                l -> 0 ->["$b$"] 1 -> r;
            };
        \end{tikzpicture}\hspace{-.3em}\right) \series     \left(\hspace{-.3em}\begin{tikzpicture}[baseline = -.5ex]
            \graph[grow right = 1.cm, branch down = 2ex, nodes={mynode}]{
            {1/{}[draw, circle]} -!-{2/{}[draw, circle]}
            };
            \node[left = .5em of 1](l){\tiny };
            \node[right = .5em of 2](r){\tiny };
            \graph[use existing nodes, edges={color=black, pos = .5, earrow}, edge quotes={fill=white, inner sep=1pt,font= \scriptsize}]{
                l -> 1 <-["$c$"] 2 -> r;
            };
        \end{tikzpicture}\hspace{-.3em}\right)
    \end{gathered}    
    \right) \series \left(\hspace{-.3em}\begin{tikzpicture}[baseline = -.5ex]
        \graph[grow right = 1.cm, branch down = 2ex, nodes={mynode}]{
        {0/{}[draw, circle]}
        };
        \node[above left = .5em of 0](l1){\tiny $1$};
        \node[below left = .5em of 0](l2){\tiny $2$};
        \node[right = .5em of 0](r){\tiny };
        \graph[use existing nodes, edges={color=black, pos = .5, earrow}, edge quotes={fill=white, inner sep=1pt,font= \scriptsize}]{
            {l1,l2} -> 0; 0 -> {r};
        };
    \end{tikzpicture}\hspace{-.3em}\right) &\quad=\quad \left(\hspace{-.3em}\begin{tikzpicture}[baseline = -2.5ex]
    \graph[grow right = 1.cm, branch down = 2ex]{
    {/, 0/{}[mynode, draw, circle]} -!- {/, /, 1/{}[mynode, draw, circle]} -!- {/, 2/{}[mynode, draw, circle]}
    };
    \node[left = .5em of 0](l){\tiny };
    \node[right = .5em of 2](r){\tiny };
    \graph[use existing nodes, edges={color=black, pos = .5, earrow}, edge quotes={fill=white, inner sep=1pt,font= \scriptsize}]{
        l -> 0; 2 -> r;
        0 ->["$a$", bend left = 10] 2;
        0 ->["$b$", bend right = 10] 1 <- ["$c$", bend right = 10] 2;
    };
    \end{tikzpicture}\hspace{-.3em}\right).
\end{align*}

For $n \in \nat$, we let $\id^{n}$ be the following \kl{graph} with $\tuple{n,n}$-interface:
$\id^{n} \defeq 
\left(
\begin{gathered}
    \begin{tikzpicture}[baseline = -3.ex]
        \graph[grow right = 1.cm, branch down = 2.5ex, nodes={}]{
        {1/{}[draw, circle,mynode], /{$\vdots$}[yshift = 0.2ex], 2/{}[draw, circle,mynode]}
        };
        \node[left = .5em of 1](l){\tiny $1$};
        \node[right = .5em of 1](r){\tiny $1$};
        \node[left = .5em of 2](l2){\tiny $n$};
        \node[right = .5em of 2](r2){\tiny $n$};
        \graph[use existing nodes, edges={color=black, pos = .5, earrow}, edge quotes={fill=white, inner sep=1pt,font= \scriptsize}]{
            l -> 1 -> r;
            l2 -> 2 -> r2;
            };
    \end{tikzpicture}
\end{gathered}\right)$.
Each $\id^{n}$ is an \intro*\kl{empty} \kl{graph} and the \kl{graph} $\id^{0}$ is the \intro*\kl{null} \kl{graph}.
Each $\id^n$ is the left identity element with respect to $\series$ on \kl{graphs} $\graph$ with $\tuple{n, m}$-interface, that is $\id^{n} \series \graph = \graph$.
The right identity element is given by $\id^{m}$, i.e., $\graph \series \id^m = \graph$.
The identity element with respect to $\parallel$ on \kl{graphs} with $\tuple{n,m}$-interface is given by $\id^{0}$, i.e., 
$\id^{0} \parallel \graph = \graph \parallel \id^{0} = \graph$.

\subsubsection*{Pathwidth and treewidth}
We recall the notions of the \kl{pathwidth} and \kl{treewidth} of \kl{structures} \cite[Definition 9.12]{courcelleGraphStructureMonadic2012},
which are defined based on the \kl{pathwidth} \cite{robertsonGraphMinorsExcluding1983} (see also, \eg, \cite[Theorem 2]{bodlaenderPartialArboretumGraphs1998}\cite[Lemma 4.6]{bojanczykDefinabilityEqualsRecognizability2016} for alternative characterizations) and the \kl{treewidth} \cite{robertsonGraphMinorsII1986} for graphs.
Below, for \kl{graphs} with bi-interface, we use a slightly generalized definition, based on \cite[Definition 2.5.3]{courcelleGraphStructureMonadic2012} for \kl{treewidth} and \cite{nakamuraPartialDerivativesGraphs2017} for \kl{pathwidth} (these definitions will be useful in the constructions in \Cref{proposition: pw le iw,proposition: tw 2}, while we will always consider \kl{structures} (\kl{graphs} without interfaces)).
For a \kl{graph} $\graph[1]$ (with bi-interface),
a \intro*\kl{path decomposition} of $\graph[1]$ is a sequence $\vec{\graph[2]} = \graph[2]_1 \dots \graph[2]_n$ of finite \kl{graphs} such that
\begin{itemize}
    \item $\univ{\graph[1]} = \bigcup_{i \in \range{n}} \univ{\graph[2]_{i}}$ and $a^{\graph[1]} = \bigcup_{i \in \range{n}} a^{\graph[2]_{i}}$ for each $a \in A$;
    \item $\univ{\graph[2]_i} \cap \univ{\graph[2]_k} \subseteq \univ{\graph[2]_j} \mbox{ for all $1 \le i \le j \le k \le n$}$;
    \item each \kl{source} \kl{vertex} of $\graph[1]$ is in $\univ{\graph[2]_1}$
    and each \kl{target} \kl{vertex} of $\graph[1]$ is in $\univ{\graph[2]_n}$.
\end{itemize}
\noindent 
The \intro*\kl(pathwidth){width} of $\vec{\graph[2]}$ is $\max_{i \in \range{n}}(\card \univ{\graph[2]_i} - 1)$.
The \intro*\kl{pathwidth} $\pw(\graph[1])$ of a \kl{graph} $\graph[1]$ is the minimum \kl(pathwidth){width} among \kl{path decompositions} of $\graph[1]$.
Each element of $\vec{\graph[2]}$ is called a \intro*\kl(pathwidth){bag}.

Similarly, for a \kl{graph} $\graph[1]$,
a \intro*\kl{tree decomposition} of $\graph[1]$ is a finite rooted tree $\vec{\graph[2]} = \set{\graph[2]_{\word}}_{\word \in \Gamma}$ of \kl{graphs}, where $\Gamma \subseteq \nat^*$ is a finite and prefix-closed ($\word \word' \in \Gamma$ implies $\word \in \Gamma$) set, such that
\begin{itemize}
    \item $\univ{\graph[1]} = \bigcup_{\word \in \Gamma} \univ{\graph[2]_{\word}}$ and $a^{\graph[1]} = \bigcup_{\word \in \Gamma} a^{\graph[2]_{\word}}$ for each $a \in A$;
    \item $\univ{\graph[2]_{\word[1]}} \cap \univ{\graph[2]_{\word[3]}} \subseteq \univ{\graph[2]_{\word[2]}}$ for all 
    $\word[1], \word[2], \word[3] \in \Gamma$ such that $\word[2]$ is on the unique path between $\word[1]$ and $\word[3]$
    (i.e., $\word[2]$ satisfies $\mathsf{LCA}(\word[1], \word[3]) \preceq_{\mathrm{pref}} \word[2] \preceq_{\mathrm{pref}} \word[1]$ or $\mathsf{LCA}(\word[1], \word[3]) \preceq_{\mathrm{pref}} \word[2] \preceq_{\mathrm{pref}} \word[3]$ where $\word[1]' \preceq_{\mathrm{pref}} \word[2]'$ denotes that $\word[1]'$ is a prefix of $\word[2]'$ and $\mathsf{LCA}(\word[1], \word[3])$ denotes the maximal common prefix \kl{word} (i.e., the least common ancestor) of $\word[1]$ and $\word[3]$);
    \item each \kl{source} \kl{vertex} and each \kl{target} \kl{vertex} of $\graph[1]$ is in the root $\univ{\graph[2]_{\eps}}$.
\end{itemize}\noindent 
The \intro*\kl(treewidth){width} of $\vec{\graph[2]}$ is $\max_{\word \in \Gamma}(\card \univ{\graph[2]_{\word}} - 1)$.
The \intro*\kl{treewidth} $\tw(\graph[1])$ of a \kl{graph} $\graph[1]$ is the minimum \kl(treewidth){width} among \kl{tree decompositions} of $\graph[1]$.
Each element of $\vec{\graph[2]}$ is called a \intro*\kl(treewidth){bag}.

Particularly when $\graph$ has no interfaces,
the \kl{pathwidth}/\kl{treewidth} of a \kl{graph} $\graph$ coincide with those of its Gaifman graph (i.e., the \kl{graph} $\tuple{\univ{\graph}, E}$ where $E$ is the binary relation $\bigcup_{a \in A} \bigcup_{\tuple{x_1, \dots, x_k} \in a^{\graph}} \set{\tuple{x_i, x_j} \mid i, j \in \range{k} \land x_i \neq x_j}$ \cite[p.\ 26]{ebbinghausFiniteModelTheory1995}), respectively.
\begin{rem}
    By definition, the number of \kl(pathwidth){bags} is finite; %
    this is sufficient in this paper as we are mainly interested in finite \kl{structures} (thanks to the \Cref{proposition: bounded model property}).
\end{rem}

\subsection{PCoR*: The Positive Calculus of Relations with Transitive Closure}\label{section: PCOR*}
We consider the positive calculus of relations with transitive closure (PCoR*).
We use $\vsig$ to denote a set of \intro*\kl{variables}.
The set of \intro*\kl{PCoR* terms} over $\vsig$ is defined as the set of terms over the signature $\set{\id_{(0)}, \emp_{(0)}, \top_{(0)}, \compo_{(2)}, \union_{(2)}, \intersection_{(2)}, \bl^{\smile}_{(1)}, \bl^{*}_{(1)}}$ and the \kl{variable} set $\vsig$:
\begin{align*}
    \term[1], \term[2], \term[3] \;\Coloneqq\; \aterm \mid \id \mid \emp  \mid \top \mid \term[1] \compo \term[2] \mid \term[1] \union \term[2] \mid \term[1] \intersection \term[2] \mid \term^{\smile}  \mid \term[1]^*. \tag{$a \in \vsig$}
\end{align*}
We often abbreviate $\term[1] \compo \term[2]$ to $\term[1] \term[2]$.
We use parentheses in ambiguous situations (where $\union$, $\compo$, and $\cap$ are left-associative).
We write $\sum_{i = 1}^{n} \term[1]_i$ for the term $\emp \union \term[1]_1 \union \dots \union \term[1]_n$,
write $\bigcompo_{i = 1}^{n} \term[1]_i$ for the term $\id \compo \term[1]_1 \compo \dots \compo \term[1]_n$,
write $\term^{n}$ for the term $\bigcompo_{i = 1}^{n} \term[1]$,
and write $\term[1]^{+}$ for the term $\term[1] \term[1]^*$.

An \intro*\kl{equation} $\term[1] = \term[2]$ is a pair of \kl{PCoR* terms}.
The \intro*\kl{inequation} $\term[1] \le \term[2]$ abbreviates the \kl{equation} $\term[1] \union \term[2] = \term[2]$.

For a \kl{PCoR* term} $\term$, the \intro*\kl{size} $\|\term\|$ is the number of symbols occurring in $\term$:
\begin{align*}
    \|\aterm\| &\defeq 1 \mbox{ for $\aterm \in \vsig \cup \set{\id, \emp, \top}$}, &
    \|\term[1]^{\mathbin{\heartsuit}}\| &\defeq 1 + \|\term[1]\| \mbox{ for $\mathbin{\heartsuit} \in \set{\smile, *}$},\\
    \|\term[1] \mathbin{\heartsuit} \term[2]\| &\defeq 1 + \|\term[1]\| + \|\term[2]\| \mbox{ for $\mathbin{\heartsuit} \in \set{\compo, \union, \intersection}$.}
\end{align*}
Additionally, the \intro*\kl{intersection width} $\iw(\term)$ \cite{gollerPDLIntersectionConverse2009} is defined as follows:
\begin{align*}
    \iw(\aterm) &\defeq 1  \mbox{ for $\aterm \in \vsig \cup \set{\id, \emp, \top}$}, & 
    \iw(\term[1]^{\mathbin{\heartsuit}}) &\defeq \iw(\term[1]) \mbox{ for $\mathbin{\heartsuit} \in \set{\smile, *}$},\\
    \iw(\term[1] \mathbin{\heartsuit} \term[2]) &\defeq \max(\iw(\term[1]), \iw(\term[2])) \mbox{  for $\mathbin{\heartsuit} \in \set{\compo, \union}$}, &
    \iw(\term[1] \intersection \term[2]) &\defeq \iw(\term[1]) + \iw(\term[2]).
\end{align*}
For instance, $\iw(a \intersection b \intersection c \intersection d) = 4$ and $\iw((a \intersection b) \compo (a \intersection c) \compo (a \intersection d)) = 2$.
\begin{prop}\label{proposition: iw}
    For all \kl{PCoR* terms} $\term$, we have $\iw(\term) \le \|\term\|$.
\end{prop}
\begin{proof}
    By easy induction on $\term$.
\end{proof}

\subsubsection{Semantics: relational models}
A \intro*\kl[structure]{structure (of binary relations)} is a non-empty \kl{graph} without interfaces over the set $\vsig$ with $\ty(a) = \tuple{1, 1}$ for $a \in \vsig$.
We use $\struc[1], \struc[2]$ to denote a \kl{structure}.
Let $\struc[1]$ be a \kl{structure}.
For a \kl{PCoR* term} $\term$,
the \intro*\kl{semantics} $\jump{\term}_{\struc} \subseteq \univ{\struc}^2$ is defined as follows:
\begin{gather*}
    \jump{a}_{\struc} \defeq a^{\struc}, \qquad
    \jump{\id}_{\struc} \defeq \triangle_{\univ{\struc}}, \qquad
    \jump{\emp}_{\struc} \defeq \emptyset, \qquad
    \jump{\top}_{\struc} \defeq \univ{\struc}^2, \qquad
    \jump{\term[1] \compo \term[2]}_{\struc} \defeq \jump{\term[1]}_{\struc} \compo \jump{\term[2]}_{\struc}, \\
    \jump{\term[1] \union \term[2]}_{\struc} \defeq \jump{\term[1]}_{\struc} \cup \jump{\term[2]}_{\struc}, \qquad
    \jump{\term[1] \intersection \term[2]}_{\struc} \defeq \jump{\term[1]}_{\struc} \cap \jump{\term[2]}_{\struc}, \qquad
    \jump{\term^{\smile}}_{\struc} \defeq \jump{\term}_{\struc}^{\smile}, \qquad
    \jump{\term[1]^*}_{\struc} \defeq \jump{\term[1]}_{\struc}^*.
\end{gather*}
We write $\REL$ for the class of all \kl{structures}.
For $\strucclass \subseteq \REL$,
we write $\strucclass \models \term[1] = \term[2]$ if $\jump{\term[1]}_{\struc} = \jump{\term[2]}_{\struc}$ for all $\struc \in \strucclass$.
In the sequel, we consider the \intro*\kl{equational theory} with respect to $\REL$.
For instance, the following \kl{equations} hold with respect to $\REL$:
\begin{center}
    \begin{minipage}[]{0.29\textwidth}
        \begin{align}
            \label{equation: KA 1} (\term[1]^* \term[2]^*)^* = (\term[1] \union \term[2])^*\\
            \label{equation: KA 2} \term[1]^* = (\term[1]\term[1])^*(\id \union \term[1])
        \end{align}
    \end{minipage}%
    \begin{minipage}[]{0.35\textwidth}
        \begin{align}
            \label{equation: allegories 1} \term[1] (\term[2] \intersection \term[3]) \le (\term[1] \term[2]) \intersection (\term[1] \term[3])\\
            \label{equation: allegories 2} (\term[1] \term[2]) \intersection \term[3] \le (\term[1] \intersection (\term[3] \term[2]^{\smile})) \term[2]
        \end{align}
    \end{minipage}%
    \begin{minipage}[]{0.36\textwidth}
        \begin{align}
            \label{equation: PCoR* 1} \term[1]^{+} \intersection \id \le (\term[1] \term[1])^{+}\\
            \label{equation: PCoR* 2} (\term[1] \intersection \term[2]^{\smile})^{+} \intersection \id \le (\term[1] \intersection \term[2]^{+})^{+}
        \end{align}
    \end{minipage}%
\end{center}

Notably,
the \kl{equational theory} of \kl{PCoR*} contains that of Kleene algebras $\set{\id, \emp, \compo, \union, \bl^{*}}$ \cite{kozenCompletenessTheoremKleene1991} (such as \Cref{equation: KA 1,equation: KA 2})
and that of allegories $\set{\id, \compo, \intersection, \bl^{\smile}}$ \cite{freydCategoriesAllegories1990,pousAllegoriesDecidabilityGraph2018} (such as \Cref{equation: allegories 1} (the semi-distributive law over $\cap$) and \Cref{equation: allegories 2} (the modularity law));
see also \cite{bloomNotesEquationalTheories1995,brunetAlgorithmsKleeneAlgebra2016} for Kleene algebras with converse $\set{\id, \emp, \compo, \union, \bl^{*}, \bl^{\smile}}$
and \cite{andrekaEquationalTheoryKleene2011,doumaneCompletenessIdentityfreeKleene2018} for \kl{identity-free Kleene lattices} $\set{\emp,\compo, \union, \intersection, \bl^{+}}$.
\Cref{equation: PCoR* 1} \cite[p.\ 14]{pousPositiveCalculusRelations2018} and \Cref{equation: PCoR* 2} are non-trivial \kl{equations} having transitive closure and intersection.

\subsubsection{An alternative semantics: graph languages}\label{section: graph languages}
We recall \kl{graph languages} for PCoR* \cite{brunetPetriAutomata2017}\cite[Definition 15]{pousPositiveCalculusRelations2018}.
The \intro*\kl{graph language} $\glang(\term)$ of a \kl{PCoR* term} $\term$ is the set of \kl{graphs} with $\tuple{1, 1}$-interface over $\vsig$, defined as follows:
\begin{align*}
    \hspace{-1.em} \glang(a)                                                                                                      & \defeq \set*{
        \begin{tikzpicture}[baseline = -.5ex]
            \graph[grow right = 1.cm, branch down = 6ex, nodes={mynode}]{
            {0/{}[draw, circle]}-!-{1/{}[draw, circle]}
            };
            \node[left = .5em of 0](l){};
            \node[right = .5em of 1](r){};
            \graph[use existing nodes, edges={color=black, pos = .5, earrow}, edge quotes={fill=white, inner sep=1pt,font= \scriptsize}]{
            0 ->["$a$"] 1;
            l -> 0; 1 -> r;
            };
        \end{tikzpicture}
    }  \mbox{ \mbox{ for $a \in \vsig$}},                                          & 
     \glang(\id)                                                                                                        &\defeq \set{\begin{tikzpicture}[baseline = -.5ex]
                                                                                                                                                                                                                                                                                        \graph[grow right = 1.cm, branch down = 6ex, nodes={mynode}]{
                                                                                                                                                                                                                                                                                        {0/{}[draw, circle]}
                                                                                                                                                                                                                                                                                        };
                                                                                                                                                                                                                                                                                        \node[left = .5em of 0](l){};
                                                                                                                                                                                                                                                                                        \node[right = .5em of 0](r){};
                                                                                                                                                                                                                                                                                        \graph[use existing nodes, edges={color=black, pos = .5, earrow}, edge quotes={fill=white, inner sep=1pt,font= \scriptsize}]{
                                                                                                                                                                                                                                                                                            l -> 0; 0 -> r;
                                                                                                                                                                                                                                                                                        };
                                                                                                                                                                                                                                                                                    \end{tikzpicture}},  \hspace{1.em}                \glang(\emp)  \defeq \emptyset,  \\
    \glang(\term[1] \compo \term[2])                                                                                 & \defeq \set*{\begin{tikzpicture}[baseline = -.5ex]
                                                                                                                                                        \graph[grow right = 1.cm, branch down = 2.5ex, nodes={mynode, font = \scriptsize}]{
                                                                                                                                                        {s1/{}[draw, circle]}
                                                                                                                                                        -!- {c/{}[draw, circle]}
                                                                                                                                                        -!- {t1/{}[draw, circle]}
                                                                                                                                                        };
                                                                                                                                                        \node[left = 4pt of s1](s1l){} edge[earrow, ->] (s1);
                                                                                                                                                        \node[right = 4pt of t1](t1l){}; \path (t1) edge[earrow, ->] (t1l);
                                                                                                                                                        \graph[use existing nodes, edges={color=black, pos = .5, earrow}, edge quotes={fill=white, inner sep=1pt,font= \scriptsize}]{
                                                                                                                                                        s1 ->["$\graph[1]$"] c;
                                                                                                                                                        c ->["$\graph[2]$"] t1;
                                                                                                                                                        };
                                                                                                                                                    \end{tikzpicture} \mid \graph[1] \in \glang(\term[1]) \land \graph[2] \in \glang(\term[2])},                                                                                            & 
    \glang(\top) &\defeq \set*{\begin{tikzpicture}[baseline = -.5ex]
        \graph[grow right = .5cm, branch down = 6ex, nodes={mynode}]{
        {0/{}[draw, circle]}-!-{1/{}[draw, circle]}
        };
        \node[left = .5em of 0](l){};
        \node[right = .5em of 1](r){};
        \graph[use existing nodes, edges={color=black, pos = .5, earrow}, edge quotes={fill=white, inner sep=1pt,font= \scriptsize}]{
            l -> 0; 1 -> r;
        };
    \end{tikzpicture}},       \\
    \hspace{-1.em}\glang(\term[1] \intersection \term[2])                                                                                  & \defeq
    \set*{\begin{tikzpicture}[baseline = -.5ex]
                    \graph[grow right = 1.cm, branch down = 2.5ex, nodes={mynode, font = \scriptsize}]{
                    {s1/{}[draw, circle]}
                    -!- {t1/{}[draw, circle]}
                    };
                    \node[left = 4pt of s1](s1l){} edge[earrow, ->] (s1);
                    \node[right = 4pt of t1](t1l){}; \path (t1) edge[earrow, ->] (t1l);
                    \graph[use existing nodes, edges={color=black, pos = .5, earrow}, edge quotes={fill=white, inner sep=1pt,font= \scriptsize}]{
                    s1 ->["$\graph[1]$", bend left = 25] t1;
                    s1 ->["$\graph[2]$", bend right = 25] t1;
                    };
                \end{tikzpicture} \mid \graph[1] \in \glang(\term[1]) \land \graph[2] \in \glang(\term[2])}, & \glang(\term[1] \union \term[2])                                                                                                & \defeq \glang(\term[1]) \cup \glang(\term[2]),                                                                                                      \\
    \hspace{-1.em}
    \glang(\term[1]^{\smile})                                                                                                    & \defeq \set{
        \begin{tikzpicture}[baseline = -.5ex]
            \graph[grow right = 1.cm, branch down = 6ex, nodes={mynode}]{
            {0/{}[draw, circle]}-!-{1/{}[draw, circle]}
            };
            \node[left = .5em of 0](l){};
            \node[right = .5em of 1](r){};
            \graph[use existing nodes, edges={color=black, pos = .5, earrow}, edge quotes={fill=white, inner sep=1pt,font= \scriptsize}]{
            1 ->["$\graph$"] 0;
            l -> 0; 1 -> r;
            };
        \end{tikzpicture} \mid \graph \in \glang(\term)
    },                                  &
    \glang(\term[1]^{*})                                                                                                          & \defeq
    \bigcup_{n \ge 0} \glang(\term[1]^n).
\end{align*}
For instance,
\begin{align*}
    \glang(\aterm[1] \intersection (\aterm[2] \aterm[3]^{\smile})) \;&=\; \set*{\begin{tikzpicture}[baseline = -.5ex]
                    \graph[grow right = 1.cm, branch down = 1.5ex, nodes={font = \scriptsize}]{
                    {s1/{}[mynode, draw, circle]}
                    -!- {/, c/{}[mynode, draw, circle]}
                    -!- {t1/{}[mynode, draw, circle]}
                    };
                    \node[left = 4pt of s1](s1l){} edge[earrow, ->] (s1);
                    \node[right = 4pt of t1](t1l){}; \path (t1) edge[earrow, ->] (t1l);
                    \graph[use existing nodes, edges={color=black, pos = .5, earrow}, edge quotes={fill=white, inner sep=1pt,font= \scriptsize}]{
                    s1 ->["$\aterm[1]$", bend left = 15] t1;
                    s1 ->["$\aterm[2]$", bend right = 15] c <-["$\aterm[3]$", bend right = 15] t1;
                    };
                \end{tikzpicture}},\\
    \glang((((\aterm[1] \aterm[1]) \intersection \id) \aterm[2] \top \aterm[3]) \union \aterm[4]) \;&=\; \set*{\begin{tikzpicture}[baseline = -2.5ex]
    \graph[grow right = 1.cm, branch down = 3.5ex, nodes={font = \scriptsize}]{
    {su/{}[mynode, draw, circle],s1/{}[mynode, draw, circle]}
    -!- {/, c1/{}[mynode, draw, circle]}
    -!- {/, c2/{}[mynode, draw, circle]}
    -!- {/, t1/{}[mynode, draw, circle]}
    };
    \node[left = 4pt of s1](s1l){} edge[earrow, ->] (s1);
    \node[right = 4pt of t1](t1l){}; \path (t1) edge[earrow, ->] (t1l);
    \graph[use existing nodes, edges={color=black, pos = .5, earrow}, edge quotes={fill=white, inner sep=1pt,font= \scriptsize}]{
    s1 ->["$\aterm[1]$", bend left = 60, pos = .45] su ->["$\aterm[1]$", bend left = 60, pos = .45] s1;
    s1 ->["$\aterm[2]$"] c1; c2 ->["$\aterm[3]$"] t1;
    };
    \end{tikzpicture}, \begin{tikzpicture}[baseline = .5ex]
        \graph[grow right = 1.cm, branch down = 6ex, nodes={mynode}]{
        {0/{}[draw, circle]}-!-{1/{}[draw, circle]}
        };
        \node[left = .5em of 0](l){};
        \node[right = .5em of 1](r){};
        \graph[use existing nodes, edges={color=black, pos = .5, earrow}, edge quotes={fill=white, inner sep=1pt,font= \scriptsize}]{
        0 ->["$\aterm[4]$"] 1;
        l -> 0; 1 -> r;
        };
    \end{tikzpicture}},\\
    \glang(\aterm^*) \;&=\; \set*{
        \begin{tikzpicture}[baseline = -.5ex]
            \graph[grow right = 1.cm, branch down = 6ex, nodes={mynode}]{
            {1/{}[draw, circle]}
            };
            \node[left = .5em of 1](l){};
            \node[right = .5em of 1](r){};
            \graph[use existing nodes, edges={color=black, pos = .5, earrow}, edge quotes={fill=white, inner sep=1pt,font= \scriptsize}]{
                l -> 1; 1 -> r;
            };
        \end{tikzpicture}, 
        \begin{tikzpicture}[baseline = -.5ex]
                \graph[grow right = 1.cm, branch down = 6ex, nodes={mynode}]{
                {0/{}[draw, circle]}-!-{1/{}[draw, circle]}
                };
                \node[left = .5em of 0](l){};
                \node[right = .5em of 1](r){};
                \graph[use existing nodes, edges={color=black, pos = .5, earrow}, edge quotes={fill=white, inner sep=1pt,font= \scriptsize}]{
                0 ->["$\aterm$"] 1;
                l -> 0; 1 -> r;
                };
            \end{tikzpicture},
            \begin{tikzpicture}[baseline = -.5ex]
                \graph[grow right = 1.cm, branch down = 6ex, nodes={mynode}]{
                {0/{}[draw, circle]}-!-{1/{}[draw, circle]}-!-{2/{}[draw, circle]}
                };
                \node[left = .5em of 0](l){};
                \node[right = .5em of 2](r){};
                \graph[use existing nodes, edges={color=black, pos = .5, earrow}, edge quotes={fill=white, inner sep=1pt,font= \scriptsize}]{
                0 ->["$\aterm$"] 1 ->["$\aterm$"] 2;
                l -> 0; 2 -> r;
                };
            \end{tikzpicture},\  \dots }.
\end{align*}

\begin{prop}\label{proposition: tw 2}
    For all \kl{PCoR* terms} $\term$ and $\graph \in \glang(\term)$, we have $\tw(\graph) \le 2$.
\end{prop}
\begin{proof}[Proof (Corollary of {\cite[Theorem 41]{bodlaenderPartialArboretumGraphs1998}})]
    By induction on $\set{\compo,\intersection}$-terms,
    we have that every series-parallel graph (with one \kl{source} and one \kl{target}) has some \kl{tree decomposition} of \kl[treewidth]{width} $2$ \cite[Theorem 41]{bodlaenderPartialArboretumGraphs1998}.
    Similar to this, by induction on \kl{PCoR* terms} $\term$,
    each \kl{graph} $\graph \in \glang(\term)$ has some \kl{tree decomposition} of \kl[treewidth]{width} $2$.
\end{proof}

\begin{prop}\label{proposition: pw le iw}
    For all \kl{PCoR* terms} $\term$ and $\graph \in \glang(\term)$, we have $\pw(\graph) \le \iw(\term)$.
\end{prop}
\begin{proof}
    Similar to the above, by easy induction on $\term$,
    each \kl{graph} $\graph \in \glang(\term)$
    has some \kl{path decomposition} $\vec{\graph[2]} = \graph[2]_1 \dots \graph[2]_n$ of \kl[pathwidth]{width} $\iw(\term)$
    (where the \kl{source} \kl{vertex} is in $\univ{\graph[2]_1}$ and the \kl{target} \kl{vertex} is in $\univ{\graph[2]_n}$).
\end{proof}

In this paper, we use the \kl{graph languages} above as an alternative semantics.
Let $\struc$ be a \kl{structure}.
For $x, y \in \univ{\struc}$,
we write $\const{G}(\struc, x, y)$ for the \kl{graph} defined by $\tuple{\univ{\struc}, \set{\jump{a}_{\struc}}_{a \in \vsig}, x, y}$.
For a \kl{graph} $\graph[2]$ and a set $\glang$ of \kl{graphs},
we define $\jump{\graph[2]}_{\struc}$ and $\jump{\glang}_{\struc}$ as follows:
\begin{align*}
    \jump{\graph[2]}_{\struc} & \;\defeq\; \set{\tuple{x, y} \mid \graph[2] \homo \const{G}(\struc, x, y)}, & \jump{\glang}_{\struc} \;\defeq\; \bigcup_{\graph[2] \in \glang} \jump{\graph[2]}_{\struc}.
\end{align*}
We then have the following.
\begin{propC}[{\cite[Lemma 2.1 and 2.2]{andrekaEquationalTheoryKleene2011}} (for PCoR* without $\top$), {\cite[Lemma 2.3]{brunetPetriAutomata2017}} (for PCoR*)]\label{proposition: graph language}
    Let $\struc$ be a \kl{structure}.
    For all \kl{PCoR* terms} $\term$, we have $\jump{\term}_{\struc} = \jump{\glang(\term)}_{\struc}$.
\end{propC}
\begin{proof}
    By easy induction on $\term$ using distributivity of $\compo$, $\intersection$, and $\bl^{\smile}$.
\end{proof}\noindent 
For instance, when $\struc = \left(\begin{tikzpicture}[baseline = 0.5ex]
    \graph[grow right = 1.cm, branch down = 6ex, nodes={mynode}]{
    {0/{$\lab[1]$}[draw, circle]}-!-{1/{$\lab[2]$}[draw, circle]}
    };
    \graph[use existing nodes, edges={color=black, pos = .5, earrow}, edge quotes={fill=white, inner sep=1pt,font= \scriptsize}]{
    0 ->["$\aterm[1]$", out = 70, in = 110, looseness = 9] 0;
    0 ->["$\aterm[2]$", bend left] 1;
    0 ->["$\aterm[3]$", bend right] 1;
    };
\end{tikzpicture} \right)$,
we have $\tuple{x, x} \in \jump{a \intersection (b c^{\smile})}_{\struc}$ by the following \kl{graph homomorphism}:
\[\glang(\aterm[1] \intersection (\aterm[2] \aterm[3]^{\smile})) \ni \begin{tikzpicture}[remember picture,baseline = -.5ex]
    \graph[grow right = 1.cm, branch down = 1.5ex, nodes={font = \scriptsize}]{
    {s1/{}[mynode, draw, circle]}
    -!- {/, c/{}[mynode, draw, circle]}
    -!- {t1/{}[mynode, draw, circle]}
    };
    \node[left = 4pt of s1](s1l){} edge[earrow, ->] (s1);
    \node[right = 4pt of t1](t1l){}; \path (t1) edge[earrow, ->] (t1l);
    \graph[use existing nodes, edges={color=black, pos = .5, earrow}, edge quotes={fill=white, inner sep=1pt,font= \scriptsize}]{
    s1 ->["$\aterm[1]$", bend left = 15] t1;
    s1 ->["$\aterm[2]$", bend right = 15] c <-["$\aterm[3]$", bend right = 15] t1;
    };
\end{tikzpicture} \hspace{4em} \begin{tikzpicture}[remember picture,baseline = -.5ex]
    \graph[grow right = 1.cm, branch down = 6ex, nodes={mynode}]{
    {0/{$\lab[1]$}[draw, circle]}-!-{1/{$\lab[2]$}[draw, circle]}
    };
    \node[above left = 3pt and 4pt of 0](l){} edge[earrow, ->] (0);
    \node[below left = 3pt and 4pt of 0](r){}; \path (0) edge[earrow, ->] (r);
    \graph[use existing nodes, edges={color=black, pos = .5, earrow}, edge quotes={fill=white, inner sep=1pt,font= \scriptsize}]{
    0 ->["$\aterm[1]$", out = 70, in = 110, looseness = 9] 0;
    0 ->["$\aterm[2]$", bend left] 1;
    0 ->["$\aterm[3]$", bend right] 1;
    };
\end{tikzpicture} = \const{G}(\struc, \lab[1], \lab[1]).
\begin{tikzpicture}[remember picture, overlay]
    \path (s1) edge [homoarrow,->, out =45, in = 165, looseness = .5] (0);
    \path (t1) edge [homoarrow,->, out =45, in = 165, looseness = .5] (0);
    \path (c) edge [homoarrow,->, out =-45, in = -120, looseness = .3] (1);
\end{tikzpicture}\]
As a corollary of \Cref{proposition: graph language}, we also have the following bounded model property.
\begin{prop}[Graph language bounded model property]\label{proposition: bounded model property}
    For all \kl{PCoR* terms} $\term[1], \term[2]$,
    \[\REL \models \term[1] \le \term[2] \;\iff\; \set{\struc \mid \exists x, y \in \univ{\struc}, \const{\graph}(\struc, x, y) \in \glang(\term[1])} \models \term[1] \le \term[2].\]
\end{prop}
\begin{proof}
    ($\Longrightarrow$):
    Trivial.
    ($\Longleftarrow$):
    We prove the contraposition.
    Assume $\tuple{x, y} \in \jump{\term[1]}_{\struc} \setminus \jump{\term[2]}_{\struc}$ for some $\struc \in \REL$ and $x, y \in \univ{\struc}$.
    By \Cref{proposition: graph language},
    $\graph' \homo \const{G}(\struc, x, y)$ for some $\graph' \in \glang(\term)$, and
    $\graph[2] \centernot\homo \const{G}(\struc, x, y)$ for all $\graph[2] \in \glang(\term[2])$.
    Let $\struc[2]$, $x'$, and $y'$ be such that  $\const{G}(\struc[2], x', y') = \graph'$.
    Since $\graph' \homo \const{G}(\struc[2], x', y')$, we have $\tuple{x', y'} \in \jump{\term[1]}_{\struc[2]}$.
    Also, by $\const{G}(\struc[2], x', y') \homo \const{G}(\struc, x, y)$,
    we have $\graph[2] \centernot\homo \const{G}(\struc[2], x', y')$ for all $\graph[2] \in \glang(\term[2])$,
    and hence $\tuple{x', y'} \not\in \jump{\term[2]}_{\struc[2]}$.
    Hence by $\const{G}(\struc[2], x', y') \in \glang(\term[1])$, this completes the proof.
\end{proof}

For a class $\strucclass \subseteq \REL$, we use the following notations:
\begin{align*}
    \strucclass_{\pw \le k} &= \set{\struc \in \strucclass \mid \pw(\struc) \le k}, &
    \strucclass_{\tw \le k} &= \set{\struc \in \strucclass \mid \tw(\struc) \le k}.
\end{align*}
As a corollary of \Cref{proposition: bounded model property}, we also have the following two model properties.
\begin{prop}[Treewidth at most $2$ model property]\label{proposition: tw 2 property}
    {\intro*\kl[treewidth at most $2$ model property]{}}%
    For all \kl{PCoR* terms} $\term[1], \term[2]$,
    \[\REL \models \term[1] \le \term[2] \;\iff\; \REL_{\tw \le 2} \models \term[1] \le \term[2].\]
\end{prop}
\begin{proof}
    ($\Longrightarrow$):
    Trivial.
    ($\Longleftarrow$):
    Because every \kl{graph} of $\glang(\term)$ has \kl{treewidth} at most $2$ (\Cref{proposition: tw 2}),
    this direction is shown by \Cref{proposition: bounded model property}.
\end{proof}
\begin{prop}[Linearly bounded pathwidth model property]\label{proposition: bounded pw property}
    {\intro*\kl[linearly bounded pathwidth model property]{}}%
    For all \kl{PCoR* terms} $\term[1], \term[2]$,
    \[\REL \models \term[1] \le \term[2] \;\iff\; \REL_{\pw \le \iw(\term[1])} \models \term[1] \le \term[2].\]
\end{prop}
\begin{proof}
    ($\Longrightarrow$):
    Trivial.
    ($\Longleftarrow$):
    Because every \kl{graph} of $\glang(\term)$ has \kl{pathwidth} at most $\iw(\term[1])$ (\Cref{proposition: pw le iw}),
    this direction is shown by \Cref{proposition: bounded model property}.
\end{proof}
Note that we can easily translate the \kl{equational theory} of \kl{PCoR* terms} into the theory of monadic second-order logic (MSO) formulas\footnote{Precisely, with respect to binary relations, the expressive power of PCoR* is equivalent to that of three-variable existential positive logic with ``variable-confined'' monadic transitive closure \cite{nakamuraExpressivePowerSuccinctness2022}.} as an analog of the standard translation from the \kl{equational theory} of the calculus of relations into the theory of three-variable first-order logic formulas \cite{tarskiCalculusRelations1941}.
Because the theory of MSO over bounded \kl{treewidth} \kl{structures} is decidable \cite{courcelleMonadicSecondorderLogic1988,courcelleMonadicSecondorderLogic1990,courcelleGraphStructureMonadic2012},
\Cref{proposition: tw 2 property} (or \Cref{proposition: bounded pw property}) implies that the \kl{equational theory} of \kl{PCoR*} is decidable.
However, the naive algorithm obtained from the above has a non-elementary complexity.
We will show that the \kl{equational theory} of \kl{PCoR*} is EXPSPACE-complete (\Cref{corollary: PCoR* EXPSPACE-complete}).

Additionally, Brunet and Pous gave the following graph-theoretic characterization of the \kl{equational theory} for PCoR*,
which generalizes the characterizations known for conjunctive queries \cite{chandraOptimalImplementationConjunctive1977}, representable allegories \cite{freydCategoriesAllegories1990} and union-free relation algebras \cite{andrekaEquationalTheoryUnionfree1995}.
\begin{propC}[\cite{brunetPetriAutomataKleene2015,brunetPetriAutomata2017}]\label{proposition: glang and equational theory}
    For all \kl{PCoR* terms} $\term[1], \term[2]$,
    \[\REL \models \term[1] \le \term[2] \;\iff\; \forall \graph[1] \in \glang(\term[1]), \exists \graph[2] \in \glang(\term[2]), \graph[2] \homo \graph[1].\]
\end{propC}

\subsubsection{Comparison to word languages}\label{section: word language}
The \kl{word language} $\ljump{\term}_{\vsig}$ of a \kl{Kleene lattice term} (i.e., terms over $\set{\id, \emp, \compo, \union, \cap, \bl^{*}}$) $\term$ is the set of \kl{words} over $\vsig$ defined as follows:
\begin{align*}
    \ljump{\aterm}_{\vsig} & \defeq \set{\aterm} \mbox{ for $\aterm \in \vsig$}, &
    \ljump{\id}_{\vsig} &\defeq \set{\eps},&
    \ljump{\emp}_{\vsig} & \defeq \emptyset, \\
    \ljump{\term[1] \union \term[2]}_{\vsig} &\defeq \ljump{\term[1]}_{\vsig} \cup \ljump{\term[2]}_{\vsig},&
    \ljump{\term[1] \compo \term[2]}_{\vsig} &\defeq \ljump{\term[1]}_{\vsig} \compo \ljump{\term[2]}_{\vsig}, &
    \ljump{\term[1]^{*}}_{\vsig} &\defeq \ljump{\term[1]}_{\vsig}^*, &
    \ljump{\term[1] \intersection \term[2]}_{\vsig} &\defeq \ljump{\term[1]}_{\vsig} \cap \ljump{\term[2]}_{\vsig}.
\end{align*}
For \kl{Kleene algebra terms} (i.e., terms over $\set{\id, \emp, \compo, \union, \bl^{*}}$) $\term[1]$ and $\term[2]$, the following is well-known (see, \eg, \cite[Theorem 4]{pousPositiveCalculusRelations2018}):
$\REL \models \term[1] \le \term[2]$ iff $\ljump{\term[1]}_{\vsig} \subseteq \ljump{\term[2]}_{\vsig}$.
This equivalence can be slightly strengthened as follows.
\begin{propC}[\cite{nakamuraPartialDerivativesGraphs2017,brunetPetriAutomata2017}]\label{proposition: glang and lang}
    For all \kl{Kleene algebra terms} $\term[1]$ and \kl{Kleene lattice terms} $\term[2]$, we have:
    \[\REL \models \term[1] \le \term[2] \;\iff\; \ljump{\term[1]}_{\vsig} \subseteq \ljump{\term[2]}_{\vsig}.\]
\end{propC}
\begin{proof}
    (In \cite[Theorem VI.3]{nakamuraPartialDerivativesGraphs2017}, this was shown via the tree unwinding technique in modal logic. The following is a simpler proof \cite[Proposition 10.2]{brunetPetriAutomata2017} based on the \kl{graph language} characterization.)
    For all \kl{Kleene lattice terms} $\term[2]$ and \kl{words} $\word \in \vsig^*$, we have:
    \[\exists \graph[2] \in \glang(\term[2]), \graph[2] \homo \const{G}(\word) \;\iff\; \word \in \ljump{\term[2]}_{\vsig}\]
    where $\const{G}(\word)$ denotes the unique \kl{graph} such that $\glang(\word) = \set{\const{G}(\word)}$.
    It is shown by easy induction on $\term[2]$.
    Thus by \Cref{proposition: glang and equational theory}, this completes the proof.
\end{proof}

\begin{rem}\label{rem: langauge model}
    The (word) language model $\LANG$ \cite{bloomNotesEquationalTheories1995,andrekaEquationalTheoryKleene2011,brunetReversibleKleeneLattices2017} is another interesting model of \kl{PCoR* terms}, where
    the constant $\top$ is interpreted as the universal language,
    the unary operator $\bl^{\smile}$ is interpreted as the reverse of the language: $L^{\smile} \defeq \set{a_n \dots a_2 a_1 \mid a_1 a_2 \dots a_n \in L}$,
    and the other operators are defined based on the above $\ljump{\bl}_{\vsig}$.
    In $\LANG$, each variable $x$ may interpret not only $\set{x}$ but also a (possibly non-singleton) language.
    Notably for \kl{identity-free Kleene lattices},
    the \kl{equational theory} with respect to $\LANG$ coincides with that with respect to $\REL$ \cite[Theorem 4.1]{andrekaEquationalTheoryKleene2011}.
    However, for \kl{Kleene lattices} and \kl{PCoR*}, this coincidence is broken, \eg, \Cref{equation: allegories 2,equation: PCoR* 1,equation: PCoR* 2} are counterexamples.
    For the \kl{equational theory} of \kl{Kleene lattices} (with complement),
    $\LANG$ can be characterized as a subclass of (linearly ordered) (generalized) relational models \cite{nakamuraFiniteRelationalSemantics2025}.
\end{rem}

\begin{rem}
The restriction of \Cref{proposition: glang and lang} is crucial.
\Cref{proposition: glang and lang} fails
when $\intersection$, $\top$, or $\bl^{\smile}$ occurs in $\term[1]$.
For instance, if $\term[1] = a \intersection b$ and $\term[2] = \emp$ where $a \neq b$, then we have $\ljump{\term[1]}_{\vsig} = \ljump{\term[2]}_{\vsig} = \emptyset$ but $a \intersection b \le \emp$ does not hold with respect to $\REL$.
\Cref{proposition: glang and lang} also fails when $\top$ or $\bl^{\smile}$ occurs in $\term[2]$ (where $\top$ expresses the universal language and $\bl^{\smile}$ expresses the reverse, according to \Cref{rem: langauge model});
let us consider, $\REL \models a \le a \top a$ or $\REL \models a \le a a^{\smile} a$.
In such cases, \kl{graph languages} are needed (\Cref{proposition: glang and equational theory}).
\end{rem}

\subsection{Kleene lattice terms: a core fragment of PCoR* terms}\label{section: Kleene lattice term}
In the sequel, we mainly consider \kl[KL terms]{Kleene lattice (KL) terms}.
We say that a \kl{PCoR* term} $\term$ is a \intro*\kl{KL term} if $\term$ contains neither $\top$ nor $\bl^{\smile}$.
Namely, the set of \kl{KL terms} is given by:
\begin{align*}
    \term[1], \term[2], \term[3] \;\Coloneqq\; \aterm \mid \id \mid \emp \mid \term[1] \compo \term[2] \mid \term[1] \union \term[2] \mid \term[1] \intersection \term[2] \mid \term[1]^*. \tag{$a \in \vsig$}
\end{align*}

\subsubsection{Second alternative semantics: run (DAG with vertices of one fan-in and one fan-out) languages}\label{section: runs}
We consider \kl{run languages}, which is a slightly modified semantics of \kl{graph languages}.
This is inspired by the runs in \kl{branching automata} \cite{lodayaSeriesparallelPosetsAlgebra1998,lodayaSeriesParallelLanguages2000,lodayaRationalityAlgebrasSeries2001}.
We use the symbols $\fork$ and $\join$ for expressing the forking and joining of runs, respectively.
Let $\vsig_{\fork, \join}$ be the set $\vsig \cup \set{\fork, \join}$ with the map $\ty$ defined by:
$\ty(a) = \tuple{1, 1} \mbox{ for } a \in \vsig$,
$\ty(\fork) = \tuple{1, 2}$, and
$\ty(\join) = \tuple{2, 1}$.
A \intro*\kl{run with $\tuple{n,m}$-interface} $\graph$ is a \kl{graph} with $\tuple{n, m}$-interface over $\vsig_{\fork, \join}$ where
\begin{itemize}
    \item $\graph$ forms a directed acyclic graph (DAG), namely, the following holds:
        \[\left(\bigcup_{a \in \vsig_{\fork, \join}}\bigcup_{i \in \range{\ty_1(a)}}\bigcup_{j \in \range{\ty_2(a)}} \set{\tuple{x_i, y_j} \mid \tuple{x_1, \dots, x_{\ty_1(a)}, y_1, \dots, y_{\ty_2(a)}} \in a^{\graph}} \right)^+ \cap \triangle_{\univ{\graph}} = \emptyset.\]
        \item $\mathsf{in}_{\graph}(z) = \mathsf{out}_{\graph}(z) = 1$ for each $z \in \univ{\graph}$, where
        \hspace{2em}
        \begin{align*}
            \mathsf{in}_{\graph}(z) &\defeq \card \set{l \in \range{n} \mid \src_l^{\graph} = z}
            +\!\!\!\! \sum_{\substack{a \in \vsig_{\fork, \join}\\ i \in \range{\ty_1(a)}}}\!\!\!\!  \card \set{\tuple{x_1,\dots, x_{\ty_1(a)}, y_1, \dots, y_{\ty_2(a)}} \! \in \! a^{\graph} \mid x_i = z},  \\
            \mathsf{out}_{\graph}(z) &\defeq \card \set{r \in \range{m} \mid \tgt_r^{\graph} = z}
            +\!\!\!\!  \sum_{\substack{a \in \vsig_{\fork, \join}\\ j \in \range{\ty_2(a)}}}\!\!\!\!  \card \set{\tuple{x_1,\dots, x_{\ty_1(a)}, y_1, \dots, y_{\ty_2(a)}} \! \in \! a^{\graph} \mid y_j = z}.
        \end{align*}
\end{itemize}

\begin{defi}\label{definition: run language}
    The \intro*\kl{run language} $\run(\term)$ of a \kl{KL term} $\term$ is a set of \kl[run with $\tuple{n,m}$-interface]{runs with $\tuple{1,1}$-interface}, defined as follows:
    \begin{align*}
        \run(\aterm)                       & \defeq \set*{
        \begin{tikzpicture}[baseline = -.5ex]
                \graph[grow right = 1.cm, branch down = 6ex, nodes={mynode}]{
                {0/{}[draw, circle]}-!-{1/{}[draw, circle]}
                };
                \node[left = .5em of 0](l){};
                \node[right = .5em of 1](r){};
                \graph[use existing nodes, edges={color=black, pos = .5, earrow}, edge quotes={fill=white, inner sep=1pt,font= \scriptsize}]{
                0 ->["$\aterm$"] 1;
                l -> 0; 1 -> r;
                };
            \end{tikzpicture}}                                     \quad \mbox{for $\aterm \in \vsig$}, &
        \run(\id)                    & \defeq \set*{
            \begin{tikzpicture}[baseline = -.5ex]
                    \graph[grow right = 1.cm, branch down = 6ex, nodes={mynode}]{
                    {1/{}[draw, circle]}
                    };
                    \node[left = .5em of 1](l){};
                    \node[right = .5em of 1](r){};
                    \graph[use existing nodes, edges={color=black, pos = .5, earrow}, edge quotes={fill=white, inner sep=1pt,font= \scriptsize}]{
                        l -> 1; 1 -> r;
                    };
                \end{tikzpicture}}, \hspace{.5em}
        \run(\emp) \defeq \emptyset,                                                                                   \\
       \run(\term[1] \compo \term[2]) & \defeq \set*{ \begin{tikzpicture}[baseline = -.5ex]
            \graph[grow right = 1.cm, branch down = 2.5ex, nodes={mynode}]{
            {s1/{}[draw, circle]}
            -!- {mid/{}[draw, circle]}
            -!- {t1/{}[draw, circle]}
            };
            \node[left = 4pt of s1](s1l){} edge[earrow, ->] (s1);
            \node[right = 4pt of t1](t1l){}; \path (t1) edge[earrow, ->] (t1l);
            \path (s1) edge [draw = white, opacity = 0] node[pos= .5, elabel](tau12){$\graph[1]$}(mid);
            \path (mid) edge [draw = white, opacity = 0] node[pos= .5, elabel](tau21){$\graph[2]$}(t1);
            \graph[use existing nodes, edges={color=black, pos = .5, earrow}, edge quotes={fill=white, inner sep=1pt,font= \scriptsize}]{
            s1 -- {tau12} -> {mid};
            {mid} -- {tau21} -> t1;
            };
        \end{tikzpicture} \mid \graph[1] \in \run(\term[1]) \land \graph[2] \in \run(\term[2])},                                                              &
        \run(\term[1] \union \term[2])  & \defeq \run(\term[1]) \cup \run(\term[2]),                                                          \\
        \run(\term[1] \intersection \term[2])  & \defeq
        \set*{\hspace{-.5em}\begin{tikzpicture}[baseline = -2.ex]
            \graph[grow right = .8cm, branch down = 1.5ex, nodes={}]{
            {/, s1/{}[mynode, draw, circle]}
            -!- {s2/{}[mynode, draw, circle], /, s3/{}[mynode, draw, circle]}
            -!- {t2/{}[mynode, draw, circle], /, t3/{}[mynode, draw, circle]}
            -!- {/, t1/{}[mynode, draw, circle]}
            };
            \node[left = 4pt of s1](s1l){} edge[earrow, ->] (s1);
            \node[right = 4pt of t1](t1l){}; \path (t1) edge[earrow, ->] (t1l);
            \path (s2) edge [draw = white, opacity = 0] node(smid){} (s3);
            \path (t2) edge [draw = white, opacity = 0] node(tmid){} (t3);
            \path (s1) edge [draw = white, opacity = 0] node[pos= .5, elabel](tau12){$\fork$}(smid);
            \path (tmid) edge [draw = white, opacity = 0] node[pos= .5, elabel](tau21){$\join$}(t1);
            \graph[use existing nodes, edges={color=black, pos = .5, earrow}, edge quotes={fill=white, inner sep=1pt,font= \scriptsize}]{
            s1 -- tau12; tau12 ->["$1$"{auto}] s2; tau12 ->["$2$"{auto, below = .5ex}] s3;
            t2 --["$1$" auto] tau21; t3 --["$2$"{auto, below = .5ex}] tau21; tau21 -> t1;
            s2 ->["$\graph[1]$"] t2;
            s3 ->["$\graph[2]$"] t3;
            };
        \end{tikzpicture}\hspace{-.5em} \mid \graph[1] \in \run(\term[1]) \land \graph[2] \in \run(\term[2])}, &
        \run(\term[1]^{*})            & \defeq
        \bigcup_{n \ge 0} \run(\term[1]^n).    %
    \end{align*}
\end{defi}
By construction, every \kl{run} of \kl{KL terms} is a directed two-terminal series-parallel graph \cite{valdesRecognitionSeriesParallel1979,eppsteinParallelRecognitionSeriesparallel1992}, by viewing labels $\fork$ and $\join$ as vertices.
For instance (\cf\ \Cref{section: graph languages}),
\begin{align*}
    \run(\aterm[1] \intersection (\aterm[2] \aterm[3])) \;&=\; \set*{\begin{tikzpicture}[baseline = -3.ex]
                 \graph[grow right = 1.cm, branch down = 2.5ex, nodes={font = \scriptsize}]{
                 {/, s1/{}[mynode, draw, circle]}
                 -!- {c1/[mynode, draw, circle], /, d1/[mynode, draw, circle]}
                 -!- {c2/[mynode, draw, circle], /, d2/[mynode, draw, circle]}
                 -!- {c3/, /, d3/[mynode, draw, circle]}
                 -!- {/, t1/{}[mynode, draw, circle]}
                 };
                 \node[left = 4pt of s1](s1l){} edge[earrow, ->] (s1);
                 \node[right = 4pt of t1](t1l){}; \path (t1) edge[earrow, ->] (t1l);
                 \path (c1) edge [draw = white, opacity = 0] node(mid1){} (d1);
                 \path (s1) edge [draw = white, opacity = 0] node(f)[opacity = 1,elabel]{$\fork$} (mid1);
                 \path (c3) edge [draw = white, opacity = 0] node(mid3){} (d3);
                 \path (mid3) edge [draw = white, opacity = 0] node(j)[opacity = 1,elabel]{$\join$} (t1);
                 \graph[use existing nodes, edges={color=black, pos = .5, earrow}, edge quotes={fill=white, inner sep=1pt,font= \scriptsize}]{
                 s1 -- f; f ->["$1$" auto] c1; f ->["$2$" {auto, below left = .3ex}] d1;
                 c1 ->["$\aterm$"] c2;
                 d1 ->["$\aterm[2]$"] d2 -> ["$\aterm[3]$"] d3;
                 c2 -- ["$1$"{auto}, pos = .8] j; d3 -- ["$2$" {auto, below = .5ex}] j; j -> t1;
                 };
             \end{tikzpicture}},\\
    \run((((\aterm[1] \aterm[1]) \intersection \id) \aterm[2] \aterm[1] \aterm[3]) \union \aterm[4]) \;&=\; \set*{\begin{tikzpicture}[baseline = -3.ex]
    \graph[grow right = .8cm, branch down = 2.5ex, nodes={font = \scriptsize}]{
    {/, s1/{}[mynode, draw, circle]}
    -!- {c1/{}[mynode, draw, circle], /, d1/{}[mynode, draw, circle]}
    -!- {c2/{}[mynode, draw, circle], /, d2/{}}
    -!- {c3/{}[mynode, draw, circle], /, d3/{}}
    -!- {/, c4/{}[mynode, draw, circle]}
    -!- {/, c5/{}[mynode, draw, circle]}
    -!- {/, c6/{}[mynode, draw, circle]}
    -!- {/, t1/{}[mynode, draw, circle]}
    };
    \node[left = 4pt of s1](s1l){} edge[earrow, ->] (s1);
    \node[right = 4pt of t1](t1l){}; \path (t1) edge[earrow, ->] (t1l);
    \path (c1) edge [draw = white, opacity = 0] node(mid1){} (d1);
    \path (s1) edge [draw = white, opacity = 0] node(f)[opacity = 1,elabel]{$\fork$} (mid1);
    \path (c3) edge [draw = white, opacity = 0] node(mid3){} (d3);
    \path (mid3) edge [draw = white, opacity = 0] node(j)[opacity = 1,elabel]{$\join$} (c4);
\graph[use existing nodes, edges={color=black, pos = .5, earrow}, edge quotes={fill=white, inner sep=1pt,font= \scriptsize}]{
    s1 -- f; f ->["$1$" auto] c1; f->["$2$"{auto, below left = .3ex}] d1;
    c1 ->["$\aterm$"] c2;
    c2 ->["$\aterm$"] c3;
    c3 --["$1$"{auto}] j; d1 --["$2$"{auto, below = .5ex},  pos = .8] j; j -> c4;
    c4 ->["$\aterm[2]$"] c5 -> ["$\aterm[1]$"] c6 ->["$\aterm[3]$"] t1;
    };
    \end{tikzpicture}, \begin{tikzpicture}[baseline = -.5ex]
        \graph[grow right = .8cm, branch down = 6ex, nodes={mynode}]{
        {0/{}[draw, circle]}-!-{1/{}[draw, circle]}
        };
        \node[left = .5em of 0](l){};
        \node[right = .5em of 1](r){};
        \graph[use existing nodes, edges={color=black, pos = .5, earrow}, edge quotes={fill=white, inner sep=1pt,font= \scriptsize}]{
        0 ->["$\aterm[4]$"] 1;
        l -> 0; 1 -> r;
        };
    \end{tikzpicture}},\\
    \run(\aterm^*) \;&=\; \set*{
        \begin{tikzpicture}[baseline = -.5ex]
            \graph[grow right = .8cm, branch down = 6ex, nodes={mynode}]{
            {1/{}[draw, circle]}
            };
            \node[left = .5em of 1](l){};
            \node[right = .5em of 1](r){};
            \graph[use existing nodes, edges={color=black, pos = .5, earrow}, edge quotes={fill=white, inner sep=1pt,font= \scriptsize}]{
                l -> 1; 1 -> r;
            };
        \end{tikzpicture}, 
        \begin{tikzpicture}[baseline = -.5ex]
                \graph[grow right = .8cm, branch down = 6ex, nodes={mynode}]{
                {0/{}[draw, circle]}-!-{1/{}[draw, circle]}
                };
                \node[left = .5em of 0](l){};
                \node[right = .5em of 1](r){};
                \graph[use existing nodes, edges={color=black, pos = .5, earrow}, edge quotes={fill=white, inner sep=1pt,font= \scriptsize}]{
                0 ->["$\aterm$"] 1;
                l -> 0; 1 -> r;
                };
            \end{tikzpicture},
            \begin{tikzpicture}[baseline = -.5ex]
                \graph[grow right = .8cm, branch down = 6ex, nodes={mynode}]{
                {0/{}[draw, circle]}-!-{1/{}[draw, circle]}-!-{2/{}[draw, circle]}
                };
                \node[left = .5em of 0](l){};
                \node[right = .5em of 2](r){};
                \graph[use existing nodes, edges={color=black, pos = .5, earrow}, edge quotes={fill=white, inner sep=1pt,font= \scriptsize}]{
                0 ->["$\aterm$"] 1 ->["$\aterm$"] 2;
                l -> 0; 2 -> r;
                };
            \end{tikzpicture}, \dots }.
\end{align*}
\noindent 
We can also use the \kl{run languages} as an alternative relational semantics.
Let $\struc$ be a \kl{structure}.
For $\vec{x}, \vec{y} \in \univ{\struc}^{*}$,
let $\tilde{\const{G}}(\struc, \vec{x}, \vec{y}) \defeq \tuple{\univ{\struc}, \set{a^{\tilde{\const{G}}(\struc, \vec{x}, \vec{y})}}_{a \in \vsig_{\fork, \join}}, \vec{x}, \vec{y}}$ where
$a^{\tilde{\const{G}}(\struc, \vec{x}, \vec{y})} = \jump{a}_{\struc} \mbox{ for } a \in \vsig$
and
$\fork^{\tilde{\const{G}}(\struc, \vec{x}, \vec{y})} = \join^{\tilde{\const{G}}(\struc, \vec{x}, \vec{y})} = \set{\tuple{z, z, z} \mid z \in \univ{\struc}}$.

For a \kl{run} $\graph[2]$ and a set $\glang$ of \kl{runs},
we define $\jump{\graph[2]}_{\struc}$ and $\jump{\glang}_{\struc}$ as follows:
\begin{align*}
    \jump{\graph[2]}_{\struc} & \defeq \set{\vec{x} \vec{y} \mid \vec{x} \in \univ{\struc}^{\ty_1(\graph[2])} \land \vec{y} \in \univ{\struc}^{\ty_2(\graph[2])} \land  \graph[2] \homo \tilde{\const{G}}(\struc, \vec{x}, \vec{y})}, & \jump{\glang}_{\struc} \defeq \bigcup_{\graph[2] \in \glang} \jump{\graph[2]}_{\struc}.
\end{align*}
We then have the following as with \Cref{proposition: graph language}.
\begin{prop}\label{proposition: run language}
    Let $\struc$ be a \kl{structure}.
    For all \kl{KL terms} $\term$, we have $\jump{\term}_{\struc} = \jump{\run(\term)}_{\struc}$.
\end{prop}
\begin{proof}
    By easy induction on $\term$ using distributivity of $\compo$ and $\intersection$.
\end{proof}

\subsubsection{Compositions on runs}\label{section: compositions of runs}
We recall the \kl{series composition} $\series$, the \kl{parallel composition} $\parallel$, and the identity elements $\id^{n}$ (\Cref{definition: series composition and parallel composition}).
The set of \intro*\kl{atomic runs} is defined as $\bar{\vsig} \defeq \set{b^{n}_{i} \mid b \in \vsig_{\fork, \join}, n \ge 1, i \in \range{n}}$ where
$a^{n}_{i} \defeq \id^{i-1} \parallel \left(\begin{tikzpicture}[baseline = -.5ex]
    \graph[grow right = 1.cm, branch down = 6ex, nodes={mynode}]{
    {0/{}[draw, circle]} -!-
    {1/{}[draw, circle]}
    };
    \node[left = .5em of 0](l){};
    \node[right = .5em of 1](r){};
    \graph[use existing nodes, edges={color=black, pos = .5, earrow}, edge quotes={fill=white, inner sep=1pt,font= \scriptsize}]{
        l -> 0 ->["$a$"] 1 -> r;
    };
\end{tikzpicture}\right) \parallel \id^{n - i}$ for $a \in \vsig$,
$\fork^{n}_{i} \defeq \id^{i-1} \parallel \left(\hspace{-.5em}\begin{gathered}
    \begin{tikzpicture}[baseline = -.5ex]
        \graph[grow right = .5cm, branch down = 2ex, nodes={}]{
        {/, 0/{}[draw,circle, mynode]} -!-
        {/, c/{\scriptsize $\fork$}[elabel]} -!-
        {10/{}[draw, circle, mynode], /, 11/{}[draw, circle, mynode]}
        };
        \node[left = .5em of 0](l){};
        \node[right = .5em of 10](r0){\tiny $1$};
        \node[right = .5em of 11](r1){\tiny $2$};
        \graph[use existing nodes, edges={color=black, pos = .5, earrow}, edge quotes={fill=white, inner sep=1pt,font= \scriptsize}]{
        l -> 0 -- c;
        c ->["$1$" auto] 10; c->["$2$"{auto, below left = .3ex}] 11;
        10 -> r0; 11 -> r1;
        };
    \end{tikzpicture}
    \end{gathered}\right) \parallel \id^{n - i}$, and
$\join^{n}_{i} \defeq \id^{i-1} \parallel \left(
    \begin{gathered}
        \begin{tikzpicture}[baseline = -.5ex]
        \graph[grow right = .5cm, branch down = 2ex, nodes={}]{
        {00/{}[draw, circle, mynode], /,01/{}[draw,circle, mynode]} -!-
        {/, c/{\scriptsize $\join$}[elabel]} -!-
        {/, 1/{}[draw, circle, mynode]}
        };
        \node[left = .5em of 00](l0){\tiny $1$};
        \node[left = .5em of 01](l1){\tiny $2$};
        \node[right = .5em of 1](r){\tiny };
        \graph[use existing nodes, edges={color=black, pos = .5, earrow}, edge quotes={fill=white, inner sep=1pt,font= \scriptsize}]{
            l0 -> 00; l1 -> 01;
            00 --["$1$" auto] c; 01--["$2$"{auto, below right = .3ex}] c;
            c -> 1 -> r;
        };
    \end{tikzpicture}
\end{gathered}\hspace{-.5em}\right) \parallel \id^{n - i}$.
Namely, $a^{n}_{i}$, $\fork^{n}_{i}$, and $\join^{n}_{i}$ are given as follows:
\begin{align*}
    a^{n}_{i} &= \left(
        \begin{gathered}
            \begin{tikzpicture}[baseline = -.5ex]
                \graph[grow right = 1.cm, branch down = 6ex, nodes={mynode}]{
                {1/{}[draw, circle]}
                };
                \node[left = .5em of 1](l){\tiny $1$};
                \node[right = .5em of 1](r){\tiny $1$};
                \graph[use existing nodes, edges={color=black, pos = .5, earrow}, edge quotes={fill=white, inner sep=1pt,font= \scriptsize}]{
                    l -> 1;
                    1 -> r;
                };
            \end{tikzpicture}\\
            \vdots\\
            \begin{tikzpicture}[baseline = -.5ex]
                \graph[grow right = 1.cm, branch down = 6ex, nodes={mynode}]{
                {1/{}[draw, circle]}
                };
                \node[left = .5em of 1](l){\tiny $i-1$};
                \node[right = .5em of 1](r){\tiny $i-1$};
                \graph[use existing nodes, edges={color=black, pos = .5, earrow}, edge quotes={fill=white, inner sep=1pt,font= \scriptsize}]{
                    l -> 1;
                    1 -> r;
                };
            \end{tikzpicture}\\
            \begin{tikzpicture}[baseline = -.5ex]
            \graph[grow right = 1.cm, branch down = 6ex, nodes={mynode}]{
            {0/{}[draw, circle]} -!-
            {1/{}[draw, circle]}
            };
            \node[left = .5em of 0](l){\tiny $i$};
            \node[right = .5em of 1](r){\tiny $i$};
            \graph[use existing nodes, edges={color=black, pos = .5, earrow}, edge quotes={fill=white, inner sep=1pt,font= \scriptsize}]{
                l -> 0 ->["$a$"] 1 -> r;
            };
        \end{tikzpicture}\\
        \begin{tikzpicture}[baseline = -.5ex]
            \graph[grow right = 1.cm, branch down = 6ex, nodes={mynode}]{
            {1/{}[draw, circle]}
            };
            \node[left = .5em of 1](l){\tiny $i+1$};
            \node[right = .5em of 1](r){\tiny $i+1$};
            \graph[use existing nodes, edges={color=black, pos = .5, earrow}, edge quotes={fill=white, inner sep=1pt,font= \scriptsize}]{
                l -> 1;
                1 -> r;
            };
        \end{tikzpicture}\\
        \vdots\\
        \begin{tikzpicture}[baseline = -.5ex]
            \graph[grow right = 1.cm, branch down = 6ex, nodes={mynode}]{
            {1/{}[draw, circle]}
            };
            \node[left = .5em of 1](l){\tiny $n$};
            \node[right = .5em of 1](r){\tiny $n$};
            \graph[use existing nodes, edges={color=black, pos = .5, earrow}, edge quotes={fill=white, inner sep=1pt,font= \scriptsize}]{
                l -> 1;
                1 -> r;
            };
        \end{tikzpicture}
    \end{gathered}\right), &
    \fork^{n}_{i} &=  \left(
        \begin{gathered}
            \begin{tikzpicture}[baseline = -.5ex]
                \graph[grow right = 1.cm, branch down = 6ex, nodes={mynode}]{
                {1/{}[draw, circle]}
                };
                \node[left = .5em of 1](l){\tiny $1$};
                \node[right = .5em of 1](r){\tiny $1$};
                \graph[use existing nodes, edges={color=black, pos = .5, earrow}, edge quotes={fill=white, inner sep=1pt,font= \scriptsize}]{
                    l -> 1;
                    1 -> r;
                };
            \end{tikzpicture}\\
            \vdots\\
            \begin{tikzpicture}[baseline = -.5ex]
                \graph[grow right = 1.cm, branch down = 6ex, nodes={mynode}]{
                {1/{}[draw, circle]}
                };
                \node[left = .5em of 1](l){\tiny $i-1$};
                \node[right = .5em of 1](r){\tiny $i-1$};
                \graph[use existing nodes, edges={color=black, pos = .5, earrow}, edge quotes={fill=white, inner sep=1pt,font= \scriptsize}]{
                    l -> 1;
                    1 -> r;
                };
            \end{tikzpicture}\\
            \begin{tikzpicture}[baseline = -.5ex]
            \graph[grow right = .5cm, branch down = 2ex, nodes={}]{
            {/, 0/{}[draw,circle, mynode]} -!-
            {/, c/{\scriptsize $\fork$}[elabel]} -!-
            {10/{}[draw, circle, mynode], /, 11/{}[draw, circle, mynode]}
            };
            \node[left = .5em of 0](l){\tiny $i$};
            \node[right = .5em of 10](r0){\tiny $i$};
            \node[right = .5em of 11](r1){\tiny $i+1$};
            \graph[use existing nodes, edges={color=black, pos = .5, earrow}, edge quotes={fill=white, inner sep=1pt,font= \scriptsize}]{
            l -> 0 -- c;
            c ->["$1$" auto] 10; c->["$2$"{auto, below left = .3ex}] 11;
            10 -> r0; 11 -> r1;
            };
        \end{tikzpicture}\\
        \begin{tikzpicture}[baseline = -.5ex]
            \graph[grow right = 1.cm, branch down = 6ex, nodes={mynode}]{
            {1/{}[draw, circle]}
            };
            \node[left = .5em of 1](l){\tiny $i+1$};
            \node[right = .5em of 1](r){\tiny $i+2$};
            \graph[use existing nodes, edges={color=black, pos = .5, earrow}, edge quotes={fill=white, inner sep=1pt,font= \scriptsize}]{
                l -> 1;
                1 -> r;
            };
        \end{tikzpicture}\\
        \vdots\\
        \begin{tikzpicture}[baseline = -.5ex]
            \graph[grow right = 1.cm, branch down = 6ex, nodes={mynode}]{
            {1/{}[draw, circle]}
            };
            \node[left = .5em of 1](l){\tiny \phantom{${}+1$}$n$};
            \node[right = .5em of 1](r){\tiny $n+1$};
            \graph[use existing nodes, edges={color=black, pos = .5, earrow}, edge quotes={fill=white, inner sep=1pt,font= \scriptsize}]{
                l -> 1;
                1 -> r;
            };
        \end{tikzpicture}
    \end{gathered}\right), &
    \join^{n}_{i} &= \left(
        \begin{gathered}
            \begin{tikzpicture}[baseline = -.5ex]
                \graph[grow right = 1.cm, branch down = 6ex, nodes={mynode}]{
                {1/{}[draw, circle]}
                };
                \node[left = .5em of 1](l){\tiny $1$};
                \node[right = .5em of 1](r){\tiny $1$};
                \graph[use existing nodes, edges={color=black, pos = .5, earrow}, edge quotes={fill=white, inner sep=1pt,font= \scriptsize}]{
                    l -> 1;
                    1 -> r;
                };
            \end{tikzpicture}\\
            \vdots\\
            \begin{tikzpicture}[baseline = -.5ex]
                \graph[grow right = 1.cm, branch down = 6ex, nodes={mynode}]{
                {1/{}[draw, circle]}
                };
                \node[left = .5em of 1](l){\tiny $i-1$};
                \node[right = .5em of 1](r){\tiny $i-1$};
                \graph[use existing nodes, edges={color=black, pos = .5, earrow}, edge quotes={fill=white, inner sep=1pt,font= \scriptsize}]{
                    l -> 1;
                    1 -> r;
                };
            \end{tikzpicture}\\
            \begin{tikzpicture}[baseline = -.5ex]
            \graph[grow right = .5cm, branch down = 2ex, nodes={}]{
            {00/{}[draw, circle, mynode], /,01/{}[draw,circle, mynode]} -!-
            {/, c/{\scriptsize $\join$}[elabel]} -!-
            {/, 1/{}[draw, circle, mynode]}
            };
            \node[left = .5em of 00](l0){\tiny $i$};
            \node[left = .5em of 01](l1){\tiny $i+1$};
            \node[right = .5em of 1](r){\tiny $i$};
            \graph[use existing nodes, edges={color=black, pos = .5, earrow}, edge quotes={fill=white, inner sep=1pt,font= \scriptsize}]{
                l0 -> 00; l1 -> 01;
                00 --["$1$" auto] c; 01--["$2$"{auto, below right = .3ex}] c;
                c -> 1 -> r;
            };
        \end{tikzpicture}\\
        \begin{tikzpicture}[baseline = -.5ex]
            \graph[grow right = 1.cm, branch down = 6ex, nodes={mynode}]{
            {1/{}[draw, circle]}
            };
            \node[left = .5em of 1](l){\tiny $i+2$};
            \node[right = .5em of 1](r){\tiny $i+1$};
            \graph[use existing nodes, edges={color=black, pos = .5, earrow}, edge quotes={fill=white, inner sep=1pt,font= \scriptsize}]{
                l -> 1;
                1 -> r;
            };
        \end{tikzpicture}\\
        \vdots\\
        \begin{tikzpicture}[baseline = -.5ex]
            \graph[grow right = 1.cm, branch down = 6ex, nodes={mynode}]{
            {1/{}[draw, circle]}
            };
            \node[left = .5em of 1](l){\tiny $n+1$};
            \node[right = .5em of 1](r){\tiny $n$\phantom{${}+1$}};
            \graph[use existing nodes, edges={color=black, pos = .5, earrow}, edge quotes={fill=white, inner sep=1pt,font= \scriptsize}]{
                l -> 1;
                1 -> r;
            };
        \end{tikzpicture}
    \end{gathered}\right).
\end{align*}
Every \kl{run} can be expressed as \kl{series compositions} of \kl{atomic runs}.
\begin{prop}\label{proposition: run atomic}
    For every \kl{run} $\graph$ with $n \ge 1$ edges,
    there exist some $a_1, \dots, a_n \in \bar{\vsig}$ such that $\graph = a_1 \series \dots \series a_n$.
\end{prop}
\begin{proof}
    By easy induction on $n$.
\end{proof}
For instance,
$\left(\begin{tikzpicture}[baseline = -3.ex]
    \graph[grow right = .8cm, branch down = 2.5ex, nodes={font = \scriptsize}]{
    {/, s1/{}[mynode, draw, circle]}
    -!- {c1/[mynode, draw, circle], /, d1/[mynode, draw, circle]}
    -!- {c2/[mynode, draw, circle], /, d2/[mynode, draw, circle]}
    -!- {c3/, /, d3/[mynode, draw, circle]}
    -!- {/, t1/{}[mynode, draw, circle]}
    };
    \node[left = 4pt of s1](s1l){} edge[earrow, ->] (s1);
    \node[right = 4pt of t1](t1l){}; \path (t1) edge[earrow, ->] (t1l);
    \path (c1) edge [draw = white, opacity = 0] node(mid1){} (d1);
    \path (s1) edge [draw = white, opacity = 0] node(f)[opacity = 1,elabel]{$\fork$} (mid1);
    \path (c3) edge [draw = white, opacity = 0] node(mid3){} (d3);
    \path (mid3) edge [draw = white, opacity = 0] node(j)[opacity = 1,elabel]{$\join$} (t1);
    \graph[use existing nodes, edges={color=black, pos = .5, earrow}, edge quotes={fill=white, inner sep=1pt,font= \scriptsize}]{
    s1 -- f;
    f ->["$1$" auto] c1;
    f ->["$2$"{auto, below left = .3ex}] d1;
    c1 ->["$\aterm$"] c2;
    d1 ->["$\aterm[2]$"] d2 -> ["$\aterm[3]$"] d3;
    c2 --["$1$"{auto}, pos = .8] j;
    d3 --["$2$"{auto, below right = .3pt}] j;
    j -> t1;
    };
\end{tikzpicture}\right) = \fork^{1}_{1} \series \aterm^{2}_{1} \series \aterm[2]^{2}_{2} \series \aterm[3]^{2}_{2} \series \join^{1}_{1}$.
This decomposition is not unique if the ``topological sort of edges'' of the \kl{run} (i.e., the linear ordering of edges such that for every \kl{vertex}, each incoming edge of the \kl{vertex} comes before each outgoing edge of the \kl{vertex} in the ordering) is not unique.
For instance,
$\fork^{1}_{1} \series \aterm[2]^{2}_{2} \series \aterm^{2}_{1} \series \aterm[3]^{2}_{2} \series \join^{1}_{1}$ and
$\fork^{1}_{1} \series \aterm[2]^{2}_{2} \series \aterm[3]^{2}_{2} \series \aterm^{2}_{1} \series \join^{1}_{1}$ also express the same \kl{run} above.

\subsubsection{Left quotients of runs}\label{section: left quotients of runs}
For \kl{runs}, we define \emph{\kl{left quotients}} with respect to \kl{series compositions}.
For \kl{runs} $\graph[1]$, $\graph[2]$, and $\graph[3]$,
we write:
\[\graph[1] \longrightarrow_{\graph[3]} \graph[2] \;\defiff\;  \mbox{$\graph[3] \series \graph[2]$ is defined and } \graph[1] = \graph[3] \series \graph[2].\]
The \intro*\kl{left quotient} $\D_{\graph[3]}(\graph)$ of a \kl{run} $\graph$ with respect to a \kl{run} $\graph[3]$ is the \kl{run language} defined by:
\begin{align*}
    \D_{\graph[3]}(\graph[1]) &\;\defeq\; \set{\graph[2] \mid \graph[1] \longrightarrow_{\graph[3]} \graph[2]}. 
\end{align*}
For a \kl{run language} $\glang$, let $\D_{\graph[3]}(\glang) \defeq \bigcup_{\graph[1] \in \glang} \D_{\graph[3]}(\graph[1])$.
For instance, when $\graph[3] = \left(\begin{tikzpicture}[baseline = -3.ex]
    \graph[grow right = .8cm, branch down = 2.5ex, nodes={font = \scriptsize}]{
    {/, s1/{}[mynode, draw, circle]}
    -!- {c1/[mynode, draw, circle], /, d1/[mynode, draw, circle]}
    -!- {c2/, /, d2/[mynode, draw, circle]}
    };
    \node[left = 4pt of s1](s1l){} edge[earrow, ->] (s1);
    \node[right = 4pt of c1](c1r){\tiny $1$}; \path (c1) edge[earrow, ->] (c1r);
    \node[right = 4pt of d2](d2r){\tiny $2$}; \path (d2) edge[earrow, ->] (d2r);
    \path (c1) edge [draw = white, opacity = 0] node(mid1){} (d1);
    \path (s1) edge [draw = white, opacity = 0] node(f)[opacity = 1,elabel]{$\fork$} (mid1);
    \graph[use existing nodes, edges={color=black, pos = .5, earrow}, edge quotes={fill=white, inner sep=1pt,font= \scriptsize}]{
    s1 -- f;
    f ->["$1$"{auto}] c1;
    f ->["$2$"{auto, below left = .3ex}] d1;
    d1 ->["$\aterm[2]$"] d2;
    };
\end{tikzpicture}\right)$, we have:
\[\D_{\graph[3]}\left(\begin{tikzpicture}[baseline = -3.ex]
    \graph[grow right = .8cm, branch down = 2.5ex, nodes={font = \scriptsize}]{
    {/, s1/{}[mynode, draw, circle]}
    -!- {c1/[mynode, draw, circle], /, d1/[mynode, draw, circle]}
    -!- {c2/[mynode, draw, circle], /, d2/[mynode, draw, circle]}
    -!- {c3/, /, d3/[mynode, draw, circle]}
    -!- {/, t1/{}[mynode, draw, circle]}
    };
    \node[left = 4pt of s1](s1l){} edge[earrow, ->] (s1);
    \node[right = 4pt of t1](t1l){}; \path (t1) edge[earrow, ->] (t1l);
    \path (c1) edge [draw = white, opacity = 0] node(mid1){} (d1);
    \path (s1) edge [draw = white, opacity = 0] node(f)[opacity = 1,elabel]{$\fork$} (mid1);
    \path (c3) edge [draw = white, opacity = 0] node(mid3){} (d3);
    \path (mid3) edge [draw = white, opacity = 0] node(j)[opacity = 1,elabel]{$\join$} (t1);
    \graph[use existing nodes, edges={color=black, pos = .5, earrow}, edge quotes={fill=white, inner sep=1pt,font= \scriptsize}]{
    s1 -- f;
    f ->["$1$"{auto}] c1;
    f ->["$2$"{auto, below left = .3ex}] d1;
    c1 ->["$\aterm$"] c2;
    d1 ->["$\aterm[2]$"] d2 -> ["$\aterm[3]$"] d3;
    c2 --["$1$"{auto}, pos = .8] j;
    d3 --["$2$"{auto, below right = .3pt}] j;
    j -> t1;
    };
\end{tikzpicture}\right) \;=\; \set*{\begin{tikzpicture}[baseline = -3.ex]
    \graph[grow right = .8cm, branch down = 2.5ex, nodes={font = \scriptsize}]{
    -!- {c1/[mynode, draw, circle], /, d1/}
    -!- {c2/[mynode, draw, circle], /, d2/[mynode, draw, circle]}
    -!- {c3/, /, d3/[mynode, draw, circle]}
    -!- {/, t1/{}[mynode, draw, circle]}
    };
    \node[left = 4pt of c1](c1l){\tiny $1$} edge[earrow, ->] (c1);
    \node[left = 4pt of d2](d2l){\tiny $2$} edge[earrow, ->] (d2);
    \node[right = 4pt of t1](t1r){}; \path (t1) edge[earrow, ->] (t1r);
    \path (c3) edge [draw = white, opacity = 0] node(mid3){} (d3);
    \path (mid3) edge [draw = white, opacity = 0] node(j)[opacity = 1,elabel]{$\join$} (t1);
    \graph[use existing nodes, edges={color=black, pos = .5, earrow}, edge quotes={fill=white, inner sep=1pt,font= \scriptsize}]{
    c1 ->["$\aterm$"] c2;
    d2 -> ["$\aterm[3]$"] d3;
    c2 --["$1$"{auto}, pos = .8] j;
    d3 --["$2$"{auto, below right = .3pt}] j;
    j -> t1;
    };
\end{tikzpicture}}.\]

Let $n \in \nat$.
For a \kl{run} $\graph$,
the \intro*\kl{emptiness property} $\EPS_{n}(\graph)$ is the \kl{truth value} defined by: 
\begin{align*}
    \EPS_{n}(\graph) &\;\defeq\; \begin{cases}
        \const{true} & (\graph = \id^n)\\
        \const{false} & (\mbox{otherwise}).
    \end{cases}
\end{align*}
For a \kl{run language} $\glang$, let $\EPS_{n}(\glang) \defeq \bigvee_{\graph \in \glang} \EPS_{n}(\graph)$.
The following is clear by definition.
\begin{prop}
    For all \kl{runs} $\graph$, we have $\EPS_{\ty_2(\graph)}(\D_{\graph}(\graph))$.
\end{prop}

\subsection{Runs on structures}
We also introduce \emph{\kl{runs} on \kl{structures}} to simplify the use of the \kl{semantics} using \kl{runs} (\Cref{proposition: run language}).
Let $\struc$ be a \kl{structure}.
An \intro*\kl{$\struc$-run} with $\tuple{\vec{\lab[1]}, \vec{\lab[2]}}$-interface is a \kl{graph homomorphism} $\trace \colon \graph[2] \homo \tilde{\const{G}}(\struc, \vec{x}, \vec{y})$, where $\graph[2]$ is some \kl{run}.
We let $\ty(\trace) \defeq \tuple{\ty_{\src}(\trace), \ty_{\tgt}(\trace)} \defeq \tuple{\vec{\lab[1]}, \vec{\lab[2]}}$.
We use $\trace[1], \trace[2], \trace[3]$ to denote \kl{$\struc$-runs}.
For brevity, we depict \kl{$\struc$-runs} as \kl{vertex}-labeled graphs where each \kl{vertex} is labeled with a \kl{vertex} of $\struc$.
For instance, let $\struc = \left(\begin{tikzpicture}[baseline = -.5ex]
    \graph[grow right = 1.cm, branch down = 6ex, nodes={mynode}]{
    {0/{$\lab[1]$}[draw, circle]}-!-{1/{$\lab[2]$}[draw, circle]}
    };
    \graph[use existing nodes, edges={color=black, pos = .5, earrow}, edge quotes={fill=white, inner sep=1pt,font= \scriptsize}]{
    0 ->["$\aterm[1]$", bend left] 1;
    1 ->["$\aterm[2]$", bend left] 0;
    };
\end{tikzpicture} \right)$ and we consider the following \kl{graph homomorphism} $\trace$: 
\[\begin{tikzpicture}[baseline = -3.ex, remember picture]
    \graph[grow right = .8cm, branch down = 2.5ex, nodes={font = \scriptsize}]{
    {/, s1/{}[mynode, draw, circle]}
    -!- {c1/[mynode, draw, circle], /, d1/[mynode, draw, circle]}
    -!- {c2/, /, d2/[mynode, draw, circle]}
    -!- {c3/, /, d3/[mynode, draw, circle]}
    -!- {c4/, /, d4/[mynode, draw, circle]}
    };
    \node[left = 4pt of s1](l){} edge[earrow, ->] (s1);
    \node[right = 4pt of c1](r1){\tiny $1$}; \path (c1) edge[earrow, ->] (r1);
    \node[right = 4pt of d4](r2){\tiny $2$}; \path (d4) edge[earrow, ->] (r2);
    \path (c1) edge [draw = white, opacity = 0] node(mid1){} (d1);
    \path (s1) edge [draw = white, opacity = 0] node(f)[opacity = 1,elabel]{$\fork$} (mid1);
    \graph[use existing nodes, edges={color=black, pos = .5, earrow}, edge quotes={fill=white, inner sep=1pt,font= \scriptsize}]{
    s1 -- f;
    f ->["$1$"{auto}] c1;
    f ->["$2$"{auto, below left = .3ex}] d1;
    d1 ->["$\aterm[1]$"] d2 ->["$\aterm[2]$"] d3 ->["$\aterm[1]$"] d4;
    };
\end{tikzpicture} \hspace{4em} \begin{tikzpicture}[baseline = -.5ex, remember picture]
    \graph[grow right = 1.5cm, branch down = 6ex, nodes={mynode}]{
    {0/{$\lab[1]$}[draw, circle]}-!-{1/{$\lab[2]$}[draw, circle]}
    };
    \node[left = 4pt of 0](0l){} edge[earrow, ->] (0);
    \node[below left = 4pt of 0](0r){\tiny $1$}; \path (0) edge[earrow, ->] (0r);
    \node[right = 4pt of 1](1r){\tiny $2$}; \path (1)  edge[earrow, ->] (1r);
    \node[above = 2ex of 0, inner sep = .3pt](0u){\tiny $\fork$};
    \node[above = 2ex of 1, inner sep = .3pt](1u){\tiny $\fork$};
    \node[below = 2ex of 0, inner sep = .3pt](0d){\tiny $\join$};
    \node[below = 2ex of 1, inner sep = .3pt](1d){\tiny $\join$};
    \graph[use existing nodes, edges={color=black, pos = .5, earrow}, edge quotes={fill=white, inner sep=1pt,font= \scriptsize}]{
    0 --[bend right] 0u;
    0u ->[bend right = 35] 0;
    0u ->[bend right = 15] 0;
    1 --[bend right] 1u;
    1u ->[bend right = 35] 1;
    1u ->[bend right = 15] 1;
    0 <-[bend right] 0d;
    0d --[bend right = 35] 0;
    0d --[bend right = 15] 0;
    1 <-[bend right] 1d;
    1d --[bend right = 35] 1;
    1d --[bend right = 15] 1;
    0 ->["$\aterm[1]$", bend left] 1;
    1 ->["$\aterm[2]$", bend left] 0;
    };
\end{tikzpicture} = \tilde{\const{G}}(\struc, \lab[1], \lab[1] \lab[2]).
\begin{tikzpicture}[remember picture, overlay]
    \path (s1) edge [homoarrow,  out =45, in = 135, looseness = .2] ($(f)+(0,.6)$);
    \path ($(f)+(0,.6)$) edge [homoarrow,->, out =45, in = 135, looseness = .2] (0);
    \path (c1) edge [homoarrow,->, out =-45, in = 135, looseness = .2] (0);
    \path (d1) edge [homoarrow,->, out =45, in = 135, looseness = .2] (0);
    \path (d2) edge [homoarrow,->, out =-45, in = -90, looseness = .5] (1);
    \path (d3) edge [homoarrow,->, out =45, in = 135, looseness = .2] (0);
    \path (d4) edge [homoarrow,->, out =-45, in = -90, looseness = .5] (1);
\end{tikzpicture}\]
We then depict $\trace$ as follows:
\[\begin{tikzpicture}[baseline = -5.ex, remember picture]
    \graph[grow right = 1.cm, branch down = 2.5ex, nodes={font = \scriptsize}]{
    {/, s1/{$\lab$}[mynode, draw, circle]}
    -!- {c1/{$\lab$}[mynode, draw, circle], /, d1/{$\lab$}[mynode, draw, circle]}
    -!- {c2/, /, d2/{$\lab[2]$}[mynode, draw, circle]}
    -!- {c3/, /, d3/{$\lab$}[mynode, draw, circle]}
    -!- {c4/, /, d4/{$\lab[2]$}[mynode, draw, circle]}
    };
    \node[left = 4pt of s1](l){} edge[earrow, ->] (s1);
    \node[right = 4pt of c1](r1){\tiny $1$}; \path (c1) edge[earrow, ->] (r1);
    \node[right = 4pt of d4](r2){\tiny $2$}; \path (d4) edge[earrow, ->] (r2);
    \path (c1) edge [draw = white, opacity = 0] node(mid1){} (d1);
    \path (s1) edge [draw = white, opacity = 0] node(f)[opacity = 1,elabel]{$\fork$} (mid1);
    \graph[use existing nodes, edges={color=black, pos = .5, earrow}, edge quotes={fill=white, inner sep=1pt,font= \scriptsize}]{
    s1 -- f;
    f ->["$1$"{auto}] c1;
    f ->["$2$"{auto, below left = .3ex}] d1;
    d1 ->["$\aterm[1]$"] d2 ->["$\aterm[2]$"] d3 ->["$\aterm[1]$"] d4;
    };
\end{tikzpicture}.\]
For two \kl{$\struc$-runs} with $\tuple{\vec{x},\vec{y}}$-interface,
$\trace[1] \colon \graph[2] \homo \tilde{\const{\graph}}(\struc, \vec{\lab[1]}, \vec{\lab[2]})$
and $\trace[2] \colon \graph[3] \homo \tilde{\const{\graph}}(\struc, \vec{\lab[1]}, \vec{\lab[2]})$,
we say that a map $f \colon \univ{\graph[2]} \to \univ{\graph[3]}$ is an \intro*\kl{$\struc$-run isomorphism} from $\trace[1]$ to $\trace[2]$, written $f \colon \trace[1] \cong \trace[2]$, if
$f \colon \graph[2] \cong \graph[3] \mbox{ and for all $x \in \univ{\graph[2]}$, $\trace[1](x) = \trace[2](f(x))$}$.
We write $\trace[1] \cong \trace[2]$ if there is some $f$ such that $f \colon \trace[1] \cong \trace[2]$.
In the sequel, we identify two \kl{$\struc$-runs} if there is an \kl{$\struc$-run isomorphism} between them;
we write $\trace[1] = \trace[2]$ when $\trace[1] \cong \trace[2]$.

For a \kl{KL term} $\term$ and $\lab[1], \lab[2] \in \univ{\struc}$,
we write $\run^{\struc}_{\lab[1], \lab[2]}(\term)$ for the \kl{$\struc$-run language} defined by:
\[\run^{\struc}_{\lab[1], \lab[2]}(\term) \;\defeq\; \set{\trace \colon \graph[2] \homo \tilde{\const{\graph}}(\struc, \lab[1], \lab[2])  \mid \graph[2] \in \run(\term)}.\]
For instance, when $\struc = \left(\begin{tikzpicture}[baseline = -.5ex]
    \graph[grow right = 1.cm, branch down = 6ex, nodes={mynode}]{
    {0/{$\lab[1]$}[draw, circle]}-!-{1/{$\lab[2]$}[draw, circle]}
    };
    \graph[use existing nodes, edges={color=black, pos = .5, earrow}, edge quotes={fill=white, inner sep=1pt,font= \scriptsize}]{
    0 ->["$\aterm[1]$", bend left] 1;
    1 ->["$\aterm[2]$", bend left] 0;
    };
\end{tikzpicture} \right)$, we have:
\begin{align*}
    \run^{\struc}_{\lab[1], \lab[2]}(\aterm[1] \intersection (\aterm[1] \aterm[2] \aterm[1])) \;&=\; \set*{\begin{tikzpicture}[baseline = -3.ex]
                 \graph[grow right = 1.cm, branch down = 2.5ex, nodes={font = \scriptsize}]{
                 {/, s1/{$\lab[1]$}[mynode, draw, circle]}
                 -!- {c1/{$\lab[1]$}[mynode, draw, circle], /, d1/{$\lab[1]$}[mynode, draw, circle]}
                 -!- {c2/{$\lab[2]$}[mynode, draw, circle], /, d2/{$\lab[2]$}[mynode, draw, circle]}
                 -!- {c3/, /, d3/{$\lab[1]$}[mynode, draw, circle]}
                 -!- {c4/, /, d4/{$\lab[2]$}[mynode, draw, circle]}
                 -!- {/, t1/{{$\lab[2]$}}[mynode, draw, circle]}
                 };
                 \node[left = 4pt of s1](s1l){} edge[earrow, ->] (s1);
                 \node[right = 4pt of t1](t1l){}; \path (t1) edge[earrow, ->] (t1l);
                 \path (c1) edge [draw = white, opacity = 0] node(mid1){} (d1);
                 \path (s1) edge [draw = white, opacity = 0] node(f)[opacity = 1,elabel]{$\fork$} (mid1);
                 \path (c4) edge [draw = white, opacity = 0] node(mid4){} (d4);
                 \path (mid4) edge [draw = white, opacity = 0] node(j)[opacity = 1,elabel]{$\join$} (t1);
                 \graph[use existing nodes, edges={color=black, pos = .5, earrow}, edge quotes={fill=white, inner sep=1pt,font= \scriptsize}]{
                s1 -- f; f ->["$1$" auto] c1; f ->["$2$" {auto, below left = .3ex}] d1;
                c1 ->["$\aterm$"] c2;
                d1 ->["$\aterm$"] d2 -> ["$\aterm[2]$"] d3 -> ["$\aterm[1]$"] d4;
                c2 -- ["$1$"{auto}, pos = .8] j; d4 --["$2$" {auto, below right = .3ex}] j; j -> t1;
                };
             \end{tikzpicture}}.
\end{align*}
Using \kl{$\struc$-runs}, we can rephrase \Cref{proposition: run language} as follows.
\begin{prop}\label{proposition: struc-run language}
    Let $\struc$ be a \kl{structure}.
    For all \kl{KL term} $\term$,
    we have:
    \[\jump{\term}_{\struc} \;=\; \set{\tuple{\lab[1],\lab[2]} \in \univ{\struc}^2 \mid \run^{\struc}_{\lab[1], \lab[2]}(\term) \neq \emptyset}.\]
\end{prop}
\begin{proof}
    By easy induction on $\term$ (or an easy consequence of \Cref{proposition: run language}, as
    $\tuple{\lab[1], \lab[2]} \in \jump{\term}_{\struc}$
    iff $\exists \graph[2] \in \run(\term), \exists \trace, \trace \colon \graph[2] \homo \tilde{\const{\graph}}(\struc, \lab[1], \lab[2])$
    iff $\run^{\struc}_{\lab[1], \lab[2]}(\term) \neq \emptyset$).
\end{proof}

\subsubsection{Compositions on $\struc$-runs}\label{section: compositions of struc-runs}
We recall the \kl{series composition} $\series$, the \kl{parallel composition} $\parallel$, and the identity element $\id^{n}$ (\Cref{definition: series composition and parallel composition}) for \kl{graphs} with bi-interface.
Similarly, we define them also for \kl{$\struc$-runs}.
\begin{defi}\label{definition: sequential composition and parallel composition of runs on structures}
    Let $\trace[1] \colon \graph[2] \homo \tilde{\const{\graph}}(\struc, \vec{x}, \vec{y})$ and 
    $\trace[2] \colon \graph[3] \homo \tilde{\const{\graph}}(\struc, \vec{x}', \vec{y}')$ be \kl{$\struc$-runs}.
    Under $\vec{y} = \vec{x}'$,
    the \intro*\kl(srun){series composition} $\trace[1] \series \trace[2] \colon (\graph[2] \series \graph[3]) \homo \tilde{\const{\graph}}(\struc, \vec{x}, \vec{y}')$ is defined (undefined when $\vec{y} \neq \vec{x}'$) as the map such that,
    for each \kl{vertex} $x \in \univ{\graph[2] \series \graph[3]}$,
    if $x$ is induced from a \kl{vertex} $y \in \univ{\graph[2]}$,
    then $(\trace[1] \series \trace[2])(x) = \trace[1](y)$,
    and
    if $x$ is induced from a \kl{vertex} $y \in \univ{\graph[3]}$,
    then $(\trace[1] \series \trace[2])(x) = \trace[2](y)$.
    The \intro*\kl(srun){parallel composition} $\trace[1] \parallel \trace[2] \colon (\graph[2] \parallel \graph[3]) \homo \tilde{\const{\graph}}(\struc, \vec{x}\vec{x}', \vec{y}\vec{y}')$
    is defined as the map such that
    if $x$ is induced from a \kl{vertex} $y \in \univ{\graph[2]}$,
    then $(\trace[1] \parallel \trace[2])(x) = \trace[1](y)$,
    and
    if $x$ is induced from a \kl{vertex} $y \in \univ{\graph[3]}$,
    then $(\trace[1] \parallel \trace[2])(x) = \trace[2](y)$.
    For each $\vec{\lab[1]} = \lab[1]_1 \dots \lab[1]_n \in \univ{\struc}^{*}$,
    let $\id^{\vec{\lab[1]}} \colon \id^{n} \homo \tilde{\const{\graph}}(\struc, \vec{\lab[1]}, \vec{\lab[1]})$ be the \kl{$\struc$-run} such that
    for each $i \in \range{n}$, the $i$-th \kl{source}/\kl{target} \kl{vertex} maps to $\lab[1]_i$.
\end{defi}
The left and right identity element with respect to $\series$ on \kl{$\struc$-runs} with $\tuple{\vec{\lab[1]}, \vec{\lab[2]}}$-interface are given by $\id^{\vec{\lab[1]}}$ and $\id^{\vec{\lab[2]}}$, respectively, i.e., $\id^{\vec{\lab[1]}} \series \trace = \trace \series \id^{\vec{\lab[2]}} = \trace$.
The identity element with respect to $\parallel$ on \kl{$\struc$-runs} with $\tuple{\vec{\lab[1]}, \vec{\lab[2]}}$-interface
is given by $\id^{\eps}$, i.e., $\id^{\eps} \parallel \trace = \trace \parallel \id^{\eps} = \trace$.

For $\vec{\lab[3]} = \lab[3]_1 \dots \lab[3]_n$,
let $a^{\vec{\lab[3]}, \lab[3]'}_{i} \defeq \id^{z_1 \dots z_{i-1}} \parallel \left(\hspace{-.3em}
        \begin{gathered}
            \begin{tikzpicture}[baseline = -.5ex]
            \graph[grow right = 1.cm, branch down = 6ex, nodes={mynode}]{
            {0/{$\lab[3]_{i}$}[draw, circle]} -!-
            {1/{$\lab[3]'$}[draw, circle, inner sep = .5pt]}
            };
            \node[left = .5em of 0](l){\tiny };
            \node[right = .5em of 1](r){\tiny };
            \graph[use existing nodes, edges={color=black, pos = .5, earrow}, edge quotes={fill=white, inner sep=1pt,font= \scriptsize}]{
                l -> 0 ->["$a$"] 1 -> r;
            };
        \end{tikzpicture}
    \end{gathered}\right) \parallel \id^{z_{i+1} \dots z_{n}}$,
$\fork^{\vec{\lab[3]}}_{i} \defeq \id^{z_1 \dots z_{i-1}} \parallel \left(\hspace{-.8em}
        \begin{gathered}
            \begin{tikzpicture}[baseline = -.5ex]
            \graph[grow right = .7cm, branch down = 2ex, nodes={}]{
            {/, 0/{$\lab[3]_{i}$}[draw,circle, mynode]} -!-
            {/, c/{\scriptsize $\fork$}[elabel]} -!-
            {10/{$\lab[3]_{i}$}[draw, circle, inner sep = .5pt, mynode], /, 11/{$\lab[3]_{i}$}[draw, circle, inner sep = .5pt, mynode]}
            };
            \node[left = .5em of 0](l){\tiny };
            \node[right = .5em of 10](r0){\tiny $1$};
            \node[right = .5em of 11](r1){\tiny $2$};
            \graph[use existing nodes, edges={color=black, pos = .5, earrow}, edge quotes={fill=white, inner sep=1pt,font= \scriptsize}]{
            l -> 0 -- c;
            c ->["$1$"{auto}] 10;
            c ->["$2$"{auto, below left = .3ex}] 11;
            10 -> r0; 11 -> r1;
            };
            \end{tikzpicture}
        \end{gathered}\right) \parallel \id^{z_{i+1} \dots z_{n}}$, and
$\join^{\vec{\lab[3]}}_{i} \defeq  \id^{z_1 \dots z_{i-1}} \parallel \left(
        \begin{gathered}
            \begin{tikzpicture}[baseline = -.5ex]
            \graph[grow right = .7cm, branch down = 2ex, nodes={}]{
            {00/{$\lab[3]_{i}$}[draw, circle, inner sep = .5pt, mynode], /,01/{$\lab[3]_{i}$}[draw,circle, mynode]} -!-
            {/, c/{\scriptsize $\join$}[elabel]} -!-
            {/, 1/{$\lab[3]_{i}$}[draw, circle, inner sep = .5pt, mynode]}
            };
            \node[left = .5em of 00](l0){\tiny $1$};
            \node[left = .5em of 01](l1){\tiny $2$};
            \node[right = .5em of 1](r){\tiny};
            \graph[use existing nodes, edges={color=black, pos = .5, earrow}, edge quotes={fill=white, inner sep=1pt,font= \scriptsize}]{
                l0 -> 00; l1 -> 01;
                00 --["$1$"{auto}] c;
                01 --["$2$"{auto, below right = .3ex}] c;
                c -> 1 -> r;
            };
        \end{tikzpicture}
    \end{gathered}\hspace{-.8em}\right) \parallel \id^{z_{i+1} \dots z_{n}}$.
The set of \intro*\kl{atomic $\struc$-runs}, written $\bar{\vsig}_{\struc}$, is defined as $\bar{\vsig}_{\struc} \defeq
\set{a^{\vec{\lab[3]}, \lab[3]'}_{i} \mid a \in \vsig, n \ge 1, i \in \range{n}, \vec{\lab[3]} \in \univ{\struc}^{n}, \lab[3]' \in \univ{\struc}, \tuple{\lab[3]_i, \lab[3]'} \in \jump{a}_{\struc}}
\cup \set{\fork^{\vec{\lab[3]}}_{i}, \join^{\vec{\lab[3]}}_{i} \mid n \ge 1, i \in \range{n}, \vec{\lab[3]} \in \univ{\struc}^{n}}$.
Every \kl{$\struc$-run} can be expressed as \kl{series compositions} of \kl{atomic $\struc$-runs}, as follows (as with \Cref{proposition: run atomic}).
\begin{prop}\label{proposition: struc-run atomic}
    Let $\struc$ be a \kl{structure}.
    For every \kl{$\struc$-run} $\trace$ with $n \ge 1$ edges,
    there exist some $a_1, \dots, a_n \in \bar{\vsig}_{\struc}$ such that $\trace = a_1 \series \dots \series a_n$.
\end{prop}
\begin{proof}
    By easy induction on $n$.
\end{proof}
For instance, when $\struc = \left(\hspace{-.5em}\begin{tikzpicture}[baseline = .5ex]
    \graph[grow right = 1.cm, branch down = 6ex, nodes={mynode}]{
    {0/{$\lab[1]$}[draw, circle]}-!-{1/{$\lab[2]$}[draw, circle]}
    };
    \graph[use existing nodes, edges={color=black, pos = .5, earrow}, edge quotes={fill=white, inner sep=1pt,font= \scriptsize}]{
    0 ->["$\aterm[1]$", bend left] 1;
    0 ->["$\aterm[3]$", bend right] 1;
    0 ->["$\aterm[2]$", out = 60, in = 120, looseness = 8] 0;
    };
\end{tikzpicture} \right)$ and $\trace = \begin{tikzpicture}[baseline = -3.ex]
    \graph[grow right = 1.cm, branch down = 2.5ex, nodes={font = \scriptsize}]{
    {/, s1/{$\lab[1]$}[mynode, draw, circle]}
    -!- {c1/{$\lab[1]$}[mynode, draw, circle], /, d1/{$\lab[1]$}[mynode, draw, circle]}
    -!- {c2/{$\lab[2]$}[mynode, draw, circle], /, d2/{$\lab[1]$}[mynode, draw, circle]}
    -!- {c3/, /, d3/{$\lab[2]$}[mynode, draw, circle]}
    -!- {/, t1/{$\lab[2]$}[mynode, draw, circle]}
    };
    \node[left = 4pt of s1](s1l){} edge[earrow, ->] (s1);
    \node[right = 4pt of t1](t1l){}; \path (t1) edge[earrow, ->] (t1l);
    \path (c1) edge [draw = white, opacity = 0] node(mid1){} (d1);
    \path (s1) edge [draw = white, opacity = 0] node(f)[opacity = 1,elabel]{$\fork$} (mid1);
    \path (c3) edge [draw = white, opacity = 0] node(mid3){} (d3);
    \path (mid3) edge [draw = white, opacity = 0] node(j)[opacity = 1,elabel]{$\join$} (t1);
    \graph[use existing nodes, edges={color=black, pos = .5, earrow}, edge quotes={fill=white, inner sep=1pt,font= \scriptsize}]{
    s1 -- f;
    f ->["$1$"{auto}, pos = .8] c1;
    f ->["$2$"{auto, below left = .3ex}] d1;
    c1 ->["$\aterm$"] c2;
    d1 ->["$\aterm[2]$"] d2 -> ["$\aterm[3]$"] d3;
    c2 --["$1$"{auto}] j;
    d3 --["$2$"{auto, below right = .3ex}] j;
    j -> t1;
    };
\end{tikzpicture}$, we have
$\trace = \fork^{\lab[1]}_{1} \series \aterm^{\lab[1]\lab[1],\lab[2]}_{1} \series \aterm[2]^{\lab[2]\lab[1],\lab[1]}_{2} \series \aterm[3]^{\lab[2]\lab[1], \lab[2]}_{2} \series \join^{\lab[2]}_{1}$.
The decomposition of $\trace$ is not unique, as
$\fork^{\lab[1]}_{1} \series \aterm[2]^{\lab[1]\lab[1], \lab[1]}_{2} \series \aterm^{\lab[1]\lab[1], \lab[2]}_{1} \series \aterm[3]^{\lab[2]\lab[1], \lab[2]}_{2} \series \join^{\lab[2]}_{1}$ and
$\fork^{\lab[1]}_{1} \series \aterm[2]^{\lab[1]\lab[1], \lab[1]}_{2} \series \aterm[3]^{\lab[1]\lab[1], \lab[2]}_{2} \series \aterm^{\lab[1]\lab[2], \lab[2]}_{1} \series \join^{\lab[2]}_{1}$
also express $\trace$.

\subsubsection{Left quotients of $\struc$-runs}\label{section: left quotients of struc-runs}
Let $\struc$ be a \kl{structure}.
Similar to \Cref{section: left quotients of runs},
we define the \kl(srun){left quotients} and \kl(srun){emptiness property} for \kl{$\struc$-runs}.
For \kl{$\struc$-runs} $\trace[1]$, $\trace[2]$, and $\trace[3]$, we write:
\[\trace[1] \longrightarrow^{\struc}_{\trace[3]} \trace[2] \;\defiff\; \mbox{$\trace[3] \series \trace[2]$ is defined and } \trace[1] \cong \trace[3] \series \trace[2].\]
For \kl{$\struc$-runs} $\trace$ and $\trace[3]$,
the \intro*\kl(srun){left quotient} $\D^{\struc}_{\trace[3]}(\trace)$ of $\trace$ with respect to $\trace[3]$
is the \kl{$\struc$-run language} defined by:
\begin{align*}
    \D^{\struc}_{\trace[3]}(\trace[1]) &\;\defeq\; \set{\trace[2] \mbox{ an \kl{$\struc$-run}} \mid \trace[1] \longrightarrow_{\trace[3]} \trace[2]}. 
\end{align*}
For an \kl{$\struc$-run language} $\glang$, let $\D^{\struc}_{\trace[3]}(\glang) \defeq \bigcup_{\trace[1] \in \glang} \D^{\struc}_{\trace[3]}(\trace[1])$.

\vskip0.5\baselineskip 
For instance, if $\trace[3] = \left(\begin{tikzpicture}[baseline = -3.ex]
    \graph[grow right = 1.cm, branch down = 2.5ex, nodes={font = \scriptsize}]{
    {/, s1/{$\lab[1]$}[mynode, draw, circle]}
    -!- {c1/{$\lab[1]$}[mynode, draw, circle], /, d1/{$\lab[1]$}[mynode, draw, circle]}
    -!- {c2/, /, d2/{$\lab[2]$}[mynode, draw, circle]}
    };
    \node[left = 4pt of s1](s1l){} edge[earrow, ->] (s1);
    \node[right = 4pt of c1](c1r){\tiny $1$}; \path (c1) edge[earrow, ->] (c1r);
    \node[right = 4pt of d2](d2r){\tiny $2$}; \path (d2) edge[earrow, ->] (d2r);
    \path (c1) edge [draw = white, opacity = 0] node(mid1){} (d1);
    \path (s1) edge [draw = white, opacity = 0] node(f)[opacity = 1,elabel]{$\fork$} (mid1);
    \graph[use existing nodes, edges={color=black, pos = .5, earrow}, edge quotes={fill=white, inner sep=1pt,font= \scriptsize}]{
    s1 -- f;
    f ->["$1$"{auto}] c1;
    f ->["$2$"{auto, below left = .3ex}] d1;
    d1 ->["$\aterm[2]$"] d2;
    };
\end{tikzpicture}\right)$, we have:
\[\D^{\struc}_{\trace[3]}\left(\begin{tikzpicture}[baseline = -3.ex]
    \graph[grow right = 1.cm, branch down = 2.5ex, nodes={font = \scriptsize}]{
    {/, s1/{$\lab[1]$}[mynode, draw, circle]}
    -!- {c1/{$\lab[1]$}[mynode, draw, circle], /, d1/{$\lab[1]$}[mynode, draw, circle]}
    -!- {c2/{$\lab[3]$}[mynode, draw, circle], /, d2/{$\lab[2]$}[mynode, draw, circle]}
    -!- {c3/, /, d3/{$\lab[3]$}[mynode, draw, circle]}
    -!- {/, t1/{$\lab[3]$}[mynode, draw, circle]}
    };
    \node[left = 4pt of s1](s1l){} edge[earrow, ->] (s1);
    \node[right = 4pt of t1](t1l){}; \path (t1) edge[earrow, ->] (t1l);
    \path (c1) edge [draw = white, opacity = 0] node(mid1){} (d1);
    \path (s1) edge [draw = white, opacity = 0] node(f)[opacity = 1,elabel]{$\fork$} (mid1);
    \path (c3) edge [draw = white, opacity = 0] node(mid3){} (d3);
    \path (mid3) edge [draw = white, opacity = 0] node(j)[opacity = 1,elabel]{$\join$} (t1);
    \graph[use existing nodes, edges={color=black, pos = .5, earrow}, edge quotes={fill=white, inner sep=1pt,font= \scriptsize}]{
    s1 -- f;
    f ->["$1$"{auto}] c1;
    f ->["$2$"{auto, below left = .3ex}] d1;
    c1 ->["$\aterm$"] c2;
    d1 ->["$\aterm[2]$"] d2 -> ["$\aterm[3]$"] d3;
    c2 --["$1$"{auto}, pos = .8] j;
    d3 --["$2$"{auto, below right = .3ex}] j;
    j -> t1;
    };
\end{tikzpicture}\right) \;=\; \set*{\begin{tikzpicture}[baseline = -3.ex]
    \graph[grow right = 1.cm, branch down = 2.5ex, nodes={font = \scriptsize}]{
    -!- {c1/{$\lab[1]$}[mynode, draw, circle], /, d1/}
    -!- {c2/{$\lab[3]$}[mynode, draw, circle], /, d2/{$\lab[2]$}[mynode, draw, circle]}
    -!- {c3/, /, d3/{$\lab[3]$}[mynode, draw, circle]}
    -!- {/, t1/{$\lab[3]$}[mynode, draw, circle]}
    };
    \node[left = 4pt of c1](c1l){\tiny $1$} edge[earrow, ->] (c1);
    \node[left = 4pt of d2](d2l){\tiny $2$} edge[earrow, ->] (d2);
    \node[right = 4pt of t1](t1r){}; \path (t1) edge[earrow, ->] (t1r);
    \path (c3) edge [draw = white, opacity = 0] node(mid3){} (d3);
    \path (mid3) edge [draw = white, opacity = 0] node(j)[opacity = 1,elabel]{$\join$} (t1);
    \graph[use existing nodes, edges={color=black, pos = .5, earrow}, edge quotes={fill=white, inner sep=1pt,font= \scriptsize}]{
    c1 ->["$\aterm$"] c2;
    d2 -> ["$\aterm[3]$"] d3;
    c2 --["$1$"{auto}, pos = .8] j;
    d3 --["$2$"{auto, below right = .3ex}] j;
    j -> t1;
    };
\end{tikzpicture}}.\]
For an \kl{$\struc$-run} $\trace[1]$ and $\vec{\lab[3]} \in \univ{\struc}^{*}$,
the \intro*\kl(srun){emptiness property} $\EPS_{\vec{\lab[3]}}(\trace[1])$ is the \kl{truth value} defined by: 
\begin{align*}
    \EPS_{\vec{\lab[3]}}(\trace[1]) &\;\defeq\; 
    \begin{cases}
        \const{true} & (\trace[1] = \id^{\vec{\lab[3]}})\\
        \const{false} & (\mbox{otherwise}).
    \end{cases}
\end{align*}
For an \kl{$\struc$-run language} $\glang$, let $\EPS_{\vec{\lab[3]}}(\glang) \defeq \bigvee_{\trace[1] \in \glang} \EPS_{\vec{\lab[3]}}(\trace)$.
The following is clear by definition.
\begin{prop}\label{proposition: left quotient}
    Let $\struc$ be a \kl{structure}.
    For all \kl{$\struc$-runs} $\trace$,
    we have $\EPS_{\ty_2(\trace)}(\D^{\struc}_{\trace}(\trace))$.
\end{prop}
We will define \kl{derivatives} based on these \kl{left quotient} and \kl{emptiness property}.

\section{Labeled Kleene Lattice Terms}\label{section: lKL}
To present \kl{derivatives} (which will be used for an automata construction for \kl{PCoR*}),
we consider simulating \kl{left quotients} of \kl{$\struc$-runs} (\Cref{section: left quotients of struc-runs}) by terms.
To this end, we extend the syntax of terms with ``\kl{labels}''.
In a nutshell, \kl{labels} are used for expressing \kl{runs} with \emph{multiple} source \kl{vertices}.

Let $\Lab$ be a set of \intro*\kl{labels}.
We use $\lab[1], \lab[2], \lab[3]$ to denote \kl{labels}.
The set of \intro*\kl{labeled Kleene lattice terms} (\kl{lKL terms}) is defined by the following grammar:
\begin{align*}
    \lterm[1], \lterm[2], \lterm[3] \in \LT_{\Lab} &\;\Coloneqq\; @\lab. \term \mid \lterm[1] \compo_{1} \term[2]  \mid \lterm[1] \intersection_{1} \lterm[2] \qquad \tag{$\lab \in \Lab$ and $\term[1], \term[2]$ are \kl{KL terms}}
\end{align*}
Each \kl{label} $\lab[1]$ indicates the \kl{vertex} $\lab[1]$ on a \kl{structure} (if exists).
The notation $@\lab. \term$ is inspired by the jump operator in hybrid logics \cite{arecesHybridLogics2007} (a minor difference is that $\lab$ directly indicates a \kl{vertex} name, here).
We explicitly use $\compo_{1}$ and $\intersection_{1}$ instead of $\compo$ and $\intersection$ in \kl{lKL terms},
as labeled terms interpret unary relations (\cf\ \Cref{section: lKL semantics}).
We write $\mathsf{lab}_{i}(\lterm)$ for the $i$-th \kl{label} occurring in $\lterm$,
and we write $\overrightarrow{\mathsf{lab}}(\lterm)$ for the sequence $\mathsf{lab}_{1}(\lterm) \dots \mathsf{lab}_{n}(\lterm)$ where 
$n$ is the number of \kl{labels} occurring in $\lterm$.
For instance, if $\lterm = ((@ \lab[1]. \lterm[2]_1) \intersection_{1} (((@\lab[1]. \lterm[2]_2) \intersection_{1} (@\lab[2]. \lterm[2]_3)) \compo_1 \term[3]_1)) \compo_{1} \term[3]_2$,
then $\mathsf{lab}_{1}(\lterm) = \lab[1]$, $\mathsf{lab}_{2}(\lterm) = \lab[1]$, $\mathsf{lab}_{3}(\lterm) = \lab[2]$, and
$\overrightarrow{\mathsf{lab}}(\lterm) = \lab[1]\lab[1]\lab[2]$.

\subsection{Semantics: relational models}\label{section: lKL semantics}
Let $\struc \in \REL$.
For an \kl{lKL term} $\lterm$,
the \intro*\kl(lterm){semantics} $\jump{\lterm}_{\struc} \subseteq \univ{\struc}$ is a unary relation defined as follows:
\begin{align*}
    \jump{@ \lab[1]. \term}_{\struc} &\defeq \set{\lab[2] \mid \tuple{\lab[1], \lab[2]} \in \jump{\term}_{\struc}}, &
    \hspace{-.7em}\jump{\lterm[1] \compo_{1} \term[2]}_{\struc} &\defeq \set{\lab[3] \mid \lab[2] \in \jump{\lterm[1]}_{\struc}, \tuple{\lab[2], \lab[3]} \in \jump{\term[2]}_{\struc}}, &
    \hspace{-.7em} \jump{\lterm[1] \intersection_{1} \lterm[2]}_{\struc} &\defeq \jump{\lterm[1]}_{\struc} \cap \jump{\lterm[2]}_{\struc}.
\end{align*}
Note that $\jump{@ \lab[1]. \term}_{\struc} = \emptyset$ when $\lab[1] \not\in \univ{\struc}$.
Also, note that $\jump{\term[1]}_{\struc}$ is a \emph{binary} relation for \kl{KL terms} $\term[1]$
and $\jump{\lterm[1]}_{\struc}$ is a \emph{unary} relation for \kl{lKL} terms $\lterm[1]$.
For $\strucclass \subseteq \REL$,
we write $\strucclass \models \lterm[1] = \lterm[2]$ if $\jump{\lterm[1]}_{\struc} = \jump{\lterm[2]}_{\struc}$ for $\struc \in \strucclass$.
We then have the following easy fact.
\begin{prop}\label{proposition: labelling}
    Let $\term[1]$ and $\term[2]$ be \kl{KL terms} and $\lab$ be a \kl{label}.
    We then have:
    \[\REL \models \term[1] \le \term[2] \;\iff\; \REL \models @\lab. \term[1] \le @\lab. \term[2].\]
\end{prop}
\begin{proof}
    ($\Longrightarrow$):
    Trivial.
    ($\Longleftarrow$):
    Let $\tuple{\lab[2], \lab[3]} \in \jump{\term}_{\struc}$.
    Let $f$ and $\struc[2]$ be such that $f \colon \struc[1] \cong \struc[2]$ and $f(\lab[2]) = \lab[1]$.
    We then have that for all \kl{KL terms} $\term[3]$,
    $\tuple{\lab[2], \lab[3]} \in \jump{\term[3]}_{\struc}$
    iff $\tuple{\lab[1], f(\lab[3])} \in \jump{\term[3]}_{\struc[2]}$ (since $\struc[1] \cong \struc[2]$ and $\term[3]$ is a (non-labelled) \kl{KL term})
    iff $f(\lab[3]) \in \jump{@ \lab.\term[3]}_{\struc[2]}$ (by the definition of $\jump{\bl}_{\struc[2]}$).
    We thus have $f(\lab[3]) \in \jump{@ \lab.\term[1]}_{\struc}$.
    By the assumption,
    $f(\lab[3]) \in \jump{@ \lab.\term[2]}_{\struc[2]}$.
    Hence, by the reverse direction of the above argument,
    we have $\tuple{\lab[2], \lab[3]} \in \jump{\term[2]}_{\struc}$.
\end{proof}

Note that for \kl{labeled KL terms} $\lterm[1]$, $\jump{\lterm[1]}_{\struc[1]}$ and $\jump{\lterm[1]}_{\struc[2]}$ may not be equal even if $\struc[1] \cong \struc[2]$, because we use the names of \kl{vertices} in the \kl(lterm){semantics} (\cf\ \kl{nominals} in hybrid logic \cite{arecesHybridLogics2007}).
For that reason, for \kl{structures}, we distinguish two \kl{structures} even if they are isomorphic, \cf\  we identify two \kl{graphs}/\kl{runs}/\kl{$\struc$-runs} if they are isomorphic.

\subsection{Runs on structures}
We use the following notation:
$X \mathop{\heartsuit} Y \defeq \set{\trace[1] \mathop{\heartsuit} \trace[2] \mid \trace[1] \in X \land \trace[2] \in Y}$
where $X$ and $Y$ are sets of \kl{runs} or sets of \kl{$\struc$-runs} and $\mathop{\heartsuit} \in \set{\series, \parallel}$.
In this set notation, we use terms as singleton sets: for instance, $\set{\trace[1], \trace[2]} \series \trace[3] = \set{\trace[1] \series \trace[3], \trace[2] \series \trace[3]}$.
\begin{defi}\label{definition: run lab}
    Let $\struc$ be a \kl{structure}.
    For an \kl{lKL term} $\lterm$ and $\lab[3] \in \univ{\struc}$,
    the \intro*\kl{$\struc$-run language} $\run^{\struc}_{\lab[3]}(\lterm)$ is defined as follows:
    \begin{align*}
        \run^{\struc}_{\lab[3]}(@ \lab[1]. \term) &\defeq \run^{\struc}_{\lab[1], \lab[3]}(\term), &
        \hspace{-.5em}\run^{\struc}_{\lab[3]}(\lterm[1] \compo_{1} \term[2]) &\defeq \bigcup_{\lab[2] \in \univ{\struc}} \run^{\struc}_{\lab[2]}(\lterm[1]) \series \run^{\struc}_{\lab[2], \lab[3]}(\term[2]),  &
        \hspace{-.5em}\run^{\struc}_{\lab[3]}(\lterm[1] \intersection_{1} \lterm[2]) &\defeq (\run^{\struc}_{\lab[3]}(\lterm[1]) \parallel \run^{\struc}_{\lab[3]}(\lterm[2])) \series \join_1^{\lab[3]}.
    \end{align*}
\end{defi}

\begin{prop}[\cf\ \Cref{proposition: struc-run language}]\label{proposition: labeled struc-run language}
    Let $\struc$ be a \kl{structure}.
    For all \kl{lKL term} $\lterm$,
    we have:
    \[\jump{\lterm}_{\struc} \;=\; \set{\lab[3] \in \univ{\struc} \mid \run^{\struc}_{\lab[3]}(\lterm) \neq \emptyset}.\]
\end{prop}
\begin{proof}
    By easy induction on $\lterm$ using \Cref{proposition: struc-run language}.
\end{proof}

\section{Derivatives on Graphs}\label{section: derivative}
We recall the \kl{left quotients} of \kl{$\struc$-runs} (defined in \Cref{section: left quotients of struc-runs}).
In this section, we now define \kl{derivatives} on $\struc$, which can simulate the \kl{left quotients}.
We first define the \kl{emptiness property} in the term level (\Cref{definition: eps}), and then define the \kl{derivatives} on \kl{graphs} (\Cref{definition: derivative,proposition: derivative on graphs}), based on derivatives on strings \cite{brzozowskiDerivativesRegularExpressions1964,antimirovPartialDerivativesRegular1996,sakarovitchElementsAutomataTheory2009}.
\begin{defi}\label{definition: eps}
    For a \kl{label} $\lab[3]$ and an \kl{lKL term} $\lterm$,
    the truth value $\EPS_{\lab[3]}(\lterm) \in \set{\const{false}, \const{true}}$ is defined as follows:%
    \footnote{The definition $\EPS_{\lab[3]}(@\lab[1].\term[2] \intersection \term[3]) \defeq \EPS_{\lab[3]}(\lterm[2] \intersection_{1} \lterm[3]) \defeq \const{false}$ is due to that each \kl{$\struc$-run} of $\run^{\struc}_{z}(@\lab[1].\term[2] \intersection \term[3])$ and $\run^{\struc}_{z}(\lterm[2] \intersection_{1} \lterm[3])$ always contains $\join^{z}_{1}$ (\Cref{definition: run lab}), so they do not contain the \kl{run} having the \kl{emptiness property}.}
    \begin{align*}
        \EPS_{\lab[3]}(@\lab[1].a)                      
        & \defeq  \EPS_{\lab[3]}(@\lab[1].\emp)
        \defeq \EPS_{\lab[3]}(@\lab[1].\term[2] \intersection \term[3])
        \defeq   \EPS_{\lab[3]}(\lterm[2] \intersection_{1} \lterm[3]) 
        \defeq \const{false},                                \span\span                                                                              \\
        \EPS_{\lab[3]}(@\lab[1].\id)                     & \defeq \EPS_{\lab[3]}(@\lab[1].\term[2]^*) \defeq \begin{cases}
            \const{true}  & (\lab[1] = \lab[3])      \\
            \const{false} & (\lab[1] \neq \lab[3]),     
        \end{cases}
        &\EPS_{\lab[3]}(\lterm[2] \compo_{1} \term[3]) &\defeq  \EPS_{\lab[3]}(\lterm[2]) \land \EPS_{\lab[3]}(@ \lab[3]. \term[3]),   \\
        \EPS_{\lab[3]}(@\lab[1].\term[2] \compo \term[3]) &\defeq  \EPS_{\lab[3]}(@\lab[1].\term[2]) \land \EPS_{\lab[3]}(@\lab[3].\term[3]),
        & \EPS_{\lab[3]}(@\lab[1].\term[2] \union \term[3]) &\defeq \EPS_{\lab[3]}(@\lab[1].\term[2]) \lor \EPS_{\lab[3]}(@\lab[1].\term[3]).
    \end{align*}
\end{defi}
For a \kl{label} $\lab[3]$ and a set $\ltset$ of \kl{lKL terms}, let $\EPS_{\lab[3]}(\ltset) \defeq \bigvee_{\lterm[1] \in \ltset} \EPS_{\lab[3]}(\lterm[1])$.
This $\EPS_{\lab[3]}$ (\Cref{definition: eps}) is defined so that we can encode the \kl{emptiness property} $\EPS_{\lab[3]}$ (in \Cref{section: left quotients of runs}) as follows.
\begin{lem}\label{lemma: eps}
    Let $\struc$ be a \kl{structure}.
    For all $\lab[3] \in \univ{\struc}$ and \kl{lKL terms} $\lterm$,
    we have:
    \[\EPS_{\lab[3]}(\run^{\struc}_{\lab[3]}(\lterm)) \;\iff\; \EPS_{\lab[3]}(\lterm).\]
\end{lem}
\begin{proof}
    By easy induction on $\lterm$, with unfolding the definitions.
    \begin{itemize}
        \item Case $\lterm = @\lab[1]. a, @\lab[1].\emp, @\lab[1].\term[2] \intersection \term[3], \lterm[2] \intersection_{1} \lterm[3]$:
        By $\EPS_{\lab[3]}(\run^{\struc}_{\lab[3]}(\lterm))$ iff $\const{false}$.

        \item Case $\lterm = @\lab[1]. \id, @\lab[1].\term[2]^*$:
        By $\EPS_{\lab[3]}(\run^{\struc}_{\lab[3]}(\lterm))$ iff $\begin{cases}
            \const{true}  & (\lab[1] = \lab[3])      \\
            \const{false} & (\lab[1] \neq \lab[3]).
        \end{cases}$

        \item Case $\lterm = @\lab[1].\term[2] \compo \term[3]$:
        We have:
        \begin{align*}
            &\EPS_{\lab[3]}(\run^{\struc}_{\lab[3]}(@\lab[1].\term[2] \compo \term[3]))\\
            &\;\iff\; \EPS_{\lab[3]}(\bigcup_{\lab[2] \in \univ{\struc}} \run^{\struc}_{\lab[2]}(@\lab[1].\term[2]) \series \run^{\struc}_{\lab[3]}(@\lab[2].\term[3])) \tag{\Cref{definition: run lab} of $\run^{\struc}_{\lab[3]}$}\\
            &\;\iff\; \bigvee_{\lab[2] \in \univ{\struc}} \EPS_{\lab[2]}(\run^{\struc}_{\lab[3]}(@\lab[1].\term[2])) \land \EPS_{\lab[3]}(\run^{\struc}_{\lab[3]}(@\lab[2].\term[3])) \tag{By $\EPS_{\lab[3]}(\trace[1] \series\trace[2])$ iff $\EPS_{\lab[3]}(\trace[1]) \land \EPS_{\lab[3]}(\trace[2])$}\\
            &\;\iff\; \EPS_{\lab[3]}(\run^{\struc}_{\lab[3]}(@\lab[1].\term[2])) \land \EPS_{\lab[3]}(\run^{\struc}_{\lab[3]}(@\lab[3].\term[3])) \tag{$\EPS_{\lab[2]}(\run^{\struc}_{\lab[3]}(@\lab[1].\term[2]))$ is $\const{false}$ when $\lab[2] \neq \lab[3]$}\\
            &\;\iff\; \EPS_{\lab[3]}(@\lab[1].\term[2]) \land \EPS_{\lab[3]}(@\lab[3].\term[3]) \tag{IH}\\
            &\;\iff\; \EPS_{\lab[3]}(@\lab[1]. \term[2] \compo \term[3]). \tag{\Cref{definition: eps}}
        \end{align*}
        \item Case $\lterm = \lterm[2] \compo_{1} \term[3]$:
        Similar to the case above.

        \item Case $\lterm = @\lab[1].\term[2] \union \term[3]$:
        We have:
        \begin{align*}
        \EPS_{\lab[3]}(\run^{\struc}_{\lab[3]}(@\lab[1].\term[2] \union \term[3])) 
        &\;\iff\; \EPS_{\lab[3]}(\run^{\struc}_{\lab[3]}(@\lab[1].\term[2]) \cup \run^{\struc}_{\lab[3]}(@\lab[3].\term[3])) \tag{\Cref{definition: run lab} of $\run^{\struc}_{\lab[3]}$}\\
        &\;\iff\; \EPS_{\lab[3]}(\run^{\struc}_{\lab[3]}(@\lab[1].\term[2])) \lor \EPS_{\lab[3]}(\run^{\struc}_{\lab[3]}(@\lab[3].\term[3])) \tag{By $\EPS_{\lab[3]}(\glang) = \bigvee_{\trace \in \glang} \EPS_{\lab[3]}(\trace)$}\\
        &\;\iff\; \EPS_{\lab[3]}(@\lab[1].\term[2]) \lor \EPS_{\lab[3]}(@\lab[3].\term[3]) \tag{IH}\\
        &\;\iff\; \EPS_{\lab[3]}(@\lab[1]. \term[2] \union \term[3]). \tag{\Cref{definition: eps}}
        \end{align*}

    \end{itemize}\noindent 
    Hence, this completes the proof.
\end{proof}

\begin{defi}\label{definition: derivative}
    Let $\struc$ be a \kl{structure}.
    The \intro*\kl{derivative} relation $\lterm[1] \longrightarrow^{\struc}_{\trace[3]} \lterm[2]$,
    where $\lterm[1]$ and $\lterm[2]$ are \kl{lKL terms} and $\trace[3]$ is an \kl{$\struc$-run},
    is defined as the smallest relation (of tuples of $\lterm[1]$, $\lterm[2]$, and $\trace[3]$) closed under the following rules:
    \begin{gather*}
        \begin{prooftree}
            \hypo{\tuple{\lab[1], \lab[2]} \in \jump{a}_{\struc}}
            \infer1{@\lab. a \longrightarrow^{\struc}_{a^{\lab[1],\lab[2]}_{1}} @\lab[2].\id}
        \end{prooftree} \mbox{ for $a \in \vsig$}
        \qquad
        \begin{prooftree}
            \hypo{@\lab.\term[1] \longrightarrow^{\struc}_{\trace[3]} \lterm[1]'}
            \infer1{@\lab.\term[1] \compo \term[2] \longrightarrow^{\struc}_{\trace[3]} \lterm[1]' \compo \term[2]}
        \end{prooftree}  
        \qquad
        \begin{prooftree}
            \hypo{\EPS_{\lab[3]}(@\lab[1].\term[1])}
            \hypo{@\lab[3].\term[2] \longrightarrow^{\struc}_{\trace[3]} \lterm[2]'}
            \infer2{@\lab[1].\term[1] \compo \term[2] \longrightarrow^{\struc}_{\trace[3]} \lterm[2]'}
        \end{prooftree}
        \\
        \begin{prooftree}
            \hypo{@\lab.\term[1] \longrightarrow^{\struc}_{\trace[3]} \lterm[1]'}
            \infer1{@\lab.\term[1] \union \term[2] \longrightarrow^{\struc}_{\trace[3]} \lterm[1]'}
        \end{prooftree}
        \qquad
        \begin{prooftree}
            \hypo{@\lab.\term[2] \longrightarrow^{\struc}_{\trace[3]} \lterm[2]'}
            \infer1{@\lab.\term[1] \union \term[2] \longrightarrow^{\struc}_{\trace[3]} \lterm[2]'}
        \end{prooftree}
        \qquad
        \begin{prooftree}
            \hypo{\lterm[1]  \longrightarrow^{\struc}_{\trace[3]} \lterm[1]'}
            \infer1{\lterm[1] \compo_{1} \term[2] \longrightarrow^{\struc}_{\trace[3]} \lterm[1]' \compo_{1} \term[2]}
        \end{prooftree}
        \qquad
        \begin{prooftree}
            \hypo{\EPS_{\lab[3]}(\lterm[1])}
            \hypo{@ \lab[3].\term[2] \longrightarrow^{\struc}_{\trace[3]} \lterm[2]'}
            \infer2{\lterm[1] \compo_{1} \term[2] \longrightarrow^{\struc}_{\trace[3]} \lterm[2]'}
        \end{prooftree}
        \\
        \begin{prooftree}
            \hypo{@\lab.\term[1] \longrightarrow^{\struc}_{\trace[3]} \lterm[1]'}
            \infer1{ @\lab.\term[1]^* \longrightarrow^{\struc}_{\trace[3]} \lterm[1]' \compo \term[1]^*}
        \end{prooftree}
        \qquad
        \begin{prooftree}
            \hypo{\mathstrut}
            \infer1{@\lab.\term[1] \intersection \term[2] \longrightarrow^{\struc}_{\fork^{\lab[1]}_{1}} (@\lab.\term[1]) \intersection_{1} (@\lab.\term[2])}
        \end{prooftree}                                                          
        \qquad
        \begin{prooftree}
            \hypo{\EPS_{\lab[3]}(\lterm[1])}
            \hypo{\EPS_{\lab[3]}(\lterm[2])}
            \infer2{\lterm[1] \intersection_{1} \lterm[2] \longrightarrow^{\struc}_{\join^{\lab[3]}_{1}} @\lab[3].\id}
        \end{prooftree}
        \\
        \begin{prooftree}[center=false] 
            \hypo{\lterm[1] \longrightarrow^{\struc}_{\trace[3]} \lterm[1]'}
            \infer1{\lterm[1] \intersection_{1} \lterm[2] \longrightarrow^{\struc}_{\trace[3] \parallel \id^{\overrightarrow{\mathsf{lab}}(\lterm[2])}} \lterm[1]' \intersection_{1} \lterm[2]}
        \end{prooftree}
        \quad
        \begin{prooftree}[center=false]
            \hypo{\lterm[2] \longrightarrow^{\struc}_{\trace[3]} \lterm[2]'}
            \infer1{\lterm[1] \intersection_{1} \lterm[2] \longrightarrow^{\struc}_{\id^{\overrightarrow{\mathsf{lab}}(\lterm[1])} \parallel \trace[3]} \lterm[1] \intersection_{1} \lterm[2]'}
        \end{prooftree}
        \quad
        \begin{prooftree}[center=false]
            \hypo{\mathstrut}
            \infer1[\labeltext{R}{rule: derivative R}]{\lterm[1] \longrightarrow^{\struc}_{\id^{\overrightarrow{\mathsf{lab}}(\lterm[1])}} \lterm[1]}
        \end{prooftree}
        \quad
        \begin{prooftree}[center=false]
            \hypo{\lterm[1] \longrightarrow^{\struc}_{\trace[3]} \lterm[3]}
            \hypo{\lterm[3] \longrightarrow^{\struc}_{\smash{\trace[3]'}} \lterm[2]}
            \infer2[\labeltext{T}{rule: derivative T}]{\lterm[1] \longrightarrow^{\struc}_{\trace[3] \series \trace[3]'} \lterm[2]}
        \end{prooftree}
        \end{gather*}
\end{defi}
Let $\struc$ be a \kl{structure} and $\trace[3]$ be an \kl{$\struc$-run}.
For an \kl{lKL term} $\lterm$ (resp.\ a set $\ltset$ of \kl{lKL terms}),
we define the \kl{derivative} of $\lterm$ (resp.\ $\ltset$) with respect to $\trace[3]$ as follows:
\begin{align*}
    \D^{\struc}_{\trace[3]}(\lterm[1]) &\;\defeq\; \set{\lterm[2] \mbox{ an \kl{lKL term}} \mid \lterm[1] \longrightarrow^{\struc}_{\trace[3]} \lterm[2]}, & \D^{\struc}_{\trace[3]}(\ltset) &\;\defeq\; \bigcup_{\lterm[1] \in \ltset} \D^{\struc}_{\trace[3]}(\lterm[1]).
\end{align*}
For notational simplicity, we use the following notation: $A \mathop{\heartsuit} B \defeq \set{a \mathop{\heartsuit} b \mid a \in A \land b \in B}$
where $A$ and $B$ are sets of \kl{lKL terms} and $\mathop{\heartsuit} \in \set{\compo_{1}, \intersection_{1}}$.
In this notation, we use terms as singleton sets: for instance, $\set{a, b} \intersection_{1} c = \set{a \intersection_{1} c, b \intersection_{1} c}$.
Using this notation, for \emph{\kl{atomic $\struc$-runs}}, \Cref{definition: derivative} can be alternatively expressed as follows (\cf\  Antimirov's derivatives \cite{antimirovPartialDerivativesRegular1996}).
\begin{prop}\label{proposition: derivative on graphs}
    Let $\struc$ be a \kl{structure}.
    For an \kl{atomic $\struc$-run} $\trace[3]$,
    we have the following:
    \begingroup%
\allowdisplaybreaks
    \begin{align*}
        \D^{\struc}_{\trace[3]}(@\lab[1].a)                       &\;=\; \begin{cases}
             \set{@\lab[2].\id \mid \tuple{\lab[1], \lab[2]} \in \jump{a}_{\struc}} & (\trace[3] = a^{\lab[1], \lab[2]}_{1})   \\
            \emptyset & (\mbox{otherwise})
        \end{cases}  \mbox{ for $a \in \vsig$},                                                                                                                                                            \\
        \D^{\struc}_{\trace[3]}(@\lab[1].\id)                       &\;=\;  \D^{\struc}_{\trace[3]}(@\lab[1].\emp) \;=\; \emptyset,                                                             \\
        \D^{\struc}_{\trace[3]}(@\lab[1].\term[2] \compo \term[3]) &\;=\;  (\D^{\struc}_{\trace[3]}(@\lab[1].\term[2]) \compo_{1} \term[3]) \cup \bigcup_{\lab[3] \in \univ{\struc}} \D^{\struc}_{\trace[3]}(\set{@ \lab[3].\term[3] \mid \EPS_{\lab[3]}(@ \lab[1]. \term[2])}),                                                        \\
        \D^{\struc}_{\trace[3]}(@\lab[1].\term[2] \union \term[3])  &\;=\; \D^{\struc}_{\trace[3]}(@\lab[1].\term[2]) \cup \D^{\struc}_{\trace[3]}(@\lab[1].\term[3]),                                                                                                                                                                  \\
        \D^{\struc}_{\trace[3]}(@\lab[1].\term[2]^*)              &\;=\; \D^{\struc}_{\trace[3]}(@\lab[1].\term[2]) \compo_{1} \term[2]^*,                                                                                                                                                                                 \\
        \D^{\struc}_{\trace[3]}(@\lab[1].\term[2] \intersection \term[3])  &\;=\; \begin{cases}
            \set{(@\lab[1].\term[2]) \intersection_{1} (@ \lab[1].\term[3])}      & (\trace[3] = \fork^{\lab[1]}_{1})        \\
            \emptyset & (\mbox{otherwise}),
        \end{cases}                                                                                                                                                               \\
        \D^{\struc}_{\trace[3]}(\lterm[2] \compo_{1} \term[3])     &\;=\;  (\D^{\struc}_{\trace[3]}(\lterm[2]) \compo_{1} \term[3]) \cup \bigcup_{\lab[3] \in \univ{\struc}} \D^{\struc}_{\trace[3]}(\set{@\lab[3]. \term[3] \mid \EPS_{\lab[3]}(\lterm[2])}),                                                                         \\
        \D^{\struc}_{\trace[3]}(\lterm[2] \intersection_{1} \lterm[3])     &\;=\;  \begin{cases}
            \D^{\struc}_{\trace[3]'}(\lterm[2]) \intersection_{1} \lterm[3] & (\mbox{$\trace[3] = (\trace[3]' \parallel \id^{\overrightarrow{\mathsf{lab}}(\lterm[3])})$ for some $\trace[3]'$})\\
            \lterm[2] \intersection_{1} \D^{\struc}_{\trace[3]'}(\lterm[3]) & (\mbox{$\trace[3] = (\id^{\overrightarrow{\mathsf{lab}}(\lterm[2])} \parallel \trace[3]')$ for some $\trace[3]'$})\\
            \set{@\lab[3]. \id \mid \EPS_{\lab[3]}(\lterm[2]) \land \EPS_{\lab[3]}(\lterm[3])} & (\trace[3] = \join^{\lab[3]}_{1})\\
            \emptyset & (\mbox{otherwise}).
        \end{cases}
    \end{align*}
    \endgroup%
\end{prop}
\begin{proof}
    By a routine verification.
\end{proof}
The \kl{derivatives} can simulate \kl{left quotients} of \kl{$\struc$-runs} in that
$\D^{\struc}_{\trace[3]}(\lterm)$ expresses the \kl{left quotient} of $\run^{\struc}_{\lab[3]}(\lterm)$ with respect to $\trace[3]$, as follows.
See \Cref{section: lemma: derivative}, for a detailed proof.
\begin{lem}\label{lemma: derivative}
    Let $\struc$ be a \kl{structure} and $\lab[3] \in \univ{\struc}$.
    For all \kl{lKL terms} $\lterm$ and \kl{$\struc$-runs} $\trace[3]$,
    we have:
    \[\D^{\struc}_{\trace[3]}(\run^{\struc}_{\lab[3]}(\lterm)) \;=\; \run^{\struc}_{\lab[3]}(\D^{\struc}_{\trace[3]}(\lterm)).\]
\end{lem}
\begin{exa}[\cf\ \Cref{section: left quotients of struc-runs}]\label{example: derivatives}
    Let $\struc = \left(\begin{tikzpicture}[baseline = -.5ex]
        \graph[grow right = 1.cm, branch down = 2.5ex]{
        {s1/{$\lab[1]$}[vert]} -!- {t1/{$\lab[2]$}[vert]}
        };
        \graph[use existing nodes, edges={color=black, pos = .5, earrow}, edge quotes={fill=white, inner sep=1pt,font= \scriptsize}]{
            s1 ->["$a$", bend right] t1;
            t1 ->["$b$", bend right] s1;
        };
    \end{tikzpicture} \right)$,
    $\lterm = @\lab[1]. (a ( (b a) \intersection \id) b) \intersection \id$,
    $\trace[3] =$
    \begin{center}
    $\left(\hspace{-.5em}\begin{tikzpicture}[baseline = -2.5ex]
        \graph[grow right = 1.cm, branch down = 2.ex]{
        {/, 11/{$\lab[1]$}[vert]} -!-
        {21/{$\lab[1]$}[vert], 2c/, 22/{$\lab[1]$}[vert]} -!-
        {31/{$\lab[2]$}[vert],}
        };
        \path (11) edge [opacity = 0] node[pos= .5, elabel](c){$\fork$}(2c);
        \node[left = .5em of 11](l1){};
        \node[right = .5em of 31](r1){\tiny $1$};
        \node[right = .5em of 22](r2){\tiny $2$};
        \graph[use existing nodes, edges={color=black, pos = .5, earrow}, edge quotes={fill=white, inner sep=1pt,font= \scriptsize}]{
            11 -- c;
            c -> ["$1$"{auto}] 21;
            c -> ["$2$"{auto, below left = .3ex}] 22;
            21 ->["$a$"] 31;
            l1 -> 11; 31 -> r1; 22 -> r2;
        };
    \end{tikzpicture}\hspace{-.5em} \right)$,
    and $\trace[3]' = \left(\hspace{-.5em}\begin{tikzpicture}[baseline = -4.ex]
        \graph[grow right = 1.cm, branch down = 2.ex]{
        {/, 11/{$\lab[2]$}[vert], /, 12/{$\lab[1]$}[vert]} -!-
        {21/{$\lab[2]$}[vert], 11mid/, 22/{$\lab[2]$}[vert], /,} -!-
        {31/{$\lab[1]$}[vert]} -!-
        {41/{$\lab[2]$}[vert], 51mid/} -!-
        {/, 51/{$\lab[2]$}[vert]} -!-
        {/, 61/{$\lab[1]$}[vert], 71mid/} -!-
        {/, /, 71/{$\lab[1]$}[vert]}
        };
        \path (11) edge [opacity = 0] node[pos= .5, elabel](11c){$\fork$}(11mid);
        \path (51mid) edge [opacity = 0] node[pos= .5, elabel](51c){$\join$}(51);
        \path (71mid) edge [opacity = 0] node[pos= .5, elabel](71c){$\join$}(71);
        \node[left = .5em of 11](l1){\tiny $1$};
        \node[left = .5em of 12](l2){\tiny $2$};
        \node[right = .5em of 71](r1){};
        \graph[use existing nodes, edges={color=black, pos = .5, earrow}, edge quotes={fill=white, inner sep=1pt,font= \scriptsize}]{
            11 -- 11c;
            11c ->["$1$"{auto}] 21;
            11c ->["$2$"{auto, below left = .3ex}] 22;
            21 ->["$b$"] 31;
            31 ->["$a$"] 41;
            {41, 22} -- 51c -> 51;
            51 ->["$b$"] 61;
            61 --["$1$"{auto}] 71c;
            12 --["$2$"{auto, below right = .3ex}, bend right = 3, pos = .9] 71c;
            71c -> 71;
            l1 -> 11; l2 -> 12; 71 -> r1;
        };
    \end{tikzpicture} \hspace{-.5em}\right)$.\end{center}

    \noindent 
    We can see $\trace[3]' \in \D_{\trace[3]}(\run^{\struc}_{\lab[1]}(\lterm))$ by $\trace[3] \series \trace[3]' \in \run^{\struc}_{\lab[1]}(\lterm)$.
    We also have $\trace[3]' \in \run^{\struc}_{\lab[1]}(\D_{\trace[3]}(\lterm))$, because by letting $\lterm' = (@\lab[2]. ( (b a) \intersection \id) b) \intersection_{1} @\lab[1]. \id$,
    we have $\trace[3]' \in \run^{\struc}(\lterm')$ and $\lterm' \in \D_{\trace[3]}(\lterm)$.
    Here, $\lterm' \in \D_{\trace[3]}(\lterm)$ is shown as follows:
    \begin{align*}
        \lterm \;=\;  @\lab[1]. (a ( (b a) \intersection \id) b) \intersection \id
        & \quad\longrightarrow^{\struc}_{\scalebox{.8}{$\left(\begin{tikzpicture}[baseline = -2.ex]
            \graph[grow right = 1.cm, branch down = 2.ex]{
            {/, 11/{$\lab[1]$}[vert]} -!-
            {21/{$\lab[1]$}[vert], 2c/, 22/{$\lab[1]$}[vert]} -!-
            };
            \path (11) edge [opacity = 0] node[pos= .5, elabel](c){$\fork$}(2c);
            \node[left = .5em of 11](l1){\tiny };
            \node[right = .5em of 21](r1){\tiny $1$};
            \node[right = .5em of 22](r2){\tiny $2$};
            \graph[use existing nodes, edges={color=black, pos = .5, earrow}, edge quotes={fill=white, inner sep=1pt,font= \scriptsize}]{
                11 -- c -> {21, 22};
                l1 -> 11; 21 -> r1; 22 -> r2;
            };
        \end{tikzpicture} \right)$}}\quad (@\lab[1]. a ( (b a) \intersection \id) b) \intersection_{1} @\lab[1]. \id\\
        &\quad\longrightarrow^{\struc}_{\scalebox{.8}{$\left(\begin{tikzpicture}[baseline = -2.ex]
            \graph[grow right = 1.cm, branch down = 2.ex]{
            {21/{$\lab[1]$}[vert], 2c/, 22/{$\lab[1]$}[vert]} -!-
            {31/{$\lab[2]$}[vert],}
            };
            \node[left = .5em of 21](l1){\tiny $1$};
            \node[left = .5em of 22](l2){\tiny $2$};
            \node[right = .5em of 31](r1){\tiny $1$};
            \node[right = .5em of 22](r2){\tiny $2$};
            \graph[use existing nodes, edges={color=black, pos = .5, earrow}, edge quotes={fill=white, inner sep=1pt,font= \scriptsize}]{
                21 ->["$a$"] 31;
                l1 -> 21; l2 -> 22; 31 -> r1; 22 -> r2;
            };
        \end{tikzpicture} \right)$}}\quad (@\lab[2]. ( (b a) \intersection \id) b) \intersection_{1} @\lab[1]. \id \;=\; \lterm'.
    \end{align*}
\end{exa}
We let $(\longrightarrow^{\struc}) \defeq \bigcup_{\text{$\trace[3]$ an \kl{$\struc$-run}}} (\longrightarrow^{\struc}_{\trace[3]})$.
For an \kl{lKL term} $\lterm$ and a set $\tilde{T}$ of \kl{lKL terms},
we let:
\begin{align*}
    \D^{\struc}(\lterm[1]) &\;\defeq\; \set{\lterm[2] \mbox{ an \kl{lKL term}}  \mid \lterm[1] \longrightarrow^{\struc} \lterm[2]}, &
    \D^{\struc}(\tilde{T}) &\;\defeq\; \bigcup_{\lterm[1] \in \tilde{T}} \D^{\struc}(\lterm[1]).
\end{align*}
By \Cref{lemma: eps,lemma: derivative}, \kl{derivatives} can give an alternative semantics, as follows.
\begin{thm}\label{theorem: derivative}
    Let $\struc$ be a \kl{structure} and let $\lterm$ be an \kl{lKL term}.
    We have:
    \[\jump{\lterm}_{\struc} \;=\; \set{z \in \univ{\struc} \mid \EPS_{\lab[3]}(\D^{\struc}(\lterm))}.\]
\end{thm}
\begin{proof}
    We have:
    \begin{align*}
       \lab[3] \in \jump{\lterm}_{\struc}
        &\;\iff\; \run^{\struc}_{\lab[3]}(\lterm) \neq \emptyset \tag{\Cref{proposition: labeled struc-run language}}\\
        &\;\iff\; \EPS_{\lab[3]}(\D^{\struc}_{\trace[3]}(\run^{\struc}_{\lab[3]}(\lterm))) \mbox{ for some \kl{$\struc$-run} $\trace[3]$} \tag{\Cref{proposition: left quotient}}\\
        &\;\iff\; \EPS_{\lab[3]}(\D^{\struc}_{\trace[3]}(\lterm)) \mbox{ for some \kl{$\struc$-run} $\trace[3]$} \tag{\Cref{lemma: derivative,lemma: eps}}\\
        &\;\iff\; \EPS_{\lab[3]}(\D^{\struc}(\lterm)). \tag{By $\D^{\struc}(\lterm) = \bigcup_{\text{$\trace[3]$ an \kl{$\struc$-run}}} \D^{\struc}_{\trace[3]}(\lterm)$}
    \end{align*}
    Hence this completes the proof.
\end{proof}

\subsection{Proof of {\Cref{lemma: derivative}}: Derivatives can simulate left quotients}\label{section: lemma: derivative}
We first prove the following proposition.
For $\vec{\lab}$ and $\vec{\lab[2]}$,
we use $X_{\vec{\lab}, \vec{\lab[2]}}, Y_{\vec{\lab}, \vec{\lab[2]}}, \dots$ to denote sets of \kl{$\struc$-runs} with $\tuple{\vec{\lab}, \vec{\lab[2]}}$-interface.
\begin{prop}\label{proposition: derivative}
    Let $\trace[3]$ be an \kl{atomic $\struc$-run}.
    We then have the following.

    \labeltext{$(\series)$}{proposition: derivative cdot}:
    $\D^{\struc}_{\trace[3]}(X_{\vec{\lab},\lab[2]} \series Y_{\lab[2], \lab[3]}) = 
    \begin{cases}
        (\D^{\struc}_{\trace[3]}(X_{\vec{\lab},\lab[2]}) \series Y_{\lab[2], \lab[3]}) \cup \D^{\struc}_{\trace[3]}(Y_{\lab[2], \lab[3]}) & (\EPS_{\lab[2]}(X_{\vec{\lab},\lab[2]}))\\
        (\D^{\struc}_{\trace[3]}(X_{\vec{\lab},\lab[2]}) \series Y_{\lab[2], \lab[3]}) & (\mbox{otherwise});
    \end{cases}$

    \labeltext{$(\parallel)$}{proposition: derivative cap}:
    $\D^{\struc}_{\trace[3]}(X_{\vec{\lab[1]}_1, \lab[3]_1} \parallel Y_{\vec{\lab[1]}_2, \lab[3]_2}) = 
    \begin{cases}
        \D^{\struc}_{\trace[3]'}(X_{\vec{\lab[1]}_1, \lab[3]_1}) \parallel Y_{\vec{\lab[1]}_2, \lab[3]_2}  & (\mbox{$\trace[3] = (\trace[3]' \parallel \id^{\overrightarrow{\mathsf{lab}}(\vec{\lab[1]_2})})$ for some $\trace[3]'$})     \\
        X_{\vec{\lab[1]}_1, \lab[3]_1} \parallel \D^{\struc}_{\trace[3]'}(Y_{\vec{\lab[1]}_2, \lab[3]_2})  & (\mbox{$\trace[3] = ( \id^{\overrightarrow{\mathsf{lab}}(\vec{\lab[1]_1})} \parallel \trace[3]')$ for some $\trace[3]'$})     \\
        \emptyset & (\mbox{otherwise}).
    \end{cases}$
\end{prop}
\begin{proof}
    For \nameref{proposition: derivative cdot}: Because for all $\trace[1], \trace[2], \trace[3]'$ such that $\ty_{2}(\trace[1]) = \ty_{1}(\trace[2]) = y$, we have
    \[\trace[1] \series \trace[2] \longrightarrow^{\struc}_{\trace[3]} \trace[3]' \;\iff\; 
    \bigvee \left\{\begin{aligned}
        & \exists \trace[1]', \trace[1] \longrightarrow^{\struc}_{\trace[3]} \trace[1]' \land \trace[3]' = \trace[1]' \series \trace[2]\\
        & \EPS_{\lab[2]}(\trace[1]) \land \trace[2] \longrightarrow^{\struc}_{\trace[3]} \trace[3]'.
    \end{aligned}\right.\]
    \noindent 
    For \nameref{proposition: derivative cap}: Because for all $\trace[1], \trace[2], \trace[3]'$, we have
    \[\trace[1] \parallel \trace[2] \longrightarrow^{\struc}_{\trace[3]} \trace[3]' \;\iff\; 
    \bigvee \left\{\begin{aligned}
        & \exists \trace[1]',\trace[3]',\ \trace[1] \longrightarrow^{\struc}_{\trace[3]''} \trace[1]' \land \trace[3]' = (\trace[1]' \parallel \trace[2]) \land \trace[3] = (\trace[3]'' \parallel \id^{\ty_1(\trace[2])}) \\
        & \exists \trace[2]',\trace[3]',\ \trace[2] \longrightarrow^{\struc}_{\trace[3]''} \trace[2]' \land \trace[3]' = (\trace[1] \parallel \trace[2]') \land \trace[3] = (\id^{\ty_1(\trace[1])} \parallel \trace[3]''). 
    \end{aligned}\right. \]

    Hence, this completes the proof. \qedhere
\end{proof}
\noindent 
We first show \Cref{lemma: derivative} for \kl{atomic $\struc$-runs}, and then we extend for \kl{$\struc$-runs}.
\begin{lem}\label{lemma: derivative atomic}
    Let $\struc$ be a \kl{structure} and $\lab[3] \in \univ{\struc}$.
    For all \kl{lKL terms} $\lterm$ and \kl{atomic $\struc$-runs} $\trace[3]$,
    \[\D^{\struc}_{\trace[3]}(\run^{\struc}_{\lab[3]}(\lterm)) \;=\; \run^{\struc}_{\lab[3]}(\D^{\struc}_{\trace[3]}(\lterm)).\]
\end{lem}
\begin{proof}
    By easy induction on $\lterm$ using \Cref{proposition: derivative}.
    \begin{itemize}
        \item Case $\lterm = @\lab[1].a$ for $a \in \vsig$:
        We have:
        \begin{align*}
            \D^{\struc}_{\trace[3]}(\run^{\struc}_{\lab[3]}(@\lab[1].a))
            &\;=\; \D^{\struc}_{\trace[3]}(\set{a^{\lab[1], \lab[3]}_{1} \mid \tuple{\lab[1], \lab[3]} \in \jump{a}_{\struc}}) \tag{\Cref{definition: run lab} of $\run^{\struc}_{\lab[3]}$}\\
            &\;=\; \begin{cases}
                \run^{\struc}_{\lab[3]}(\set{@\lab[3].\id \mid \tuple{\lab[1], \lab[3]} \in \jump{a}_{\struc}}) & (\trace[3] = a^{\lab[1], \lab[3]}_{1})   \\
               \emptyset & (\mbox{otherwise})
           \end{cases} \;=\; \run^{\struc}_{\lab[3]}(\D^{\struc}_{\trace[3]}(@\lab[1].a)).
        \end{align*}

        \item Case $\lterm = @\lab[1].\id, @\lab[1].\emp$:
        We have: $\D^{\struc}_{\trace[3]}(\run^{\struc}_{\lab[3]}(\lterm)) = \emptyset = \run^{\struc}_{\lab[3]}(\D^{\struc}_{\trace[3]}(\lterm))$.

        \item Case $\lterm = @\lab[1].\term[2] \union \term[3]$:
        We have:
        \begin{align*}
            \D^{\struc}_{\trace[3]}(\run^{\struc}_{\lab[3]}(@\lab[1].\term[2] \union \term[3]))
            &\;=\; \D^{\struc}_{\trace[3]}(\run^{\struc}_{\lab[3]}(@\lab[1].\term[2])) \cup  \D^{\struc}_{\trace[3]}(\run^{\struc}_{\lab[3]}(@\lab[1].\term[3]))   \tag{\Cref{definition: run lab} of $\run^{\struc}_{\lab[3]}$}                      \\
            &\;=\; \run^{\struc}_{\lab[3]}(\D^{\struc}_{\trace[3]}(@\lab[1].\term[2])) \cup \run^{\struc}_{\lab[3]}(\D^{\struc}_{\trace[3]}(@\lab[1].\term[3]))
            = \run^{\struc}_{\lab[3]}(\D^{\struc}_{\trace[3]}(@\lab[1].\term[2] \union \term[3])).       \tag{IH} 
        \end{align*}

        \item Case $\lterm = @\lab[1].\term[2] \compo \term[3]$:
        We have:
        \begin{align*}
            &\D^{\struc}_{\trace[3]}(\run^{\struc}_{\lab[3]}(@\lab[1].\term[2] \compo \term[3])) \;=\; \bigcup_{\lab[2] \in \univ{\struc}} \D^{\struc}_{\trace[3]}(\run^{\struc}_{\lab[2]}(@\lab[1].\term[2]) \series \run^{\struc}_{\lab[2], \lab[3]}(\term[3])) \tag{\Cref{definition: run lab} of $\run^{\struc}_{\lab[3]}$}\\
            &= (\bigcup_{\lab[2] \in \univ{\struc}} \D^{\struc}_{\trace[3]}(\run^{\struc}_{\lab[2]}(@\lab[1].\term[2])) \series \run^{\struc}_{\lab[2], \lab[3]}(\term[3]))
            \cup \bigcup_{\lab[2] \in \univ{\struc}; \EPS_{\lab[2]}(\run^{\struc}_{\lab[2]}(@\lab[1].\term[2]))} \D^{\struc}_{\trace[3]}( \run^{\struc}_{\lab[3]}(@ \lab[2]. \term[3]))
            \tag{\Cref{proposition: derivative} \nameref{proposition: derivative cdot}}  \\
            &= (\bigcup_{\lab[2] \in \univ{\struc}} \run^{\struc}_{\lab[2]}(\D^{\struc}_{\trace[3]}(@\lab[1].\term[2])) \series \run^{\struc}_{\lab[2], \lab[3]}(\term[3]))
            \cup \bigcup_{\lab[2] \in \univ{\struc}; \EPS_{\lab[2]}(@\lab[1].\term[2])} \run^{\struc}_{\lab[3]}(\D^{\struc}_{\trace[3]}(@ \lab[2]. \term[3]))               \tag{\Cref{lemma: eps}, IH}                \\
            &= \run^{\struc}_{\lab[3]}\left((\D^{\struc}_{\trace[3]}(@\lab[1].\term[2]) \compo_{1} \term[3]) \cup \bigcup_{\lab[2] \in \univ{\struc}; \EPS_{\lab[2]}(@\lab[1].\term[2])} \D^{\struc}_{\trace[3]}(@\lab[2]. \term[3]) \right)  = \run^{\struc}_{\lab[3]}(\D^{\struc}_{\trace[3]}(@\lab[1].\term[2] \compo \term[3])).                             \tag{\Cref{definition: run lab}}
        \end{align*}

        \item Case $\lterm = @\lab[1].\term[2]^*$:
        We have:
        \begin{align*}
            \D^{\struc}_{\trace[3]}(\run^{\struc}_{\lab[3]}(@\lab[1].\term[2]^*))
            &\;=\; \D^{\struc}_{\trace[3]}(\bigcup_{\lab[2] \in \univ{\struc}} \set{\trace \in \run^{\struc}_{\lab[2]}(@ \lab[1]. \term[2]) \mid \lnot \EPS_{\lab[2]}(\trace)} \series \run^{\struc}_{\lab[2], \lab[3]}(\term[2]^{*})) \tag{As $\trace[3]$ is not empty}\\
            &\;=\; \bigcup_{\lab[2] \in \univ{\struc}} \D^{\struc}_{\trace[3]}(\set{\trace \in \run^{\struc}_{\lab[2]}(@ \lab[1]. \term[2]) \mid \lnot \EPS_{\lab[2]}(\trace)}) \series \run^{\struc}_{\lab[2], \lab[3]}(\term[2]^{*}) \tag{\Cref{proposition: derivative} \nameref{proposition: derivative cdot}} \\
            &\;=\; \bigcup_{\lab[2] \in \univ{\struc}} \D^{\struc}_{\trace[3]}(\run^{\struc}_{\lab[2]}(@ \lab[1]. \term[2])) \series \run^{\struc}_{\lab[2], \lab[3]}(\term[2]^{*})        \tag{As $\trace[3]$ is not empty}                     \\
            &\;=\; \bigcup_{\lab[2] \in \univ{\struc}} \run^{\struc}_{\lab[2]}(\D^{\struc}_{\trace[3]}(@\lab[1].\term[2])) \series \run^{\struc}_{\lab[2], \lab[3]}(\term[2]^*)         \tag{IH}                     \\
            &\;=\; \run^{\struc}_{\lab[3]}(\D^{\struc}_{\trace[3]}(@\lab[1].\term[2]) \compo_{1} \term[2]^*) \;=\; \run^{\struc}_{\lab[3]}(\D^{\struc}_{\trace[3]}(@\lab[1].\term[2]^*)). \tag{\Cref{definition: run lab}}
        \end{align*}

        \item Case $\lterm = @\lab[1].\term[2] \intersection \term[3]$:
        We have:
        \begin{align*}
            \D^{\struc}_{\trace[3]}(\run^{\struc}_{\lab[3]}(@\lab[1].\term[2] \intersection \term[3]))
            &\;=\; \D^{\struc}_{\trace[3]}(\fork^{\lab[1]}_{1} \series (\run^{\struc}_{\lab[3]}(@\lab[1].\term[2]) \parallel \run^{\struc}_{\lab[3]}(@ \lab[1].\term[3])) \series \join^{\lab[3]}_{1}) \tag{\Cref{definition: run lab} of $\run^{\struc}_{\lab[3]}$} \\
            &\;=\; \begin{cases}
                (\run^{\struc}_{\lab[3]}(@\lab[1].\term[2]) \parallel \run^{\struc}_{\lab[3]}(@ \lab[1].\term[3])) \series \join^{\lab[3]}_{1}      & (\trace[3] = \fork^{\lab[1]}_1)        \\
                \emptyset & (\mbox{otherwise})
            \end{cases}
            &= \run^{\struc}_{\lab[3]}(\D^{\struc}_{\trace[3]}(@\lab[1].\term[2] \intersection \term[3])).
        \end{align*}

        \item Case $\lterm = \lterm[2] \compo_{1} \term[3]$:
        Similar to Case $\lterm = @\lab[1].\term[2] \compo \term[3]$.

        \item Case $\lterm = \lterm[2] \intersection_{1} \lterm[3]$:
        We have:
        \begin{align*}
            &\D^{\struc}_{\trace[3]}(\run^{\struc}_{\lab[3]}(\lterm[2] \intersection_{1} \lterm[3]))
            \;=\; \D^{\struc}_{\trace[3]}((\run^{\struc}_{\lab[3]}(\lterm[2]) \parallel \run^{\struc}_{\lab[3]}(\lterm[3])) \series \join_1^{\lab[3]})  \tag{\Cref{definition: run lab} of $\run^{\struc}_{\lab[3]}$}\\
            &\;=\;  \begin{cases}
                (\run^{\struc}_{\lab[3]}(\D^{\struc}_{\trace[3]'}(\lterm[2])) \parallel \run^{\struc}_{\lab[3]}(\lterm[3])) \series \join_1^{\lab[3]}  & (\trace[3] = (\trace[3]' \parallel \id^{\overrightarrow{\mathsf{lab}}(\lterm[3])}) \mbox{ for some $\trace[3]'$})\\
                (\run^{\struc}_{\lab[3]}(\lterm[2]) \parallel \run^{\struc}_{\lab[3]}(\D^{\struc}_{\trace[3]'}(\lterm[3]))) \series \join_1^{\lab[3]} & (\trace[3] = (\id^{\overrightarrow{\mathsf{lab}}(\lterm[2])} \parallel \trace[3]') \mbox{ for some $\trace[3]'$})\\
                \set{\id^{\lab[3]} \mid \EPS_{\lab[3]}(\lterm[2]) \land \EPS_{\lab[3]}(\lterm[3])} & (\trace[3] = \join_1^{\lab[3]})\\
                \emptyset & (\mbox{otherwise})
            \end{cases} \tag{\Cref{proposition: derivative}, IH}   \\
            &\;=\;  \begin{cases}
                \run^{\struc}_{\lab[3]}(\D^{\struc}_{\trace[3]'}(\lterm[2]) \intersection_{1} \lterm[3]) & (\trace[3] = (\trace[3]' \parallel \id^{\overrightarrow{\mathsf{lab}}(\lterm[3])}) \mbox{ for some $\trace[3]'$})\\
                \run^{\struc}_{\lab[3]}(\lterm[2] \intersection_{1} \D^{\struc}_{\trace[3]'}(\lterm[3])) & (\trace[3] = (\id^{\overrightarrow{\mathsf{lab}}(\lterm[2])} \parallel \trace[3]') \mbox{ for some $\trace[3]'$})\\
                \run^{\struc}_{\lab[3]}(\set{@\lab[3]. \id \mid \EPS_{\lab[3]}(\lterm[2]) \land \EPS_{\lab[3]}(\lterm[3])}) & (\trace[3] = \join_1^{\lab[3]})\\
                \emptyset & (\mbox{otherwise})
            \end{cases} \\ 
            &\;=\; \run^{\struc}_{\lab[3]}(\D^{\struc}_{\trace[3]}(\lterm[2] \intersection_{1} \lterm[3])).
        \end{align*}
    \end{itemize}
    Hence this completes the proof.
\end{proof}

\begin{proof}[Proof of \Cref{lemma: derivative}]
    By easy induction on the number $k$ of edges occurring in $\trace[3]$ using \Cref{lemma: derivative atomic}.
    We distinguish the following cases:
    \begin{itemize}
        \item Case $k = 0$:
        We have:
        \begin{align*}
            \D^{\struc}_{\trace[3]}(\run^{\struc}_{\lab[3]}(\lterm))
            &\;=\; \begin{cases}
                \run^{\struc}_{\lab[3]}(\lterm) & (\trace[3] = \id^{\overrightarrow{\mathsf{lab}}(\lterm)})\\
                \emptyset & (\mbox{otherwise})
            \end{cases} \;=\; \run^{\struc}_{\lab[3]}(\D^{\struc}_{\trace[3]}(\lterm)). \tag{By the rule (\nameref{rule: derivative R})}
        \end{align*}
        \item Case $k \ge 1$:
        By \Cref{proposition: struc-run atomic},
        let $a_1, \dots, a_k$ be \kl{atomic $\struc$-runs} such that $\trace[3] = a_1 \series \dots \series a_k$.
        We have:
        \begin{align*}
            \D^{\struc}_{a_1 \series \dots \series a_k}(\run^{\struc}_{\lab[3]}(\lterm)) 
            &\;=\; \D^{\struc}_{a_k}(\D^{\struc}_{a_{k-1}}( \dots (\D^{\struc}_{a_1}(\run^{\struc}_{\lab[3]}(\lterm))))) \tag{By the rule (\nameref{rule: derivative T})}\\
            &\;=\; \run^{\struc}_{\lab[3]}(\D^{\struc}_{a_k}(\D^{\struc}_{a_{k-1}}( \dots (\D^{\struc}_{a_1}(\lterm))))) \tag{By \Cref{lemma: derivative atomic}, iteratively}\\
            &\;=\; \run^{\struc}_{\lab[3]}(\D^{\struc}_{a_1 \series \dots \series a_k}(\lterm)). \tag{By the rule (\nameref{rule: derivative T})}
        \end{align*}
    \end{itemize}
    Hence this completes the proof.
\end{proof}

\subsection{Closure property of derivatives}
Moreover, the \kl{derivatives} above have a closure property like Antimirov's derivatives \cite{antimirovPartialDerivativesRegular1996,bastosStateComplexityPartial2016}.
\begin{defi}\label{definition: closure}
    Let $L$ be a set (of \kl{labels}).
    For a \kl{KL term} $\term$, the \intro*\kl{closure} $\cl_{L}(\term)$ is the set of \kl{lKL terms} defined by:
    \begin{align*}
        \cl_{L}(a)                       &\;\defeq\; \set{@\lab[2]. a, @\lab[2].\id \mid \lab[2] \in L} \mbox{ for $a \in \vsig$},                                                                                                                                                                 \\
        \cl_{L}(\id)                       &\;\defeq\; \set{@\lab[2].\id \mid \lab[2] \in L},\\
        \cl_{L}(\emp) &\;\defeq\;  \set{\emp},                                                \\
        \cl_{L}(\term[2] \compo \term[3]) &\;\defeq\; \set{@\lab[2].\term[2] \compo \term[3] \mid \lab[2] \in L} \cup (\cl_{L}(\term[2]) \compo_{1} \term[3]) \cup \cl_{L}(\term[3]), \\
        \cl_{L}(\term[2] \union \term[3])  &\;\defeq\; \set{@\lab[2].\term[2] \union \term[3] \mid \lab[2] \in L} \cup \cl_{L}(\term[2]) \cup \cl_{L}(\term[3]),                                                                                                                                                                  \\
        \cl_{L}(\term[2]^*)              &\;\defeq\; \set{@\lab[2].\term[2]^* \mid \lab[2] \in L} \cup (\cl_{L}(\term[2]) \compo_{1} \term[2]^*),                                                                                                                                                                                 \\
        \cl_{L}(\term[2] \intersection \term[3])  &\;\defeq\; \set{@\lab[2].\term[2] \intersection \term[3], @\lab[2]. \id \mid \lab[2] \in L} \cup (\cl_{L}(\term[2]) \intersection_{1} \cl_{L}(\term[3])).
    \end{align*}
    For an \kl{lKL term} $\lterm \in \LT_{L}$, the \kl{closure} $\cl_{L}(\lterm)$ is the set of \kl{lKL terms} defined by:
    \begin{align*}
        \cl_{L}(@\lab[1].\term)                       &\;\defeq\; \cl_{L}(\term),                                                                                        \\
        \cl_{L}(\lterm[2] \compo_{1} \term[3])     &\;\defeq\; (\cl_{L}(\lterm[2]) \compo_{1} \term[3]) \cup \cl_{L}(\term[3]),    \\
        \cl_{L}(\lterm[2] \intersection_{1} \lterm[3])     &\;\defeq\; \set{@\lab[2]. \id \mid \lab[2] \in L} \cup (\cl_{L}(\lterm[2]) \intersection_{1} \cl_{L}(\lterm[3])).
    \end{align*}
    We then extend $\cl_{L}$ for sets of \kl{lKL terms}, by $\cl_{L}(\ltset) \defeq \bigcup_{\lterm \in \ltset} \cl_{L}(\lterm)$.
\end{defi}
The function $\cl_{L}$ is a closure operator, as follows.
\begin{prop}\label{proposition: closure operator}
    For each set $L$, the function $\cl_{L} \colon \wp(\LT_{L}) \to \wp(\LT_{L})$ is a closure operator.
\end{prop}
\begin{proof}
    Because we have the following, respectively:
    \begin{itemize}
        \item Extensivity $\ltset[1] \subseteq \cl_{L}(\ltset[1])$:
        Because $\lterm \in \cl_{L}(\lterm)$ (by definition).
        \item Monotonicity $\ltset[1] \subseteq \ltset[2] \Longrightarrow \cl_{L}(\ltset[1]) \subseteq \cl_{L}(\ltset[2])$:
        By $\cl_{L}(\ltset[1]) = \bigcup_{\lterm \in \ltset[1]} \cl_{L}(\lterm) \subseteq \bigcup_{\lterm \in \ltset[2]} \cl_{L}(\lterm) = \cl_{L}(\ltset[2])$.
        \item Idempotency $\cl_{L}(\cl_{L}(\ltset[1])) \subseteq \cl_{L}(\ltset[1])$:
        It suffices to prove that $\cl_{L}(\cl_{L}(\term)) \subseteq \cl_{L}(\term)$ and $\cl_{L}(\cl_{L}(\lterm)) \subseteq \cl_{L}(\lterm)$.
        This is shown by induction on the \kl{size} of $\term$ (resp.\ $\lterm$). \qedhere
    \end{itemize}
\end{proof}
\noindent 
Additionally, the closure operator $\cl_{L}$ is \emph{algebraic}, i.e., $\cl_{L}(\ltset) = \bigcup_{\ltset[2] \subseteq \ltset; \text{$\ltset[2]$ is finite}} \cl_{L}(\ltset[2])$ holds, which is clear by $\cl_{L}(\ltset[2]) = \bigcup_{\lterm \in \ltset[2]} \cl_{L}(\lterm)$ for each $\ltset[2]$.

We then have that the \kl{lKL terms} obtained by \kl{derivatives} of an \kl{lKL term} $\lterm$ are always in the closure of $\lterm$, as follows.
\begin{prop}\label{proposition: closure}
    Let $\struc$ be a \kl{structure} and let $\trace[3]$ be a \kl{$\struc$-run}.
    For all \kl{lKL terms} $\lterm \in \LT_{\univ{\struc}}$,
    we have
    $\D^{\struc}_{\trace[3]}(\lterm) \subseteq \cl_{\univ{\struc}}(\lterm)$.
\end{prop}
\begin{proof}
    When $\trace[3]$ is an \kl{atomic $\struc$-run}, this is shown by easy induction on the \kl{size} of $\lterm$ (by comparing \Cref{proposition: derivative on graphs} and \Cref{definition: closure}).
    By extensivity and idempotency (\Cref{proposition: closure operator}), we can extend this property
    from \kl{atomic $\struc$-runs} to \kl{$\struc$-runs}.
\end{proof}

Moreover, on the closure size, we have the following result.
\begin{prop}\label{proposition: closure size}
    Let $L$ be a set.
    For all \kl{KL terms} $\term$,
        $\# \cl_{L}(\term) \le (2 \#L \|\term\|)^{\iw(\term)}$.
\end{prop}
\begin{proof}
    By easy induction on $\term$.
    \begin{itemize}
        \item Case $\term = \aterm$ where $\aterm \in \vsig \cup \set{\id, \emp}$:                                        
        By $\# \cl_{L}(\term) \le 2 \# L \le (2 \# L \|\term\|)^{1}$.
        \item Case $\term = \term[2] \mathbin{\heartsuit} \term[3]$ where $\mathbin{\heartsuit} \in \set{\compo, \union}$:
        We have:
        \begin{align*}
            \# \cl_{L}(\term[2] \mathbin{\heartsuit} \term[3])
            & \le \#L + \# \cl_{L}(\term[2]) + \# \cl_{L}(\term[3])  \tag{\Cref{definition: closure} of $\cl_{L}$}\\
            & \le \#L + (2 \# L \|\term[2]\|)^{\iw(\term[2])} + (2 \# L \|\term[3]\|)^{\iw(\term[3])}  \tag{IH}\\
            & \le (2 \# L)^{\max(\iw(\term[2]), \iw(\term[3]))} (1 + \|\term[2]\|^{\max(\iw(\term[2]), \iw(\term[3]))} + \|\term[3]\|^{\max(\iw(\term[2]), \iw(\term[3]))}) \\
            & \le (2 \# L \|\term[2] \mathbin{\heartsuit} \term[3]\|)^{\iw(\term[2] \mathbin{\heartsuit} \term[3])}.
        \end{align*}

        \item Case $\term = \term[2]^*$:
        We have:
        \begin{align*}
            \# \cl_{L}(\term[2]^*)
            & \le \#L + \# \cl_{L}(\term[2])  \tag{\Cref{definition: closure} of $\cl_{L}$}\\
            & \le \#L + (2 \# L \|\term[2]\|)^{\iw(\term[2])} \tag{IH}\\
            & \le (2 \# L)^{\iw(\term[2])} (1 + \|\term[2]\|^{\iw(\term[2])}) \le (2 \# L \|\term[2]^*\|)^{\iw(\term[2]^*)}.
        \end{align*}

        \item Case $\term = \term[2] \intersection \term[3]$:
        We have:
        \begin{align*}
            \# \cl_{L}(\term[2] \intersection \term[3])
            & \le \#L + \# \cl_{L}(\term[2]) \# \cl_{L}(\term[3]) + \# L  \tag{\Cref{definition: closure} of $\cl_{L}$}\\
            & \le 2 \#L + (2 \# L \|\term[2]\|)^{\iw(\term[2])} (2 \# L \|\term[3]\|)^{\iw(\term[3])} \tag{IH}\\
            & \le (2 \#L)^{\iw(\term[2]) + \iw(\term[3])} (1 + \|\term[2]\|^{\iw(\term[2])} \|\term[3]\|^{\iw(\term[3])})\\
            & \le (2 \#L)^{\iw(\term[2]) + \iw(\term[3])} (1 + \|\term[2]\| + \|\term[3]\|)^{\iw(\term[2]) + \iw(\term[3])}         
            \le (2 \# L \|\term[2] \intersection \term[3]\|)^{\iw(\term[2] \intersection \term[3])}.
        \end{align*}
    \end{itemize}
    Hence, this completes the proof.
\end{proof}
\section{Decomposing Derivatives: An Automata Construction} \label{section: automata construction}
In this section, we show that the \kl{derivatives} are decomposable on \kl{path decompositions} (\Cref{theorem: decomposition} in \Cref{section: decomposing derivatives}),
and we give an automata construction (\Cref{section: automata construction}).

We let $\STR_{k} \defeq \set{\struc \in \STR \mid \univ{\struc} \subseteq \range{k}}$.
For $\strucclass \subseteq \STR$,
we write $\isoc(\strucclass)$ for the \intro*\kl{isomorphic closure} of $\strucclass$, i.e., $\isoc(\strucclass) \defeq \set{\struc[2] \in \STR \mid \struc \in \strucclass \mbox{ and } \struc \cong \struc[2]}$.

\subsection{Gluing operator for path decompositions}
For $\struc_1, \dots, \struc_n \in \STR$,
the \intro*\kl{disjoint union} $\bigsqcup_{i = 1}^{n} \struc_i$ is defined as
the \kl{structure} $\tuple{\set{\tuple{i, x} \mid i \in \range{n}, x \in \univ{\struc_i}},
\set{\tuple{\tuple{i, x}, \tuple{i, y}} \mid i \in \range{n}, \tuple{x, y} \in a^{\struc_i}}_{a \in \vsig}}$.
We consider the following gluing operator.
\begin{defi}\label{definition: gluing structures}
    Let $\vec{\struc} = \struc_1 \dots \struc_n \in \STR^{+}$.
    The \kl{structure} $\odot \vec{\struc}$ (or written $\odot_{i = 1}^{n} \struc_i$) is defined as follows:
    \[\odot \vec{\struc} \;\defeq\; \left(\bigsqcup_{i = 1}^{n} \struc_i \right)/{\sim_{\vec{\struc}}}\]
    where $\sim_{\vec{\struc}}$ is the minimal equivalence relation on the set $\set{\tuple{i, x} \mid i \in \range{n}, x \in \univ{\struc_i}}$ closed under the following rule:
    \begin{quotation}
        For all $x$ and $i \in \range{n-1}$, if $x \in \univ{\struc_i} \cap \univ{\struc_{i+1}}$, then $\tuple{i, x} \sim_{\vec{\struc}} \tuple{i+1, x}$.
    \end{quotation}
    (For notational simplicity, we abbreviate $\sim_{\vec{\struc}}$ to $\sim$ when $\vec{\struc}$ is clear from the context.)
\end{defi}
For instance, $\bigodot \left(\begin{tikzpicture}[baseline = -.5ex]
    \graph[grow right = 1.cm, branch down = 2.5ex]{
    {s1/{$1$}[vert]} -!- {t1/{$2$}[vert]}
    };
    \graph[use existing nodes, edges={color=black, pos = .5, earrow}, edge quotes={fill=white, inner sep=1pt,font= \scriptsize}]{
        s1 ->["$a$", bend right] t1;
        t1 ->["$b$", bend right] s1;
    };
\end{tikzpicture} \right) \left(\begin{tikzpicture}[baseline = -.5ex]
    \graph[grow right = 1.cm, branch down = 2.5ex]{
    {s1/{$2$}[vert]} -!- {t1/{$3$}[vert]}
    };
    \graph[use existing nodes, edges={color=black, pos = .5, earrow}, edge quotes={fill=white, inner sep=1pt,font= \scriptsize}]{
        s1 ->["$a$", bend right] t1;
        t1 ->["$b$", bend right] s1;
    };
\end{tikzpicture} \right) \left(\begin{tikzpicture}[baseline = -.5ex]
    \graph[grow right = 1.cm, branch down = 2.5ex]{
    {s1/{$3$}[vert]} -!- {t1/{$1$}[vert]}
    };
    \graph[use existing nodes, edges={color=black, pos = .5, earrow}, edge quotes={fill=white, inner sep=1pt,font= \scriptsize}]{
        s1 ->["$a$", bend right] t1;
        t1 ->["$b$", bend right] s1;
    };
\end{tikzpicture} \right) = \left(\begin{tikzpicture}[baseline = -.5ex]
    \graph[grow right = 1.cm, branch down = 2.5ex]{
    {1/{}[vert]} -!- {2/{}[vert]} -!- {3/{}[vert]} -!- {4/{}[vert]}
    };
    \graph[use existing nodes, edges={color=black, pos = .5, earrow}, edge quotes={fill=white, inner sep=1pt,font= \scriptsize}]{
        1 ->["$a$", bend right] 2; 2 ->["$b$", bend right] 1;
        2 ->["$a$", bend right] 3; 3 ->["$b$", bend right] 2;
        3 ->["$a$", bend right] 4; 4 ->["$b$", bend right] 3;
    };
    \end{tikzpicture}
    \right)$.
This gluing operator transforms a given \kl{path decomposition} into the original \kl{structure} up to \kl[graph isomorphism]{isomorphisms}.
\begin{prop}\label{proposition: gluing structures}
    Let $\struc$ be a \kl{structure}.
    If $\vec{\struc[2]}$ is a \kl{path decomposition} of $\struc$,
    then $\struc \cong \odot \vec{\struc[2]}$.
\end{prop}
\begin{proof}
    Easy by definition.
\end{proof}
From this, we have:
\[\REL_{\pw \le k} = \isoc(\set{\odot \vec{\struc} \mid \vec{\struc} \in \STR_{k+1}^+}).\]
The direction ($\supseteq$) is clear by the definition of \kl{pathwidth} (\Cref{proposition: gluing structures}).
For the direction ($\subseteq$), let $\vec{\struc}$ be a \kl{path decomposition} of \kl(pathwidth){width} $k$.
We then can replace names of \kl{vertices} in each \kl(pathwidth){bag} with numbers in $\range{k+1}$
by introducing an intermediate \kl(pathwidth){bag} whose \kl{vertices} consist of the intersection of the two adjacent \kl(pathwidth){bags}.
For instance, we can transform $\vec{\struc}$ into $\vec{\struc[2]}$ with preserving $\odot \vec{\struc} \cong \odot \vec{\struc[2]}$ as follows:
\begin{align*}
&\vec{\struc} = \left(\begin{tikzpicture}[baseline = -.5ex]
    \graph[grow right = .9cm, branch down = 2.5ex]{
    {s1/{$1$}[vert]} -!- {s2/{$2$}[vert]} -!- {s3/{$3$}[vert]} -!- {s4/{$4$}[vert]}
    };
    \graph[use existing nodes, edges={color=black, pos = .4, earrow}, edge quotes={fill=white, inner sep=1pt,font= \scriptsize}]{
        s1 ->["$a$"] s2 ->["$a$"] s3 ->["$a$"] s4 -> ["$a$", bend right] s1;
    };
\end{tikzpicture}\right)\hspace{-.4em}
\left(\begin{tikzpicture}[baseline = -.5ex]
    \graph[grow right = .9cm, branch down = 2.5ex]{
    {s1/{$3$}[vert]} -!- {s2/{$4$}[vert]} -!- {s3/{$x$}[vert]} -!- {s4/{$y$}[vert]}
    };
    \graph[use existing nodes, edges={color=black, pos = .4, earrow}, edge quotes={fill=white, inner sep=1pt,font= \scriptsize}]{
        s1 ->["$a$"] s2 ->["$a$"] s3 ->["$a$"] s4 -> ["$a$", bend right] s1;
    };
\end{tikzpicture}\right)\hspace{-.4em}\left(\begin{tikzpicture}[baseline = -.5ex]
    \graph[grow right = .9cm, branch down = 2.5ex]{
    {s1/{$x$}[vert]} -!- {s2/{$y$}[vert]} -!- {s3/{$z$}[vert]} -!- {s4/{$w$}[vert]}
    };
    \graph[use existing nodes, edges={color=black, pos = .4, earrow}, edge quotes={fill=white, inner sep=1pt,font= \scriptsize}]{
        s1 ->["$a$"] s2 ->["$a$"] s3 ->["$a$"] s4 -> ["$a$", bend right] s1;
    };
\end{tikzpicture}\right) \quad\leadsto \\
&\vec{\struc[2]} = \left(\begin{tikzpicture}[baseline = -.5ex]
    \graph[grow right = .9cm, branch down = 2.5ex]{
    {s1/{$1$}[vert]} -!- {s2/{$2$}[vert]} -!- {s3/{$3$}[vert]} -!- {s4/{$4$}[vert]}
    };
    \graph[use existing nodes, edges={color=black, pos = .4, earrow}, edge quotes={fill=white, inner sep=1pt,font= \scriptsize}]{
        s1 ->["$a$"] s2 ->["$a$"] s3 ->["$a$"] s4 -> ["$a$", bend right] s1;
    };
\end{tikzpicture}\right)\hspace{-.4em}
\left(\begin{tikzpicture}[baseline = -.5ex]
    \graph[grow right = .9cm, branch down = 2.5ex]{
    {s1/{$3$}[vert]} -!- {s2/{$4$}[vert]}
    };
    \graph[use existing nodes, edges={color=black, pos = .4, earrow}, edge quotes={fill=white, inner sep=1pt,font= \scriptsize}]{
        s1 ->["$a$"] s2;
    };
\end{tikzpicture}\right)\hspace{-.4em}
\left(\begin{tikzpicture}[baseline = -.5ex]
    \graph[grow right = .9cm, branch down = 2.5ex]{
    {s1/{$3$}[vert]} -!- {s2/{$4$}[vert]} -!- {s3/{$1$}[vert]} -!- {s4/{$2$}[vert]}
    };
    \graph[use existing nodes, edges={color=black, pos = .4, earrow}, edge quotes={fill=white, inner sep=1pt,font= \scriptsize}]{
        s1 ->["$a$"] s2 ->["$a$"] s3 ->["$a$"] s4 -> ["$a$", bend right] s1;
    };
\end{tikzpicture}\right)\hspace{-.4em}
\left(\begin{tikzpicture}[baseline = -.5ex]
    \graph[grow right = .9cm, branch down = 2.5ex]{
    {s1/{$1$}[vert]} -!- {s2/{$2$}[vert]}
    };
    \graph[use existing nodes, edges={color=black, pos = .4, earrow}, edge quotes={fill=white, inner sep=1pt,font= \scriptsize}]{
        s1 ->["$a$"] s2;
    };
\end{tikzpicture}\right)\hspace{-.4em}
\left(\begin{tikzpicture}[baseline = -.5ex]
    \graph[grow right = .9cm, branch down = 2.5ex]{
    {s1/{$1$}[vert]} -!- {s2/{$2$}[vert]} -!- {s3/{$3$}[vert]} -!- {s4/{$4$}[vert]}
    };
    \graph[use existing nodes, edges={color=black, pos = .4, earrow}, edge quotes={fill=white, inner sep=1pt,font= \scriptsize}]{
        s1 ->["$a$"] s2 ->["$a$"] s3 ->["$a$"] s4 -> ["$a$", bend right] s1;
    };
\end{tikzpicture}\right).
\end{align*}
By the \kl{linearly bounded pathwidth model property} (\Cref{proposition: bounded pw property}),
the \kl{equational theory} of \kl{KL terms} can be reformulated as follows:
\[\REL \models \term[1] \le \term[2] \;\iff\; \set{\odot \vec{\struc} \mid \vec{\struc} \in \STR_{\iw(\term[1]) + 1}^+} \models \term[1] \le \term[2].\]
Because \kl{path decompositions} are sequences (\kl{words}) of \kl{structures}, this characterization is compatible with the automata construction given later.

\subsection{Decomposing derivatives}\label{section: decomposing derivatives}
Let $\bullet$ be a fresh \kl{label} for denoting a special isolated \kl{vertex}.
For a \kl{structure} $\struc$, let $\struc_{\bullet} \defeq \tuple{\univ{\struc} \uplus \set{\bullet}, \set{a^{\struc}}_{a \in \vsig}}$.
Clearly, for all $\lterm[1], \lterm[2] \in \LT_{\univ{\struc}}$,
$\lterm[1] \longrightarrow^{\struc} \lterm[2]$ iff $\lterm[1] \longrightarrow^{\struc_{\bullet}} \lterm[2]$.
Let $\vec{\struc} = \struc_1 \dots \struc_n \in \STR^{+}$.
In the sequel, we consider the \kl{structure} $\odot ((\struc_1)_{\bullet} \dots (\struc_n)_{\bullet})$; we use $(\odot \vec{\struc})_{\bullet}$ to denote this \kl{structure} (as they are isomorphic) and we may abbreviate $[\tuple{i, \bullet}]_{\sim_{(\odot \vec{\struc})_{\bullet}}}$ to $\bullet$. 
For $i \in \range{n}$, we let $\LT^{\vec{\struc}}_{(i)} \defeq \LT_{\set{[\tuple{i, \lab}]_{\sim_{(\odot \vec{\struc})_{\bullet}}} \mid \lab \in \univ{(\struc_i)_{\bullet}}}}$.
For an \kl{lKL term} $\lterm$ and $i \in \range{n}$, let $\lterm_{(i)}$ be the \kl{lKL term} $\lterm$ in which each $\lab$ has been replaced with $[\tuple{i, \lab}]_{\sim_{(\odot \vec{\struc})_{\bullet}}}$.
We now consider the following composition of \kl{derivative} relations.
\begin{defi}\label{definition: decomposition}
    Let $\vec{\struc} = \struc_1 \dots \struc_n \in \STR^{+}$.
    The relation $\lterm[1] \mathrel{(\odot_{i = 1}^{n} \longrightarrow^{(\struc_{i})_{\bullet}})} \lterm[2]$,
    where $\tuple{\lterm[1], \lterm[2]} \in \bigcup_{i = 1}^{n} (\LT^{\vec{\struc}}_{(i)})^2$,
    is defined as the smallest relation (of tuples of $i \in \range{1, n}$, $\lterm[1] \in \LT^{\vec{\struc}}_{(i)}$, and $\lterm[2] \in \LT^{\vec{\struc}}_{(i)}$) closed under the following rules:
    \begin{gather*}
        \begin{prooftree}
            \hypo{\lterm[1] \longrightarrow^{(\struc_{j})_{\bullet}} \lterm[2]}
            \infer1[\labeltext{D}{rule: decomposition D}]{\lterm[1]_{(j)} \mathrel{(\odot_{i = 1}^{n} \longrightarrow^{(\struc_{i})_{\bullet}})} \lterm[2]_{(j)}}
        \end{prooftree} \hspace{5em}
        \begin{prooftree}
            \hypo{\lterm[1][\bullet/l] \mathrel{(\odot_{i = 1}^{n} \longrightarrow^{(\struc_{i})_{\bullet}})} \lterm[2][\bullet/r]}
            \infer1[\labeltext{L}{rule: decomposition L}]{\lterm[1][\lab/l] \mathrel{(\odot_{i = 1}^{n} \longrightarrow^{(\struc_{i})_{\bullet}})} \lterm[2][\lab/r]}
        \end{prooftree}\\
        \begin{prooftree}
            \hypo{\lterm[1] \mathrel{(\odot_{i = 1}^{n} \longrightarrow^{(\struc_{i})_{\bullet}})} \lterm[3]}
            \hypo{\lterm[3] \mathrel{(\odot_{i = 1}^{n} \longrightarrow^{(\struc_{i})_{\bullet}})} \lterm[2]}
            \infer2[\labeltext{T}{rule: decomposition T}]{\lterm[1] \mathrel{(\odot_{i = 1}^{n} \longrightarrow^{(\struc_{i})_{\bullet}})} \lterm[2]}
        \end{prooftree}.
    \end{gather*}
    Here, $\lterm[1][\lab[2]/j]$ denotes the \kl{lKL term} $\lterm[1]$ in which the $j$-th \kl{label} of $\lterm[1]$ has been replaced with $\lab[2]$.
\end{defi}
\begin{exa}\label{example: decomposition}
    Let $\vec{\struc} = \textcolor{blue}{\struc_1} \textcolor{red!60}{\struc_2}$ where
    $\textcolor{blue}{\struc_1} = \left(\begin{tikzpicture}[baseline = -.5ex]
        \graph[grow right = 1.cm, branch down = 2.5ex]{
        {s1/{$1$}[vert]} -!- {t1/{$2$}[vert]}
        };
        \graph[use existing nodes, edges={color=black, pos = .5, earrow}, edge quotes={fill=white, inner sep=1pt,font= \scriptsize}]{
            s1 ->["$a$", bend right] t1;
            t1 ->["$b$", bend right] s1;
        };
    \end{tikzpicture} \right)$
     and 
    $\textcolor{red!60}{\struc_2} = \left(\begin{tikzpicture}[baseline = -.5ex]
        \graph[grow right = 1.cm, branch down = 2.5ex]{
        {s1/{$2$}[vert]} -!- {t1/{$3$}[vert]}
        };
        \graph[use existing nodes, edges={color=black, pos = .5, earrow}, edge quotes={fill=white, inner sep=1pt,font= \scriptsize}]{
            s1 ->["$a$", bend right] t1;
            t1 ->["$b$", bend right] s1;
        };
    \end{tikzpicture} \right)$.
    Let $\lterm[1] = @[\tuple{1,1}]_{\sim}. (a^2 ( (b^2 a^2) \intersection \id) b^2) \intersection \id$
    and let $\lterm[2] = @[\tuple{1, 1}]_{\sim}. \id$.
    We then have $\lterm[1] \longrightarrow^{(\odot \vec{\struc})_{\bullet}} \lterm[2]$ by the following \kl[$\struc$-run]{$(\odot \vec{\struc})_{\bullet}$-run}:
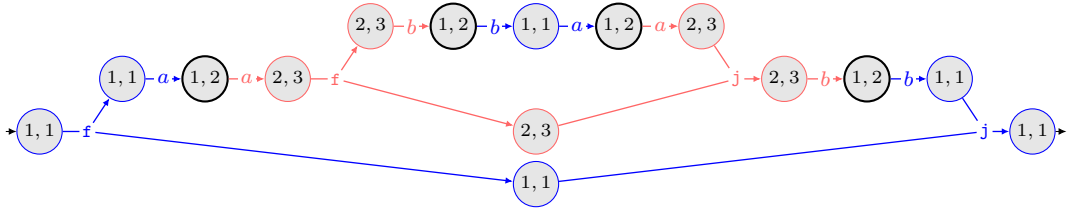
    \begin{center}
        \begin{tikzpicture}[baseline = -.5ex, yscale=-1]
            \foreach \x/\y/\label/\style/\dc in {
                0/2/{$1, 1$}//{blue},1/1/{$1, 1$}//{blue},2/1/{$1, 2$}/{thick}/,3/1/{$2, 3$}//{red!60},
                4/0/{$2, 3$}//{red!60},5/0/{$1, 2$}/{thick}/,6/0/{$1, 1$}//{blue},6/2/{$2, 3$}//{red!60},6/3/{$1, 1$}//{blue},
                7/0/{$1, 2$}/{thick}/,8/0/{$2, 3$}//{red!60},9/1/{$2, 3$}//{red!60},10/1/{$1, 2$}/{thick}/,11/1/{$1, 1$}//{blue},
                12/2/{$1, 1$}//{blue}}{
               \node[mynode,draw = \dc,circle, \style]  (\x X\y) at (1.15*\x,.7*\y) {\tiny \label};
            }
            \node[left = 4pt of 0X2](s){}; \path (s) edge[earrow, ->] (0X2);
            \node[right = 4pt of 12X2](t){}; \path (12X2) edge[earrow, ->] (t);
            \foreach \xp/\x/\y/\col in {0/1/2/blue,3/4/1/red!60}{
                \node  (\x X\y) at (1.15*\x,.7*\y) {};
                \path (\xp X\y) edge [opacity = 0] node[pos= .5, elabel, color = \col](\xp X\y Xf){\tiny $\fork$}(\x X\y);
            }
            \foreach \xp/\x/\y/\col in {8/9/1/red!60,11/12/2/blue}{
                \node (\xp X\y) at (1.15*\xp,.7*\y) {};
                \path (\xp X\y) edge [opacity = 0] node[pos= .5, elabel, color = \col](\x X\y Xj){\tiny $\join$}(\x X\y);
            }
            \graph[use existing nodes, edges={color=black, pos = .4, earrow, line width = .5pt, 
            },edge quotes={fill=white, inner sep=1pt,font= \scriptsize}]{
            {0X2} --[color = blue] {0X2Xf} ->[color = blue] {1X1, 6X3};
            {1X1} ->["$a$", color = blue] {2X1} ->["$a$", color = red!60] {3X1};
            {3X1} --[color = red!60] {3X1Xf} -> [color = red!60] {4X0, 6X2};
            {4X0} ->["$b$", color = red!60] {5X0} ->["$b$", color = blue] {6X0};
            {6X0} ->["$a$", color = blue] {7X0} ->["$a$", color = red!60] {8X0};
            {8X0,6X2} --[color = red!60] {9X1Xj} -> [color = red!60] {9X1};
            {9X1} ->["$b$", color = red!60] {10X1} ->["$b$", color = blue] {11X1};
            {11X1,6X3} --[color = blue] {12X2Xj} -> [color = blue] {12X2};
            };
        \end{tikzpicture}
    \end{center}
    where we abbreviate $[\tuple{i,x}]_{\sim}$ to ``$i, x$'' in each \kl{vertex}.
    We let
   {\scriptsize\begin{align*}
        \lterm_1 &= (@[\tuple{1,2}]_{\sim}. a ( (b^2 a^2) \intersection \id) b^2) \intersection_{1} (@[\tuple{1,1}]_{\sim}. \id),&
        \lterm_1' &= (@[\tuple{1,2}]_{\sim}. a ( (b^2 a^2) \intersection \id) b^2) \intersection_{1} (@\bullet. \id),\\
        \lterm_2' &= (( (@[\tuple{2,2}]_{\sim}. b a^2) \intersection_{1} (@[\tuple{2,3}]_{\sim}. \id)) b^2) \intersection_{1} (@ \bullet. \id),&
        \lterm_2'' &= (( (@[\tuple{1,2}]_{\sim}. b a^2) \intersection_{1} (@\bullet. \id)) b^2) \intersection_{1} (@ \bullet. \id),\\
        \lterm_3'' &= (( (@[\tuple{1,2}]_{\sim}. a) \intersection_{1} (@\bullet. \id)) b^2) \intersection_{1} (@ \bullet. \id),&
        \lterm_3' &= (( (@[\tuple{2,2}]_{\sim}. a) \intersection_{1} (@[\tuple{2,3}]_{\sim}. \id)) b^2) \intersection_{1} (@ \bullet. \id),\\
        \lterm_4' &= (@[\tuple{2,2}]_{\sim}. b) \intersection_{1} (@ \bullet. \id),&
        \lterm_4 &= (@[\tuple{1,2}]_{\sim}. b) \intersection_{1} (@ [\tuple{1,1}]_{\sim}. \id).
    \end{align*}}
    We also have $\lterm[1] \mathrel{(\odot_{i = 1}^{2} \longrightarrow^{(\struc_{i})_{\bullet}})} \lterm[2]$, as follows:
    \begin{center}
    \begin{prooftree}
        \hypo{\mathstrut}
        \infer1[\nameref{rule: decomposition D} ($\textcolor{blue}{\struc_1}$)]{\lterm[1] \leadsto \lterm_1}
        \hypo{\mathstrut}
        \infer1[\nameref{rule: decomposition D} ($\textcolor{red!60}{\struc_2}$)]{\lterm_1' \leadsto \lterm_2'}
        \hypo{\mathstrut}
        \infer1[\nameref{rule: decomposition D} ($\textcolor{blue}{\struc_1}$)]{\lterm_2'' \leadsto \lterm_3''}
        \infer1[\nameref{rule: decomposition L}]{\lterm_2' \leadsto \lterm_3'}
        \hypo{\mathstrut}
        \infer1[\nameref{rule: decomposition D} ($\textcolor{red!60}{\struc_2}$)]{\lterm_3' \leadsto \lterm_4'}
        \infer[double]3[\nameref{rule: decomposition T}]{\lterm_1' \leadsto \lterm_4'}
        \infer1[\nameref{rule: decomposition L}]{\lterm_1 \leadsto \lterm_4}
        \hypo{\mathstrut}
        \infer1[\nameref{rule: decomposition D} ($\textcolor{blue}{\struc_1}$)]{\lterm_4 \leadsto \lterm[2]}
        \infer[double]3[\nameref{rule: decomposition T}]{\lterm[1] \leadsto \lterm[2]}
    \end{prooftree}
    \end{center}
    Note that we cannot use pairs of \kl{lKL terms} that both \kl{labels} $[\tuple{1,1}]_{\sim}$ and $[\tuple{2,3}]_{\sim}$ appear in each \kl{vertex} of the derivation tree,
    because these \kl{labels} do not occur simultaneously in $\textcolor{blue}{\struc_1}$ or $\textcolor{red!60}{\struc_2}$;
    so, we should consider an appropriate strategy.
\end{exa}
The following decomposition theorem shows that by composing \kl{derivative} relations $\longrightarrow^{{\struc_{i}}}$ using the rules of \Cref{definition: decomposition},
we can obtain the \kl{derivative} relation $\longrightarrow^{(\odot \vec{\struc})_{\bullet}}$.
Namely, we can decompose a \kl{derivative} relation on a large \kl{structure} into \kl{derivative} relations on small \kl{structures}, as follows.
The soundness ($\Longleftarrow$) is easy.
For the completeness ($\Longrightarrow$), we appropriately decompose \kl{derivative} relations, as in \Cref{example: decomposition}.
The proof will be given in \Cref{section: theorem: decomposition}.
\begin{restatable}[Decomposition theorem]{thm}{decompositiontheorem}\label{theorem: decomposition}
    Let $\vec{\struc} = \struc_1 \dots \struc_n \in \STR^{+}$.
    For all $j \in \range{n}$ and \kl{lKL terms} $\lterm[1], \lterm[2] \in \LT^{\vec{\struc}}_{(j)}$,
    we have:
    \[\lterm[1] \longrightarrow^{(\odot \vec{\struc})_{\bullet}} \lterm[2] \;\iff\; \lterm[1] \mathrel{(\odot_{i = 1}^{n} \longrightarrow^{(\struc_{i})_{\bullet}})} \lterm[2].\]
\end{restatable}

\subsection{Reducing to 2AFAs}
Using \Cref{theorem: decomposition}, we can give a reduction from the \kl{equational theory} of PCoR*
into the inclusion problem of \kl{two-way alternating finite word automata} (\kl{2AFAs}) \cite{ladnerAlternatingPushdownAutomata1978,ladnerAlternatingPushdownStack1984} (the following definition is based on \cite{geffertTransformingTwoWayAlternating2014}).\footnote{See also \cite{kapoutsisAlternationTwowayFinite2021} for a comparison of definitions of ``2AFA''.
In their classification, our definition corresponds to monotone 2BFA, as we use monotone (but, possibly non-basic) formulas in transition functions.}

\subsubsection{2AFAs}
We use $\triangleright$ and $\triangleleft$ as the special characters denoting the left end-marker and right end-marker, respectively.
A \intro*\kl{2AFA} $\automaton$ over a finite set $A$ is a tuple $\automaton = \tuple{\univ{\automaton}, \delta^{\automaton}, \src^{\automaton}}$, where 
\begin{itemize}
    \item $\univ{\automaton}$ is a finite set of states;
    \item $\delta^{\automaton} \colon \univ{\automaton} \times (A \uplus \set{\triangleright, \triangleleft}) \to \mathbb{B}_{+}(\univ{\automaton} \times \set{-1, 0, 1})$ is a \emph{transition function},
    where $\mathbb{B}_{+}(X)$ denotes the set of positive boolean formulas over a set $X$ given by
    \begin{align*}
        \varphi, \psi \in \mathbb{B}_{+}(X) &\;\Coloneqq\; p \mid \const{f} \mid \const{t} \mid \varphi \lor \psi \mid \varphi \land \psi \tag*{($p \in X$);}
    \end{align*}
\item $\src^{\automaton} \in \univ{\automaton}$ is the initial state.
\end{itemize}\noindent 
For a \kl{2AFA} $\automaton$ and a \kl{word} $w = a_{0} \dots a_{n-1}$ over $A \uplus \set{\triangleright, \triangleleft}$,
the set $S^{\automaton}_{\word} \subseteq \univ{\automaton} \times \range{0,n-1}$ is defined as the smallest set closed under the following rule:
    For $\tuple{q, i} \in \univ{\automaton} \times \range{0,n-1}$
    and propositional variables $\tuple{q_1, i_1}, \dots, \tuple{q_{m}, i_{m}} \in \univ{\automaton} \times \set{-1, 0, 1}$,
    when the positive boolean formula $\delta^{\automaton}(q, a_i)$ is $\const{true}$ under that each $\tuple{q_k, i_k}$ is $\const{true}$, then 
    \[\begin{prooftree}
        \hypo{\tuple{q_1, i+ i_1} \in S^{\automaton}_{\word}}
        \hypo{\dots}
        \hypo{\tuple{q_{m}, i + i_{m}} \in S^{\automaton}_{\word}}
        \infer3{\tuple{q, i} \in S^{\automaton}_{\word}}
    \end{prooftree}.\]
For a \kl{2AFA} $\automaton$, the \kl{language} is defined as $\ljump{\automaton} \defeq \set{\word \in A^* \mid \tuple{\src^{\automaton}, 0} \in S^{\automaton}_{{\triangleright}\word{\triangleleft}}}$.
We define the \intro*\kl(2AFA){size} $\|\automaton\|$ as $\sum_{\tuple{q, a} \in \univ{\automaton} \times (A \uplus \set{\triangleright, \triangleleft})} \|\delta^{\automaton}(q, a)\|$
where $\|\varphi\|$ denotes the number of symbols occurring in the positive boolean formula $\varphi$.%
\footnote{Note that $\|\delta^{\automaton}(q, a)\|$ has no limit in general.
Also, even if we take $\delta^{\automaton}(q, a)$ as the minimum one among equivalent positive boolean formulas,
the \kl{size} is not bounded by a polynomial in $(\#\univ{\automaton} \#A)$,
because the number of monotone boolean functions (a.k.a.\ the Dedekind number) with $2 n$ variables,
is $\Omega(2^{\binom{2 n}{n}})$ by considering antichains in which each set has exactly $n$ elements (and hence, $2^{2^{\Omega(n)}}$ by Stirling's approximation). %
Due to this, we count up each $\|\delta^{\automaton}(q, a)\|$ in the definition of the \kl{size} $\|\automaton\|$.
}
\begin{prop}\label{proposition: 2AFA PSPACE}
    The \intro*\kl{inclusion problem} for \kl{2AFAs}---given a finite set $A$ and given two \kl{2AFAs} $\automaton[1]$ and $\automaton[2]$ over $A$, does $\ljump{\automaton[1]} \subseteq \ljump{\automaton[2]}$ hold?---is decidable in $\mathrm{PSPACE}$.
\end{prop}
\begin{proof}
    Let $\automaton[1]'$ be the (one-way) non-deterministic finite automata (\kl{1NFA}) of $2^{\mathcal{O}(n \log n)}$ states
    such that $\ljump{\automaton[1]'} = \ljump{\automaton[1]}$ obtained by \cite[Lemma 1]{geffertTransformingTwoWayAlternating2014},
    and let $\bar{\automaton[2]}'$ be the \kl{1NFA} of $2^{\mathcal{O}(n)}$ states such that $\ljump{\bar{\automaton[2]}'} = A^* \setminus \ljump{\automaton[2]}$ obtained by \cite[Lemma 5]{geffertTransformingTwoWayAlternating2014} (see also \cite[Theorem 8]{calvaneseViewBasedQueryAnswering2002}).
    Here, $n$ is the number of states in the input \kl{2AFA}. %
    We have $\ljump{\automaton[1]} \not\subseteq \ljump{\automaton[2]}$ iff $\ljump{\automaton[1]'} \cap \ljump{\bar{\automaton[2]}'} \neq \emptyset$.
    Then the right-hand side can be decided in a non-deterministic polynomial space algorithm
    by the ``on-the-fly'' checking of the non-emptiness problem of the product \kl{1NFA} of $\automaton[1]'$ and $\bar{\automaton[2]'}$.\footnote{
    See \cite{geffertTransformingTwoWayAlternating2014} for the precise constructions of $\automaton[1]'$ and $\bar{\automaton[2]'}$.
    For $\automaton[1]'$ \cite[Lemma 1]{geffertTransformingTwoWayAlternating2014},
    we can encode each state of $\automaton[1]'$ (expressed as a sequence over subsets of $\univ{\automaton[1]}$) in $\mathcal{O}(n \log n)$-space because its length is at most $2 n$ and each state of $\automaton[1]$ occurs at most $2$ times.
    Given $q, q' \in \univ{\automaton[1]'}$ and $a \in A \uplus \set{\triangleright,\triangleleft}$,
    the membership $\tuple{q, q'} \in a^{\automaton[1]'}$ can be decided in polynomial space, by construction.
    Similarly for $\bar{\automaton[2]'}$ \cite[Lemma 5]{geffertTransformingTwoWayAlternating2014},
    the membership $\tuple{q, q'} \in a^{\bar{\automaton[2]'}}$ can be decided in polynomial space.
    Thus, the ``on-the-fly'' checking is indeed possible.
    }
    Hence, by Savitch's theorem ($\mathrm{(co\text{-})NPSPACE} = \mathrm{PSPACE}$) \cite{savitchRelationshipsNondeterministicDeterministic1970}, this completes the proof.
\end{proof}

\subsubsection{2AFAs construction}
Using \kl{2AFAs}, we can naturally encode the rules of \Cref{definition: decomposition}.
\begin{defi}\label{definition: 2AFA}
    For $k \ge 1$ and a \kl{KL term} $\term[3]$,
    we define the \kl{2AFA} $\automaton_{k}^{\term[3]}$ over $\STR_{k}$ as follows:
    \begin{itemize}
        \item $\univ{\automaton_{k}^{\term[3]}} \defeq \set{\src} \uplus (\cl_{\range{k} \cup \set{\bullet}}(\term[3])^2 \times \set{?, \checkmark})$, we abbreviate $\tuple{\lterm[1], \lterm[2], p}$ to $\tuple{\lterm[1], \lterm[2]}_{p}$;
        \item $\src^{\automaton_{k}^{\term[3]}} \defeq \src$;
        \item $\delta^{\automaton_{k}^{\term[3]}}(q, a) \defeq \bigvee X(q, a)$ where
        $X(q, a) \subseteq \mathbb{B}_{+}(\univ{\automaton_{k}^{\term[3]}} \times \set{-1, 0, +1})$ is the smallest set closed under the following rules:
        \begin{gather*}
            \begin{prooftree}
                \hypo{\EPS_{\lab[2]}(\lterm[2])}
                \infer1[\labeltext{I}{rule: automata I} ($x, y \in \range{k}$)]{\tuple{\tuple{@ \lab[1].\ \term[3], \lterm[2]}_{?}, +1} \in X(\src^{\automaton_{k}^{\term[3]}}, \triangleright)}
            \end{prooftree}
            \hspace{1.5em}
            \begin{prooftree}
                \hypo{\overrightarrow{\mathsf{lab}}(\lterm[1]), \overrightarrow{\mathsf{lab}}(\lterm[2]) \in \univ{\struc_{\bullet}}^{+}}
                \infer1[\labeltext{$\checkmark$}{rule: automata check}]{\tuple{\tuple{\lterm[1], \lterm[2]}_{\checkmark}, 0} \in X(\tuple{\lterm[1], \lterm[2]}_{?}, \struc)}
            \end{prooftree}\\
            \begin{prooftree}
                \hypo{\mathstrut}
                \infer1[\labeltext{$-1$}{rule: automata -1}]{\tuple{\tuple{\lterm[1], \lterm[2]}_{?}, -1} \in X(\tuple{\lterm[1], \lterm[2]}_{\checkmark}, \struc)}
            \end{prooftree}\hspace{1.5em}
            \begin{prooftree}
                \hypo{\mathstrut}
                \infer1[\labeltext{$+1$}{rule: automata +1}]{\tuple{\tuple{\lterm[1], \lterm[2]}_{?}, +1} \in X(\tuple{\lterm[1], \lterm[2]}_{\checkmark}, \struc)}
            \end{prooftree}\hspace{1.5em}
            \begin{prooftree}
                \hypo{\lterm[1] \longrightarrow^{\struc_{\bullet}} \lterm[2]}
                \infer1[\labeltext{D}{rule: automata D}]{\const{t} \in X(\tuple{\lterm[1], \lterm[2]}_{\checkmark}, \struc)}
            \end{prooftree}\\
            \begin{prooftree}
                \hypo{\overrightarrow{\mathsf{lab}}(\lterm[3]') \in \univ{\struc_{\bullet}}^+}
                \infer1[\labeltext{T}{rule: automata T}]{\tuple{\tuple{\lterm[1], \lterm[3]'}_{\checkmark}, 0} \land \tuple{\tuple{\lterm[3]', \lterm[2]}_{\checkmark}, 0} \in X(\tuple{\lterm[1], \lterm[2]}_{\checkmark}, \struc)}
            \end{prooftree} \ 
            \begin{prooftree}
                \hypo{\mathstrut}
                \infer1[\labeltext{L}{rule: automata L}]{\tuple{\tuple{\lterm[1][\bullet/l], \lterm[2][\bullet/r]}_{\checkmark}, 0} \in X(\tuple{\lterm[1][x/l], \lterm[2][x/r]}_{\checkmark}, \struc)}
            \end{prooftree}
        \end{gather*}
    \end{itemize}
\end{defi}
The rules are defined based on \Cref{definition: decomposition}.
Intuitively, the check mark ``$\checkmark$'' in $\tuple{\lterm[1], \lterm[2]}_{\checkmark}$ expresses that $\overrightarrow{\mathsf{lab}}(\lterm[1]), \overrightarrow{\mathsf{lab}}(\lterm[2]) \in \univ{\struc_{\bullet}}^{+}$ holds (where $\struc$ is the \kl{structure} in the current position);
the question mark ``$?$'' in $\tuple{\lterm[1], \lterm[2]}_{?}$ expresses that it has not yet been checked.
They are introduced to check it after the rule $-1$ or $+1$ (as \kl{2AFAs} cannot read sibling \kl{structures} in one step).
\begin{exa}
    Let $\vec{\struc} = \textcolor{blue}{\struc_1} \textcolor{red!60}{\struc_2}$ where
    $\textcolor{blue}{\struc_1} = \left(\begin{tikzpicture}[baseline = -.5ex]
        \graph[grow right = 1.cm, branch down = 2.5ex]{
        {s1/{$1$}[vert]} -!- {t1/{$2$}[vert]}
        };
        \graph[use existing nodes, edges={color=black, pos = .5, earrow}, edge quotes={fill=white, inner sep=1pt,font= \scriptsize}]{
            s1 ->["$a$", bend right] t1;
            t1 ->["$b$", bend right] s1;
        };
    \end{tikzpicture} \right)$
     and 
    $\textcolor{red!60}{\struc_2} = \left(\begin{tikzpicture}[baseline = -.5ex]
        \graph[grow right = 1.cm, branch down = 2.5ex]{
        {s1/{$2$}[vert]} -!- {t1/{$3$}[vert]}
        };
        \graph[use existing nodes, edges={color=black, pos = .5, earrow}, edge quotes={fill=white, inner sep=1pt,font= \scriptsize}]{
            s1 ->["$a$", bend right] t1;
            t1 ->["$b$", bend right] s1;
        };
    \end{tikzpicture} \right)$.
    Let
    \begin{align*}
        \lterm_1 &= @1. (a^2 b^2) \intersection \id,&
        \lterm_2 &= (@2. a b^2) \intersection_1 (@1.\id),&
        \lterm_2' &= (@2. a b^2) \intersection_1 (@\bullet.\id),\\
        \lterm_3' &= (@2. b) \intersection_1 (@ \bullet.\id),&
        \lterm_3 &= (@2. b) \intersection_1 (@1.\id),&
        \lterm_4 &= @1. \id.
    \end{align*}
    Then $(\lterm_1)_{(1)} \leadsto (\lterm_4)_{(1)}$ holds by the following derivation tree (with respect to the rules \Cref{definition: decomposition}),
    where we abbreviate ${\mathrel{(\odot_{i = 1}^{2} \longrightarrow^{(\struc_{i})_{\bullet}})}}$ to $(\leadsto)$:
    \begin{center}
        \begin{prooftree}
            \hypo{\mathstrut}
            \infer1[D ($\textcolor{blue}{\struc_1})$]{(\lterm_1)_{(1)} \leadsto (\lterm_2)_{(1)}}
            \hypo{\mathstrut}
            \infer1[D ($\textcolor{red!60}{\struc_2})$]{(\lterm_2')_{(2)} \leadsto (\lterm_3')_{(2)}}
            \infer1[L]{(\lterm_2)_{(1)} \leadsto (\lterm_3)_{(1)}}
            \hypo{\mathstrut}
            \infer1[D ($\textcolor{blue}{\struc_1})$]{(\lterm_3)_{(1)} \leadsto (\lterm_4)_{(1)}}
            \infer[double]3[T]{(\lterm_1)_{(1)} \leadsto (\lterm_4)_{(1)}}
        \end{prooftree}
    \end{center}
    Let $\term[3] = a^2 b^2 \intersection \id$.
    This tree can be simulated in $\automaton_{k}^{\term[3]}$ as follows,
    where we omit ``$\in S^{{\automaton_{k}^{\term[3]}}}_{\triangleright \vec{\struc} \triangleleft}$'':
    \begin{center}
        \begin{prooftree}
            \hypo{\mathstrut}
            \infer1[D ($\textcolor{blue}{\struc_1})$]{\tuple{\tuple{\lterm_1, \lterm_2}_{\checkmark}, 1}}
            \hypo{\mathstrut}
            \infer1[D ($\textcolor{red!60}{\struc_2})$]{\tuple{\tuple{\lterm_2', \lterm_3'}_{\checkmark}, 2}}
            \infer1[$\checkmark$]{\tuple{\tuple{\lterm_2', \lterm_3'}_{?}, 2}}
            \infer1[$+1$]{\tuple{\tuple{\lterm_2', \lterm_3'}_{\checkmark}, 1}}
            \infer1[L]{\tuple{\tuple{\lterm_2, \lterm_3}_{\checkmark}, 1}}
            \hypo{\mathstrut}
            \infer1[D ($\textcolor{blue}{\struc_1})$]{\tuple{\tuple{\lterm_3, \lterm_4}_{\checkmark}, 1}}
            \infer[double]3[T]{\tuple{\tuple{\lterm_1, \lterm_4}_{\checkmark}, 1}}
            \infer1[$\checkmark$]{\tuple{\tuple{\lterm_1, \lterm_4}_{?}, 1}}
            \infer1[I]{\tuple{\src^{\automaton_{k}^{\term[3]}}, 0}}
        \end{prooftree}
    \end{center}
\end{exa}
Since the rules of $\automaton_{k}^{\term[3]}$ are defined based on the rules of \Cref{definition: decomposition},
we have the following, as expected.
Both directions are shown by easy induction on the derivative trees.
\begin{prop}[\Cref{section: proposition: decomposition and 2AFA}]\label{proposition: decomposition and 2AFA}
    Let $k \ge 1$ and let $\term[3]$ be an \kl{lKL term}.
    Let $\struc_1 \dots \struc_n \in \STR_{k}^{+}$.
    For all $j \in \range{n}$ and $\lterm[1], \lterm[2] \in \cl_{\univ{\struc_{j}} \cup \set{\bullet}}(\term[3])$, we have: 
    \[\lterm[1]_{(j)} \mathrel{(\odot_{i = 1}^{n} \longrightarrow^{(\struc_{i})_{\bullet}})} \lterm[2]_{(j)} \;\iff\; \tuple{\tuple{\lterm[1], \lterm[2]}_{\checkmark}, j} \in S^{{\automaton_{k}^{\term[3]}}}_{\triangleright \struc_1 \dots \struc_{n} \triangleleft}.\]
\end{prop}
Thus, we have the following lemma.
\begin{lem}\label{lemma: 2AFA 2}
    For $k \ge 1$ and a \kl{KL term} $\term$,
    we have:
    $\ljump{\automaton_{k}^{\term[1]}} = \set{\vec{\struc} = \struc_1 \dots \struc_n \in \STR_{k}^{+} \mid \tuple{[\tuple{1,x}]_{\sim}, [\tuple{1,y}]_{\sim}} \in \jump{\term}_{\odot \vec{\struc}} \mbox{ for some $x, y \in \univ{\struc_1}$}}$.
\end{lem}
\begin{proof}
    Let $\vec{\struc} = \struc_1 \dots \struc_n \in \STR_{k}^{+}$.
    We have:
    \begin{align*}
        &\vec{\struc} \in \ljump{\automaton_{k}^{\term}}\\
        &\iff
        \exists x, y \in \univ{\struc_1}, \exists \lterm[2] \in \cl_{[k]}(\term[1]),\  \tuple{\tuple{@x.\term,\lterm[2]}_{\checkmark}, 1} \in S^{\automaton_{k}^{\term}}_{\triangleright \vec{\struc} \triangleleft} \land \EPS_{y}(\lterm[2])
        \tag{By the form of $\automaton_{k}^{\term}$}\\
        &\iff \exists x, y \in \univ{\struc_1}, \exists \lterm[2] \in \cl_{[k]}(\term[1]),\  (@x.\term)_{(1)} \mathrel{(\odot_{i = 1}^{n} \longrightarrow^{(\struc_{i})_{\bullet}})} \lterm[2]_{(1)} \land \EPS_{y}(\lterm[2]) \tag{\Cref{proposition: decomposition and 2AFA}}\\
        &\iff \exists x, y \in \univ{\struc_1}, \exists \lterm[2] \in \cl_{[k]}(\term[1]),\  (@x.\term)_{(1)} \longrightarrow^{(\odot \vec{\struc})_{\bullet}} \lterm[2]_{(1)} \land \EPS_{y}(\lterm[2]) \tag{\Cref{theorem: decomposition}}\\
        &\iff \exists x, y \in \univ{\struc_1}, \exists \lterm[2] \in \cl_{[k]}(\term[1]),\  (@x.\term)_{(1)} \longrightarrow^{\odot \vec{\struc}} \lterm[2]_{(1)} \land \EPS_{[\tuple{1, y}]_{\sim}}(\lterm[2]_{(1)})\\
        &\iff \exists x, y \in \univ{\struc_1},\  [\tuple{1, y}]_{\sim} \in \jump{@[\tuple{1, x}]_{\sim}.\term}_{\odot \vec{\struc}}. \tag{\Cref{theorem: derivative}}
    \end{align*}
    Hence, this completes the proof.
\end{proof}

When ${l}$ and ${r}$ are fresh \kl{variables}, we have: $\REL \models \term[1] \le \term[2]$ iff $\REL \models \top {l} \term[1] {r} \top \le \top {l} \term[2] {r} \top$.
For \kl{PCoR* terms} of the form $\top \term[3] \top$, either $\jump{\top \term[3] \top}_{\struc} = \emptyset$ or $\jump{\top \term[3] \top}_{\struc} = \univ{\struc}^2$ holds.
Thus, if there is a \kl{KL term} $\term[3]_{\top}$ such that $\jump{\term[3]_{\top}}_{\struc} = \jump{\top}_{\struc}$ (note that $\top$ is not a \kl{KL term}),
then we have:
$\REL_{\pw \le k-1} \models \term[1] \le \term[2]$ iff
$\set{\vec{\struc} \in \STR_{k}^{+} \mid \jump{\term[3]_{\top} {l} \term {r} \term[3]_{\top}}_{\odot \vec{\struc}} \neq \emptyset} \subseteq \set{\vec{\struc} \in \STR_{k}^{+} \mid \jump{\term[3]_{\top} {l} \term[2] {r} \term[3]_{\top}}_{\odot \vec{\struc}} \neq \emptyset}$
iff $\ljump{\automaton_{k}^{\term[3]_{\top} {l} \term[1] {r} \term[3]_{\top}}} \subseteq \ljump{\automaton_{k}^{\term[3]_{\top} {l} \term[2] {r} \term[3]_{\top}}}$ (by \Cref{lemma: 2AFA 2}).

From this, if we can encode $\top$, then we can encode the \kl{equational theory} of \kl{KL terms}.
In the following, we consider encoding $\top$ on our automata construction.

\subsubsection{Normalized path decompositions}\label{section: normal form}
Let $\mathcal{L}^{\mathrm{N}}_{k} \defeq \STR_{k}^{+} \setminus (\mathcal{L}^{\mathrm{Inac}}_{k} \cup \mathcal{L}^{\mathrm{Incon}}_{k})$ where
\begin{align*}
    \mathcal{L}^{\mathrm{Inac}}_{k} &\defeq \STR_{k}^* \compo \set*{\struc_1 \struc_2 \in \STR_{k}^{2} \mid \univ{\struc_1} \cap \univ{\struc_2} = \emptyset} \compo \STR_{k}^*,\\
    \mathcal{L}^{\mathrm{Incon}}_{k} &\defeq \STR_{k}^* \compo \set*{\struc_1 \struc_2 \in \STR_{k}^{2} \mid \exists b \in \vsig,\ b^{\struc_{1}} \cap (\univ{\struc_{1}} \cap \univ{\struc_{2}})^{2} \neq b^{\struc_{2}} \cap (\univ{\struc_{1}} \cap \univ{\struc_{2}})^{2}} \compo \STR_{k}^*.
\end{align*}
Intuitively,
$\mathcal{L}^{\mathrm{Inac}}_{k}$ detects adjacent \kl{structures} not sharing any \kl{vertices}, and
$\mathcal{L}^{\mathrm{Incon}}_{k}$ detects adjacent \kl{structures} such that their relations are not the same in the sharing part.
The following proposition shows that $\mathcal{L}^{\mathrm{N}}_{k}$ is sufficient to enumerate all \kl{structures} of \kl{pathwidth} at most $k-1$.
\begin{prop}\label{prop: normal form}
    For $k \ge 2$,
    we have
    $\REL_{\pw \le k-1} = \isoc(\set{\odot \vec{\struc} \mid \vec{\struc} \in \mathcal{L}^{\mathrm{N}}_{k}})$.
\end{prop}
\begin{proof}
    Let $\vec{\struc[1]} = \struc_{1} \dots \struc_{n} \in \STR_{k}^{+}$.
    If $\univ{\struc_{i}} \cap \univ{\struc_{i+1}} = \emptyset$ for some $i$, then by putting a new component consisting of one \kl{vertex} in $\struc_i$ and one \kl{vertex} in $\struc_{i+1}$ between $\struc_i$ and $\struc_{i+1}$,
    we can avoid the condition of $\mathcal{L}^{\mathrm{Inac}}_{k}$.
    Moreover, for each $\tuple{[\tuple{i, x}]_{\sim}, [\tuple{i, y}]_{\sim}} \in b^{\odot \vec{\struc[1]}}$, we add a new edge for $\tuple{x, y} \in b^{\struc[1]_{i}}$ in $\struc[1]_{i}$;
    then, we can avoid the condition of $\mathcal{L}^{\mathrm{Incon}}_{k}$.
\end{proof}
In the sequel, we consider $\mathcal{L}^{\mathrm{N}}_{k}$ instead of $\STR_{k}^{+}$.
This is useful for encoding additional operators (\Cref{section: encoding KL,section: encoding variants}).
We use the condition of $\mathcal{L}^{\mathrm{Inac}}_{k}$ to encode $\top$
and we use the condition of $\mathcal{L}^{\mathrm{Incon}}_{k}$ to encode \kl{tests} and \kl{nominals}.

\subsubsection{Encoding the equational theory of KL}\label{section: encoding KL}
Let $c_{\top}$ be a special \kl{letter} for encoding $\top$.
We consider the following \kl{language}:
\[\mathcal{L}^{\top}_{k} \;\defeq\; \STR_{k}^{*} \compo \set{\struc \in \STR_{k} \mid {\aterm[3]_{\top}^{\struc} \neq \univ{\struc}^2}} \compo \STR_{k}^{*}.\]
Note that for every $\vec{\struc} \in \STR_{k}^{+} \setminus (\mathcal{L}^{\top}_{k} \cup \mathcal{L}^{\mathrm{Inac}}_{k})$,
we have $\jump{\aterm[3]_{\top}^{*}}_{\odot \vec{\struc}} = \jump{\top}_{\odot \vec{\struc}}$.
We then can reduce the \kl{equational theory} into the \kl{inclusion problem} for \kl{2AFAs}, as follows.
\begin{lem}\label{lemma: encoding the equational theory of KL}
    Let $k \ge 2$.
    Let $\term[1]$ and $\term[2]$ be \kl{KL terms}.
    Let ${l}, {r}, \aterm[3]_{\top}$ be fresh \kl{variables}.
    Then,
    \[\REL_{\pw \le k-1} \models \term[1] \le \term[2] \;\iff\; 
    \ljump{\automaton_{k}^{\aterm[3]_{\top}^* {l} \term[1] {r} \aterm[3]_{\top}^*}} \subseteq \ljump{\automaton_{k}^{\aterm[3]_{\top}^* {l} \term[2] {r} \aterm[3]_{\top}^*}} \cup \mathcal{L}^{\mathrm{Inac}}_{k} \cup \mathcal{L}^{\mathrm{Incon}}_{k} \cup \mathcal{L}^{\top}_{k}.\]
\end{lem}
\begin{proof}
    We have:
    \begin{align*}
        &\REL_{\pw \le k-1} \models \term[1] \le \term[2]
        \;\iff\; \REL_{\pw \le k-1} \models \top {l} \term[1] {r} \top \le \top {l} \term[2] {r} \top\\
        &\;\iff\; \set{\odot \vec{\struc[1]} \mid \vec{\struc[1]} \in \mathcal{L}^{\mathrm{N}}_{k} \setminus \mathcal{L}^{\top}_{k}} \models \top {l} \term[1] {r} \top \le \top {l} \term[2] {r} \top \tag{\Cref{section: normal form} and $c_\top$ is fresh}\\
        &\;\iff\; \forall \vec{\struc} \in \mathcal{L}^{\mathrm{N}}_{k} \setminus \mathcal{L}^{\top}_{k},
        \jump{\aterm[3]_{\top}^* {l} \term[1] {r} \aterm[3]_{\top}^*}_{\odot \vec{\struc}} \neq \emptyset
        \mbox{ implies }
        \jump{\aterm[3]_{\top}^* {l} \term[2] {r} \aterm[3]_{\top}^*}_{\odot \vec{\struc}} \neq \emptyset  \tag{$\jump{\top}_{\odot \vec{\struc}} = \jump{\aterm[3]_{\top}^*}_{\odot \vec{\struc}}$}\\
        &\;\iff\; \ljump{\automaton_{k}^{\aterm[3]_{\top}^* {l} \term[1] {r} \aterm[3]_{\top}^*}} \cap (\mathcal{L}^{\mathrm{N}}_{k} \setminus \mathcal{L}^{\top}_{k}) \subseteq \ljump{\automaton_{k}^{\aterm[3]_{\top}^* {l} \term[2] {r} \aterm[3]_{\top}^*}} \tag{\Cref{lemma: 2AFA 2}, and $\jump{\top \term[3] \top}_{\odot \vec{\struc}}$ is $\emptyset$ or $\univ{\odot \vec{\struc}}^2$}\\
        &\;\iff\; \ljump{\automaton_{k}^{\aterm[3]_{\top}^* {l} \term[1] {r} \aterm[3]_{\top}^*}} \subseteq \ljump{\automaton_{k}^{\aterm[3]_{\top}^* {l} \term[2] {r} \aterm[3]_{\top}^*}} \cup \mathcal{L}^{\mathrm{Inac}}_{k} \cup \mathcal{L}^{\mathrm{Incon}}_{k} \cup \mathcal{L}^{\top}_{k}. \tag{By $\mathcal{L}^{\mathrm{N}}_{k} \defeq \STR_{k}^{+} \setminus (\mathcal{L}^{\mathrm{Inac}}_{k} \cup \mathcal{L}^{\mathrm{Incon}}_{k})$}
    \end{align*}
    Hence, this completes the proof.
\end{proof}

\subsubsection{Complexity: On the size of the translated 2AFA}
For the number of states of $\automaton_{k}^{\term[1]}$,
by \Cref{proposition: closure size},
we have $\# \univ{\automaton_{k}^{\term[1]}} = 2^{\mathcal{O}(\iw(\term) \log(k \|\term\|))}$.
For the alphabet size of $\automaton_{k}^{\term[1]}$, we have $\# \STR_{k} = 2^{\mathcal{O}(k^2 \#\vsig)}$ (the number of \kl{structures} with universe $S$ is $2^{\#S^2 \#\vsig}$, as every element of $\vsig$ denotes a binary relation;
thus, $\#\STR_{k} = \sum_{S \subseteq \range{k}; S \neq \emptyset} 2^{\#S^2 \#\vsig} \le 2^{k} \times 2^{k^2 \#\vsig} = 2^{\mathcal{O}(k^2 \#\vsig)}$).
For the sizes of the positive boolean formulas,
we have $\|\delta^{\automaton_{k}^{\term[1]}}(q, a)\| = \mathcal{O}(k^2 \# \univ{\automaton_{k}^{\term[1]}}) = 2^{\mathcal{O}(\iw(\term) \log(k \|\term\|))}$ (the rule ($\const{I}$) and the rule (T) are critical).
Thus, we have $\|\automaton_{k}^{\term[1]}\| = 2^{\mathcal{O}(\iw(\term) \log(k \|\term\|)+ k^2 \#\vsig)}$.
Hence, under $k = \mathcal{O}(\|\term\|)$ and $\#\vsig = \mathcal{O}(\|\term\|)$,
we have $\|\automaton_{k}^{\term[1]}\| = 2^{\mathrm{poly}(\|\term\|)}$.

Additionally, for $\mathcal{L}^{\mathrm{Inac}}_{k}$, there is a \kl{1DFA} such that the number of states is $\mathcal{O}(2^{k})$ (each state corresponds to the universe in the current \kl{structure}).
For $\mathcal{L}^{\mathrm{Incon}}_{k}$, there is a \kl{1DFA} such that the number of states is $\mathcal{O}(\#\STR_{k})$ (each state corresponds to the current \kl{structure}).
For $\mathcal{L}^{\top}_{k}$, there is a \kl{1DFA} such that the number of states is $2$ (hence, a constant size with respect to $k$).
Finally, we have obtained the following complexity result.
\begin{cor}\label{corollary: KL EXPSPACE-complete}
    The \kl{equational theory} for \kl{KL terms} with respect to binary relations---given a finite set $\vsig$ and \kl{KL terms} $\term[1]$ and $\term[2]$ over $\vsig$, does $\REL \models \term[1] \le \term[2]$ hold?---is $\mathrm{EXPSPACE}$-complete.
\end{cor}
\begin{proof}
    (In $\mathrm{EXPSPACE}$):
    By \Cref{lemma: encoding the equational theory of KL} with the \kl{linearly bounded pathwidth model property} (\Cref{proposition: bounded pw property}), we can reduce the problem into the \kl{inclusion problem} for \kl{2AFAs}:
    \[\REL \models \term[1] \le \term[2] \;\iff\; \ljump{\automaton_{\iw(\term) + 1}^{\aterm[3]_{\top}^* {l} \term[1] {r} \aterm[3]_{\top}^*}} \subseteq \ljump{\automaton_{\iw(\term) + 1}^{\aterm[3]_{\top}^* {l} \term[2] {r} \aterm[3]_{\top}^*}} \cup \mathcal{L}^{\mathrm{Inac}}_{\iw(\term) + 1} \cup \mathcal{L}^{\mathrm{Incon}}_{\iw(\term) + 1} \cup \mathcal{L}^{\top}_{\iw(\term) + 1}.\]
    Here, we take the union construction.
    By the discussion above with $k = \iw(\term) + 1 \le \|\term\| + 1$ (\Cref{proposition: iw}),
    the \kl{2AFAs} have exponential sizes to the input size.
    By \Cref{proposition: 2AFA PSPACE}, this completes the proof.
    (Hardness) \cite{nakamuraPartialDerivativesGraphs2017,brunetPetriAutomata2017}: 
    By $\REL \models (\sum_{a \in \vsig} a)^* \le \term$ iff $\vsig^* \subseteq \ljump{\term}_{\vsig}$ (\Cref{proposition: glang and lang}) iff $\vsig^* = \ljump{\term}_{\vsig}$,
    we can give a reduction from the universality problem for regular expressions with intersection, which is $\mathrm{EXPSPACE}$-hard \cite[Theorem 2]{furerComplexityInequivalenceProblem1980}.
\end{proof}
\begin{rem}\label{remark: EXPSPACE-hard unary}
    In \Cref{corollary: KL EXPSPACE-complete},
    by \cite[Theorem 2]{furerComplexityInequivalenceProblem1980}, the \kl{equational theory} for \kl{KL terms} with respect to binary relations is $\mathrm{EXPSPACE}$-hard even when $\#\vsig = 2$.
    Furthermore, even when $\#\vsig = 1$, the \kl{equational theory} is $\mathrm{EXPSPACE}$-hard;
    given \kl{KL terms} $\term[1]$ and $\term[2]$,
    by letting $\term[1]'$ and $\term[2]'$ be the terms $\term[1]$ and $\term[2]$
    in which each variable $a_i$ (where $\vsig = \set{a_1, a_2}$) has been replaced with $\term[3]_i \defeq b^{3} \intersection b^{3 + i}$,
    we have $\REL \models \term[1] \le \term[2]$ iff $\REL \models \term[1]' \le \term[2]'$.
    (($\Longrightarrow$):
    Because $\term[1]' \le \term[2]'$ is obtained by a substitution from $\term[1] \le \term[2]$.
    ($\Longleftarrow$):
    We show the contraposition.
    Assume $\jump{\term[1]}_{\struc} \not\subseteq \jump{\term[2]}_{\struc}$.
    Let $\struc[2]$ be the \kl{structure} obtained by extending fresh \kl{vertices} for $\tuple{x, y} \in \jump{\term[3]_i}_{\struc[2]}$,
    for each edge $\tuple{x, y} \in \jump{a_i}_{\struc}$;
    for instance, if $\struc$ is $\begin{tikzpicture}[baseline = -.5ex]
        \graph[grow right = 1.2cm, branch down = 6ex, nodes={mynode}]{
        {0/{}[draw, circle]}-!-{1/{}[draw, circle]}-!-{2/{}[draw, circle]}
        };
        \graph[use existing nodes, edges={color=black, pos = .5, earrow}, edge quotes={fill=white, inner sep=1pt,font= \scriptsize}]{
        0 ->["$a_1$"] 1;
        1 ->["$a_2$"] 2;
        };
    \end{tikzpicture}$,
    then $\struc[2]$ is $\begin{tikzpicture}[baseline = -.5ex]
        \graph[grow right = 1.5cm, branch down = 6ex, nodes={mynode}]{
        {0/{}[draw, circle]}-!-{1/{}[draw, circle]}-!-{2/{}[draw, circle]}
        };
        \node[above right = 0 and .5cm of 0, draw, circle, mynode, inner sep = 1pt](0u1){};
        \node[above right = 0 and 1.cm of 0, draw, circle, mynode, inner sep = 1pt](0u2){};
        \node[below right = 0 and .3cm of 0, draw, circle, mynode, inner sep = 1pt](0d1){};
        \node[below right = 0 and .6cm of 0, draw, circle, mynode, inner sep = 1pt](0d2){};
        \node[below right = 0 and .9cm of 0, draw, circle, mynode, inner sep = 1pt](0d3){};
        \node[above right = 0 and .5cm of 1, draw, circle, mynode, inner sep = 1pt](1u1){};
        \node[above right = 0 and 1.cm of 1, draw, circle, mynode, inner sep = 1pt](1u2){};
        \node[below right = 0 and .2cm of 1, draw, circle, mynode, inner sep = 1pt](1d1){};
        \node[below right = 0 and .5cm of 1, draw, circle, mynode, inner sep = 1pt](1d2){};
        \node[below right = 0 and .8cm of 1, draw, circle, mynode, inner sep = 1pt](1d3){};
        \node[below right = 0 and 1.1cm of 1, draw, circle, mynode, inner sep = 1pt](1d4){};
        \graph[use existing nodes, edges={color=black, pos = .5, earrow}, edge quotes={fill=white, inner sep=1pt,font= \scriptsize}]{
        0 -> 0u1 -> 0u2 -> 1;
        0 -> 0d1 -> 0d2 -> 0d3 -> 1;
        1 -> 1u1 -> 1u2 -> 2;
        1 -> 1d1 -> 1d2 -> 1d3 -> 1d4 -> 2;
        };
    \end{tikzpicture}$.
    We then have
    $\jump{\term'}_{\struc[2]} = \jump{\term}_{\struc}$
    or $\jump{\term'}_{\struc[2]} = \jump{\term}_{\struc} \cup \triangle_{\univ{\struc[2]} \setminus \univ{\struc[1]}}$,
    by easy induction on $\term$;
    thus, $\jump{\term[1]'}_{\struc[2]} \not\subseteq \jump{\term[2]'}_{\struc[2]}$.)
\end{rem}  %
\section{Encoding Additional Operators}\label{section: encoding variants}
In this section, similar to how $\top$ was encoded in the reduction of \Cref{lemma: encoding the equational theory of KL} above,
we provide encodings for additional operators within our automata construction.
As a consequence, we have that the \kl{equational theory} of \kl{PCoR*} is decidable in EXPSPACE.

\subsection{Encoding PCoR*}\label{section: encoding PCoR*}
In this subsection, we give encodings of $\top$ and $\bl^{\smile}$.
Using them, we can encode the \kl{equational theory} of \kl{PCoR*}.

\subsubsection{Encoding top $(\top)$}\label{section: encode top}
We recall $\mathcal{L}^{\top}_{k}$ and $\aterm[3]_{\top}$ (\Cref{section: encoding KL}).
For a term $\term$, let $\term'$ be the unique normal form with respect to the term rewriting system:
$\top \leadsto \aterm[3]_{\top}^{*}$.
For $\vec{\struc} \in \STR_{k}^{+} \setminus (\mathcal{L}^{\top}_{k} \cup \mathcal{L}^{\mathrm{Inac}}_{k})$,
by $\jump{\aterm[3]_{\top}^{*}}_{\odot \vec{\struc}} = \jump{\top}_{\odot \vec{\struc}}$,
we have $\jump{\term'}_{\odot \vec{\struc}} = \jump{\term}_{\odot \vec{\struc}}$.
Thus we can encode $\top$.

\subsubsection{Encoding converse $(\bl^{\smile})$}\label{section: encode converse}
For each \kl{variable} $a \in \vsig$, we introduce a fresh \kl{variable} $\widebreve{a}$.
Let $\REL^{\smile} \defeq \set{\struc \in \REL \mid \forall a \in \vsig, \widebreve{a} = a^{\smile}}$.
We consider the following \kl{language}:
\begin{align*}
    \mathcal{L}^{\smile}_{k} &
    \;\defeq\;
    \STR_{k}^{*} \compo
    \set{\struc \in \STR_{k} \mid \widebreve{a}^{\struc} \neq (a^{\struc})^{\smile} \mbox{ for some $a \in \vsig$}} \compo 
    \STR_{k}^{*}.
\end{align*}
We then have $\REL_{\pw \le k - 1}^{\smile} = \set{\odot \vec{\struc} \mid \vec{\struc} \in \STR_{k}^{+} \setminus 
(\mathcal{L}^{\mathrm{Incon}}_{k} \cup \mathcal{L}^{\smile}_{k})}$.
Moreover, we consider the \intro*\kl{converse normal form} \cite{brunetPetriAutomata2017} of \kl{PCoR*}.
For a \kl{PCoR* term} $\term$, we consider the unique normal form $\term'$ with respect to the following term rewriting system:
\begin{gather*}
    a^{\smile} \leadsto \widebreve{a}, \qquad
    \id^{\smile} \leadsto \id, \qquad
    \emp^{\smile} \leadsto \emp, \qquad
    \top^{\smile} \leadsto \top, \qquad
    (\term[1] \compo \term[2])^{\smile} \leadsto \term[2]^{\smile} \compo \term[1]^{\smile},\\
    (\term[1] \union \term[2])^{\smile} \leadsto \term[1]^{\smile} \union \term[2]^{\smile}, \qquad
    (\term[1] \intersection \term[2])^{\smile} \leadsto \term[1]^{\smile} \intersection \term[2]^{\smile}, \qquad
    (\term[1]^{\smile})^{\smile} \leadsto \term[1], \qquad
    (\term[1]^{*})^{\smile} \leadsto (\term[1]^{\smile})^{*}.
\end{gather*}
We then have $\REL^{\smile} \models \term' = \term$.
Thus, within $\REL^{\smile}$, we can encode $\bl^{\smile}$.
Additionally, for the \kl{language} $\mathcal{L}^{\smile}_{k}$, there is a \kl{1DFA} such that the number of states is $2$ (a constant size).

\begin{exa}\label{ex: conv}
    As a toy example, consider $\REL \models a \le a a^{\smile} a$ \cite[(10)]{bloomNotesEquationalTheories1995}. 
    We observe that 
    this holds
    iff $\REL_{\pw \le 1} \models a \le a a^{\smile} a$ (\Cref{proposition: pw le iw})
    iff $\REL_{\pw \le 1}^{\smile} \models a \le a \breve{a} a$ (because $\breve{a}$ is fresh)
    iff $\ljump{\automaton_{2}^{\aterm[3]_{\top}^* {l} a {r} \aterm[3]_{\top}^*}} \subseteq \ljump{\automaton_{2}^{\aterm[3]_{\top}^* {l} a \breve{a} a {r} \aterm[3]_{\top}^*}} \cup \mathcal{L}^{\mathrm{Inac}}_{2} \cup \mathcal{L}^{\mathrm{Incon}}_{2} \cup \mathcal{L}^{\top}_{2} \cup \mathcal{L}^{\smile}_{2}$ (by \Cref{lemma: encoding the equational theory of KL} where $l, r, \aterm[3]_{\top}$ are fresh \kl{variables}).
    If $\vec{\struc} \in \ljump{\automaton_{2}^{\aterm[3]_{\top}^* {l} a {r} \aterm[3]_{\top}^*}} \setminus (\mathcal{L}^{\mathrm{Inac}}_{2} \cup \mathcal{L}^{\mathrm{Incon}}_{2} \cup \mathcal{L}^{\top}_{2} \cup \mathcal{L}^{\smile}_{2})$,
    then by \Cref{lemma: 2AFA 2} with
    $\jump{\aterm[3]_{\top}^*}_{\odot \vec{\struc}} = \jump{\top}_{\odot \vec{\struc}}$ (by $\mathcal{L}^{\top}_{2}$),
    there are $z_l$, $z_1$, $z_2$, and $z_r$ such that
    $\tuple{z_l, z_1} \in \jump{l}_{\odot \vec{\struc}}$,
    $\tuple{z_1, z_2} \in \jump{a}_{\odot \vec{\struc}}$, and
    $\tuple{z_2, z_r} \in \jump{r}_{\odot \vec{\struc}}$.
    By $\jump{\breve{a}}_{\odot \vec{\struc}} = \jump{a}_{\odot \vec{\struc}}^{\smile}$ (by $\mathcal{L}^{\smile}_{2}$),
    we have $\tuple{z_2, z_1} \in \jump{\widebreve{a}}_{\odot \vec{\struc}}$.
    Thus, the \kl{structure} $\odot\vec{\struc}$ can be illustrated as follows:
    \[\dots \hspace{.5em} \begin{tikzpicture}[baseline = -1.ex]
        \graph[grow right = 1.2cm, branch down = 2.5ex]{
            {l/{$z_l$}[font= \tiny, vert, xshift = .6cm]} -!- {s1/{$z_1$}[font= \tiny, vert]} -!- {t1/{$z_2$}[font= \tiny, vert]} -!- {r/{$z_r$}[font= \tiny, vert, xshift = -.6cm]}
        };
        \graph[use existing nodes, edges={color=black, pos = .5, earrow}, edge quotes={fill=white, inner sep=1pt,font= \scriptsize}]{
            l->["$l$", bend left = 80] s1;
            s1 ->["$a$", bend right] t1;
            t1 ->["$\breve{a}$", bend right] s1;
            t1->["$r$", bend left = 80] r;
        };
    \end{tikzpicture} \hspace{.5em} \dots\]
    where some of them may be the same \kl{vertex} and some \kl{vertices}/edges (including the universal relation $\jump{\aterm[3]_{\top}^*}_{\odot \vec{\struc}}$) are omitted.
    By the form of $\odot\vec{\struc}$ with \Cref{lemma: 2AFA 2},
    we have $\vec{\struc} \in \ljump{\automaton_{2}^{\aterm[3]_{\top}^* {l} a \breve{a} a {r} \aterm[3]_{\top}^*}}$.
    Hence, from this language inclusion, we have $\REL \models a \le a a^{\smile} a$.
\end{exa}

\begin{exa}\label{ex: modularity}
    We recall $\REL \models (\aterm[1] \aterm[2]) \intersection \aterm[3] \le (\aterm[1] \intersection (\aterm[3] \aterm[2]^{\smile})) \aterm[2]$ (\cf\ \Cref{equation: allegories 2}).
    Similar to \Cref{ex: conv},
    this holds
    iff $\ljump{\automaton_{3}^{\aterm[3]_{\top}^* {l} (\aterm[1] \aterm[2] \intersection \aterm[3]) {r} \aterm[3]_{\top}^*}} \subseteq \ljump{\automaton_{3}^{\aterm[3]_{\top}^* {l} (\aterm[1] \intersection \aterm[3] \widebreve{\aterm[2]}) \aterm[2] {r} \aterm[3]_{\top}^*}} \cup \mathcal{L}^{\mathrm{Inac}}_{3} \cup \mathcal{L}^{\mathrm{Incon}}_{3} \cup \mathcal{L}^{\top}_{3} \cup \mathcal{L}^{\smile}_{3}$.
    If $\vec{\struc} \in \ljump{\automaton_{3}^{\aterm[3]_{\top}^* {l} (\aterm[1] \aterm[2] \intersection \aterm[3]) {r} \aterm[3]_{\top}^*}} \setminus (\mathcal{L}^{\mathrm{Inac}}_{3} \cup \mathcal{L}^{\mathrm{Incon}}_{3} \cup \mathcal{L}^{\top}_{3} \cup \mathcal{L}^{\smile}_{3})$,
    then by \Cref{lemma: 2AFA 2} with
    $\jump{\aterm[3]_{\top}^*}_{\odot \vec{\struc}} = \jump{\top}_{\odot \vec{\struc}}$ (by $\mathcal{L}^{\top}_{3}$),
    there are $z_l$, $z_1$, $z_2$, $z_3$, and $z_r$ such that
    $\tuple{z_l, z_1} \in \jump{l}_{\odot \vec{\struc}}$,
    $\tuple{z_1, z_2} \in \jump{a}_{\odot \vec{\struc}}$,
    $\tuple{z_2, z_3} \in \jump{b}_{\odot \vec{\struc}}$,
    $\tuple{z_1, z_3} \in \jump{c}_{\odot \vec{\struc}}$, and
    $\tuple{z_3, z_r} \in \jump{r}_{\odot \vec{\struc}}$.
    By $\jump{\breve{b}}_{\odot \vec{\struc}} = \jump{b}_{\odot \vec{\struc}}^{\smile}$ (by $\mathcal{L}^{\smile}_{3}$),
    we have $\tuple{z_3, z_2} \in \jump{\widebreve{b}}_{\odot \vec{\struc}}$.
    Thus, the \kl{structure} $\odot\vec{\struc}$ can be illustrated as follows:
    \[\dots \hspace{.5em} \begin{tikzpicture}[baseline = -1.ex]
        \graph[grow right = 1.2cm, branch down = 2.5ex]{
            {1/{$z_l$}[vert, xshift = .6cm]} -!- {2/{$z_1$}[vert]} -!- {3/{$z_2$}[vert, yshift = 2ex]} -!- {4/{$z_3$}[vert]} -!- {5/{$z_r$}[vert, xshift = -.6cm]}
        };
        \graph[use existing nodes, edges={color=black, pos = .5, earrow}, edge quotes={fill=white, inner sep=1pt,font= \scriptsize}]{
            1->["$l$", bend left = 80] 2;
            2 ->["$a$", bend right = 15] 3;
            3 ->["$b$", bend right = 15] 4;
            4 ->["$\breve{b}$", bend right = 20] 3;
            2 ->["$c$", bend right = 15] 4;
            4->["$r$", bend left = 80] 5;
        };
    \end{tikzpicture} \hspace{.5em} \dots\]
    where some of them may be the same \kl{vertex} and some \kl{vertices}/edges are omitted.
    By the form of $\odot\vec{\struc}$ with \Cref{lemma: 2AFA 2},
    we have $\vec{\struc} \in \ljump{\automaton_{3}^{\aterm[3]_{\top}^* {l} (\aterm[1] \intersection \aterm[3] \widebreve{\aterm[2]}) \aterm[2] {r} \aterm[3]_{\top}^*}}$.
    Hence, from this language inclusion, we have $\REL \models (\aterm[1] \aterm[2]) \intersection \aterm[3] \le (\aterm[1] \intersection (\aterm[3] \aterm[2]^{\smile})) \aterm[2]$.
\end{exa}

\subsubsection{Encoding the equational theory of PCoR*}
By using the encodings above, we can encode the \kl{equational theory} of \kl{PCoR*}.
Thus, we have the following.
\begin{cor}\label{corollary: PCoR* EXPSPACE-complete}
    The \kl{equational theory} of \kl{PCoR*} is decidable and $\mathrm{EXPSPACE}$-complete.
\end{cor}
\begin{proof}
    ($\mathrm{EXPSPACE}$-hardness):
    By \Cref{corollary: KL EXPSPACE-complete}.
    (In $\mathrm{EXPSPACE}$):
    Let $\term[1]$ and $\term[2]$ be the \kl{KL terms} obtained from given \kl{PCoR* terms} $\term[1]'$ and $\term[2]'$ by applying the translations of \Cref{section: encode top,section: encode converse}.
    Similar to \Cref{corollary: KL EXPSPACE-complete}, we can reduce into the \kl{inclusion problem} for \kl{2AFAs}, as follows:
    $\REL \models \term[1]' \le \term[2]'$ iff $\REL_{\pw \le \iw(\term[1]')} \models \term[1]' \le \term[2]'$ (\Cref{proposition: bounded pw property}) iff 
    \[\ljump{\automaton_{\iw(\term) + 1}^{\aterm[3]_{\top}^* {l} \term[1] {r} \aterm[3]_{\top}^*}} \;\subseteq\; \ljump{\automaton_{\iw(\term) + 1}^{\aterm[3]_{\top}^* {l} \term[2] {r} \aterm[3]_{\top}^*}} \cup \mathcal{L}^{\mathrm{Inac}}_{\iw(\term) + 1} \cup \mathcal{L}^{\mathrm{Incon}}_{\iw(\term) + 1} \cup \mathcal{L}^{\top}_{\iw(\term) + 1} \cup \mathcal{L}^{\smile}_{\iw(\term) + 1}.\]
    Because these \kl{2AFAs} have exponential sizes, this completes the proof by \Cref{proposition: 2AFA PSPACE}.
\end{proof}

\subsection{Encoding tests and nominals}\label{section: encoding more PCoR*}
In this subsection, moreover, we give encodings of \kl{tests} in Kleene algebra with tests \cite[Section 4]{kozenKleeneAlgebraTests1996} and \kl{nominals} in hybrid logic \cite{arecesHybridLogics2007}.

\subsubsection{Encoding tests}\label{section: encode tests}
Let $B \subseteq \vsig$ be a finite set of \intro*\kl{atomic tests}.
Let $\REL^{\mathrm{tests}_{B}} \defeq \set{\struc \in \REL \mid \forall b \in B, b^{\struc} \subseteq \triangle_{\univ{\struc}}}$.
The set of \intro*\kl{PCoR* terms with tests} is given by:
\begin{align*}
    \term[1], \term[2] &\;\Coloneqq\; p \mid a \mid \id \mid \emp \mid \term[1] \compo \term[2] \mid \term[1] \union \term[2] \mid \term[1]^{*}, \tag{$a \in \vsig \setminus B$} \\
    p, q &\;\Coloneqq\; b \mid \id \mid \emp \mid p \compo q \mid p \union q \mid p^{-}. \tag{$b \in B$} 
\end{align*}
where the \kl[PCoR* term with tests]{term} $p^{-}$ expresses the complement of $p$ with respect to the identity---the \kl(lterm){semantics} is extended by $\jump{p^{-}}_{\struc} = \triangle_{\univ{\struc}} \setminus \jump{p}_{\struc}$ (the others are the same as \kl{PCoR* terms}).
Using \kl{tests}, we can encode propositional while programs \cite{fischerPropositionalDynamicLogic1979,kozenKleeneAlgebraTests1997}:
\begin{align*}
    \mbox{\textbf{while }} p \mbox{\textbf{ do }} \term[1] &\;\defeq\; (p \term[1])^{*} p^{-}, &
    \mbox{\textbf{if }} p \mbox{\textbf{ then }} \term[1] \mbox{\textbf{ else }} \term[2] &\;\defeq\; (p \term[1]) \union (p^{-} \term[2]).
\end{align*}

For each $b \in B$, we introduce a fresh \kl{variable} $\bar{b}$; let $\bar{B}$ be the set $\set{\bar{b} \mid b \in B}$.
Let $\REL^{\mathrm{tests}_{B, \bar{B}}} \defeq \set{\struc \in \REL \mid \forall b \in B, b^{\struc} \uplus \bar{b}^{\struc} = \triangle_{\univ{\struc}}}$.
For a \kl[PCoR* term with tests]{term} $\term$, we consider the unique normal form $\term'$ in \kl{PCoR* terms} with respect to the following term rewriting system:
\begin{gather*}
    b^{-} \leadsto \bar{b}, \qquad
    \id^{-} \leadsto \emp, \qquad
    \emp^{-} \leadsto \id, \qquad
    (p \compo q)^{-} \leadsto p^{-} \union q^{-}, \qquad
    (p \union q)^{-} \leadsto p^{-} \compo q^{-}.
\end{gather*}
Then for $\struc \in \REL^{\mathrm{tests}_{B, \bar{B}}}$, we have $\jump{\term[1]'}_{\struc} = \jump{\term[1]}_{\struc}$.
Let \[\mathcal{L}^{\mathrm{tests}_{B, \bar{B}}}_{k} \;\defeq\;
\STR_{k}^{*} \compo
\set{\struc \in \STR_{k} \mid \exists b \in B,\ b^{\struc} \uplus \bar{b}^{\struc} \neq \triangle_{\univ{\struc}} } \compo 
\STR_{k}^{*}.\]
Then,
$\REL^{\mathrm{tests}_{B, \bar{B}}}_{\pw \le k-1} = \set{\odot \vec{\struc} \mid \vec{\struc} \in \STR_{k}^{+} \setminus (\mathcal{L}^{\mathrm{tests}_{B, \bar{B}}}_{k} \cup \mathcal{L}^{\mathrm{Incon}}_{k})}$.
For $\mathcal{L}^{\mathrm{tests}_{B, \bar{B}}}_{k}$, there is a \kl{1NFA} such that the number of states is a constant size.
Thus, we can encode \kl{tests}.
\begin{rem}\label{rem: tests}
    For the \kl{1NFA} of $\mathcal{L}_k^{\mathrm{tests}_{B,\bar{B}}}$,
    the smallest number of states is $2$ (hence, a constant size) but the number of transitions is naively exponential in the size of the alphabet for \kl{tests}.
    In \Cref{section: alphabet size}, to obtain the PSPACE upper bound,
    we will give an additional technique for reducing the number of the alphabet size to polynomial to the input (\Cref{section: alphabet size}).
    The modified encoding works even if the number of \kl{tests} is not fixed.
\end{rem}

\subsubsection{Encoding nominals}\label{section: encode nominals}
Let $L \subseteq \vsig$ be a finite set of \intro*\kl{nominals}.
Let $\REL^{\mathrm{noms}_{L}} \defeq \set{\struc \in \REL \mid \forall l \in L, \exists x \in \univ{\struc}, l^{\struc} = \set{\tuple{x, x}}}$.
Using \kl{nominals}, we can point a \kl{vertex}.
In $\REL^{\mathrm{noms}_{L}}$, we can also encode the jump operator ``$@ l. \term$'' in the definition of \cite{arecesHybridLogics2007}, by $\REL^{\mathrm{noms}_{L}} \models @ l. \term = \top l \term$.

We define the following \kl{language}:
\begin{align*}
    \mathcal{L}^{\mathrm{noms}_{L}}_{k} &\;\defeq\;
    \bigcup_{l \in L}
    \left(\begin{aligned}
       & \set{\struc \in \STR_{k} \mid {l}^{\struc} = \emptyset}^{+} \cup (\STR_{k}^{*} \compo \set{\struc \in \STR_{k} \mid {l}^{\struc} \not\subseteq \triangle_{\univ{\struc}}} \compo \STR_{k}^{*}) \cup {} \\
       & \STR_{k}^{*} \compo \set*{
            \begin{aligned}
                \struc_1 \dots \struc_{n} \in \STR_{k}^{+} \mid \exists x, \exists y, \tuple{x, x} \in l^{\struc_1} \land \tuple{y, y} \in l^{\struc_n}   \\
                \land  (x \neq y \lor (x = y \land \exists j \in \range{n}, x \not\in \univ{\struc_{j}}))
            \end{aligned}} \compo \STR_{k}^{*}
    \end{aligned}\right).
\end{align*}
Then, $\REL^{\mathrm{noms}_{L}}_{\pw \le k - 1} = \set{\odot \vec{\struc} \mid \vec{\struc} \in \STR_{k}^{+} \setminus \mathcal{L}^{\mathrm{noms}_{L}}_{k}}$.
(By the first line, $l^{\odot\vec{\struc}} \neq \emptyset$ and $l^{\odot\vec{\struc}} \subseteq \triangle_{\univ{\odot\vec{\struc}}}$.
By the second line, $\# \set{x \mid \tuple{x, x} \in l^{\odot\vec{\struc}}} \le 1$.)
For $\mathcal{L}^{\mathrm{noms}_{L}}_{k}$, there is a \kl{1NFA} such that the number of states is $\mathcal{O}(\# L \times k)$ (intuitively, we use the states to remember $l \in L$, $x \in \range{k}$, and to track whether $x \not\in \univ{\struc_{j}}$ held in past \kl{structures} $\struc_j$).
Thus, we can encode \kl{nominals}.

\subsubsection{Encoding the equational theory of PCoR* with tests and nominals}
From the above two, to use \kl{tests} and \kl{nominals},
we consider \kl{PCoR* terms} in the class $\REL^{\mathrm{tests}_{B, \bar{B}}, \mathrm{noms}_{L}} \defeq \REL^{\mathrm{tests}_{B, \bar{B}}} \cap \REL^{\mathrm{noms}_{L}}$.
For the class, we have the following proposition (\cf\ \Cref{proposition: bounded pw property}).
\begin{prop}\label{proposition: bounded pw property for nominals}
    For all \kl{PCoR* terms} $\term[1], \term[2]$, we have:
    \[\REL^{\mathrm{tests}_{B, \bar{B}}, \mathrm{noms}_{L}} \models \term[1] \le \term[2] \;\iff\; \REL^{\mathrm{tests}_{B, \bar{B}}, \mathrm{noms}_{L}}_{\pw \le \iw(\term[1]) + \#L} \models \term[1] \le \term[2].\]
\end{prop}
\begin{proof}
    ($\Longrightarrow$):
    Trivial.
    ($\Longleftarrow$):
    We prove the contraposition.
    Assume $\tuple{x, y} \in \jump{\term[1]}_{\struc} \setminus \jump{\term[2]}_{\struc}$ for some $\struc \in \REL^{\mathrm{tests}_{B, \bar{B}}, \mathrm{noms}_{L}}$.
    By \Cref{proposition: graph language},
    $\graph' \homo \const{G}(\struc, x, y)$ for some $\graph' \in \glang(\term[1])$, and 
    $\graph[2] \centernot\homo \const{G}(\struc, x, y)$ for all $\graph[2] \in \glang(\term[2])$.
    Let $\graph$ be the \kl{graph} $\graph'$ extended with $L$ isolated \kl{vertices}
    having $l_1, \dots, l_{\#L}$ labelled looping edges, respectively, where $L = \set{l_1, \dots, l_{\#L}}$.
    Then there is a \kl{graph homomorphism} $h \colon \graph \homo \const{G}(\struc, x, y)$ by extending $\graph' \homo \const{G}(\struc, x, y)$.
    Let $\mathcal{S}_{B, \bar{B}}(h) \defeq \tuple{\univ{\graph}, \set{a^{\mathcal{S}_{B, \bar{B}}(h)}}_{a \in \vsig}, \src^{\graph[1]}, \tgt^{\graph[1]}}$ where
    the binary relation $a^{\mathcal{S}_{B, \bar{B}}(h)}$ is defined by 
    \[a^{\mathcal{S}_{B, \bar{B}}(h)} \;\defeq\; \begin{cases}
        a^{\graph} & (a \in \vsig \setminus (B \cup \bar{B}))\\
        a^{\graph} \cup \set{\tuple{x, x} \in \triangle_{\univ{\graph}} \mid \tuple{h(x), h(x)} \in a^{\graph[2]}} & (a \in B \cup \bar{B}).
    \end{cases}\]
    Note that $h \colon \mathcal{S}_{B, \bar{B}}(h) \homo \const{G}(\struc, x, y)$.
    Let $\sim$ be the minimal equivalence relation satisfying the following two conditions:
    \begin{itemize}
        \item for each $l \in L$, if $\tuple{x, y},\tuple{x',y'} \in l^{\mathcal{S}_{B, \bar{B}}(h)}$, then $x \sim y'$ (hence, $x \sim y \sim x' \sim y'$);
        \item for each $b \in B \cup \bar{B}$, if $\tuple{x, y} \in b^{\mathcal{S}_{B, \bar{B}}(h)}$, then $x \sim y$.
    \end{itemize}\noindent 
    Then $\mathcal{S}_{B, \bar{B}}(h)/{\sim} \homo \const{G}(\struc, x, y)$ holds because
    $h(x) = h(y)$ holds for all $x \sim y$.
    Let $\struc[2]$, $x'$, and $y'$ be such that $\const{G}(\struc[2], x', y') = \mathcal{S}_{B, \bar{B}}(h)/{\sim}$.
    By $\graph' \homo \graph \homo \mathcal{S}_{B, \bar{B}}(h) \homo \const{G}(\struc[2], x', y')$,
    we have $\tuple{x', y'} \in \jump{\term[1]}_{\struc[2]}$ (\Cref{proposition: graph language}).
    By $\const{G}(\struc[2], x', y') \homo \const{G}(\struc, x, y)$,
    we also have $\tuple{x', y'} \not\in \jump{\term[2]}_{\struc[2]}$ (\Cref{proposition: graph language}).
    Because $\struc[2] \in \REL^{\mathrm{tests}_{B, \bar{B}}, \mathrm{noms}_{L}}_{\pw \le \iw(\term[1]) + \#L}$
    (note that by $\pw(\mathcal{S}_{B, \bar{B}}(h)) = \pw(\graph) = \pw(\graph') \le \iw(\term)$ (\Cref{proposition: pw le iw}), we have $\pw(\mathcal{S}_{B, \bar{B}}(h)/{\sim}) \le \iw(\term) + \#L$),
    this completes the proof.
\end{proof}
Using this bounded pathwidth model property, we have the following corollary.
\begin{cor}\label{corollary: PCoR* with tests EXPSPACE-complete}
    The \kl{equational theory} of \kl{PCoR* terms with tests and nominals} is decidable and $\mathrm{EXPSPACE}$-complete.
\end{cor}
\begin{proof}
    ($\mathrm{EXPSPACE}$-hardness):
    By \Cref{corollary: KL EXPSPACE-complete}.
    (In $\mathrm{EXPSPACE}$):
    Given \kl[PCoR* terms with tests and nominals]{terms} $\term[1]''$ and $\term[2]''$
    (where $B$ and $\bar{B}$ are used for \kl{tests} and $L$ is used for \kl{nominals}),
    let $\term[1]'$ and $\term[2]'$ be the \kl{PCoR* terms} obtained by applying the translations of \Cref{section: encode tests},
    and let $\term[1]$ and $\term[2]$ be the \kl{KL terms} obtained by applying the translations of \Cref{section: encode top,section: encode converse}, respectively.
    Similar to \Cref{corollary: KL EXPSPACE-complete}, we can reduce into the \kl{inclusion problem} for \kl{2AFAs}, as follows:
    $\REL^{\mathrm{tests}_{B}, \mathrm{noms}_{L}} \models \term[1]'' \le \term[2]''$ iff
    $\REL^{\mathrm{tests}_{B, \bar{B}}, \mathrm{noms}_{L}} \models \term[1]' \le \term[2]'$ (\Cref{section: encode tests}) iff
    $\REL^{\mathrm{tests}_{B, \bar{B}}, \mathrm{noms}_{L}}_{\pw \le \iw(\term[1]) + \# L} \models \term[1]' \le \term[2]'$ (\Cref{proposition: bounded pw property for nominals}) iff
    \[\ljump{\automaton_{k}^{\aterm[3]_{\top}^* {l} \term[1] {r} \aterm[3]_{\top}^*}} \;\subseteq\; \ljump{\automaton_{k}^{\aterm[3]_{\top}^* {l} \term[2] {r} \aterm[3]_{\top}^*}} \cup \mathcal{L}^{\mathrm{Inac}}_{k} \cup \mathcal{L}^{\mathrm{Incon}}_{k} \cup \mathcal{L}^{\top}_{k} \cup \mathcal{L}^{\smile}_{k} \cup \mathcal{L}^{\mathrm{tests}_{B, \bar{B}}}_{k} \cup \mathcal{L}^{\mathrm{noms}_{L}}_{k}\]
    where $k \defeq \iw(\term) + \#L + 1$.
    Because these \kl{2AFAs} have exponential sizes, this completes the proof by \Cref{proposition: 2AFA PSPACE}.
\end{proof}

\subsection{Reducing the alphabet size for PSPACE upper bound}\label{section: alphabet size}
In this subsection, we note that we can reduce the alphabet size from $\# \STR_{k} = 2^{\mathcal{O}(\#\vsig \times k^2)}$ into $\#\vsig \times 2^{\mathcal{O}(k^2)}$,
by reducing the number of \kl{variables} occurring at the same time.
For instance, if $\struc = \left(\begin{tikzpicture}[baseline = -.5ex]
    \graph[grow right = 1.cm, branch down = 6ex, nodes={mynode}]{
    {1/{$1$}[draw, circle]}-!-{2/{$2$}[draw, circle]}-!-{3/{$3$}[draw, circle]}
    };
    \graph[use existing nodes, edges={color=black, pos = .5, earrow}, edge quotes={fill=white, inner sep=1pt,font= \scriptsize}]{
    1 ->["$\aterm[1]$", bend left] 2  ->["$\aterm[2]$", bend left] 3;
    3 ->["$\aterm[3]$", bend left] 2  ->["$\aterm[2]$", bend left] 1;
    };
\end{tikzpicture} \right)$, then by letting
$\vec{\struc[2]} =  \left(\begin{tikzpicture}[baseline = -.5ex]
    \graph[grow right = 1.cm, branch down = 6ex, nodes={mynode}]{
    {1/{$1$}[draw, circle]}-!-{2/{$2$}[draw, circle]}-!-{3/{$3$}[draw, circle]}
    };
    \graph[use existing nodes, edges={color=black, pos = .5, earrow}, edge quotes={fill=white, inner sep=1pt,font= \scriptsize}]{
    1 ->["$\aterm[1]$", bend left] 2;
    };
\end{tikzpicture} \right)
\left(\begin{tikzpicture}[baseline = -.5ex]
    \graph[grow right = 1.cm, branch down = 6ex, nodes={mynode}]{
    {1/{$1$}[draw, circle]}-!-{2/{$2$}[draw, circle]}-!-{3/{$3$}[draw, circle]}
    };
    \graph[use existing nodes, edges={color=black, pos = .5, earrow}, edge quotes={fill=white, inner sep=1pt,font= \scriptsize}]{
 2  ->["$\aterm[2]$", bend left] 3;
 2  ->["$\aterm[2]$", bend left] 1;
 };
\end{tikzpicture} \right)
\left(\begin{tikzpicture}[baseline = -.5ex]
    \graph[grow right = 1.cm, branch down = 6ex, nodes={mynode}]{
    {1/{$1$}[draw, circle]}-!-{2/{$2$}[draw, circle]}-!-{3/{$3$}[draw, circle]}
    };
    \graph[use existing nodes, edges={color=black, pos = .5, earrow}, edge quotes={fill=white, inner sep=1pt,font= \scriptsize}]{
    3 ->["$\aterm[3]$", bend left] 2;
    };
\end{tikzpicture} \right)$, we have $\struc \cong \odot \vec{\struc[2]}$.
In the sequel, we consider partitions of $\vsig$ given as follows:
\[\mbox{$\set{b, \bar{b}, \widebreve{b}, \widebreve{\bar{b}}}$ for $b \in B$,} \qquad
\mbox{$\set{l, \widebreve{l}}$ for $l \in L$,} \qquad
\mbox{$\set{a, \widebreve{a}}$ for other \kl{variables} $a$}.\]
We enumerate them as a sequence $\vsig_0, \dots, \vsig_{m-1}$.
Let $\mathsf{supp}(\struc) \defeq \set{a \in \vsig \mid a^{\struc} \neq \emptyset}$ and 
let $\STR'_{k} \defeq \set{\struc \in \STR_{k} \mid \exists i < m, \mathsf{supp}(\struc) \subseteq \vsig_i}$.
As above, to enumerate all \kl{structures} of \kl{pathwidth} at most $k-1$, it suffices to consider $\STR'_{k}$.
Since $\#\vsig_i$ is bounded by a fixed number ($\# \vsig_i \le 4$), we have $\# \STR'_{k} = \#\vsig \times 2^{\mathcal{O}(k^2)}$.
However, this decomposition breaks the condition of $\mathcal{L}^{\mathrm{Incon}}_{k}$.
For that reason, we consider the following well-aligned universe:
\[\mathcal{L}^{\mathrm{U}'}_{k} \;\defeq\; \set{\struc_{0} \dots \struc_{m-1} \mid \univ{\struc_0} = \dots = \univ{\struc_{m-1}} \land \forall i < m, \mathsf{supp}(\struc_i) \subseteq \vsig_i}^{+},\]
and then we replace $\mathcal{L}^{\mathrm{Incon}}_{k}$ with the \kl{language} $\mathcal{L}^{\mathrm{Incon}'}_{k}$ defined by:
\begin{align*}
    &(\STR'_{k})^* \compo \bigcup_{a \in \vsig} \set*{
    \struc_0 \dots \struc_{m} \in (\STR'_{k})^{m+1} \;\middle|\;
    \begin{aligned}
    & a^{\struc_{0}} \cap (\univ{\struc_{0}} \cap \univ{\struc_{m}})^{2}\\
    & \neq a^{\struc_{m}} \cap (\univ{\struc_{0}} \cap \univ{\struc_{m}})^{2}
    \end{aligned}} \compo (\STR'_{k})^*.
\end{align*}
According to this, similarly for the others, we construct them as follows:
\begingroup%
\allowdisplaybreaks
\begin{align*}
    \mathcal{L}^{\mathrm{Inac}'}_{k} &\;\defeq\; (\STR'_{k})^* \compo \set{\struc_1 \struc_2 \in (\STR'_{k})^{2} \mid \univ{\struc_1} \cap \univ{\struc_2} = \emptyset} \compo (\STR'_{k})^*,\\
    \mathcal{L}^{\top'}_{k} &\;\defeq\; ((\STR'_{k})^{m})^{*} \compo (\STR'_{k})^{i_{\top}-1} \compo \set{\struc \in \STR'_{k} \mid  \aterm[3]_{\top}^{\struc} \neq \univ{\struc}^2} \compo (\STR'_{k})^{*} \tag*{where $i_{\top}$ is such that  $\aterm[3]_{\top} \in \vsig_{i_{\top}}$,}\\
    \mathcal{L}^{\smile'}_{k} &\;\defeq\; 
(\STR'_{k})^{*} \compo
\set{\struc \in \STR'_{k} \mid \widebreve{a}^{\struc} \neq (a^{\struc})^{\smile}} \compo 
(\STR'_{k})^{*},\\
    \mathcal{L}^{\mathrm{tests}_{B, \bar{B}}'}_{k} &\;\defeq\;
    ((\STR'_{k})^{m})^{*} \compo 
    (\bigcup_{b \in B}
    (\STR'_{k})^{i_b-1} \compo \set{\struc \in \STR'_{k} \mid b^{\struc} \uplus \bar{b}^{\struc} \neq \triangle_{\univ{\struc}}}) \compo (\STR'_{k})^{*} \tag*{where $i_b$ is such that  $b \in \vsig_{i_b}$,}\\
    \mathcal{L}^{\mathrm{noms}_{L}'}_{k} &\;\defeq\;
    \bigcup_{l \in L}
    \left(\begin{aligned}
        & \set{\struc \in \STR'_{k} \mid {l}^{\struc} = \emptyset}^{+} \cup {}\\
        & ((\STR'_{k})^{*} \compo \set{\struc \in \STR'_{k} \mid  {l}^{\struc} \not\subseteq \triangle_{\univ{\struc}}} \compo (\STR'_{k})^{*}) \cup {} \\
        & (\STR'_{k})^{*} \compo
        \set*{
            \begin{aligned}
                &\struc_1 \dots \struc_{n} \in (\STR'_{k})^{+} \mid \exists x, \exists y, \\
                & \tuple{x, x} \in l^{\struc_1} \land \tuple{y, y} \in l^{\struc_n}\\
                &  \land ({x \neq y} \lor ({x = y} \land \exists j \in \range{n}, x \not\in \univ{\struc_{j}}))
            \end{aligned}
        } \compo
        (\STR'_{k})^{*}
    \end{aligned}\right).
\end{align*}%
\endgroup
Each of them is a regular language,
and there exists a \kl{1NFA} of $\mathcal{O}(2^{k} m)$ states for $\mathcal{L}^{\mathrm{U}'}_{k}$,
$\mathcal{O}(2^{k} m)$ for $\mathcal{L}^{\mathrm{Incon}'}_{k}$,
$\mathcal{O}(2^{k})$ for $\# \mathcal{L}^{\mathrm{Inac}'}_{k}$,
$\mathcal{O}(m)$ for $\# \mathcal{L}^{\top'}_{k}$,
$\mathcal{O}(1)$ for $\# \mathcal{L}^{\smile'}_{k}$,
$\mathcal{O}(m)$ for $\# \mathcal{L}^{\mathrm{tests}_{B, \bar{B}}'}_{k}$, and
$\mathcal{O}(k \# L)$ for $\# \mathcal{L}^{\mathrm{noms}_{L}'}_{k}$.
Thus, assuming $k$ is fixed, 
for each language, there is a \kl{1NFA} such that the number of states is $\mathcal{O}(\|\term\|)$.
The alphabet size $\STR'_{k}$ is also $\mathcal{O}(\|\term\|)$.

\subsubsection{The bounded intersection width fragment}\label{section: bounded intersection width}
We moreover show that when the \kl{intersection width} \cite{gollerPDLIntersectionConverse2009} is fixed, the \kl{equational theory} is $\mathrm{PSPACE}$-complete.
\begin{cor}\label{corollary: PCoR* with tests v iw fixed PSPACE-complete}
    Let $k \ge 1$.
    The \kl{equational theory} of \kl{PCoR* terms with tests and nominals} of \kl{intersection width} at most $k$ is $\mathrm{PSPACE}$-complete.
\end{cor}
\begin{proof}
    ($\mathrm{PSPACE}$-hardness):
    Because the universality problem for regular expressions is $\mathrm{PSPACE}$-complete \cite{meyerEquivalenceProblemRegular1972}.
    (In $\mathrm{PSPACE}$):
    Similar to \Cref{corollary: PCoR* with tests EXPSPACE-complete},
    from given \kl{PCoR* terms with tests and nominals} $\term[1]'$ and $\term[2]'$,
    let $\term[1]$ and $\term[2]$ be the \kl{KL terms} obtained by applying the translations of \Cref{section: encode top,section: encode converse,section: encode tests}.
    By using the automata above with $\iw(\term) \le k$, we can reduce into the \kl{inclusion problem} for \kl{2AFAs}, as follows:
    $\REL \models \term[1]' \le \term[2]'$ iff
    \[\ljump{\automaton_{k + 1}^{\aterm[3]_{\top}^* {l} \term[1]' {r} \aterm[3]_{\top}^*}} \cap \mathcal{L}^{\mathrm{U}'}_{k + 1} \subseteq \ljump{\automaton_{k + 1}^{\aterm[3]_{\top}^* {l} \term[2]' {r} \aterm[3]_{\top}^*}} \cup \mathcal{L}^{\mathrm{Inac}'}_{k + 1} \cup \mathcal{L}^{\mathrm{Incon}'}_{k + 1} \cup \mathcal{L}^{\top'}_{k + 1} \cup \mathcal{L}^{\smile'}_{k + 1} \cup \mathcal{L}^{\mathrm{tests}'_{B, \bar{B}}}_{k + 1} \cup \mathcal{L}^{\mathrm{noms}'_{L}}_{k + 1}.\]
    Because these \kl{2AFAs} have polynomial sizes, this completes the proof by \Cref{proposition: 2AFA PSPACE}.
\end{proof}
Particularly when $k = 1$, we have the following,
which slightly extends the PSPACE upper bound of Kleene algebra with top and \kl{tests} \cite{pousCompletenessTheoremsKleene2024}
and of Kleene algebra with top and converse \cite{nakamuraExistentialCalculiRelations2023}.
\begin{cor}\label{corollary: KAtopconverse}
    The \kl{equational theory} for Kleene algebra with top, converse, \kl{tests}, and \kl{nominals} with respect to binary relations is $\mathrm{PSPACE}$-complete.
\end{cor}

\section{Proof of {\Cref{theorem: decomposition}}: The Decomposition Theorem}\label{section: theorem: decomposition}
In this section, we prove \Cref{theorem: decomposition}, which is the only remaining part.
\decompositiontheorem*
We first show the decomposition theorem for a subclass of \kl{runs} (\Cref{lemma: decomposition lsubSPR} in \Cref{section: decomposition run}),
and then we extend this result for \kl{lKL terms}.

\subsection{Decomposition for local sub-SP runs}\label{section: decomposition run}
In this subsection,
we first show the decomposition theorem for \emph{\kl{local} \kl{sub-series-parallel runs}} (\kl{local} \kl{sub-SP runs}).
In \Cref{section: subSPR}, we define \kl{sub-SP runs}.
In \Cref{section: local}, we define \emph{\kl{local}} \kl{sub-SP runs}.
In \Cref{section: decomposition of lsubSPR}, we show this theorem.

\subsubsection{Sub-SP runs}\label{section: subSPR}
The set $\SPR$ of \intro*\kl[SP runs]{series-parallel (SP) runs} is the subclass of \kl[runs]{runs with $\tuple{1,1}$-interface} defined by:
\begin{align*}
    \graph[1], \graph[2], \graph[3] \in \SPR &\;\Coloneqq\; \id^{1} \mid a^{1}_{1}
    \mid \graph[1] \series \graph[2] \mid \fork^{1}_{1} \series (\graph[1] \parallel \graph[2]) \series \join^{1}_{1} 
    \tag{$a \in \vsig$}
\end{align*}
For \kl{runs} $\graph[1]$ and $\graph[1]'$,
we say that $\graph'$ is a \intro*\kl{sub-run} of $\graph$ if there are \kl{runs} $\graph[2]$ and $\graph[3]$ such that $\graph[2] \series \graph' \series \graph[3] = \graph$.
We write $\subSPR$ for the set of all \intro*\kl{sub-series-parallel runs} (\kl{sub-SP runs}):
\[\subSPR \;\defeq\; \set{\graph[1]' \mid \mbox{there are \kl{runs} $\graph[2]$, $\graph[3]$ and an \kl{SP} \kl{run} $\graph[1]$ such that $\graph[1] = \graph[2] \series \graph[1]' \series \graph[3]$}}.\]
Similarly, we define the sets of \intro*\kl{sub-SP-l runs} and \intro*\kl{sub-SP-r runs}, written $\subSPRl$ and $\subSPRr$, as follows:
\begin{align*}
    \subSPRl &\;\defeq\; \set{\graph[1]' \mid \mbox{there are a \kl{run} $\graph[2]$ and an \kl{SP} \kl{run} $\graph[1]$ such that $\graph[1] = \graph[2] \series \graph[1]'$}},\\
    \subSPRr &\;\defeq\; \set{\graph[1]' \mid \mbox{there are a \kl{run} $\graph[3]$ and an \kl{SP} \kl{run} $\graph[1]$ such that $\graph[1] = \graph[1]' \series \graph[3]$}}.
\end{align*}
Namely, \kl{sub-SP-l runs} are \kl{runs} obtained by \kl{left quotients} of \kl{SP runs}.
Similarly, \kl{sub-SP-r runs} are \kl{runs} obtained by ``right quotients'' of \kl{SP runs}.

Below, we present alternative definitions of $\subSPR$, $\subSPRl$, and $\subSPRr$.
\begin{align*}
    \graph[1], \graph[2], \graph[3] \in \subSPR'
    &\;\Coloneqq\;
    (\graph[2]^{\mathrm{l}} \series \graph[3]^{\mathrm{r}})
    \mid (\graph[2] \parallel \graph[3])
    \tag{$\graph[2]^{\mathrm{l}} \in \subSPRl'$, $\graph[3]^{\mathrm{r}} \in \subSPRr'$, $\graph[2], \graph[3] \in \subSPR'$}\\
    \graph[1]^{\mathrm{l}}, \graph[2]^{\mathrm{l}}, \graph[3]^{\mathrm{l}} \in \subSPRl'
    &\;\Coloneqq\; \graph[3] \mid (\graph[1]^{\mathrm{l}} \parallel \graph[2]^{\mathrm{l}}) \series \join^{1}_{1} \series \graph[3]
    \tag{$\graph[1]^{\mathrm{l}}, \graph[2]^{\mathrm{l}} \in \subSPRl'$, $\graph[3] \in \SPR$}\\
    \graph[1]^{\mathrm{r}}, \graph[2]^{\mathrm{r}}, \graph[3]^{\mathrm{r}} \in \subSPRr'
    &\;\Coloneqq\; \graph[3] \mid \graph[3] \series \fork^{1}_{1} \series (\graph[1]^{\mathrm{r}} \parallel \graph[2]^{\mathrm{r}})
    \tag{$\graph[1]^{\mathrm{r}}, \graph[2]^{\mathrm{r}} \in \subSPRr'$, $\graph[3] \in \SPR$}
\end{align*}
\begin{prop}[\Cref{section: proposition: subSPR equiv}]\label{proposition: subSPR equiv}
We have $\subSPRl = \subSPRl'$, $\subSPRr = \subSPRr'$, and $\subSPR = \subSPR'$.
\end{prop}
Thanks to this alternative definition,
we have that every \kl{sub-SP run} is obtained by the (multiple) \kl{parallel compositions} of \kl{SP runs} having some \kl{vertex} such that all \kl{sources} are reachable to the \kl{vertex} and all \kl{targets} are reachable from the \kl{vertex} (note that $\ty_2(\graph) = 1$ for every $\graph \in \subSPRl$ and $\ty_1(\graph) = 1$ for every $\graph \in \subSPRr$).
For instance,
let us consider the \kl{sub-SP run} illustrated as follows (where $\graph_i \in \SPR$ for each $i$):
\begin{center}
    $\left(\begin{tikzpicture}[baseline = -8.5ex, yscale=-1]
        \foreach \x/\y/\style in {
            0/0/,0/2/,0/5/,
            1/0/,1/2/,
            2/1/,
            3/1/,3/5/,
            4/3/,
            5/3/{minimum width = 1.5ex, line width = 1.5pt},
            6/3/,
            7/1/,7/5/,
            8/1/,8/5/,
            0/6/, 0/8/,
            3/6/, 3/8/,
            4/7/,
            5/7/{minimum width = 1.5ex, line width = 1.5pt},
            8/7/}{
           \node[mynode,draw,circle,apply style/.expand once = \style]  (\x X\y) at (1.1*\x,.3*\y) {};
        }
        \foreach \x/\y/\label in {
            1/1,
            3/3, 3/7,
            7/3}{
           \node[inner sep = 0, outer sep = 0, minimum width = 0] (\x X\y) at (1.1*\x,.3*\y) {};
        }
        \node[left = 4pt of 0X0](s1){\tiny $1$}; \path (s1) edge[earrow, ->] (0X0);
        \node[left = 4pt of 0X2](s2){\tiny $2$}; \path (s2) edge[earrow, ->] (0X2);
        \node[left = 4pt of 0X5](s3){\tiny $3$}; \path (s3) edge[earrow, ->] (0X5);
        \node[left = 4pt of 0X6](s4){\tiny $4$}; \path (s4) edge[earrow, ->] (0X6);
        \node[left = 4pt of 0X8](s5){\tiny $5$}; \path (s5) edge[earrow, ->] (0X8);
        \node[right = 4pt of 8X1](t1){\tiny $1$}; \path (8X1) edge[earrow, ->] (t1);
        \node[right = 4pt of 8X5](t2){\tiny $2$}; \path (8X5) edge[earrow, ->] (t2);
        \node[right = 4pt of 8X7](t3){\tiny $3$}; \path (8X7) edge[earrow, ->] (t3);
        \foreach \xp/\x/\y in {6/7/3}{
            \node  (\x X\y) at (1.1*\x,.3*\y) {};
            \path (\xp X\y) edge [opacity = 0] node[pos= .5, elabel](\xp X\y Xf){$\fork$}(\x X\y);
        }
        \foreach \xp/\x/\y in {1/2/1,3/4/3, 3/4/7}{
            \node (\xp X\y) at (1.1*\xp,.3*\y) {};
            \path (\xp X\y) edge [opacity = 0] node[pos= .5, elabel](\x X\y Xj){$\join$}(\x X\y);
        }
        \graph[use existing nodes, edges={color=black, pos = .5, earrow, line width = .5pt, 
        },edge quotes={fill=white, inner sep=1pt,font= \scriptsize}]{
        {0X0} ->["$\graph_1$"] {1X0};
        {0X2} ->["$\graph_2$"] {1X2};
        {1X0} --["$1$"{auto}] {2X1Xj};
        {1X2} --["$2$"{auto, below right= .3pt}] {2X1Xj};
        {2X1Xj} -> {2X1};
        {2X1} ->["$\graph_3$"] {3X1};
        {0X5} ->["$\graph_4$"] {3X5};
        {3X1} --["$1$"{auto}] {4X3Xj};
        {3X5} --["$2$"{auto, below right= .3pt}] {4X3Xj};
        {4X3Xj} ->{4X3};
        {4X3} ->["$\graph_5$"] {5X3};
        {5X3} ->["$\graph_{5'}$"] {6X3};
        {6X3} -- {6X3Xf};
        {6X3Xf} ->["$1$"{auto}]{7X1};
        {6X3Xf} ->["$2$"{auto, below left= .3pt}]{7X5};
        {7X1} ->["$\graph_{3'}$"] {8X1};
        {7X5} ->["$\graph_{4'}$"] {8X5};
        {0X6} ->["$\graph_{6}$"] {3X6};
        {0X8} ->["$\graph_{7}$"] {3X8};
        {3X6} --["$1$"{auto}] {4X7Xj};
        {3X8} --["$2$"{auto, below right= .3pt}] {4X7Xj};
        {4X7Xj} ->{4X7};
        {4X7} ->["$\graph_{8}$"] {5X7};
        {5X7} ->["$\graph_{8'}$"] {8X7};
        };
    \end{tikzpicture}\right)$.
\end{center}
This is clearly a \kl{sub-SP run} by filling up both the left-hand side and the right-hand side, appropriately.
This \kl{sub-SP run} is decomposed to the \kl{parallel compositions} of \kl{runs} $\graph[1] \series \graph[2]$ where $\graph[1] \in \subSPRl$ and $\graph[2] \in \subSPRr$, as follows.
\begin{center}
    $\left(\begin{gathered}
        \left(\left(\begin{tikzpicture}[baseline = -5.5ex, yscale=-1]
        \foreach \x/\y/\style in {
            0/0/,0/2/,0/5/,
            1/0/,1/2/,
            2/1/,
            3/1/,3/5/,
            4/3/,
            5/3/{minimum width = 1.5ex, line width = 1.5pt}}{
           \node[mynode,draw,circle,apply style/.expand once = \style]  (\x X\y) at (1.1*\x,.3*\y) {};
        }
        \foreach \x/\y/\label in {
            1/1,
            3/3}{
           \node[inner sep = 0, outer sep = 0, minimum width = 0] (\x X\y) at (1.1*\x,.3*\y) {};
        }
        \node[left = 4pt of 0X0](s1){\tiny $1$}; \path (s1) edge[earrow, ->] (0X0);
        \node[left = 4pt of 0X2](s2){\tiny $2$}; \path (s2) edge[earrow, ->] (0X2);
        \node[left = 4pt of 0X5](s3){\tiny $3$}; \path (s3) edge[earrow, ->] (0X5);
        \node[right = 4pt of 5X3](t1){}; \path (5X3) edge[earrow, ->] (t1);
        \foreach \xp/\x/\y in {1/2/1,3/4/3}{
            \node (\xp X\y) at (1.1*\xp,.3*\y) {};
            \path (\xp X\y) edge [opacity = 0] node[pos= .5, elabel](\x X\y Xj){$\join$}(\x X\y);
        }
        \graph[use existing nodes, edges={color=black, pos = .5, earrow, line width = .5pt, 
        },edge quotes={fill=white, inner sep=1pt,font= \scriptsize}]{
        {0X0} ->["$\graph_1$"] {1X0};
        {0X2} ->["$\graph_2$"] {1X2};
        {1X0} --["$1$"{auto}] {2X1Xj};
        {1X2} --["$2$"{auto, below right= .3pt}] {2X1Xj};
        {2X1Xj} -> {2X1};
        {2X1} ->["$\graph_3$"] {3X1};
        {0X5} ->["$\graph_4$"] {3X5};
        {3X1} --["$1$"{auto}] {4X3Xj};
        {3X5} --["$2$"{auto, below right= .3pt}] {4X3Xj};
        {4X3Xj} ->{4X3};
        {4X3} ->["$\graph_5$"] {5X3};
        };
    \end{tikzpicture}\right) \series \left(\begin{tikzpicture}[baseline = -5.5ex, yscale=-1]
        \foreach \x/\y/\style in {
            5/3/{minimum width = 1.5ex, line width = 1.5pt},
            6/3/,
            7/1/,7/5/,
            8/1/,8/5/}{
           \node[mynode,draw,circle,apply style/.expand once = \style]  (\x X\y) at (1.1*\x,.3*\y) {};
        }
        \foreach \x/\y/\label in {
            7/3}{
           \node[inner sep = 0, outer sep = 0, minimum width = 0] (\x X\y) at (1.1*\x,.3*\y) {};
        }
        \node[left = 4pt of 5X3](s1){}; \path (s1) edge[earrow, ->] (5X3);
        \node[right = 4pt of 8X1](t1){\tiny $1$}; \path (8X1) edge[earrow, ->] (t1);
        \node[right = 4pt of 8X5](t2){\tiny $2$}; \path (8X5) edge[earrow, ->] (t2);
        \foreach \xp/\x/\y in {6/7/3}{
            \node  (\x X\y) at (1.1*\x,.3*\y) {};
            \path (\xp X\y) edge [opacity = 0] node[pos= .5, elabel](\xp X\y Xf){$\fork$}(\x X\y);
        }
        \graph[use existing nodes, edges={color=black, pos = .5, earrow, line width = .5pt, 
        },edge quotes={fill=white, inner sep=1pt,font= \scriptsize}]{
        {5X3} ->["$\graph_{5'}$"] {6X3};
        {6X3} -- {6X3Xf};
        {6X3Xf} ->["$1$"{auto}]{7X1};
        {6X3Xf} ->["$2$"{auto, below left= .3pt}]{7X5};
        {7X1} ->["$\graph_{3'}$"] {8X1};
        {7X5} ->["$\graph_{4'}$"] {8X5};
        };
    \end{tikzpicture}\right)\right)\\
    \parallel\\  
    \left(\left(\begin{tikzpicture}[baseline = -2.5ex, yscale=-1]
        \foreach \x/\y/\style in {
            0/6/, 0/8/,
            3/6/, 3/8/,
            4/7/,
            5/7/{minimum width = 1.5ex, line width = 1.5pt}}{
           \node[mynode,draw,circle,apply style/.expand once = \style]  (\x X\y) at (1.1*\x,.3*\y - 1.8) {};
        }
        \node[left = 4pt of 0X6](s4){\tiny $1$}; \path (s4) edge[earrow, ->] (0X6);
        \node[left = 4pt of 0X8](s5){\tiny $2$}; \path (s5) edge[earrow, ->] (0X8);
        \node[right = 4pt of 5X7](t1){}; \path (5X7) edge[earrow, ->] (t1);
        \foreach \xp/\x/\y in {3/4/7}{
            \node (\xp X\y) at (1.1*\xp,.3*\y - 1.8) {};
            \path (\xp X\y) edge [opacity = 0] node[pos= .5, elabel](\x X\y Xj){$\join$}(\x X\y);
        }
        \graph[use existing nodes, edges={color=black, pos = .5, earrow, line width = .5pt, 
        },edge quotes={fill=white, inner sep=1pt,font= \scriptsize}]{
        {0X6} ->["$\graph_{6}$"] {3X6};
        {0X8} ->["$\graph_{7}$"] {3X8};
        {3X6} --["$1$"{auto}] {4X7Xj};
        {3X8} --["$2$"{auto, below right= .3pt}] {4X7Xj};
        {4X7Xj} ->{4X7};
        {4X7} ->["$\graph_{8}$"] {5X7};
        };
    \end{tikzpicture}\right) \series \left(\begin{tikzpicture}[baseline = -2.5ex, yscale=-1]
        \foreach \x/\y/\style in {
            5/7/{minimum width = 1.5ex, line width = 1.5pt},
            8/7/}{
           \node[mynode,draw,circle,apply style/.expand once = \style]  (\x X\y) at (1.1*\x,.3*\y - 1.8) {};
        }
        \node[left = 4pt of 5X7](s1){}; \path (s1) edge[earrow, ->] (5X7);
        \node[right = 4pt of 8X7](t3){\tiny $1$}; \path (8X7) edge[earrow, ->] (t3);
        \graph[use existing nodes, edges={color=black, pos = .5, earrow, line width = .5pt, 
        },edge quotes={fill=white, inner sep=1pt,font= \scriptsize}]{
        {5X7} ->["$\graph_{8'}$"] {8X7};
        };
    \end{tikzpicture}\right)\right)
    \end{gathered}\right).$
\end{center}

\begin{rem}\label{remark: N-free}
    Every \kl[sub-SP run]{(sub-)SP run} does not contain the four pairwise distinct \kl{vertices} of the form
    $\left(\begin{tikzpicture}[baseline = -1.5ex]
        \graph[grow right = .8cm, branch down = 1.ex, nodes={}]{
        {L1/{}[mynode,draw,circle], /, L2/{}[mynode,draw,circle]} -!- {R1/{}[mynode,draw,circle], /, R2/{}[mynode,draw,circle]} };
        \node[left = 4pt of L1](s1){}; \path (s1) edge[earrow, ->] (L1);
        \node[left = 4pt of L2](s2){}; \path (s2) edge[earrow, ->] (L2);
        \node[right = 4pt of R1](t1){}; \path (R1) edge[earrow, ->] (t1);
        \node[right = 4pt of R2](t2){}; \path (R2) edge[earrow, ->] (t2);
        \graph[use existing nodes, edges={color=black, pos = .5, earrow, line width = .5pt,
        decorate, decoration={snake, segment length=2mm, amplitude=.2mm}
        },
        edge quotes={fill=white, inner sep=1pt,font= \scriptsize}]{
            {L1,L2} -> {R1,R2};
            {L2} -> {R1};
        };
    \end{tikzpicture}\right)$, \cf\ ``N-free'' \cite[p.\ 17]{valdesRecognitionSeriesParallel1979}\cite{lodayaSeriesparallelPosetsAlgebra1998}.
    Here,  
    $\begin{tikzpicture}[baseline = -0.5ex]
        \graph[grow right = 1.cm, branch down = 4.ex, nodes={}, edges={color=black, pos = .5, earrow, line width = .5pt,
        decorate, decoration={snake, segment length=2mm, amplitude=.2mm}
        }]{{L/{}[mynode,draw,circle]} -> {R/{}[mynode,draw,circle]}};
    \end{tikzpicture}$
    expresses that the target \kl{vertex} is reachable from the source \kl{vertex}.
\end{rem}

\subsubsection{Labeled KL terms and local sub-SP runs}\label{section: local}
We first note the following proposition.
\begin{prop}\label{proposition: subSP}
    Let $\struc$ be a \kl{structure}.
    \begin{itemize}
        \item For all \kl{KL terms} $\term$,
        every \kl{$\struc$-run} of $\term$ is \kl{SP}.
        \item For all \kl{lKL terms} $\lterm$,
        every \kl{$\struc$-run} of $\lterm$ is \kl{sub-SP-l}.
    \end{itemize}
\end{prop}
\begin{proof}
    By easy induction on terms.
\end{proof}
Moreover, each \kl[$\struc$-run]{$(\odot \vec{\struc})_{\bullet}$-run} of \kl{lKL terms} is \emph{\kl{local}} in the following sense.
\begin{defi}\label{definition: local}
    Let $\vec{\struc} = \struc_1 \dots \struc_n$.
    We say that an \kl[$\struc$-run]{$(\odot \vec{\struc})_{\bullet}$-run} $\trace \colon \graph \homo (\odot\vec{\struc})_{\bullet}$ is \intro*\kl{local} if
    for all $\tuple{v_1, \dots, v_k} \in a^{\graph}$,
    there is some $i \in \range{n}$ s.t.\  $\set{\trace(v_1), \dots, \trace(v_k)} \subseteq
    \set{[\tuple{i, \lab}]_{\sim} \mid \lab \in \univ{\struc_i}}$.
\end{defi}
\begin{prop}\label{proposition: local}
    Let $\lterm$ be an \kl{lKL term}.
    Every \kl[$\struc$-run]{$(\odot \vec{\struc})_{\bullet}$-run} of $\lterm$ is \kl{local} \kl{sub-SP-l}.
\end{prop}
\begin{proof}
    For ``\kl{sub-SP-l}'':
    By \Cref{proposition: subSP}.
    For ``\kl{local}'':
    By definition of \kl[$\struc$-run]{$(\odot \vec{\struc})_{\bullet}$-run}.
\end{proof}

\subsubsection{Decomposition of local sub-SP runs}\label{section: decomposition of lsubSPR}
Let $\vec{\struc} = \struc_1 \dots \struc_n \in \STR^{+}$.
Let $\lsubSPR^{(\odot \vec{\struc})_{\bullet}}$ be the set of \kl{local} \kl{sub-SP} \kl[$\struc$-runs]{$(\odot \vec{\struc})_{\bullet}$-runs}.
Similarly, we use $\lSPR^{(\odot \vec{\struc})_{\bullet}}$, $\lsubSPRl^{(\odot \vec{\struc})_{\bullet}}$, and $\lsubSPRr^{(\odot \vec{\struc})_{\bullet}}$.
We now define the decompositions of \kl{local} \kl{sub-SP} \kl[$\struc$-runs]{$(\odot \vec{\struc})_{\bullet}$-runs}.
For $i \in \range{n}$,
we let $\lsubSPR^{\vec{\struc}}_{(i)} \defeq \set{\trace \in \lsubSPR^{(\odot \vec{\struc})_{\bullet}} \mid 
\ty_1(\trace)\ty_2(\trace) \in \set{[\tuple{i, \lab}]_{\sim} \mid \lab \in \univ{(\struc_i)_{\bullet}}}^*}$.
For $i \in \range{n}$ and an \kl[$\struc$-run]{$(\struc_i)_{\bullet}$-run} $\trace$,
we write $\trace_{(i)}$ for the \kl[$\struc$-run]{$(\odot \vec{\struc})_{\bullet}$-run} $\trace$ in which each label $\lab$ has been replaced with $[\tuple{i, \lab}]_{\sim}$.
We then consider the following decomposition of \kl[$\struc$-run]{$(\odot \vec{\struc})_{\bullet}$-run}.
\begin{defi}\label{definition: decomposition lsubSPR}
    Let $\vec{\struc} = \struc_1 \dots \struc_n \in \STR^{+}$.
    The set $\odot_{i = 1}^{n} \lsubSPR^{(\struc_i)_{\bullet}}$
    is defined as the smallest subset of $\bigcup_{i = 1}^{n} \lsubSPR^{\vec{\struc}}_{(i)}$ closed under the following rules:
    \begin{gather*}
        \begin{prooftree}
            \hypo{\trace[1] \in \lsubSPR^{\struc_j}}
            \infer1[D]{\trace[1]_{(j)} \in \odot_{i = 1}^{n} \lsubSPR^{(\struc_i)_{\bullet}}}
        \end{prooftree}
        \hspace{3em}
        \begin{prooftree}
            \hypo{\trace[1][\bullet/\tuple{l, r}] \in \odot_{i = 1}^{n} \lsubSPR^{(\struc_i)_{\bullet}}}
            \infer1[L]{\trace[1][\lab/\tuple{l, r}] \in \odot_{i = 1}^{n} \lsubSPR^{(\struc_i)_{\bullet}}}
        \end{prooftree}\\
        \begin{prooftree}
            \hypo{\trace[1] \in \odot_{i = 1}^{n} \lsubSPR^{(\struc_i)_{\bullet}}}
            \hypo{\trace[2] \in \odot_{i = 1}^{n} \lsubSPR^{(\struc_i)_{\bullet}}}
            \infer2[T]{\trace[1] \series \trace[2] \in \odot_{i = 1}^{n} \lsubSPR^{(\struc_i)_{\bullet}}}
        \end{prooftree}
    \end{gather*}
    Here, $\trace[1][\lab[2]/\tuple{l, r}]$ denotes the \kl[$\struc$-run]{$(\odot \vec{\struc})_{\bullet}$-run} $\trace[1]$ in which $\trace[1](\src^{\trace[1]}_{l})$ and $\trace[1](\tgt^{\trace[1]}_{r})$
    has been replaced with $\lab[2]$.
\end{defi}
We then have that every \kl[$\struc$-run]{$(\odot \vec{\struc})_{\bullet}$-run} of $\bigcup_{i = 1}^{n} \lsubSPR^{\vec{\struc}}_{(i)}$
can be derived using the rules of \Cref{definition: decomposition lsubSPR} (\Cref{lemma: decomposition lsubSPR}).
Note that each \kl[$\struc$-run]{$(\odot \vec{\struc})_{\bullet}$-run} occurring in the derivation tree should be in $\lsubSPR^{\vec{\struc}}_{(i)}$ for some $i$.
\begin{exa}\label{example: definition: decomposition lsubSPR}
    We recall the sequence $\vec{\struc} = \textcolor{blue}{\struc_1} \textcolor{red!60}{\struc_2}$ and the following \kl[$\struc$-run]{$(\odot \vec{\struc})_{\bullet}$-run} $\trace \in \lsubSPR^{\vec{\struc}}_{(1)}$ in \Cref{example: decomposition}:
    \begin{center}
        \begin{tikzpicture}[baseline = -.5ex, yscale=-1]
            \foreach \x/\y/\label/\style/\dc in {
                0/2/{$1, 1$}//{blue},2/1/{$1, 2$}/{thick}/,
                5/0/{$1, 2$}/{thick}/,6/2/{$2, 3$}//{red!60},6/3/{$1, 1$}//{blue},
                7/0/{$1, 2$}/{thick}/,10/1/{$1, 2$}/{thick}/,
                12/2/{$1, 1$}//{blue}}{
               \node[mynode,draw,circle, \style, draw = \dc, inner sep = .5pt]  (\x X\y) at (.6*\x,.6*\y) {\tiny \label};
            }
            \foreach \x/\y/\label in {
                1/1/{$1, 1$},1/3/{},3/1/{$2, 3$},
                4/0/{$2, 3$},4/2/{}, 6/0/{$1, 1$},
                8/0/{$2, 3$},8/2/{}, 9/1/{$2, 3$},11/3/{},11/1/{$1, 1$}}{
               \node[inner sep = 0, outer sep = 0, minimum width = 0] (\x X\y) at (.6*\x,.6*\y) {};
            }
            \node[left = 4pt of 0X2](s){}; \path (s) edge[earrow, ->] (0X2);
            \node[right = 4pt of 12X2](t){}; \path (12X2) edge[earrow, ->] (t);
            \foreach \xp/\x/\y/\col in {0/1/2/blue,3/4/1/red!60}{
                \node  (\x X\y) at (.6*\x,.6*\y) {};
                \path (\xp X\y) edge [opacity = 0] node[pos= .9, elabel, color = \col, inner sep = 0, outer sep = 0, minimum width = 0](\xp X\y Xf){}(\x X\y);
            }
            \foreach \xp/\x/\y/\col in {8/9/1/red!60,11/12/2/blue}{
                \node (\xp X\y) at (.6*\xp,.6*\y) {};
                \path (\xp X\y) edge [opacity = 0] node[pos= .1, elabel, color = \col, inner sep = 0, outer sep = 0, minimum width = 0](\x X\y Xj){}(\x X\y);
            }
            \graph[use existing nodes, edges={color=black, pos = .5, earrow, line width = .5pt, 
            decorate, decoration={snake, segment length=2mm, amplitude=.2mm}
            },edge quotes={fill=white, inner sep=1pt,font= \scriptsize}]{
            {0X2} --[color = blue] {0X2Xf} --[color = blue] {1X1, 1X3}; {1X3} ->[color = blue] {6X3};
            {1X1} ->[ color = blue] {2X1} --[ color = red!60] {3X1};
            {3X1} --[color = red!60] {3X1Xf} -- [color = red!60] {4X0, 4X2}; {4X2} ->[color=red!60] {6X2};
            {4X0} ->[ color = red!60] {5X0} --[ color = blue] {6X0}; {6X3} --[color=blue] {11X3};
            {6X0} ->[ color = blue] {7X0} --[ color = red!60] {8X0}; {6X2} --[color=red!60] {8X2};
            {8X0, 8X2} --[color = red!60] {9X1Xj} -- [color = red!60] {9X1};
            {9X1} ->[ color = red!60] {10X1} --[ color = blue] {11X1};
            {11X1, 11X3} --[color = blue] {12X2Xj} -> [color = blue] {12X2};
            };
        \end{tikzpicture}
    \end{center}
    We then also have $\tau \in \odot_{i = 1}^{n} \lsubSPR^{(\struc_i)_{\bullet}}$ by the following derivation tree (\cf\ \Cref{example: decomposition}):
    \begin{center}
        \begin{prooftree}
            \hypo{\mathstrut}
            \infer1{\begin{tikzpicture}[baseline = -.5ex, yscale=-1]
                \foreach \x/\y/\style/\dc in {
                    0/2//{blue},2/1/{thick}/,6/3//{blue}}{
                   \node[mynode,draw = \dc,circle, \style]  (\x X\y) at (.25*\x,.25*\y) {};
                }
                \foreach \x/\y in {
                    1/1,1/3,3/1,4/0,4/2}{
                   \node[inner sep = 0, outer sep = 0, minimum width = 0] (\x X\y) at (.25*\x,.25*\y) {};
                }
                \node[left = 4pt of 0X2](s){}; \path (s) edge[earrow, ->] (0X2);
                \node[right = 4pt of 2X1](t1){}; \path (2X1) edge[earrow, ->] (t1);
                \node[right = 4pt of 6X3](t2){}; \path (6X3) edge[earrow, ->] (t2);
                \foreach \xp/\x/\y/\col in {0/1/2/blue}{
                    \node  (\x X\y) at (.25*\x+.1,.25*\y) {};
                    \path (\xp X\y) edge [opacity = 0] node[pos= .9, elabel, color = \col, inner sep = 0, outer sep = 0, minimum width = 0](\xp X\y Xf){}(\x X\y);
                }
                \graph[use existing nodes, edges={color=black, pos = .5, earrow, line width = .5pt, 
                decorate, decoration={snake, segment length=2mm, amplitude=.2mm}
                },edge quotes={fill=white, inner sep=1pt,font= \scriptsize}]{
                {0X2} --[color = blue] {0X2Xf} --[color = blue] {1X1, 1X3}; {1X1} --[ color = blue] {2X1};
                {1X3} --[color=blue] {6X3};
                };
            \end{tikzpicture}}
            \hypo{\mathstrut}
        \infer1{\begin{tikzpicture}[baseline = -.5ex, yscale=-1]
            \foreach \x/\y/\col/\style/\dc in {
                2/1/gray!20/{thick}/,5/0/gray!20/{thick}/,6/2/gray!20//{red!60},6/3/black//}{
               \node[mynode, draw = \dc, circle, fill = \col, \style]  (\x X\y) at (.25*\x,.25*\y) {};
            }
            \foreach \x/\y in {
                3/1,4/0,4/2}{
               \node[inner sep = 0, outer sep = 0, minimum width = 0] (\x X\y) at (.25*\x,.25*\y) {};
            }
            \node[left = 4pt of 2X1](s1){}; \path (s1) edge[earrow, ->] (2X1);
            \node[left = 4pt of 6X3](s2){}; \path (s2) edge[earrow, ->] (6X3);
            \node[right = 4pt of 5X0](t1){}; \path (5X0) edge[earrow, ->] (t1);
            \node[right = 4pt of 6X2](t2){}; \path (6X2) edge[earrow, ->] (t2);
            \node[right = 4pt of 6X3](t3){}; \path (6X3) edge[earrow, ->] (t3);
            \foreach \xp/\x/\y/\col in {3/4/1/red!60}{
                \node  (\x X\y) at (.25*\x+.1,.25*\y) {};
                \path (\xp X\y) edge [opacity = 0] node[pos= .9, elabel, color = \col, inner sep = 0, outer sep = 0, minimum width = 0](\xp X\y Xf){}(\x X\y);
            }
            \graph[use existing nodes, edges={color=black, pos = .5, earrow, line width = .5pt, 
            decorate, decoration={snake, segment length=2mm, amplitude=.2mm}
            },edge quotes={fill=white, inner sep=1pt,font= \scriptsize}]{
            {2X1} --[ color = red!60] {3X1};
            {3X1} --[color = red!60] {3X1Xf} -- [color = red!60] {4X0, 4X2}; {4X2} --[ color = red!60] {6X2};
            {4X0} ->[ color = red!60] {5X0};
            };
        \end{tikzpicture}}
        \hypo{\mathstrut}
        \infer1{
            \begin{tikzpicture}[baseline = -.5ex, yscale=-1]
                \foreach \x/\y/\col/\style/\dc in {
                    5/0/gray!20/{thick}/,6/2/black//,6/3/black//,7/0/gray!20/{thick}/}{
                   \node[mynode,draw = \dc,circle, fill = \col, \style]  (\x X\y) at (.25*\x,.25*\y) {};
                }
                \foreach \x/\y in {
                    6/0}{
                   \node[inner sep = 0, outer sep = 0, minimum width = 0] (\x X\y) at (.25*\x,.25*\y) {};
                }
                \node[left = 4pt of 5X0](s1){}; \path (s1) edge[earrow, ->] (5X0);
                \node[left = 4pt of 6X2](s2){}; \path (s2) edge[earrow, ->] (6X2);
                \node[left = 4pt of 6X3](s3){}; \path (s3) edge[earrow, ->] (6X3);
                \node[right = 4pt of 7X0](t1){}; \path (7X0) edge[earrow, ->] (t1);
                \node[right = 4pt of 6X2](t2){}; \path (6X2) edge[earrow, ->] (t2);
                \node[right = 4pt of 6X3](t3){}; \path (6X3) edge[earrow, ->] (t3);
                \graph[use existing nodes, edges={color=black, pos = .5, earrow, line width = .5pt, 
                decorate, decoration={snake, segment length=2mm, amplitude=.2mm}
                },edge quotes={fill=white, inner sep=1pt,font= \scriptsize}]{
                {5X0} --[ color = blue] {6X0};
                {6X0} ->[ color = blue] {7X0};
                };
            \end{tikzpicture}}
        \infer1[L]{\begin{tikzpicture}[baseline = -.5ex, yscale=-1]
            \foreach \x/\y/\col/\style/\dc in {
                5/0/gray!20/{thick}/,6/2/gray!20//{red!60},6/3/black//,7/0/gray!20/{thick}/}{
                \node[mynode,draw = \dc,circle, fill = \col, \style]  (\x X\y) at (.25*\x,.25*\y) {};
            }
            \foreach \x/\y in {
                6/0}{
               \node[inner sep = 0, outer sep = 0, minimum width = 0] (\x X\y) at (.25*\x,.25*\y) {};
            }
            \node[left = 4pt of 5X0](s1){}; \path (s1) edge[earrow, ->] (5X0);
            \node[left = 4pt of 6X2](s2){}; \path (s2) edge[earrow, ->] (6X2);
            \node[left = 4pt of 6X3](s3){}; \path (s3) edge[earrow, ->] (6X3);
            \node[right = 4pt of 7X0](t1){}; \path (7X0) edge[earrow, ->] (t1);
            \node[right = 4pt of 6X2](t2){}; \path (6X2) edge[earrow, ->] (t2);
            \node[right = 4pt of 6X3](t3){}; \path (6X3) edge[earrow, ->] (t3);
            \graph[use existing nodes, edges={color=black, pos = .5, earrow, line width = .5pt, 
            decorate, decoration={snake, segment length=2mm, amplitude=.2mm}
            },edge quotes={fill=white, inner sep=1pt,font= \scriptsize}]{
            {5X0} --[ color = blue] {6X0};
            {6X0} ->[ color = blue] {7X0};
            };
        \end{tikzpicture}}
        \hypo{\mathstrut}
        \infer1{\begin{tikzpicture}[baseline = -.5ex, yscale=-1]
                \foreach \x/\y/\col/\style/\dc in {
                    6/2/gray!20//{red!60},6/3/black//,7/0/gray!20/{thick}/,10/1/gray!20/{thick}/}{
                   \node[mynode,draw = \dc,circle, fill = \col, \style]  (\x X\y) at (.25*\x,.25*\y) {};
                }
                \foreach \x/\y in {
                    8/0,8/2,9/1}{
                   \node[inner sep = 0, outer sep = 0, minimum width = 0] (\x X\y) at (.25*\x,.25*\y) {};
                }
                \node[left = 4pt of 7X0](s1){}; \path (s1) edge[earrow, ->] (7X0);
                \node[left = 4pt of 6X2](s2){}; \path (s2) edge[earrow, ->] (6X2);
                \node[left = 4pt of 6X3](s3){}; \path (s3) edge[earrow, ->] (6X3);
                \node[right = 4pt of 10X1](t1){}; \path (10X1) edge[earrow, ->] (t1);
                \node[right = 4pt of 6X3](t2){}; \path (6X3) edge[earrow, ->] (t2);
                \foreach \xp/\x/\y/\col in {8/9/1/red}{
                    \node (\xp X\y) at (.25*\xp,.25*\y) {};
                    \path (\xp X\y) edge [opacity = 0] node[pos= .1, elabel, color = \col, inner sep = 0, outer sep = 0, minimum width = 0](\x X\y Xj){}(\x X\y);
                }
                \graph[use existing nodes, edges={color=black, pos = .5, earrow, line width = .5pt, 
                decorate, decoration={snake, segment length=2mm, amplitude=.2mm}
                },edge quotes={fill=white, inner sep=1pt,font= \scriptsize}]{
                {7X0} --[ color = red!60] {8X0}; {6X2} --[ color = red!60] {8X2};
                {8X0,8X2} --[color = red!60] {9X1Xj} -- [color = red!60] {9X1};
                {9X1} ->[ color = red!60] {10X1};
                };
            \end{tikzpicture}
        }
        \infer[double]3[T]{
              \begin{tikzpicture}[baseline = -.5ex, yscale=-1]
                \foreach \x/\y/\col/\style/\dc in {
                    2/1/gray!20/{thick}/,5/0/gray!20/{thick}/,6/2/gray!20//{red!60},6/3/black//,7/0/gray!20/{thick}/,10/1/gray!20/{thick}/}{
                   \node[mynode,draw = \dc,circle, fill = \col, \style]  (\x X\y) at (.25*\x,.25*\y) {};
                }
                \foreach \x/\y in {
                    1/1,1/3,3/1,4/0,4/2,6/0,8/0,8/2,9/1,11/3,11/1}{
                   \node[inner sep = 0, outer sep = 0, minimum width = 0] (\x X\y) at (.25*\x,.25*\y) {};
                }
                \node[left = 4pt of 2X1](s1){}; \path (s1) edge[earrow, ->] (2X1);
                \node[left = 4pt of 6X3](s2){}; \path (s2) edge[earrow, ->] (6X3);
                \node[right = 4pt of 10X1](t1){}; \path (10X1) edge[earrow, ->] (t1);
                \node[right = 4pt of 6X3](t2){}; \path (6X3) edge[earrow, ->] (t2);
                \foreach \xp/\x/\y/\col in {3/4/1/red!60}{
                    \node  (\x X\y) at (.25*\x+.1,.25*\y) {};
                    \path (\xp X\y) edge [opacity = 0] node[pos= .9, elabel, color = \col, inner sep = 0, outer sep = 0, minimum width = 0](\xp X\y Xf){}(\x X\y);
                }
                \foreach \xp/\x/\y/\col in {8/9/1/red!60}{
                    \node (\xp X\y) at (.25*\xp-.1,.25*\y) {};
                    \path (\xp X\y) edge [opacity = 0] node[pos= .1, elabel, color = \col, inner sep = 0, outer sep = 0, minimum width = 0](\x X\y Xj){}(\x X\y);
                }
                \graph[use existing nodes, edges={color=black, pos = .5, earrow, line width = .5pt, 
                decorate, decoration={snake, segment length=2mm, amplitude=.2mm}
                },edge quotes={fill=white, inner sep=1pt,font= \scriptsize}]{
                {2X1} --[ color = red!60] {3X1};
                {3X1} --[color = red!60] {3X1Xf} -- [color = red!60] {4X0, 4X2}; {4X2} --[ color = red!60] {6X2};
                {4X0} ->[ color = red!60] {5X0} --[ color = blue] {6X0};
                {6X0} ->[ color = blue] {7X0} --[ color = red!60] {8X0}; {6X2} --[ color = red!60] {8X2};
                {8X0,8X2} --[color = red!60] {9X1Xj} -- [color = red!60] {9X1};
                {9X1} ->[ color = red!60] {10X1};
                };
            \end{tikzpicture}
        }
        \infer1[L]{  \begin{tikzpicture}[baseline = -.5ex, yscale=-1]
            \foreach \x/\y/\col/\style/\dc in {
                2/1/gray!20/{thick}/,5/0/gray!20/{thick}/,6/2/gray!20//{red!60},6/3/gray!20//{blue},7/0/gray!20/{thick}/,10/1/gray!20/{thick}/}{
                \node[mynode,draw = \dc,circle, fill = \col, \style]  (\x X\y) at (.25*\x,.25*\y) {};
            }
            \foreach \x/\y in {
                1/1,1/3,3/1,4/0,4/2,6/0,8/0,8/2,9/1,11/3,11/1}{
               \node[inner sep = 0, outer sep = 0, minimum width = 0] (\x X\y) at (.25*\x,.25*\y) {};
            }
            \node[left = 4pt of 2X1](s1){}; \path (s1) edge[earrow, ->] (2X1);
            \node[left = 4pt of 6X3](s2){}; \path (s2) edge[earrow, ->] (6X3);
            \node[right = 4pt of 10X1](t1){}; \path (10X1) edge[earrow, ->] (t1);
            \node[right = 4pt of 6X3](t2){}; \path (6X3) edge[earrow, ->] (t2);
            \foreach \xp/\x/\y/\col in {3/4/1/red!60}{
                \node  (\x X\y) at (.25*\x+.1,.25*\y) {};
                \path (\xp X\y) edge [opacity = 0] node[pos= .9, elabel, color = \col, inner sep = 0, outer sep = 0, minimum width = 0](\xp X\y Xf){}(\x X\y);
            }
            \foreach \xp/\x/\y/\col in {8/9/1/red!60}{
                \node (\xp X\y) at (.25*\xp-.1,.25*\y) {};
                \path (\xp X\y) edge [opacity = 0] node[pos= .1, elabel, color = \col, inner sep = 0, outer sep = 0, minimum width = 0](\x X\y Xj){}(\x X\y);
            }
            \graph[use existing nodes, edges={color=black, pos = .5, earrow, line width = .5pt, 
            decorate, decoration={snake, segment length=2mm, amplitude=.2mm}
            },edge quotes={fill=white, inner sep=1pt,font= \scriptsize}]{
            {2X1} --[ color = red!60] {3X1};
            {3X1} --[color = red!60] {3X1Xf} -- [color = red!60] {4X0, 4X2}; {4X2} --[ color = red!60] {6X2};
            {4X0} ->[ color = red!60] {5X0} --[ color = blue] {6X0};
            {6X0} ->[ color = blue] {7X0} --[ color = red!60] {8X0}; {6X2} --[ color = red!60] {8X2};
            {8X0,8X2} --[color = red!60] {9X1Xj} -- [color = red!60] {9X1};
            {9X1} ->[ color = red!60] {10X1};
            };
        \end{tikzpicture}}
        \hypo{\mathstrut}
        \infer1{ \begin{tikzpicture}[baseline = -.5ex, yscale=-1]
            \foreach \x/\y/\col/\style/\dc in {
               6/3/gray!20//{blue},10/1/gray!20/{thick}/,12/2/gray!20//{blue}}{
               \node[mynode,draw = \dc,circle, fill = \col, \style]  (\x X\y) at (.25*\x,.25*\y) {};
            }
            \foreach \x/\y in {
               11/3,11/1}{
               \node[inner sep = 0, outer sep = 0, minimum width = 0] (\x X\y) at (.25*\x,.25*\y) {};
            }
            \node[left = 4pt of 10X1](s1){}; \path (s1) edge[earrow, ->] (10X1);
            \node[left = 4pt of 6X3](s2){}; \path (s2) edge[earrow, ->] (6X3);
            \node[right = 4pt of 12X2](t){}; \path (12X2) edge[earrow, ->] (t);
            \foreach \xp/\x/\y/\col in {11/12/2/blue}{
                \node (\xp X\y) at (.25*\xp-.1,.25*\y) {};
                \path (\xp X\y) edge [opacity = 0] node[pos= .1, elabel, color = \col, inner sep = 0, outer sep = 0, minimum width = 0](\x X\y Xj){}(\x X\y);
            }
            \graph[use existing nodes, edges={color=black, pos = .5, earrow, line width = .5pt, 
            decorate, decoration={snake, segment length=2mm, amplitude=.2mm}
            },edge quotes={fill=white, inner sep=1pt,font= \scriptsize}]{
            {10X1} --[ color = blue] {11X1}; {6X3} --[color=blue] {11X3};
            {11X1,11X3} --[color = blue] {12X2Xj} -> [color = blue] {12X2};
            };
        \end{tikzpicture}}
        \infer[double]3[T]{
            \begin{tikzpicture}[baseline = -.5ex, yscale=-1]
                \foreach \x/\y/\col/\style/\dc in {
                    0/2/gray!20//{blue},2/1/gray!20/{thick}/,5/0/gray!20/{thick}/,6/2/gray!20//{red!60},6/3/gray!20//{blue},7/0/gray!20/{thick}/,10/1/gray!20/{thick}/,12/2/gray!20//{blue}}{
                   \node[mynode,draw = \dc,circle, fill = \col, \style]  (\x X\y) at (.25*\x,.25*\y) {};
                }
                \foreach \x/\y in {
                    1/1,1/3,3/1,4/0,4/2,6/0,8/0,8/2,9/1,11/3,11/1}{
                   \node[inner sep = 0, outer sep = 0, minimum width = 0] (\x X\y) at (.25*\x,.25*\y) {};
                }
                \node[left = 4pt of 0X2](s){}; \path (s) edge[earrow, ->] (0X2);
                \node[right = 4pt of 12X2](t){}; \path (12X2) edge[earrow, ->] (t);
                \foreach \xp/\x/\y/\col in {0/1/2/blue,3/4/1/red!60}{
                    \node  (\x X\y) at (.25*\x+.1,.25*\y) {};
                    \path (\xp X\y) edge [opacity = 0] node[pos= .9, elabel, color = \col, inner sep = 0, outer sep = 0, minimum width = 0](\xp X\y Xf){}(\x X\y);
                }
                \foreach \xp/\x/\y/\col in {8/9/1/red!60,11/12/2/blue}{
                    \node (\xp X\y) at (.25*\xp-.1,.25*\y) {};
                    \path (\xp X\y) edge [opacity = 0] node[pos= .1, elabel, color = \col, inner sep = 0, outer sep = 0, minimum width = 0](\x X\y Xj){}(\x X\y);
                }
                \graph[use existing nodes, edges={color=black, pos = .5, earrow, line width = .5pt, 
                decorate, decoration={snake, segment length=2mm, amplitude=.2mm}
                },edge quotes={fill=white, inner sep=1pt,font= \scriptsize}]{
                {0X2} --[color = blue] {0X2Xf} --[color = blue] {1X1, 1X3}; {1X3} --[color=blue] {6X3};
                {1X1} ->[ color = blue] {2X1} --[ color = red!60] {3X1};
                {3X1} --[color = red!60] {3X1Xf} -- [color = red!60] {4X0, 4X2}; {4X2} --[ color = red!60] {6X2};
                {4X0} ->[ color = red!60] {5X0} --[ color = blue] {6X0};
                {6X0} ->[ color = blue] {7X0} --[ color = red!60] {8X0}; {6X2} --[ color = red!60] {8X2};
                {8X0,8X2} --[color = red!60] {9X1Xj} -- [color = red!60] {9X1};
                {9X1} ->[ color = red!60] {10X1} --[ color = blue] {11X1}; {6X3} --[color=blue] {11X3};
                {11X1,11X3} --[color = blue] {12X2Xj} -> [color = blue] {12X2};
                };
            \end{tikzpicture}
        }
    \end{prooftree}
    \end{center}
\end{exa}
The most technical point in \Cref{example: definition: decomposition lsubSPR} is to find an appropriate interpolation.
In the following, we show that we can always take an appropriate interpolation (\Cref{lemma: decomposition lsubSPR}).

For $\vec{\struc} = \struc_1 \dots \struc_n \in \STR^{+}$ and $i \in \range{n}$,
we consider the following three disjoint sets:
\begin{align*}
    \univ{\vec{\struc}}_{(i)} &\defeq \set{[\tuple{i, \lab}]_{\sim} \mid \lab \in \univ{\struc_i}}, &
    \univ{\vec{\struc}}_{(<i)} &\defeq (\bigcup_{j < i} \univ{\vec{\struc}}_{(j)}) \setminus \univ{\vec{\struc}}_{(i)}, &
    \univ{\vec{\struc}}_{(>i)} &\defeq (\bigcup_{j > i} \univ{\vec{\struc}}_{(j)}) \setminus \univ{\vec{\struc}}_{(i)}.
\end{align*}
The following is an illustration of the three sets:
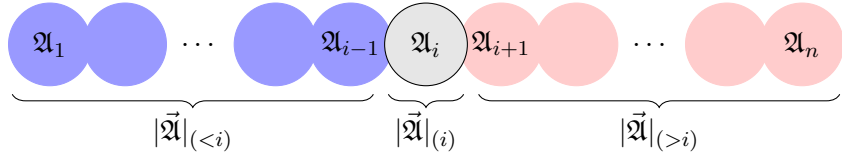
\begin{center}
    \begin{tikzpicture}
        \filldraw[fill=blue!40, color = blue!40] (-5.0,0) circle (.55cm); \node at (-5.0,0) {$\struc_1$};
        \filldraw[fill=blue!40, color = blue!40] (-4.0,0) circle (.55cm);
        \node at (-3.0,0) {$\dots$};
        \filldraw[fill=blue!40, color = blue!40] (-2.0,0) circle (.55cm);
        \filldraw[fill=blue!40, color = blue!40] (-1.0,0) circle (.55cm); \node at (-1.0,0) {$\struc_{i-1}$};
        \filldraw[fill=red!20, color = red!20] (5.0,0) circle (.55cm); \node at (5.0,0) {$\struc_{n}$};
        \filldraw[fill=red!20, color = red!20] (4.0,0) circle (.55cm);
        \node at (3.0,0) {$\dots$};
        \filldraw[fill=red!20, color = red!20] (2.0,0) circle (.55cm);
        \filldraw[fill=red!20, color = red!20] (1.0,0) circle (.55cm); \node at (1.0,0) {$\struc_{i+1}$};
        \filldraw[fill=gray!20] (0,0) circle (.55cm); \node at (0,0) {$\struc_i$};
        \draw [decorate,decoration={brace,amplitude=5pt,mirror,raise=4ex}](-5.5,0) -- (-.7,0) node[midway,yshift=-3em]{$\univ{\vec{\struc}}_{(<i)}$};
        \draw [decorate,decoration={brace,amplitude=5pt,mirror,raise=4ex}](+.7,0) -- (5.5,0) node[midway,yshift=-3em]{$\univ{\vec{\struc}}_{(>i)}$};
        \draw [decorate,decoration={brace,amplitude=5pt,mirror,raise=4ex}](-.5,0) -- (+.5,0) node[midway,yshift=-3em]{$\univ{\vec{\struc}}_{(i)}$};
    \end{tikzpicture}
\end{center}

For $\trace[1] \in \lsubSPR^{(\odot \vec{\struc})_{\bullet}}$,
we say that $\trace[1]$ is \intro*\kl{$i$-split} if
there are some $x \in \univ{\vec{\struc}}_{(<i)}$ and $y \in \univ{\vec{\struc}}_{(>i)}$ occurring in the sequence $\ty_1(\trace[1])\ty_2(\trace[1])$.
We first show that
we can decompose each \kl{$i$-split} $\lSPR^{(\odot \vec{\struc})_{\bullet}}$ \kl[$\struc$-run]{run} into
two $\lSPR^{(\odot \vec{\struc})_{\bullet}}$ \kl[$\struc$-runs]{runs} such that each glued \kl{vertex} is $\univ{\vec{\struc}}_{(i)}$-labeled.
\begin{lem}\label{lemma: interpolation SPR}
    Let $\vec{\struc} = \struc_1 \dots \struc_n \in \STR^{+}$ and $i \in \range{n}$.
    For all \kl{$i$-split} $\trace[1] \in \lSPR^{(\odot \vec{\struc})_{\bullet}}$,
    there are non-empty $\trace[2]$ and $\trace[3]$ such that
    $\trace[1] = \trace[2] \series \trace[3]$ and 
    for each $k$,
    we have
    $\trace[2](\tgt^{\trace[2]}_k) \in \univ{\vec{\struc}}_{(i)}$.
\end{lem}
\begin{proof}
    By induction on the derivation tree of $\graph \in \SPR$ where $\trace[1] \colon \graph \homo (\odot \vec{\struc})_{\bullet}$.
    \begin{itemize}
        \item Case $\trace[1] = a^{x,y}_{1}$, Case $\trace[1] = \id^{x}$:
        Then $\trace[1]$ is \kl{local}.
        However, as this contradicts that $\trace[1]$ is \kl{$i$-split}, this case does not occur.
        \item Case $\trace[1] = 
        \left(\begin{tikzpicture}[baseline = -.5ex]
            \graph[grow right = .6cm, branch down = 1.5ex, nodes={}]{
            {, L'12/[mynode,draw,circle]} -!-
            {LC'12/} -!-
            {C'12/{$z$}[mynode,draw,circle]} -!-
            {CR'12/} -!-
            {R'12/[mynode,draw,circle]}};
            \node[left = 4pt of L'12](s1){}; \path (s1) edge[earrow, ->] (L'12);
            \node[right = 4pt of R'12](t1){}; \path (R'12) edge[earrow, ->] (t1);
            \graph[use existing nodes, edges={color=black, pos = .5, earrow, line width = .5pt, }, edge quotes={fill=white, inner sep=1pt,font= \scriptsize}]{
            {L'12} --[color = black] {LC'12};
            {LC'12} ->[color = black] {C'12} --[color = black] {CR'12};
            {CR'12} ->[color = black] {R'12};
            {L'12} ->["$\trace[1]'$"] {C'12};
            {C'12} ->["$\trace[1]''$"] {R'12};
            };
        \end{tikzpicture}\right)$:
        We distinguish the following cases:
        \begin{itemize}
            \item Case $z \in \univ{\vec{\struc}}_{(i)}$:
            By letting $\trace[2] = \trace[1]'$ and $\trace[3] = \trace[1]''$, this case has been proved.
            \item Otherwise:
            Because $z$ is $\univ{\vec{\struc}}_{(<i)}$- or $\univ{\vec{\struc}}_{(>i)}$-labeled \kl{vertex},
            either $\trace[1]'$ or $\trace[1]''$ is \kl{$i$-split}.
            \begin{itemize}
                \item Case $\trace[1]'$ is \kl{$i$-split}:
                Let $\trace[2]'$ and $\trace[3]'$ be the ones obtained by IH with respect to $\trace[1]'$.
                Then by letting $\trace[2] = \trace[2]'$ and $\trace[3] = \trace[3]' \series \trace[1]''$, this case has been proved.
                \item Case $\trace[1]''$ is \kl{$i$-split}:
               Similarly, let $\trace[2]''$ and $\trace[3]''$ be the ones obtained by IH with respect to $\trace[1]''$.
               Then by letting $\trace[2] = \trace[1]' \series \trace[2]''$ and $\trace[3] = \trace[3]''$, this case has been proved.
           \end{itemize}
        \end{itemize}
        
        \item Case $\trace[1] = \fork^{\lab}_{1} \series (\trace[1]' \parallel \trace[1]'') \series \join^{\lab[2]}_{1}$:
        Since both $\trace[1]'$ and $\trace[1]''$ are \kl{$i$-split},
        let $\trace[2]'$ and $\trace[3]'$ be the ones obtained by IH with respect to $\trace[1]'$
        and let $\trace[2]''$ and $\trace[3]''$ be the ones obtained by IH with respect to $\trace[1]''$.
        Then by letting $\trace[2] = \fork^{\lab}_{1} \series (\trace[2]' \parallel \trace[2]'')$ and $\trace[3] = (\trace[3]' \parallel \trace[3]'') \series \join^{\lab[2]}_{1}$, this case has been proved.
    \end{itemize}
    Hence, this completes the proof.
\end{proof}
We then extend the lemma above for (slightly specialized) $\lsubSPR$.
This is the most crucial lemma for \Cref{lemma: decomposition lsubSPR}.
\begin{lem}[Interpolation lemma]\label{lemma: interpolation lsubSPR 2}
    Let $\vec{\struc} = \struc_1 \dots \struc_n \in \STR^{+}$ and let $i \in \range{n}$.
    For all $\trace[1] \in \lsubSPR^{\vec{\struc}}_{(i)}$ of the form
    $\left(\begin{tikzpicture}[baseline = -4.5ex]
        \graph[grow right = 1.5cm, branch down = 3.5ex, nodes={}]{
        {L1/{$x_1$}[mynode,draw,circle, inner sep = .5pt], L12/{\tiny $\vdots$}[yshift = .5ex], L2/{$x_n$}[mynode,draw,circle, inner sep = .5pt]} -!-
        {L'1/{$x_1'$}[mynode,draw,circle, inner sep = .5pt], L'12/{\tiny $\vdots$}[yshift = .5ex], L'2/{$x_n'$}[mynode,draw,circle, inner sep = .5pt]} -!-
        {/,/,/} -!-
        {R'1/{$y_1'$}[mynode,draw,circle, inner sep = .5pt], R'12/{\tiny $\vdots$}[yshift = .5ex], R'2/{$y_m'$}[mynode,draw,circle, inner sep = .5pt]} -!-
        {R1/{$y_1$}[mynode,draw,circle, inner sep = .5pt], R12/{\tiny $\vdots$}[yshift = .5ex], R2/{$y_m$}[mynode,draw,circle, inner sep = .5pt]} -!-
        };
        \node[left = 4pt of L1](s1){\tiny $1$}; \path (s1) edge[earrow, ->] (L1);
        \node[left = 4pt of L2](s2){\tiny $n$}; \path (s2) edge[earrow, ->] (L2);
        \node[right = 4pt of R1](t1){\tiny $1$}; \path (R1) edge[earrow, ->] (t1);
        \node[right = 4pt of R2](t2){\tiny $m$}; \path (R2) edge[earrow, ->] (t2);
        \graph[use existing nodes, edges={color=black, pos = .5, earrow, line width = .5pt, }, edge quotes={fill=white, inner sep=1pt,font= \scriptsize}]{
        {L1} ->["$a_1$", color = black] {L'1} -> {R'1} ->["$b_1$", color = black] {R1};
        {L2} ->["$a_n$", color = black] {L'2} -> {R'2} ->["$b_m$", color = black] {R2};
        };
        \draw[fill = white] ($(L'1)+(.5,.15)$) -- ($(L'2)+(.5,-.15)$) -- ($(R'2)+(-.5,-.15)$) -- ($(R'1)+(-.5,.15)$) -- cycle;
        \path (L'12) edge[pos =.5, opacity = 0] node[opacity = 1]{\scriptsize  $\trace_0$} (R'12);
    \end{tikzpicture}\right)$
    such that $\trace_0$ is \kl{$i$-split} and \emph{\kl{connected}} (i.e., $\trace_0 \in \lsubSPRl^{(\odot\vec{\struc})_{\bullet}} \series \lsubSPRr^{(\odot\vec{\struc})_{\bullet}}$ by \Cref{proposition: subSPR equiv}),
    there are non-empty $\trace[2]$ and $\trace[3]$ such that
    $\trace[1] = \trace[2] \series \trace[3]$ and
    $\trace[2],\trace[3] \in \lsubSPR^{\vec{\struc}}_{(i)}$.
\end{lem}
\begin{proof}
    Let $\trace_0 = (\trace[2]_0 \series \trace[3]_0)$
    where $\trace[2]_0 \in \lsubSPRl^{(\odot\vec{\struc})_{\bullet}}$ and $\trace[3]_0 \in \lsubSPRr^{(\odot\vec{\struc})_{\bullet}}$.
    Note that $\trace[2]_0$ and $\trace[3]_0$ can be viewed as \kl{trees} of $\lSPR^{(\odot\vec{\struc})_{\bullet}}$ runs.
    Let $v_k$ be such \kl{vertices} in the \kl{tree} and let $\lab[3]_k$ be the \kl{label} of $v_k$.
    Let $\textcolor{red!60}{\lab[3]_{<}}$ and $\textcolor{blue}{\lab[3]_{>}}$ be 
    an $\univ{\vec{\struc}}_{(<i)}$-labeled \kl{vertex} and an $\univ{\vec{\struc}}_{(>i)}$-labeled \kl{vertex} 
    in \kl{source} or \kl{target} \kl{vertices} of $\trace[1]$ (as $\trace[1]$ is \kl{$i$-split}).
    For instance, we consider the following $\trace[2]_0$ and $\trace[3]_0$ (where $\textcolor{red!60}{\lab[3]_{<}} = \textcolor{red!60}{\lab[3]_{2}}$ and $\textcolor{blue}{\lab[3]_{>}} = \textcolor{blue}{\lab[3]_{15}}$):
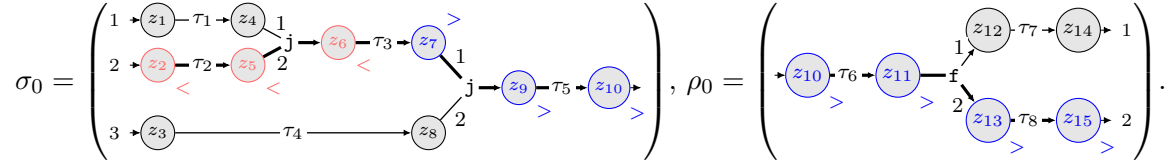
    \begin{center}
        $\trace[2]_0 = \left(\hspace{-.4em}\begin{tikzpicture}[baseline = -5.5ex, yscale=-1]
            \foreach \x/\y/\col/\nlab in {
                0/0/black/$\lab[3]_1$,0/2/red!60/$\lab[3]_2$,0/5/black/$\lab[3]_3$,
                1/0/black/$\lab[3]_4$,1/2/red!60/$\lab[3]_5$,
                2/1/red!60/$\lab[3]_6$,
                3/1/blue/$\lab[3]_7$,3/5/black/$\lab[3]_8$,
                4/3/blue/$\lab[3]_9$,
                5/3/blue/$\lab[3]_{10}$}{
               \node[mynode,draw = \col,text = \col, circle]  (\x X\y) at (1.2*\x,.3*\y) {\tiny \nlab};
            }
            \foreach \x/\y/\label in {
                1/1,
                3/3}{
               \node[inner sep = 0, outer sep = 0, minimum width = 0] (\x X\y) at (1.2*\x,.3*\y) {};
            }
            \foreach \x/\y/\col/\nlab in {
                0/2/red!60/$\lab[3]_2$,1/2/red!60/$\lab[3]_5$,2/1/red!60/$\lab[3]_6$}{
               \node[below right = -0.6ex of \x X\y, text = \col](\x X\y p) {\tiny $<$};
            }
            \foreach \x/\y/\col/\nlab in {
                4/3/blue/$\lab[3]_9$,5/3/blue/$\lab[3]_{10}$}{
               \node[below right = -0.6ex of \x X\y,text = \col](\x X\y p) {\tiny $>$};
            }
            \foreach \x/\y/\col/\nlab in {
                3/1/blue/$\lab[3]_7$}{
               \node[above right = -0.6ex of \x X\y,text = \col](\x X\y p) {\tiny $>$};
            }
            \node[left = 4pt of 0X0](s1){\tiny $1$}; \path (s1) edge[earrow, ->] (0X0);
            \node[left = 4pt of 0X2](s2){\tiny $2$}; \path (s2) edge[earrow, ->] (0X2);
            \node[left = 4pt of 0X5](s3){\tiny $3$}; \path (s3) edge[earrow, ->] (0X5);
            \node[right = 4pt of 5X3](t1){}; \path (5X3) edge[earrow, ->] (t1);
            \foreach \xp/\x/\y in {1/2/1,3/4/3}{
                \node (\xp X\y) at (1.2*\xp,.3*\y) {};
                \path (\xp X\y) edge [opacity = 0] node[pos= .5, elabel](\x X\y Xj){$\join$}(\x X\y);
            }
            \graph[use existing nodes, edges={color=black, pos = .5, earrow, line width = .5pt, 
            },edge quotes={fill=white, inner sep=1pt,font= \scriptsize}]{
            {0X0} ->["$\trace_1$"] {1X0};
            {0X2} ->["$\trace_2$", line width = 1.2pt] {1X2};
            {1X0} --["$1$"{auto}]{2X1Xj};
            {1X2} --["$2$"{auto, below right= .3pt}, line width = 1.2pt] {2X1Xj};
            {2X1Xj} ->[line width = 1.2pt] {2X1};
            {2X1} ->["$\trace_3$", line width = 1.2pt] {3X1};
            {0X5} ->["$\trace_4$"] {3X5};
            {3X1} --["$1$"{auto}, line width = 1.2pt] {4X3Xj};
            {3X5} --["$2$"{auto, below right= .3pt}] {4X3Xj} ->[line width = 1.2pt]{4X3};
            {4X3} ->["$\trace_5$", line width = 1.2pt] {5X3};
            };
        \end{tikzpicture}\hspace{-.7em}\right)$,
        $\trace[3]_0 = \left(\hspace{-.9em}\begin{tikzpicture}[baseline = -6.5ex, yscale=-1]
            \foreach \x/\y/\col/\nlab in {
                5/3/blue/$\lab[3]_{10}$,
                6/3/blue/$\lab[3]_{11}$,
                7/1/black/$\lab[3]_{12}$,7/5/blue/$\lab[3]_{13}$,
                8/1/black/$\lab[3]_{14}$,8/5/blue/$\lab[3]_{15}$}{
               \node[mynode,draw = \col,text = \col, circle]  (\x X\y) at (1.2*\x,.3*\y) {\scriptsize \nlab};
            }
            \foreach \x/\y/\label in {
                7/3}{
               \node[inner sep = 0, outer sep = 0, minimum width = 0] (\x X\y) at (1.2*\x,.3*\y) {};
            }
            \foreach \x/\y/\col/\nlab in {
                5/3/blue/$\lab[3]_{10}$, 6/3/blue/$\lab[3]_{11}$,
                7/5/blue/$\lab[3]_{13}$, 8/5/blue/$\lab[3]_{15}$
            }{
               \node[below right = -0.6ex of \x X\y,text = \col](\x X\y p) {\tiny $>$};
            }
            \node[left = 4pt of 5X3](s1){}; \path (s1) edge[earrow, ->] (5X3);
            \node[right = 4pt of 8X1](t1){\tiny $1$}; \path (8X1) edge[earrow, ->] (t1);
            \node[right = 4pt of 8X5](t2){\tiny $2$}; \path (8X5) edge[earrow, ->] (t2);
            \foreach \xp/\x/\y in {6/7/3}{
                \node  (\x X\y) at (1.22*\x,.3*\y) {};
                \path (\xp X\y) edge [opacity = 0] node[pos= .5, elabel](\xp X\y Xf){$\fork$}(\x X\y);
            }
            \graph[use existing nodes, edges={color=black, pos = .5, earrow, line width = .5pt, 
            },edge quotes={fill=white, inner sep=1pt,font= \scriptsize}]{
            {5X3} ->["$\trace_{6}$", pos = .4, line width = 1.2pt] {6X3};
            {6X3} --[line width = 1.2pt] {6X3Xf} ->["$1$"{auto}]{7X1};
            {6X3Xf} ->["$2$"{auto, below left= .3pt}, line width = 1.2pt] {7X5};
            {7X1} ->["$\trace_{7}$", pos = .4, ] {8X1};
            {7X5} ->["$\trace_{8}$", pos = .4, line width = 1.2pt] {8X5};
            };
        \end{tikzpicture}\hspace{-.4em}\right)$.
    \end{center}
    We distinguish the following cases:
    \begin{itemize}
        \item Case $\lab[3]_k \in \univ{\vec{\struc}}_{(i)}$ for some $v_k$ on $\trace[2]_0$:
        By letting $\trace[2]$ and $\trace[3]$ be such that $\trace[1] = \trace[2] \series \trace[3]$ and $\trace[2]$ consists of the edges on the left of $\lab[3]_k$, this case has been proved.
        For instance, if $\lab[3]_6 \in \univ{\vec{\struc}}_{(i)}$ in the instance above, then $\trace[2] = \left(\hspace{-.3em}\begin{tikzpicture}[baseline = -4.4ex, yscale=-1]
            \foreach \x/\y/\col/\nlab in {
                0/0/black/$\lab[3]_1$,0/2/black/$\lab[3]_2$, 
                1/0/black/$\lab[3]_4$,1/2/black/$\lab[3]_5$,
                2/1/black/$\lab[3]_6$}{
               \node[mynode,draw = \col,text = \col, circle]  (\x X\y) at (1.2*\x,.3*\y) {\tiny \nlab};
            }
            \foreach \x/\y/\label in {
                1/1}{
               \node[inner sep = 0, outer sep = 0, minimum width = 0] (\x X\y) at (1.2*\x,.3*\y) {};
            }
            \node[mynode,draw = black,text = black, circle] (a1) at (-1.5,0) {\tiny $x_1$};
            \node[mynode,draw = black,text = black, circle] (a2) at (-1.5,0.6) {\tiny $x_2$};
            \node[mynode,draw = black,text = black, circle] (a3) at (-1.5,1.2) {\tiny $x_3$};
            \node[left = 4pt of a1](s1){\tiny $1$}; \path (s1) edge[earrow, ->] (a1);
            \node[left = 4pt of a2](s2){\tiny $2$}; \path (s2) edge[earrow, ->] (a2);
            \node[left = 4pt of a3](s3){\tiny $3$}; \path (s3) edge[earrow, ->] (a3);
            \node[right = 4pt of 2X1](t1){\tiny $1$}; \path (2X1) edge[earrow, ->] (t1);
            \node[right = 4pt of a3](t2){\tiny $2$}; \path (a3) edge[earrow, ->] (t2);
            \foreach \xp/\x/\y in {1/2/1}{
                \node (\xp X\y) at (1.2*\xp,.3*\y) {};
                \path (\xp X\y) edge [opacity = 0] node[pos= .5, elabel](\x X\y Xj){$\join$}(\x X\y);
            }
            \graph[use existing nodes, edges={color=black, pos = .5, earrow, line width = .5pt, 
            },edge quotes={fill=white, inner sep=1pt,font= \scriptsize}]{
            {a1} ->["$a_1$"] {0X0};
            {a2} ->["$a_2$"] {0X2};
            {0X0} ->["$\trace_1$"] {1X0};
            {0X2} ->["$\trace_2$"] {1X2};
            {1X0} --["$1$"{auto}] {2X1Xj};
            {1X2} --["$2$"{auto, below right= .3pt}] {2X1Xj} -> {2X1};
            };
        \end{tikzpicture}\hspace{-.5em}\right)$.
        \item Case $\lab[3]_k \in \univ{\vec{\struc}}_{(i)}$ for some $v_k$ on $\trace[3]_0$:
        In the same way as above.
        \item Otherwise, every $\lab[3]_k$ is either $\univ{\vec{\struc}}_{(<i)}$- or $\univ{\vec{\struc}}_{(>i)}$-labeled.
        We consider a path between $\textcolor{red!60}{\lab[3]_{<}}$ and $\textcolor{blue}{\lab[3]_{>}}$ (the bold edges in the above).
        Then on the path, there is some $k$ such that 
        there are an $\univ{\vec{\struc}}_{(<i)}$-labeled \kl{vertex}
        and an $\univ{\vec{\struc}}_{(>i)}$-labeled \kl{vertex} in \kl{source} or \kl{target} \kl{vertices} of $\trace_k$
        (\eg, in the instance above, $\trace_{3}$ is an expected $\lSPR^{(\odot\vec{\struc})_{\bullet}}$ run where each $\textcolor{red!60}{<}$ and $\textcolor{blue}{>}$ indicates $\univ{\vec{\struc}}_{(<i)}$-labeled and $\univ{\vec{\struc}}_{(>i)}$-labeled, respectively).
        We distinguish the following cases:
        \begin{itemize}
            \item Case $\trace_k$ is on $\trace[2]_0$:
            Let $\trace[2]_k$ and $\trace[3]_k$ be the ones obtained from $\trace[1]_k$ by \Cref{lemma: interpolation SPR}.
            Let $\trace[2]',\trace[3]', \vec{z}_1, \vec{z}_2$ be such that
            $\trace[2]_0 = \trace[2]' \series (\id^{\vec{z}_1} \parallel \trace_k \parallel \id^{\vec{z}_2}) \series \trace[3]'$ and
            $\trace[2]'$ consists of the edges on the left of the \kl{source} of $\trace_k$.
            Then by letting $\trace[2] = \trace[2]' \series (\id^{\vec{z}_1} \parallel \trace[2]_k \parallel \id^{\vec{z}_2})$ and $\trace[3] = (\id^{\vec{z}_1} \parallel \trace[3]_k \parallel \id^{\vec{z}_2}) \series \trace[3]' \series \trace[3]_0$, this case has been proved.
            \item Otherwise (Case $\trace_k$ is on $\trace[3]_0$):
            In the same way as above.
        \end{itemize}
    \end{itemize}
    Hence, this completes the proof.
\end{proof}
By \Cref{lemma: interpolation lsubSPR 2}, we can show the decomposition theorem for \kl{local} \kl{sub-SP runs}, as follows.
\begin{lem}[Decomposition theorem for local sub-SP runs]\label{lemma: decomposition lsubSPR}
    Let $\vec{\struc} = \struc_1 \dots \struc_n \in \STR^{+}$.
    \[\bigcup_{i = 1}^{n} \lsubSPR^{\vec{\struc}}_{(i)}  \;=\; \odot_{i = 1}^{n} \lsubSPR^{\struc_i}.\]
\end{lem}
\begin{proof}
    ($\supseteq$):
    Trivial by the definition of them.
    ($\subseteq$):
    We show the following by induction on the number of edges of $\trace$:
    for all $i \in \range{n}$,
    if $\trace \in \lsubSPR^{\vec{\struc}}_{(i)}$, then $\trace \in \odot_{i = 1}^{n} \lsubSPR^{\struc_i}$.
    We distinguish the following cases.
    \begin{itemize}
        \item Case $\trace$ has no edges:
        Then $\tau = \id^{\vec{z}}$.
        By the rule (R), this case has been proved.

        \item Case there is an $a$-labeled edge of $\struc_i$ adjacent to a \kl{source} \kl{vertex}:
        \begin{itemize}
            \item Case $a \in \vsig$:
            Then $\trace = a^{\vec{z},z'}_{k} \series \trace'$ where $\vec{z}z' \in \univ{\vec{\struc}}_{(i)}^{+}$.
            We then have:
            \begin{center}
            \begin{prooftree}
            \hypo{\mathstrut}
            \infer1[D]{a^{\vec{z},z'}_{k}}
            \hypo{\mathstrut}
            \infer[double]1[(IH)]{\trace'}
            \infer2[T]{a^{\vec{z}, z'}_{k} \series \trace'}
            \end{prooftree}.            
            \end{center}
            
            \item Case $a = \fork$:
            In the same way as above.

            \item Case $a = \join$:
            Then $\trace = ((\id^{\vec{z}_1} \parallel (\id^{z} \parallel \rho) \parallel \id^{\vec{z}_2}) \series \join^{\vec{z}_1 z\vec{z}_2}_{k}) \series \trace'$
            or $\trace = ((\id^{\vec{z}_1} \parallel (\rho \parallel \id^{z}) \parallel \id^{\vec{z}_2}) \series \join^{\vec{z}_1 z\vec{z}_2}_{k}) \series \trace'$
            where 
            $\rho \in \lsubSPRl^{\vec{\struc}}_{(i)}$, $\vec{z}_1 z\vec{z}_2 \in \univ{\vec{\struc}}_{(i)}^{+}$, and $k = \|\vec{z}_1\| + 1$.
            Then we have:
            \begin{center}
                \begin{prooftree}
                    \hypo{\mathstrut}
                    \infer[double]1[(IH)]{\id^{\vec{z}_1} \parallel \rho' \parallel \id^{\vec{z}_2}}
                    \hypo{\mathstrut}
                \infer1[D]{\join^{\vec{z}_1 z\vec{z}_2}_{k}}
                \hypo{\mathstrut}
                \infer[double]1[(IH)]{\trace'}
                \infer[double]3[T]{((\id^{\vec{z}_1} \parallel \rho' \parallel \id^{\vec{z}_2}) \series \join^{\vec{z}_1 z\vec{z}_2}_{k}) \series \trace'}
                \end{prooftree}
                where $\rho' = \id^{z} \parallel \rho$ or $\rho' = \rho \parallel \id^{z}$.
            \end{center}
        \end{itemize}

        \item Case there is an edge of $\struc_i$ adjacent to a \kl{target} \kl{vertex}:
        Similar to the case above.

        \item Case $\trace$ has two or more weakly connected components, that is, $\trace = \trace_1 \parallel \dots \parallel \trace[1]_p$,
        where
        $p \ge 1$,
        $\trace[1]_1, \dots, \trace[1]_p \in \lsubSPRl^{(\odot\vec{\struc})_{\bullet}} \series \lsubSPRr^{(\odot\vec{\struc})_{\bullet}}$, and
        $\trace_i$ is non-empty for two or more $i \in \range{p}$:
        Then we can easily take some non-empty $\trace'$ and $\trace''$ such that $\trace = \trace' \parallel \trace''$ holds.
        By letting $\trace[2] = \trace' \parallel \id^{\ty_1(\trace'')}$ and $\trace[3] = \id^{\ty_2(\trace')} \parallel \trace''$,
        we have:
        \begin{center}
            \begin{prooftree}
            \hypo{\mathstrut}
            \infer[double]1[(IH)]{\trace[2]}
            \hypo{\mathstrut}
            \infer[double]1[(IH)]{\trace[3]}
            \infer2[T]{\trace[2] \series \trace[3]}
            \end{prooftree}. 
        \end{center}

        \item 
        Otherwise, after applying the rule (L) for each \kl{vertex} on both \kl{source} and \kl{target} \kl{vertices},
        $\trace = \id^{\bullet \dots \bullet} \parallel (\trace' \series \trace'' \series \trace''') \parallel \id^{\bullet \dots \bullet}$ holds (by the fourth case) where 
        \[\trace' \!=\! \left(\!\hspace{-.5em}\begin{tikzpicture}[baseline = -4.5ex]
            \graph[grow right = 1.5cm, branch down = 3.5ex, nodes={}]{
            {L1/{$x_1'$}[mynode,draw,circle], L12/{\tiny $\vdots$}[yshift=.5ex], L2/{$x_n'$}[mynode,draw,circle]} -!-
            {L'1/{$x_1$}[mynode,draw,circle, inner sep =0.5pt], L'12/{\tiny $\vdots$}[yshift=.5ex], L'2/{$x_n$}[mynode,draw,circle, inner sep =0.5pt]} -!-
            };
            \node[left = 4pt of L1](s1){}; \path (s1) edge[earrow, ->] (L1);
            \node[left = 4pt of L2](s2){}; \path (s2) edge[earrow, ->] (L2);
            \node[right = 4pt of L'1](t1){}; \path (L'1) edge[earrow, ->] (t1);
            \node[right = 4pt of L'2](t2){}; \path (L'2) edge[earrow, ->] (t2);
            \graph[use existing nodes, edges={color=black, pos = .5, earrow, line width = .5pt, }, edge quotes={fill=white, inner sep=1pt,font= \scriptsize}]{
            {L1} ->["$a_1$", color = black] {L'1};
            {L2} ->["$a_n$", color = black] {L'2};
            };
        \end{tikzpicture}\!\right)\!,
        \trace'' \!=\! \left(\!\begin{tikzpicture}[baseline = -4.5ex]
            \graph[grow right = 1.5cm, branch down = 3.5ex, nodes={}]{
            {L'1/{$x_1$}[mynode,draw,circle, inner sep =0.5pt], L'12/{\tiny $\vdots$}[yshift=.5ex], L'2/{$x_n$}[mynode,draw,circle, inner sep =0.5pt]} -!-
            {/,/,/} -!-
            {R'1/{$y_1$}[mynode,draw,circle, inner sep =0.5pt], R'12/{\tiny $\vdots$}[yshift=.5ex], R'2/{$y_m$}[mynode,draw,circle, inner sep =0.5pt]} -!-
            };
            \node[left = 4pt of L'1](s1){}; \path (s1) edge[earrow, ->] (L'1);
            \node[left = 4pt of L'2](s2){}; \path (s2) edge[earrow, ->] (L'2);
            \node[right = 4pt of R'1](t1){}; \path (R'1) edge[earrow, ->] (t1);
            \node[right = 4pt of R'2](t2){}; \path (R'2) edge[earrow, ->] (t2);
            \graph[use existing nodes, edges={color=black, pos = .5, earrow, line width = .5pt, }, edge quotes={fill=white, inner sep=1pt,font= \scriptsize}]{
            {L'1} -> {R'1};
            {L'2} -> {R'2};
            };
            \draw[fill = white] ($(L'1)+(.5,.15)$) -- ($(L'2)+(.5,-.15)$) -- ($(R'2)+(-.5,-.15)$) -- ($(R'1)+(-.5,.15)$) -- cycle;
            \path (L'12) edge[pos =.5, opacity = 0] node[opacity = 1]{\scriptsize  $\trace''$} (R'12);
        \end{tikzpicture}\!\right)\!,
        \trace''' \!=\! \left(\!\begin{tikzpicture}[baseline = -4.5ex]
            \graph[grow right = 1.5cm, branch down = 3.5ex, nodes={}]{
            {R'1/{$y_1$}[mynode,draw,circle, inner sep =0.5pt], R'12/{\tiny $\vdots$}[yshift=.5ex], R'2/{$y_m$}[mynode,draw,circle, inner sep =0.5pt]} -!-
            {R1/{$y_1'$}[mynode,draw,circle], R12/{\tiny $\vdots$}[yshift=.5ex], R2/{$y_m'$}[mynode,draw,circle]} -!-
            };
            \node[left = 4pt of R'1](s1){}; \path (s1) edge[earrow, ->] (R'1);
            \node[left = 4pt of R'2](s2){}; \path (s2) edge[earrow, ->] (R'2);
            \node[right = 4pt of R1](t1){}; \path (R1) edge[earrow, ->] (t1);
            \node[right = 4pt of R2](t2){}; \path (R2) edge[earrow, ->] (t2);
            \graph[use existing nodes, edges={color=black, pos = .5, earrow, line width = .5pt, }, edge quotes={fill=white, inner sep=1pt,font= \scriptsize}]{
            {R'1} ->["$b_1$", color = black] {R1};
            {R'2} ->["$b_m$", color = black] {R2};
            };
        \end{tikzpicture}\hspace{-.5em}\!\right)\!.\]
        Here, $n, m \ge 1$ (as $\trace[1]$ has some edge by the first case) and 
        $x_1, \dots, x_n, y_1, \dots, y_m \in \univ{\vec{\struc}}_{(<i)} \cup \univ{\vec{\struc}}_{(>i)}$ (by the second and third cases).
        We then distinguish the following cases:
        \begin{itemize}
            \item Case $x_1, \dots, x_n, y_1, \dots, y_m \in \univ{\vec{\struc}}_{(<i)}$:
            Let $j < i$ be the maximum $j$ such that some edge of $a_1, \dots, a_n, b_1, \dots, b_m$ occurs in $\struc_{j}$.
            Let $\trace[2]'$ and $\trace[3]'$ be such that $\trace[1]' = \trace[2]' \series \trace[3]'$
            and $\trace[2]'$ contains edges of $\struc_{j}$.
            Similarly, let $\trace[2]'''$ and $\trace[3]'''$ be such that $\trace[1]''' = \trace[2]''' \series \trace[3]'''$
            and $\trace[3]'''$ contains edges of $\struc_{j}$.
            Then either $\trace[2]'$ or $\trace[3]'''$ is not empty and
            $\trace[2]', \trace[3]' \series \trace[1]'' \series \trace[2]''', \trace[3]''' \in \lsubSPR^{\vec{\struc}}_{(j)}$ ($x_1, \dots, x_n, y_1, \dots, y_m \in \univ{\vec{\struc}}_{(j)}$ holds by the maximality of $j$).
            Then we have:
            \begin{center}
                \begin{prooftree}
                    \hypo{\mathstrut}
                    \infer[double]1[(IH)]{\id^{\bullet \dots \bullet} \parallel \trace[2]' \parallel \id^{\bullet \dots \bullet}}
                \hypo{\mathstrut}
                \infer[double]1[(IH)]{\id^{\bullet \dots \bullet} \parallel (\trace[3]' \series \trace[1]'' \series \trace[2]''') \parallel \id^{\bullet \dots \bullet}}
                \hypo{\mathstrut}
                \infer[double]1[(IH)]{\id^{\bullet \dots \bullet} \parallel \trace[3]''' \parallel \id^{\bullet \dots \bullet}}
                \infer[double]3[T]{\id^{\bullet \dots \bullet} \parallel (\trace[2]' \series (\trace[3]' \series \trace[1]'' \series \trace[2]''') \series \trace[3]''') \parallel \id^{\bullet \dots \bullet}}
                \end{prooftree}.
            \end{center}
            \item Case $x_1, \dots, x_n, y_1, \dots, y_m \in \univ{\vec{\struc}}_{(> i)}$:
            In the same way as above.
            
            \item Otherwise,
            by the interpolation lemma (\Cref{lemma: interpolation lsubSPR 2}),
            there are non-empty $\trace[2]$ and $\trace[3]$ such that
            $\trace[1] = \id^{\bullet \dots \bullet} \parallel (\trace[2] \series \trace[3]) \parallel \id^{\bullet \dots \bullet}$ and
            $\trace[2],\trace[3] \in \lsubSPR^{\vec{\struc}}_{(i)}$.
            Thus we have:
            \begin{center}
                \begin{prooftree}
                \hypo{\mathstrut}
                \infer[double]1[(IH)]{\id^{\bullet \dots \bullet} \parallel \trace[2] \parallel \id^{\bullet \dots \bullet}}
                \hypo{\mathstrut}
                \infer[double]1[(IH)]{\id^{\bullet \dots \bullet} \parallel \trace[3] \parallel \id^{\bullet \dots \bullet}}
                \infer2[T]{\id^{\bullet \dots \bullet} \parallel (\trace[2] \series \trace[3]) \parallel \id^{\bullet \dots \bullet}}
                \end{prooftree}.
            \end{center}
        \end{itemize}
    \end{itemize}
    Hence, this completes the proof.
\end{proof}

\begin{rem}[On non-sub-SP runs]\label{remark: non-sub-SP runs}
    If a \kl{local} \kl[$\struc$-run]{$\odot \vec{\struc}$-run} $\trace$ is not \kl{sub-SP},
    we may not be able to decompose $\trace$ (and hence, \Cref{lemma: decomposition lsubSPR} fails for non-sub-SP \kl{runs}).
    For instance, we cannot decompose for the following $\trace$:
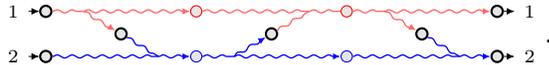
     \begin{tikzpicture}[baseline = -2.5ex, yscale=-1]
            \foreach \x/\y/\col/\nlab/\style/\m in {
                0/0/black/{}/{thick}/, 0/2/black/{}/{thick}/,
                2/1/black/{}/{thick}/, 
                4/0/red/{}, 4/2/blue/{}/,
                6/1/black/{}/{thick}/, 
                8/0/red/{}, 8/2/blue/{}/,
                10/1/black/{}/{thick}/, 
                12/0/black/{}/{thick}/, 12/2/black/{}/{thick}/
                }{
               \node[mynode,draw = \col,text = \col, circle, \style]  (\x X\y) at (0.5*\x,.3*\y) {\tiny \nlab};
            }
            \foreach \x/\y/\label in {
                1/0/,
                3/2/,
                5/2/,
                7/0/,
                9/0/,
                11/2/
                }{
               \node[inner sep = 0, outer sep = 0, minimum width = 0] (\x X\y) at (0.5*\x,.3*\y) {};
            }
            \node[left = 4pt of 0X0](s1){\tiny $1$}; \path (s1) edge[earrow, ->] (0X0);
            \node[left = 4pt of 0X2](s2){\tiny $2$}; \path (s2) edge[earrow, ->] (0X2);
            \node[right = 4pt of 12X0](t1){\tiny $1$}; \path (12X0) edge[earrow, ->] (t1);
            \node[right = 4pt of 12X2](t2){\tiny $2$}; \path (12X2) edge[earrow, ->] (t2);
            \graph[use existing nodes, edges={color=black, pos = .5, earrow, line width = .5pt, 
            decorate, decoration={snake, segment length=2mm, amplitude=.2mm}
            },edge quotes={fill=white, inner sep=1pt,font= \scriptsize}]{
            {0X0} --[color = red!60] {1X0} ->[color=red!60] {4X0,2X1};
            {4X0,6X1} --[color = red!60] {7X0} --[color = red!60] {8X0} --[color=red!60] {9X0} ->[color=red!60] {10X1, 12X0};
            {0X2,2X1} --[color = blue] {3X2} ->[color=blue] {4X2};
            {4X2} --[color = blue] {5X2} ->[color= blue] {6X1, 8X2}; 
            {10X1,8X2} --[color = blue] {11X2} ->[color= blue] {12X2};
            };
        \end{tikzpicture}.
\end{rem}

\subsection{Properties of derivatives}\label{section: derivative properties}
In this subsection, we get back to \kl{derivatives}, and we show some properties of \kl{derivatives}.
To prove \Cref{theorem: decomposition}, we will use them for each rule: \Cref{lemma: i} for (D), \Cref{lemma: L} for (L), and \Cref{lemma: T} for (T), respectively.

\begin{lem}\label{lemma: i}
    Let $\vec{\struc} = \struc_1 \dots \struc_n \in \STR^{+}$ and $j \in \range{n}$.
    For all \kl{lKL terms} $\lterm[1],\lterm[2]$ and  \kl[$\struc$-runs]{$(\struc_j)_{\bullet}$-runs} $\trace[1]$,
    if $\lterm[1]_{(j)} \longrightarrow^{(\odot \vec{\struc})_{\bullet}}_{\trace[1]_{(j)}} \lterm[2]_{(j)}$,
    then $\lterm[1] \longrightarrow^{(\struc[1]_{j})_{\bullet}}_{\trace[1]} \lterm[2]$.
\end{lem}
\begin{proof}
    By easy induction on the derivation tree with respect to \Cref{definition: derivative}.
    Note that $\tuple{\lterm[1]_{(j)}, \trace[1]_{(j)}, \lterm[2]_{(j)}}$ and $\tuple{\lterm[1], \trace[1], \lterm[2]}$ are the same up to changing the names of \kl{vertices}' \kl{labels}.
\end{proof}
\begin{lem}\label{lemma: L}
    Let $\struc$ be a \kl{structure}.
    For all \kl{lKL terms} $\lterm[1]$ and $\lterm[2]$, \kl[$\struc$-runs]{$\struc_{\bullet}$-runs} $\trace[1]$, $l$ and $r$ such that $\src^{\trace[1]}_{l} = \tgt^{\trace[1]}_{r}$,
    and $\lab[1], \lab[2] \in \univ{\struc_{\bullet}}$,
    if $\lterm[1][\lab[1]/l] \longrightarrow^{\struc_{\bullet}}_{\trace[1][\lab[1]/\tuple{l, r}]} \lterm[2][\lab[1]/r]$,
    then $\lterm[1][\lab[2]/l] \longrightarrow^{\struc_{\bullet}}_{\trace[1][\lab[2]/\tuple{l, r}]} \lterm[2][\lab[2]/r]$.
\end{lem}
\begin{proof}
    By easy induction on the derivation tree with respect to \Cref{definition: derivative}.
    Since $\src^{\trace[1]}_{l} = \tgt^{\trace[1]}_{r} = v$ for some $v$, the \kl{label} $\lab[1]$ is not used on the derivation tree of $\longrightarrow^{\struc_{\bullet}}_{\trace[1][\lab[1]/\tuple{l, r}]}$,
    thus we can replace the \kl{label} of $v$ with any \kl{label} (in particular, with $\bullet$).
\end{proof}

\begin{lem}\label{lemma: T}
    Let $\struc$ be a \kl{structure}.
    For all \kl{lKL terms} $\lterm[1]$ and $\lterm[2]$ and \kl{$\struc$-runs} $\trace[1]$ and $\trace[2]$,
    if $\lterm[1] \longrightarrow^{\struc}_{\trace[1] \series \trace[2]} \lterm[2]$,
    then there exists some \kl{lKL term} $\lterm[3]$ such that
    $\lterm[1] \longrightarrow^{\struc}_{\trace[1]} \lterm[3]$
    and 
    $\lterm[3] \longrightarrow^{\struc}_{\trace[2]} \lterm[2]$.
\end{lem}
To prove \Cref{lemma: T}, we use the following alternative definition of $(\longrightarrow^{\struc}_{\trace})$.
Intuitively, this is the ``big-step'' version of $(\longrightarrow^{\struc}_{\trace})$, obtained by eliminating the transitivity rule (T).
\begin{defi}\label{definition: multistep derivative}
    Let $\struc$ be a \kl{structure}.
    The \kl{derivative} relation $\lterm[1] \longleadsto^{\struc}_{\trace[3]} \lterm[2]$,
    where $\lterm[1]$ and $\lterm[2]$ are \kl{lKL terms} and $\trace[3]$ is an \kl{$\struc$-run},
    is defined as the smallest relation closed under the rules:
    \begin{gather*}
        \begin{prooftree}
            \hypo{\tuple{\lab[1], \lab[2]} \in a^{\struc}}
            \infer1{@\lab. a \longleadsto^{\struc}_{a^{x, y}_{1}} @\lab[2].\id}
        \end{prooftree} \mbox{ for $a \in \vsig$}
        \quad
        \begin{prooftree}
            \hypo{@\lab.\term[1] \longleadsto^{\struc}_{\trace[3]} \lterm[1]'}
            \infer1{@\lab.\term[1] \compo \term[2] \longleadsto^{\struc}_{\trace[3]} \lterm[1]' \compo_{1} \term[2]}
        \end{prooftree}  
        \quad
        \begin{prooftree}
            \hypo{@\lab[1].\term[1] \longleadsto^{\struc}_{\trace[3]} \lterm[1]'}
            \hypo{\EPS_{\lab[3]}(\lterm[1]')}
            \hypo{@\lab[3].\term[2] \longleadsto^{\struc}_{\trace[3]'} \lterm[2]'}
            \infer3{@\lab[1].\term[1] \compo \term[2] \longleadsto^{\struc}_{\trace[3] \series \trace[3]'} \lterm[2]'}
        \end{prooftree}
        \\
        \begin{prooftree}
            \hypo{@\lab[1]_0.\term[1] \longleadsto^{\struc}_{\trace[3]_0} \lterm[1]_1'}
            \hypo{\EPS_{\lab[1]_1}(\lterm[1]_1')}
            \hypo{\dots}
            \hypo{@\lab_{n-1}.\term[1] \longleadsto^{\struc}_{\trace[3]_{n-1}} \lterm[1]_n'}
            \hypo{\EPS_{\lab[1]_n}(\lterm[1]_n')}
            \hypo{@\lab[1]_{n}.\term[1] \longleadsto^{\struc}_{\trace[3]_{n}} \lterm[1]_{n+1}'}
            \infer6{ @\lab[1]_0.\term[1]^* \longleadsto^{\struc}_{\trace[3]_0 \series \dots \series \trace[3]_n} \lterm[1]_{n+1}' \compo_{1} \term[1]^*}
        \end{prooftree}
        \\
        \begin{prooftree}
            \hypo{@\lab[1]_0.\term[1] \longleadsto^{\struc}_{\trace[3]_0} \lterm[1]_1'}
            \hypo{\EPS_{\lab[1]_1}(\lterm[1]_1')}
            \hypo{\dots}
            \hypo{@\lab_{n-1}.\term[1] \longleadsto^{\struc}_{\trace[3]_{n-1}} \lterm[1]_n'}
            \hypo{\EPS_{\lab[1]_n}(\lterm[1]_n')}
            \infer5{ @\lab[1]_0.\term[1]^* \longleadsto^{\struc}_{\trace[3]_0 \series \dots \series \trace[3]_{n-1}} @\lab[1]_{n}. \term[1]^*}
        \end{prooftree}
        \\
        \begin{prooftree}
            \hypo{@\lab.\term[1] \longleadsto^{\struc}_{\trace[3]} \lterm[1]'}
            \infer1{@\lab.\term[1] \union \term[2] \longleadsto^{\struc}_{\trace[3]} \lterm[1]'}
        \end{prooftree}
        \quad
        \begin{prooftree}
            \hypo{@\lab.\term[2] \longleadsto^{\struc}_{\trace[3]} \lterm[2]'}
            \infer1{@\lab.\term[1] \union \term[2] \longleadsto^{\struc}_{\trace[3]} \lterm[2]'}
        \end{prooftree}
        \quad                         
        \begin{prooftree}
            \hypo{@\lab.\term[1] \longleadsto^{\struc}_{\trace[3]} \lterm[1]'}
            \hypo{\EPS_{\lab[3]}(\lterm[1]')}
            \hypo{@\lab.\term[2] \longleadsto^{\struc}_{\trace[3]'} \lterm[2]'}
            \hypo{\EPS_{\lab[3]}(\lterm[2]')}
            \infer4{@\lab.\term[1] \intersection \term[2] \longleadsto^{\struc}_{\fork^{x}_{1} \series (\trace[3] \parallel \trace[3]') \series \join^{\lab[3]}_{1}}  @\lab[3].\id}
        \end{prooftree}\\               
        \begin{prooftree}
            \hypo{@\lab.\term[1] \longleadsto^{\struc}_{\trace[3]} \lterm[1]'}
            \hypo{@\lab.\term[2] \longleadsto^{\struc}_{\trace[3]'} \lterm[2]'}
            \infer2{@\lab.\term[1] \intersection \term[2] \longleadsto^{\struc}_{\fork^{x}_{1} \series (\trace[3] \parallel \trace[3]')} \lterm[1]' \intersection_{1} \lterm[2]'}
        \end{prooftree}
        \quad                         
        \begin{prooftree}
            \hypo{\lterm[1] \longleadsto^{\struc}_{\trace[3]} \lterm[1]' \; \EPS_{\lab[3]}(\lterm[1]')}
            \hypo{\lterm[2] \longleadsto^{\struc}_{\trace[3]'} \lterm[2]' \; \EPS_{\lab[3]}(\lterm[2]')}
            \infer2{\lterm[1] \intersection_{1} \lterm[2] \longleadsto^{\struc}_{(\trace[3] \parallel \trace[3]') \series \join^{\lab[3]}_{1}}  @\lab[3].\id}
        \end{prooftree}
        \quad
        \begin{prooftree}
            \hypo{\lterm[1] \longleadsto^{\struc}_{\trace[3]} \lterm[1]'}
            \hypo{\lterm[2] \longleadsto^{\struc}_{\trace[3]'} \lterm[2]'}
            \infer2{\lterm[1] \intersection_{1} \lterm[2] \longleadsto^{\struc}_{(\trace[3] \parallel \trace[3]')} \lterm[1]' \intersection_{1} \lterm[2]'}
        \end{prooftree}
        \\            
        \begin{prooftree}[center=false]
            \hypo{\lterm[1]  \longleadsto^{\struc}_{\trace[3]} \lterm[1]'}
            \infer1{\lterm[1] \compo_{1} \term[2] \longleadsto^{\struc}_{\trace[3]} \lterm[1]' \compo_{1} \term[2]}
        \end{prooftree}
        \quad
        \begin{prooftree}[center=false]
            \hypo{\lterm[1] \longleadsto^{\struc}_{\trace[3]} \lterm[1]'}
            \hypo{\EPS_{\lab[3]}(\lterm[1]')}
            \hypo{@ \lab[3]. \term[2] \longleadsto^{\struc}_{\trace[3]'} \lterm[2]}
            \infer3{\lterm[1] \compo_{1} \term[2] \longleadsto^{\struc}_{\trace[3] \series \trace[3]'} \lterm[2]}
        \end{prooftree}
        \quad
        \begin{prooftree}[center=false]
            \hypo{\mathstrut}
            \infer1[R]{\lterm[1] \longleadsto^{\struc}_{\id^{\overrightarrow{\mathsf{lab}}(\lterm[1])}} \lterm[1]}
        \end{prooftree}
    \end{gather*}
\end{defi}

\begin{prop}[\Cref{section: proposition: transitivity elimination}]\label{proposition: transitivity elimination}
    For all \kl{$\struc$-runs} $\trace$, we have:
    $(\longrightarrow^{\struc}_{\trace}) = (\longleadsto^{\struc}_{\trace})$.
\end{prop}
\begin{proof}[Proof Sketch]
    ($\supseteq$):
    Using the transitivity rule (T), we have that each rule of $\longleadsto^{\struc}_{\trace}$ is admissible in $\longrightarrow^{\struc}_{\trace}$.
    ($\subseteq$):
    We have that the transitivity rule $\begin{prooftree}
        \hypo{\lterm[1] \longleadsto^{\struc}_{\trace[1]} \lterm[3]}
        \hypo{\lterm[3] \longleadsto^{\struc}_{\trace[2]} \lterm[2]}
        \infer2[T]{\lterm[1] \longleadsto^{\struc}_{\trace[1] \series \trace[2]} \lterm[2]}
    \end{prooftree}$ is admissible in $\longleadsto^{\struc}_{\trace}$,
    which is shown by induction on the size of the derivation tree.
\end{proof}
Using \Cref{proposition: transitivity elimination}, we can show \Cref{lemma: T}, as follows.\footnote{Our method can be viewed as an analog of Maehara's method \cite[p.\ 237]{maeharaCraigsInterpolationTheorem1961} in that, after eliminating the rule (T) (which corresponds to the cut rule), we construct an interpolant by induction on derivation trees.}
\begin{proof}[Proof of \Cref{lemma: T}]
    By induction on the derivation tree of $\longleadsto^{\struc}_{\trace[3]}$.

    \begin{itemize}
        \item Case $\begin{prooftree}
            \hypo{\tuple{\lab[1], \lab[2]} \in a^{\struc}}
            \infer1{@\lab. a \longleadsto^{\struc}_{\trace[3]} @\lab[2].\id}
        \end{prooftree}$:
        Then $\trace[3] = \trace[1] \series \trace[2] = a^{\lab[1], \lab[2]}_{1}$.
        We distinguish the following sub-cases:
        \begin{itemize}
            \item Case $\trace[1] = a^{\lab[1], \lab[2]}_{1}$ and $\trace[2] = \id^{\lab[2]}$:
            By letting $\lterm[3] = @\lab[2].\id$.
            \item Case $\trace[1] = \id^{\lab[1]}$ and $\trace[2] = a^{\lab[1], \lab[2]}_{1}$:
            By letting $\lterm[3] = @ \lab[1]. a$.
        \end{itemize}
        
        \item Case $\begin{prooftree}
            \hypo{@\lab.\term[1] \longleadsto^{\struc}_{\trace[3]} \lterm[1]'}
            \infer1{@\lab.\term[1] \union \term[2] \longleadsto^{\struc}_{\trace[3]} \lterm[1]'}
            \end{prooftree}$:
            By IH, there is some $\lterm[1]''$ such that  $@\lab.\term[1] \longleadsto^{\struc}_{\trace[1]} \lterm[1]''$ and $\lterm[1]'' \longleadsto^{\struc}_{\trace[2]} \lterm[1]'$.
            We then have \begin{prooftree}
                    \hypo{@\lab.\term[1] \longleadsto^{\struc}_{\trace[1]} \lterm[1]''}
                    \infer1{@\lab.\term[1] \union \term[2] \longleadsto^{\struc}_{\trace[1]} \lterm[1]''}
                    \end{prooftree}.
            Hence by letting $\lterm[3] = \lterm[1]''$, this case has been proved.
        \item Case $\begin{prooftree}
            \hypo{@\lab.\term[2] \longleadsto^{\struc}_{\trace[3]} \lterm[2]'}
            \infer1{@\lab.\term[1] \union \term[2] \longleadsto^{\struc}_{\trace[3]} \lterm[2]'}
        \end{prooftree}$: 
        Similar to the case above.
        
        \item Case $\begin{prooftree}
            \hypo{@\lab.\term[1] \longleadsto^{\struc}_{\trace[3]} \lterm[1]'}
            \infer1{@\lab.\term[1] \compo \term[2] \longleadsto^{\struc}_{\trace[3]} \lterm[1]' \compo_{1} \term[2]}
        \end{prooftree}$:
        By IH, there is some $\lterm[1]''$ such that $@\lab.\term[1] \longleadsto^{\struc}_{\trace[1]} \lterm[1]''$ and $\lterm[1]'' \longleadsto^{\struc}_{\trace[2]} \lterm[1]'$.
        We then have:
        \begin{gather*}
            \begin{prooftree}
                \hypo{@\lab.\term[1] \longleadsto^{\struc}_{\trace[1]} \lterm[1]''}
                \infer1{@\lab.\term[1] \compo \term[2] \longleadsto^{\struc}_{\trace[1]} \lterm[1]'' \compo_{1} \term[2]}
            \end{prooftree} 
            \quad\mbox{ and }\quad
            \begin{prooftree}
                \hypo{\lterm[1]'' \longleadsto^{\struc}_{\trace[2]} \lterm[1]'}
                \infer1{\lterm[1]'' \compo_{1} \term[2] \longleadsto^{\struc}_{\trace[2]} \lterm[1]' \compo_{1} \term[2]}
            \end{prooftree}.
        \end{gather*}
        Thus by letting $\lterm[3] = \lterm[1]'' \compo_{1} \term[2]$, this case has been proved.
        \item Case $\begin{prooftree}
            \hypo{@\lab[1].\term[1] \longleadsto^{\struc}_{\trace[3]} \lterm[1]'}
            \hypo{\EPS_{\lab[3]}(\lterm[1]')}
            \hypo{@\lab[3].\term[2] \longleadsto^{\struc}_{\trace[3]'} \lterm[2]'}
            \infer3{@\lab[1].\term[1] \compo \term[2] \longleadsto^{\struc}_{\trace[3] \series \trace[3]'} \lterm[2]'}
        \end{prooftree}$:
        We distinguish the following cases:
        \begin{itemize}
            \item Case $\trace[3] = \trace[1] \series \trace[3]''$ (and $\trace[2] = \trace[3]'' \series \trace[3]'$):
            By IH, there is $\lterm[1]''$ such that $@\lab[1].\term[1] \longleadsto^{\struc}_{\trace[1]} \lterm[1]''$ and $\lterm[1]'' \longleadsto^{\struc}_{\trace[3]''} \lterm[1]'$.
            Then,
            \begin{prooftree}[]
                    \hypo{\lterm[1]'' \longleadsto^{\struc}_{\trace[3]''} \lterm[1]'}
                    \hypo{\EPS_{\lab[3]}(\lterm[1]')}
                    \hypo{@\lab[3].\term[2] \longleadsto^{\struc}_{\trace[3]'} \lterm[2]'}
                    \infer3{\lterm[1]'' \longleadsto^{\struc}_{\trace[3]'' \series \trace[3]'} \lterm[2]'}
                \end{prooftree}.
            By letting $\lterm[3] = \lterm[1]''$, this case has been proved.
    
            \item Case $\trace[3]' = \trace[3]'' \series \trace[2]$ (and $\trace[1] = \trace[3] \series \trace[3]''$):
            By IH, there is $\lterm[2]''$ such that $@\lab[3].\term[2] \longleadsto^{\struc}_{\trace[3]''} \lterm[2]''$ and $\lterm[2]'' \longleadsto^{\struc}_{\trace[2]} \lterm[2]'$.
            Then,
            \begin{prooftree}[]
                \hypo{@\lab[1].\term[1] \longleadsto^{\struc}_{\trace[3]} \lterm[1]'}
                \hypo{\EPS_{\lab[3]}(\lterm[1]')}
                \hypo{@\lab[3].\term[2] \longleadsto^{\struc}_{\trace[3]''} \lterm[2]''}
                \infer3{@\lab[1].\term[1] \longleadsto^{\struc}_{\trace[3] \series \trace[3]''} \lterm[2]''}
            \end{prooftree}.
            By letting $\lterm[3] = \lterm[2]''$, this case has been proved.
        \end{itemize}

        \item Case $\begin{prooftree}
            \hypo{\lterm[1]  \longleadsto^{\struc}_{\trace[3]} \lterm[1]'}
            \infer1{\lterm[1] \compo_{1} \term[3] \longleadsto^{\struc}_{\trace[3]} \lterm[1]' \compo_{1} \term[3]}
        \end{prooftree}$, Case $\begin{prooftree}
            \hypo{\lterm[1] \longleadsto^{\struc}_{\trace[3]} \lterm[1]'}
            \hypo{\EPS_{\lab[3]}(\lterm[1]')}
            \hypo{@ \lab[3]. \term[2] \longleadsto^{\struc}_{\trace[3]'} \lterm[2]'}
            \infer3{\lterm[1] \compo_{1} \term[2] \longleadsto^{\struc}_{\trace[3] \series \trace[3]'} \lterm[2]'}
        \end{prooftree}$:
        Similar to the two cases above.
    
        \item Case $\begin{prooftree}[separation=0.8em]
            \hypo{@\lab[1]_0.\term[1] \longleadsto^{\struc}_{\trace[3]_0} \lterm[1]_1'}
            \hypo{\EPS_{\lab[1]_1}(\lterm[1]_1')}
            \hypo{\dots}
            \hypo{@\lab_{n-1}.\term[1] \longleadsto^{\struc}_{\trace[3]_{n-1}} \lterm[1]_n'}
            \hypo{\EPS_{\lab[1]_n}(\lterm[1]_n')}
            \hypo{@\lab[1]_{n}.\term[1] \longleadsto^{\struc}_{\trace[3]_{n}} \lterm[1]_{n+1}'}
            \infer6{ @\lab[1]_0.\term[1]^* \longleadsto^{\struc}_{\trace[3]_0 \series \dots \series \trace[3]_n} \lterm[1]_{n+1}' \compo_{1} \term[1]^*}
        \end{prooftree}$:\\
        Then there are some $i \in \range{0, n}$,
        $\trace[3]_i'$, and $\trace[3]_i''$ such that
        $\trace[3]_i = \trace[3]_i' \series \trace[3]_i''$,
        $\trace[1] = \trace[3]_0 \series \dots \series \trace[3]_{i-1} \series \trace[3]_i'$, and 
        $\trace[2] = \trace[3]_i'' \series \trace[3]_{i+1} \series \dots \series \trace[3]_{n}$.
        By IH, there is some $\lterm[3]'$ such that  $@\lab_{i}.\term[1]_{i} \longleadsto^{\struc}_{\trace[3]_{i}'} \lterm[3]'$ and $\lterm[3]' \longleadsto^{\struc}_{\trace[3]_{i}''} \lterm[1]_{i+1}'$.
        We then have:
        \begin{gather*}
            \begin{prooftree}[separation=0.8em]
                \hypo{@\lab[1]_0.\term[1] \longleadsto^{\struc}_{\trace[3]_0} \lterm[1]_1'}
                \hypo{\EPS_{\lab[1]_1}(\lterm[1]_1')}
                \hypo{\dots}
                \hypo{@\lab_{i-1}.\term[1] \longleadsto^{\struc}_{\trace[3]_{i-1}} \lterm[1]_{i}'}
                \hypo{\EPS_{\lab[1]_{i}}(\lterm[1]_{i}')}
                \hypo{@\lab[1]_{i}.\term[1] \longleadsto^{\struc}_{\trace[3]_{i}'} \lterm[3]'}
                \infer6{ @\lab[1]_0.\term[1]^* \longleadsto^{\struc}_{\trace[3]_0 \series \dots \series \trace[3]_{i-1} \series \trace[3]_i'} \lterm[3]' \compo_{1} \term[1]^*}
            \end{prooftree} \mbox{ and } 
            \\
            \begin{prooftree}[separation=0.7em]
                \hypo{\lterm[3]' \longleadsto^{\struc}_{\trace[3]_i''} \lterm[1]_{i+1}'}
                \hypo{\EPS_{\lab[1]_{i+1}}(\lterm[1]_{i+1}')}
                \hypo{@\lab_{i+1}.\term[1] \longleadsto^{\struc}_{\trace[3]_{i+1}} \lterm[1]_{i+2}'}
                \hypo{\EPS_{\lab[1]_{i+2}}(\lterm[1]_{i+2}')}
                \hypo{\dots}
                \hypo{\EPS_{\lab[1]_n}(\lterm[1]_n')}
                \hypo{@\lab[1]_{n}.\term[1] \longleadsto^{\struc}_{\trace[3]_n} \lterm[1]_{n+1}'}
                \infer5{@\lab[1]_{i+1}.\term[1]^* \longleadsto^{\struc}_{\trace[3]_{i+1} \series \dots \series \trace[3]_n} \lterm[1]_{n+1}' \compo_{1} \term[1]^*}
                \infer3{\lterm[3]' \compo_{1} \term[1]^* \longleadsto^{\struc}_{\trace[3]_{i}'' \series \trace[3]_{i+1} \series \dots \series \trace[3]_n} \lterm[1]_{n+1}' \compo_{1} \term[1]^*}
            \end{prooftree}
        \end{gather*}
        Hence by letting $\lterm[3] = \lterm[3]' \compo_{1} \term[1]^*$, this case has been proved.
    
        \item Case \begin{prooftree}[separation = .5em]
            \hypo{@\lab[1]_0.\term[1] \longleadsto^{\struc}_{\trace[3]_0} \lterm[1]_1'}
            \hypo{\EPS_{\lab[1]_1}(\lterm[1]_1')}
            \hypo{\dots}
            \hypo{@\lab_{n-1}.\term[1] \longleadsto^{\struc}_{\trace[3]_{n-1}} \lterm[1]_n'}
            \hypo{\EPS_{\lab[1]_n}(\lterm[1]_n')}
            \infer5{ @\lab[1]_0.\term[1]^* \longleadsto^{\struc}_{\trace[3]_0 \series \dots \series \trace[3]_{n-1}} @\lab[1]_{n}. \term[1]^*}
        \end{prooftree}:
        Similar to the case above.

        \item Case $\begin{prooftree}
            \hypo{@\lab.\term[1] \longleadsto^{\struc}_{\trace[3]'} \lterm[1]'}
            \hypo{\EPS_{\lab[3]}(\lterm[1]')}
            \hypo{@\lab.\term[2] \longleadsto^{\struc}_{\trace[3]''} \lterm[2]'}
            \hypo{\EPS_{\lab[3]}(\lterm[2]')}
            \infer4{@\lab.\term[1] \intersection \term[2] \longleadsto^{\struc}_{\fork^{\lab[1]}_{1} \series (\trace[3]' \parallel \trace[3]'') \series \join^{\lab[3]}_{1} }  @\lab[3].\id}
        \end{prooftree}$:
        We distinguish the following cases:
        \begin{itemize}
            \item Case $\trace[1] = \fork^{\lab[1]}_{1} \series (\trace[3]' \parallel \trace[3]'') \series \join^{\lab[3]}_{1}$ and $\trace[2] = \id^{\lab[3]}$:
            By $@\lab.\term[1] \intersection \term[2] \longleadsto^{\struc}_{\trace[1]} @\lab[3].\id$ and $@\lab[3].\id \longleadsto^{\struc}_{\trace[2]} @\lab[3].\id$.

            \item Case $\trace[1] = \id^{\lab[1]}$ and $\trace[2] = \fork^{\lab[1]}_{1} \series (\trace[3]' \parallel \trace[3]'') \series \join^{\lab[3]}_{1}$:
            Similar to the case above.

            \item Case $\trace[1] = \fork^{\lab[1]}_{1} \series (\trace[1]' \parallel \trace[1]'')$ and $\trace[2] = (\trace[2]' \parallel \trace[2]'') \series \join^{\lab[3]}_{1}$
            (then, $\trace[3]' = \trace[1]' \series \trace[2]'$ and $\trace[3]'' = \trace[1]'' \series \trace[2]''$):
            By IH with $\trace[3]' = \trace[1]' \series \trace[2]'$,
            there is some $\lterm[1]''$ such that $@\lab.\term[1] \longleadsto^{\struc}_{\trace[1]'} \lterm[1]''$ and $\lterm[1]'' \longleadsto^{\struc}_{\trace[2]'} \lterm[1]'$.
            By IH with $\trace[3]'' = \trace[1]'' \series \trace[2]''$,
            there is some $\lterm[2]''$ such that $@\lab.\term[2] \longleadsto^{\struc}_{\trace[1]''} \lterm[2]''$ and $\lterm[2]'' \longleadsto^{\struc}_{\trace[2]''} \lterm[2]'$.
            Then we have:
            \begin{gather*}
                \begin{prooftree}
                    \hypo{@\lab.\term[1] \longleadsto^{\struc}_{\trace[1]'} \lterm[1]''}
                    \hypo{@\lab.\term[2] \longleadsto^{\struc}_{\trace[1]''} \lterm[2]''}
                    \infer2{@\lab.\term[1] \intersection \term[2] \longleadsto^{\struc}_{\fork^{\lab}_{1} \series (\trace[1]' \parallel \trace[1]'')} \lterm[1]'' \intersection_{1} \lterm[2]''}
                \end{prooftree}
                \quad\mbox{and}\quad
                \begin{prooftree}
                    \hypo{\lterm[1]'' \longleadsto^{\struc}_{\trace[2]'} \lterm[1]' \; \EPS_{\lab[3]}(\lterm[1]')}
                    \hypo{\lterm[2]'' \longleadsto^{\struc}_{\trace[2]''} \lterm[2]' \; \EPS_{\lab[3]}(\lterm[2]')}
                    \infer2{\lterm[1]'' \intersection_{1} \lterm[2]'' \longleadsto^{\struc}_{(\trace[2]' \parallel \trace[2]'') \series \join^{\lab[3]}_{1} }  @\lab[3].\id}
                \end{prooftree}
            \end{gather*}
            Hence by letting $\lterm[3] = \lterm[1]'' \intersection_{1} \lterm[2]''$, this case has been proved.
        \end{itemize}

        \item Case $\begin{prooftree}
            \hypo{@\lab.\term[1] \longleadsto^{\struc}_{\trace[3]'} \lterm[1]'}
            \hypo{@\lab.\term[2] \longleadsto^{\struc}_{\trace[3]''} \lterm[2]'}
            \infer2{@\lab.\term[1] \intersection \term[2] \longleadsto^{\struc}_{\fork^{\lab[1]}_{1} \series (\trace[3]' \parallel \trace[3]'')} \lterm[1]' \intersection_{1} \lterm[2]'}
        \end{prooftree}$, Case $\begin{prooftree}
            \hypo{\lterm[1] \longleadsto^{\struc}_{\trace[3]'} \lterm[1]' \hspace{1.em} \EPS_{\lab[3]}(\lterm[1]')}
            \hypo{\lterm[2] \longleadsto^{\struc}_{\trace[3]''} \lterm[2]' \hspace{1.em} \EPS_{\lab[3]}(\lterm[2]')}
            \infer2{\lterm[1] \intersection_{1} \lterm[2] \longleadsto^{\struc}_{(\trace[3]' \parallel \trace[3]'') \series \join^{\lab[3]}_{1}}  @\lab[3].\id}
        \end{prooftree}$, and Case\\ 
        \hspace{2em}$\begin{prooftree}
            \hypo{\lterm[1] \longleadsto^{\struc}_{\trace[3]'} \lterm[1]'}
            \hypo{\lterm[2] \longleadsto^{\struc}_{\trace[3]''} \lterm[2]'}
            \infer2{\lterm[1] \intersection_{1} \lterm[2] \longleadsto^{\struc}_{(\trace[3]' \parallel \trace[3]'')} \lterm[1]' \intersection_{1} \lterm[2]'}
        \end{prooftree}$:
        Similar to the case above.
        
        \item Case $\begin{prooftree}
            \hypo{\mathstrut}
            \infer1{\lterm[1] \longleadsto^{\struc}_{\id^{\overrightarrow{\mathsf{lab}}(\lterm)}} \lterm[1]}
        \end{prooftree}$:
        Then $\trace[1] = \trace[2] = \id^{\overrightarrow{\mathsf{lab}}(\lterm)}$.
        This case has been proved by the assumption.
    \end{itemize}\noindent 
    Hence, this completes the proof.
\end{proof}

\subsection{Proof of \Cref{theorem: decomposition}: decomposition of derivatives}
We now prove \Cref{theorem: decomposition}.
First, we extend \Cref{definition: decomposition} with annotations $(-)_{\trace}$ where $\trace$ is a \kl[$\struc$-run]{$(\odot \vec{\struc})_{\bullet}$-run}.
\begin{defi}\label{definition: decomposition'}
    Let $\vec{\struc} = \struc_1 \dots \struc_n \in \STR^{+}$.
    The relation $\lterm[1] \mathrel{(\odot_{i = 1}^{n} \longrightarrow^{(\struc_{i})_{\bullet}})}_{\trace} \lterm[2]$,
    where $\tuple{\lterm[1], \trace, \lterm[2]} \in \bigcup_{j = 1}^{n} (\LT^{\vec{\struc}}_{(j)} \times \lsubSPR^{\vec{\struc}}_{(j)} \times \LT^{\vec{\struc}}_{(j)})$,
    is defined as the smallest relation (of tuples of $\lterm[1]$, $\trace$, and $\lterm[2]$) closed under the following rules:
    \begin{gather*}
        \begin{prooftree}
            \hypo{\lterm[1] \longrightarrow^{(\struc_{j})_{\bullet}}_{\trace[1]} \lterm[2]}
            \infer1[D]{\lterm[1]_{(j)} \mathrel{(\odot_{i = 1}^{n} \longrightarrow^{(\struc_{i})_{\bullet}})}_{\trace[1]_{(j)}} \lterm[2]_{(j)}}
        \end{prooftree}
        \hspace{3em}
        \begin{prooftree}
            \hypo{\lterm[1][\bullet/l] \mathrel{(\odot_{i = 1}^{n} \longrightarrow^{(\struc_{i})_{\bullet}})}_{\trace[1][\bullet/\tuple{l, r}]} \lterm[2][\bullet/r]}
            \infer1[L]{\lterm[1][\lab/l] \mathrel{(\odot_{i = 1}^{n} \longrightarrow^{(\struc_{i})_{\bullet}})}_{\trace[1][\lab[1]/\tuple{l, r}]} \lterm[2][\lab/r]}
        \end{prooftree}
        \\
        \begin{prooftree}
            \hypo{\lterm[1] \mathrel{(\odot_{i = 1}^{n} \longrightarrow^{(\struc_{i})_{\bullet}})}_{\trace[1]} \lterm[3]}
            \hypo{\lterm[3] \mathrel{(\odot_{i = 1}^{n} \longrightarrow^{(\struc_{i})_{\bullet}})}_{\trace[2]} \lterm[2]}
            \infer2[T]{\lterm[1] \mathrel{(\odot_{i = 1}^{n} \longrightarrow^{(\struc_{i})_{\bullet}})}_{\trace[1] \series \trace[2]} \lterm[2]}
        \end{prooftree}.
    \end{gather*}
\end{defi}
The rules of \Cref{definition: decomposition'} $\mathrel{(\odot_{i = 1}^{n} \longrightarrow^{(\struc_{i})_{\bullet}})}_{\trace}$ are the same as those of \Cref{definition: decomposition} $\mathrel{(\odot_{i = 1}^{n} \longrightarrow^{(\struc_{i})_{\bullet}})}$ except that an \kl[$\struc$-run]{$(\odot \vec{\struc})_{\bullet}$-run} is annotated.
Thus, the following immediately holds:
\begin{prop}\label{proposition: decomposition'}
    Let $\vec{\struc} = \struc_1 \dots \struc_n \in \STR^{+}$.
    For all $j \in \range{n}$ and $\lterm[1], \lterm[2] \in \LT^{\vec{\struc}}_{(j)}$,
    we have:
    \[\lterm[1] \mathrel{(\odot_{i = 1}^{n} \longrightarrow^{(\struc_{i})_{\bullet}})} \lterm[2] \;\iff\; \exists \trace[1] \in \lsubSPR^{\vec{\struc}}_{(j)},\ \lterm[1] \mathrel{(\odot_{i = 1}^{n} \longrightarrow^{(\struc_{i})_{\bullet}})}_{\trace[1]} \lterm[2].\]
\end{prop}
\begin{proof}
    Both directions are shown by easy induction on the derivation trees.
\end{proof}
We also have the following by using the properties shown in \Cref{section: derivative properties}.
\begin{lem}\label{lemma: decomposition' 2}
    Let $\vec{\struc} = \struc_1 \dots \struc_n \in \STR^{+}$.
    For $\tuple{\lterm[1], \trace[1], \lterm[2]} \in \bigcup_{j = 1}^{n} \LT^{\vec{\struc}}_{(j)} \times \lsubSPR^{\vec{\struc}}_{(j)} \times \LT^{\vec{\struc}}_{(j)}$,
    \[\lterm[1] \longrightarrow^{(\odot \vec{\struc})_{\bullet}}_{\trace[1]} \lterm[2] \;\iff\; \lterm[1] \mathrel{(\odot_{i = 1}^{n} \longrightarrow^{(\struc_{i})_{\bullet}})}_{\trace[1]} \lterm[2].\]
\end{lem}
\begin{proof}
    ($\Longleftarrow$):
    By easy induction on the derivation tree of \Cref{definition: decomposition'}.
    ($\Longrightarrow$):
    Using \Cref{lemma: decomposition lsubSPR},
    we show this by induction on the derivation tree of $\trace[1] \in \odot_{i = 1}^{n} \lsubSPR^{(\struc_i)_{\bullet}}$ (of the rules in \Cref{definition: decomposition lsubSPR}).
    \begin{itemize}
        \item Case \begin{prooftree}
            \hypo{\trace[1] \in \lsubSPR^{(\struc_j)_{\bullet}}}
            \infer1[D]{\trace[1]_{(j)} \in \odot_{i = 1}^{n} \lsubSPR^{(\struc_i)_{\bullet}}}
        \end{prooftree}:
        By $\lterm[1] \longrightarrow^{(\odot \vec{\struc})_{\bullet}}_{\trace[1]_{(j)}} \lterm[2]$ (so $\lterm[1], \lterm[2] \in \LT^{\vec{\struc}}_{(j)}$),
        let $\lterm[1]'$ and $\lterm[2]'$ be such that $\lterm[1] = \lterm[1]'_{(j)}$ and $\lterm[2] = \lterm[2]'_{(j)}$.
        By $\lterm[1]'_{(j)} \longrightarrow^{(\odot \vec{\struc})_{\bullet}}_{\trace[1]_{(j)}} \lterm[2]'_{(j)}$,
        we have $\lterm[1]' \longrightarrow^{(\struc[1]_{j})_{\bullet}}_{\trace[1]} \lterm[2]'$ (\Cref{lemma: i}).
        Thus, by the rule (D), this case has been proved.
        
        \item Case \begin{prooftree}
            \hypo{\trace[1][\bullet/\tuple{l, r}] \in \odot_{i = 1}^{n} \lsubSPR^{(\struc_i)_{\bullet}}}
            \infer1[L]{\trace[1][\lab/\tuple{l, r}] \in \odot_{i = 1}^{n} \lsubSPR^{(\struc_i)_{\bullet}}}
        \end{prooftree}:
        By $\lterm[1] \longrightarrow^{(\odot \vec{\struc})_{\bullet}}_{\trace[1][\lab/\tuple{l, r}]} \lterm[2]$,
        let $\lterm[1]'$ and $\lterm[2]'$ be such that $\lterm[1] = \lterm[1]'[\lab/l]$ and  $\lterm[2] = \lterm[2]'[\lab/r]$.
        Then we have $\lterm[1]'[\bullet/l] \longrightarrow^{(\odot \vec{\struc})_{\bullet}}_{\trace[1][\bullet/\tuple{l, r}]} \lterm[2]'[\bullet/r]$ (\Cref{lemma: L}).
        Thus, by the rule (L) with IH, this case has been proved.

        \item Case \begin{prooftree}
            \hypo{\trace[1] \in \odot_{i = 1}^{n} \lsubSPR^{(\struc_i)_{\bullet}}}
            \hypo{\trace[2] \in \odot_{i = 1}^{n} \lsubSPR^{(\struc_i)_{\bullet}}}
            \infer2[T]{\trace[1] \series \trace[2] \in \odot_{i = 1}^{n} \lsubSPR^{(\struc_i)_{\bullet}}}
        \end{prooftree}:
        By $\lterm[1] \longrightarrow^{(\odot \vec{\struc})_{\bullet}}_{\trace[1] \series \trace[2]} \lterm[2]$,
        there exists some \kl{lKL term} $\lterm[3]$ such that 
        $\lterm[1] \longrightarrow^{(\odot \vec{\struc})_{\bullet}}_{\trace[1]} \lterm[3]$
        and
        $\lterm[3] \longrightarrow^{(\odot \vec{\struc})_{\bullet}}_{\trace[2]} \lterm[2]$ (\Cref{lemma: T}).
        Thus by the rule (T) with IH, this case has been proved.
   \end{itemize}\noindent 
    Hence, this completes the proof.
\end{proof}

\begin{proof}[Proof of \Cref{theorem: decomposition}]
    For all $\lterm[1], \lterm[2] \in \LT^{\vec{\struc}}_{(j)}$, we have:
    \begin{align*}
        \lterm[1] \longrightarrow^{(\odot \vec{\struc})_{\bullet}} \lterm[2]
        &\;\iff\; \exists \trace[1] \in \lsubSPR^{\vec{\struc}}_{(j)},\ \lterm[1] \longrightarrow^{(\odot \vec{\struc})_{\bullet}}_{\trace[1]} \lterm[2] \tag{\Cref{proposition: local}}\\
        &\;\iff\; \exists \trace[1] \in \lsubSPR^{\vec{\struc}}_{(j)},\ \lterm[1] \mathrel{(\odot_{i = 1}^{n} \longrightarrow^{(\struc_{i})_{\bullet}})}_{\trace[1]} \lterm[2] \tag{\Cref{lemma: decomposition' 2}}\\
        &\;\iff\; \lterm[1] \mathrel{(\odot_{i = 1}^{n} \longrightarrow^{(\struc_{i})_{\bullet}})} \lterm[2]. \tag{\Cref{proposition: decomposition'}}
    \end{align*}
    Hence, this completes the proof.
\end{proof}

\subsection{Incompleteness of the decomposition rules in \cite{nakamuraPartialDerivativesGraphs2017}}\label{section: incomplete}
In this subsection, we point out an error of the decomposition of \kl[$\struc$-runs]{$(\odot \vec{\struc})_{\bullet}$-runs} in \cite[Section IV (Lemma IV.10)]{nakamuraPartialDerivativesGraphs2017}.
In \cite{nakamuraPartialDerivativesGraphs2017}, \kl{left quotients} of \kl[$\struc$-runs]{$(\odot \vec{\struc})_{\bullet}$-runs} with respect to SGCPs are considered, but events of SGCPs are read in one-way.
This approach essentially corresponds to the following decomposition rules of \kl[$\struc$-runs]{$(\odot \vec{\struc})_{\bullet}$-runs}
where $\vec{\struc} = \struc_1 \dots \struc_n$ and each $X_j$ ($j \in \range{n}$) is the smallest subset of $\bigcup_{i = 1}^{n} \lsubSPR^{\vec{\struc}}_{(i)}$ closed under the following rules (\cf\ \Cref{definition: decomposition lsubSPR}):
\begin{gather*}
    \begin{prooftree}
        \hypo{\trace[1] \in \lsubSPR^{\struc_j}}
        \infer1[D]{\trace[1]_{(j)} \in X_j}
    \end{prooftree}\hspace{1.5em} 
    \begin{prooftree}
        \hypo{\trace[1][\bullet/\tuple{l, r}] \in X_j}
        \infer1[L]{\trace[1][\lab/\tuple{l, r}] \in X_j}
    \end{prooftree}\hspace{1.5em}
    \begin{prooftree}
        \hypo{\trace[1] \in X_j}
        \hypo{\trace[2] \in X_j}
        \infer2[T]{\trace[1] \series \trace[2] \in X_j}
    \end{prooftree}\hspace{1.5em}
    \begin{prooftree}
        \hypo{\trace[1] \in X_{j-1}}
        \infer1[$+1$]{\trace[1] \in X_j}
    \end{prooftree}
\end{gather*}
(In a nutshell, the rule ``$-1$'' does not exist.)
However, the set $\bigcup_{i = 1}^{n} X_i$ is incomplete, as the direction $\Longleftarrow$ of \Cref{lemma: decomposition lsubSPR} does not hold.
We recall \Cref{example: definition: decomposition lsubSPR} and the sequence $\vec{\struc} = \textcolor{blue}{\struc_1} \textcolor{red!60}{\struc_2}$.
For instance, the \kl[$\struc$-run]{$(\odot \vec{\struc})_{\bullet}$-run} of the form $\begin{tikzpicture}[baseline = -2.ex, yscale=-1]
    \foreach \x/\y/\col/\style/\dc in {
        2/1/gray!20//{red!60},5/0/gray!20/{thick}/,6/2/gray!20//{red!60},7/0/gray!20/{thick}/,10/1/gray!20//{red!60}}{
       \node[mynode,draw=\dc,circle, fill = \col, \style]  (\x X\y) at (.25*\x,.25*\y) {};
    }
    \foreach \x/\y in {
        1/1,1/3,3/1,4/0,4/2,6/0,8/0,8/2,9/1,11/3,11/1}{
       \node[inner sep = 0, outer sep = 0, minimum width = 0] (\x X\y) at (.25*\x,.25*\y) {};
    }
    \node[left = 4pt of 2X1](s1){}; \path (s1) edge[earrow, ->] (2X1);
    \node[right = 4pt of 10X1](t1){}; \path (10X1) edge[earrow, ->] (t1);
    \foreach \xp/\x/\y/\col in {3/4/1/red}{
        \node  (\x X\y) at (.25*\x+.1,.25*\y) {};
        \path (\xp X\y) edge [opacity = 0] node[pos= .9, elabel, color = \col, inner sep = 0, outer sep = 0, minimum width = 0](\xp X\y Xf){}(\x X\y);
    }
    \foreach \xp/\x/\y/\col in {8/9/1/red}{
        \node (\xp X\y) at (.25*\xp-.1,.25*\y) {};
        \path (\xp X\y) edge [opacity = 0] node[pos= .1, elabel, color = \col, inner sep = 0, outer sep = 0, minimum width = 0](\x X\y Xj){}(\x X\y);
    }
    \graph[use existing nodes, edges={color=black, pos = .5, earrow, line width = .5pt, 
    decorate, decoration={snake, segment length=2mm, amplitude=.2mm}
    },edge quotes={fill=white, inner sep=1pt,font= \scriptsize}]{
    {2X1} --[ color = red!60] {3X1};
    {3X1} --[color = red!60] {3X1Xf} -- [color = red!60] {4X0, 4X2}; {4X2} --[ color = red!60] {6X2};
    {4X0} ->[ color = red!60] {5X0} --[ color = blue] {6X0};
    {6X0} ->[ color = blue] {7X0} --[ color = red!60] {8X0}; {6X2} --[ color = red!60] {8X2};
    {8X0,8X2} --[color = red!60] {9X1Xj} -- [color = red!60] {9X1};
    {9X1} ->[ color = red!60] {10X1};
    };
\end{tikzpicture}$ is derived by:
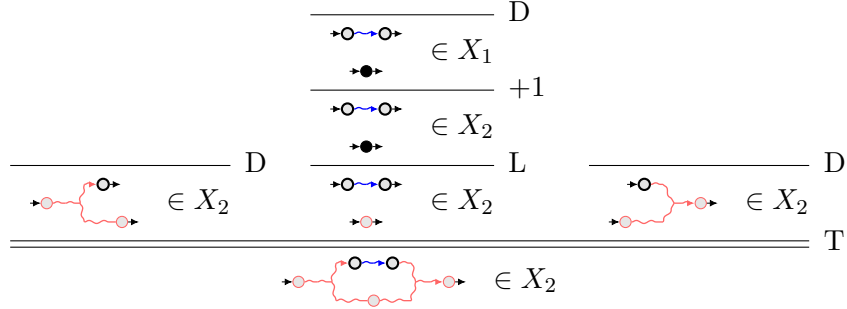
\begin{center}
    \begin{prooftree}
        \hypo{\mathstrut}
        \infer1[D]{\begin{tikzpicture}[baseline = -2.ex, yscale=-1]
            \foreach \x/\y/\col/\style/\dc in {
                2/1/gray!20//{red!60},5/0/gray!20/{thick}/,6/2/gray!20//{red!60}}{
               \node[mynode,draw = \dc,circle, fill = \col, \style]  (\x X\y) at (.25*\x,.25*\y) {};
            }
            \foreach \x/\y in {
                3/1,4/0,4/2}{
               \node[inner sep = 0, outer sep = 0, minimum width = 0] (\x X\y) at (.25*\x,.25*\y) {};
            }
            \node[left = 4pt of 2X1](s1){}; \path (s1) edge[earrow, ->] (2X1);
            \node[right = 4pt of 5X0](t1){}; \path (5X0) edge[earrow, ->] (t1);
            \node[right = 4pt of 6X2](t2){}; \path (6X2) edge[earrow, ->] (t2);
            \foreach \xp/\x/\y/\col in {3/4/1/red!60}{
                \node  (\x X\y) at (.25*\x+.1,.25*\y) {};
                \path (\xp X\y) edge [opacity = 0] node[pos= .9, elabel, color = \col, inner sep = 0, outer sep = 0, minimum width = 0](\xp X\y Xf){}(\x X\y);
            }
            \graph[use existing nodes, edges={color=black, pos = .5, earrow, line width = .5pt, 
            decorate, decoration={snake, segment length=2mm, amplitude=.2mm}
            },edge quotes={fill=white, inner sep=1pt,font= \scriptsize}]{
            {2X1} --[ color = red!60] {3X1};
            {3X1} --[color = red!60] {3X1Xf} -- [color = red!60] {4X0, 4X2}; {4X2} --[ color = red!60] {6X2};
            {4X0} ->[ color = red!60] {5X0};
            };
        \end{tikzpicture} \in X_2}
        \hypo{\mathstrut}
        \infer1[D]{
            \begin{tikzpicture}[baseline = -2.ex, yscale=-1]
                \foreach \x/\y/\col/\style/\dc in {
                    5/0/gray!20/{thick}/,6/2/black//,7/0/gray!20/{thick}/}{
                   \node[mynode,draw = \dc,circle, fill = \col, \style]  (\x X\y) at (.25*\x,.25*\y) {};
                }
                \foreach \x/\y in {
                    6/0}{
                   \node[inner sep = 0, outer sep = 0, minimum width = 0] (\x X\y) at (.25*\x,.25*\y) {};
                }
                \node[left = 4pt of 5X0](s1){}; \path (s1) edge[earrow, ->] (5X0);
                \node[left = 4pt of 6X2](s2){}; \path (s2) edge[earrow, ->] (6X2);
                \node[right = 4pt of 7X0](t1){}; \path (7X0) edge[earrow, ->] (t1);
                \node[right = 4pt of 6X2](t2){}; \path (6X2) edge[earrow, ->] (t2);
                \graph[use existing nodes, edges={color=black, pos = .5, earrow, line width = .5pt, 
                decorate, decoration={snake, segment length=2mm, amplitude=.2mm}
                },edge quotes={fill=white, inner sep=1pt,font= \scriptsize}]{
                {5X0} --[ color = blue] {6X0};
                {6X0} ->[ color = blue] {7X0};
                };
            \end{tikzpicture}\in X_1}
            \infer1[+1]{
                \begin{tikzpicture}[baseline = -2.ex, yscale=-1]
                    \foreach \x/\y/\col/\style/\dc in {
                        5/0/gray!20/{thick}/,6/2/black//,7/0/gray!20/{thick}/}{
                       \node[mynode,draw = \dc,circle, fill = \col, \style]  (\x X\y) at (.25*\x,.25*\y) {};
                    }
                    \foreach \x/\y in {
                        6/0}{
                       \node[inner sep = 0, outer sep = 0, minimum width = 0] (\x X\y) at (.25*\x,.25*\y) {};
                    }
                    \node[left = 4pt of 5X0](s1){}; \path (s1) edge[earrow, ->] (5X0);
                    \node[left = 4pt of 6X2](s2){}; \path (s2) edge[earrow, ->] (6X2);
                    \node[right = 4pt of 7X0](t1){}; \path (7X0) edge[earrow, ->] (t1);
                    \node[right = 4pt of 6X2](t2){}; \path (6X2) edge[earrow, ->] (t2);
                    \graph[use existing nodes, edges={color=black, pos = .5, earrow, line width = .5pt, 
                    decorate, decoration={snake, segment length=2mm, amplitude=.2mm}
                    },edge quotes={fill=white, inner sep=1pt,font= \scriptsize}]{
                    {5X0} --[ color = blue] {6X0};
                    {6X0} ->[ color = blue] {7X0};
                    };
                \end{tikzpicture} \in X_2}
        \infer1[L]{\begin{tikzpicture}[baseline = -2.ex, yscale=-1]
            \foreach \x/\y/\col/\style/\dc in {
                5/0/gray!20/{thick}/,6/2/gray!20//{red!60},7/0/gray!20/{thick}/}{
               \node[mynode,draw =\dc,circle, fill = \col, \style]  (\x X\y) at (.25*\x,.25*\y) {};
            }
            \foreach \x/\y in {
                6/0}{
               \node[inner sep = 0, outer sep = 0, minimum width = 0] (\x X\y) at (.25*\x,.25*\y) {};
            }
            \node[left = 4pt of 5X0](s1){}; \path (s1) edge[earrow, ->] (5X0);
            \node[left = 4pt of 6X2](s2){}; \path (s2) edge[earrow, ->] (6X2);
            \node[right = 4pt of 7X0](t1){}; \path (7X0) edge[earrow, ->] (t1);
            \node[right = 4pt of 6X2](t2){}; \path (6X2) edge[earrow, ->] (t2);
            \graph[use existing nodes, edges={color=black, pos = .5, earrow, line width = .5pt, 
            decorate, decoration={snake, segment length=2mm, amplitude=.2mm}
            },edge quotes={fill=white, inner sep=1pt,font= \scriptsize}]{
            {5X0} --[ color = blue] {6X0};
            {6X0} ->[ color = blue] {7X0};
            };
        \end{tikzpicture} \in X_2}
        \hypo{\mathstrut}
        \infer1[D]{
            \begin{tikzpicture}[baseline = -2.ex, yscale=-1]
                \foreach \x/\y/\col/\style/\dc in {
                    6/2/gray!20//{red!60},7/0/gray!20/{thick}/,10/1/gray!20//{red!60}}{
                   \node[mynode,draw=\dc,circle, fill = \col, \style]  (\x X\y) at (.25*\x,.25*\y) {};
                }
                \foreach \x/\y in {
                    8/0,8/2,9/1}{
                   \node[inner sep = 0, outer sep = 0, minimum width = 0] (\x X\y) at (.25*\x,.25*\y) {};
                }
                \node[left = 4pt of 7X0](s1){}; \path (s1) edge[earrow, ->] (7X0);
                \node[left = 4pt of 6X2](s2){}; \path (s2) edge[earrow, ->] (6X2);
                \node[right = 4pt of 10X1](t1){}; \path (10X1) edge[earrow, ->] (t1);
                \foreach \xp/\x/\y/\col in {8/9/1/red!60}{
                    \node (\xp X\y) at (.25*\xp,.25*\y) {};
                    \path (\xp X\y) edge [opacity = 0] node[pos= .1, elabel, color = \col, inner sep = 0, outer sep = 0, minimum width = 0](\x X\y Xj){}(\x X\y);
                }
                \graph[use existing nodes, edges={color=black, pos = .5, earrow, line width = .5pt, 
                decorate, decoration={snake, segment length=2mm, amplitude=.2mm}
                },edge quotes={fill=white, inner sep=1pt,font= \scriptsize}]{
                {7X0} --[ color = red!60] {8X0}; {6X2} --[ color = red!60] {8X2};
                {8X0,8X2} --[color = red!60] {9X1Xj} -- [color = red!60] {9X1};
                {9X1} ->[ color = red!60] {10X1};
                };
            \end{tikzpicture} \in X_2
        }
        \infer[double]3[T]{
              \begin{tikzpicture}[baseline = -2.ex, yscale=-1]
                \foreach \x/\y/\col/\style/\dc in {
                    2/1/gray!20//{red!60},5/0/gray!20/{thick}/,6/2/gray!20//{red!60},7/0/gray!20/{thick}/,10/1/gray!20//{red!60}}{
                   \node[mynode,draw =\dc,circle, fill = \col, \style]  (\x X\y) at (.25*\x,.25*\y) {};
                }
                \foreach \x/\y in {
                    1/1,1/3,3/1,4/0,4/2,6/0,8/0,8/2,9/1,11/3,11/1}{
                   \node[inner sep = 0, outer sep = 0, minimum width = 0] (\x X\y) at (.25*\x,.25*\y) {};
                }
                \node[left = 4pt of 2X1](s1){}; \path (s1) edge[earrow, ->] (2X1);
                \node[right = 4pt of 10X1](t1){}; \path (10X1) edge[earrow, ->] (t1);
                \foreach \xp/\x/\y/\col in {3/4/1/red!60}{
                    \node  (\x X\y) at (.25*\x+.1,.25*\y) {};
                    \path (\xp X\y) edge [opacity = 0] node[pos= .9, elabel, color = \col, inner sep = 0, outer sep = 0, minimum width = 0](\xp X\y Xf){}(\x X\y);
                }
                \foreach \xp/\x/\y/\col in {8/9/1/red!60}{
                    \node (\xp X\y) at (.25*\xp-.1,.25*\y) {};
                    \path (\xp X\y) edge [opacity = 0] node[pos= .1, elabel, color = \col, inner sep = 0, outer sep = 0, minimum width = 0](\x X\y Xj){}(\x X\y);
                }
                \graph[use existing nodes, edges={color=black, pos = .5, earrow, line width = .5pt, 
                decorate, decoration={snake, segment length=2mm, amplitude=.2mm}
                },edge quotes={fill=white, inner sep=1pt,font= \scriptsize}]{
                {2X1} --[ color = red!60] {3X1};
                {3X1} --[color = red!60] {3X1Xf} -- [color = red!60] {4X0, 4X2}; {4X2} --[ color = red!60] {6X2};
                {4X0} ->[ color = red!60] {5X0} --[ color = blue] {6X0};
                {6X0} ->[ color = blue] {7X0} --[ color = red!60] {8X0}; {6X2} --[ color = red!60] {8X2};
                {8X0,8X2} --[color = red!60] {9X1Xj} -- [color = red!60] {9X1};
                {9X1} ->[ color = red!60] {10X1};
                };
            \end{tikzpicture} \in X_2
        }
    \end{prooftree}
\end{center}
(Here, blue-colored edges and red-colored edges are the edges induced from $\textcolor{blue}{\struc_1}$ and $\textcolor{red!60}{\struc_2}$, respectively.)
However, we cannot derive the \kl[$\struc$-run]{$(\odot \vec{\struc})_{\bullet}$-run} $\trace$ of the form \begin{tikzpicture}[baseline = -2.ex, yscale=-1]
                \foreach \x/\y/\col/\style/\dc in {
                    2/1/gray!20//{blue},5/0/gray!20/{thick}/,6/2/gray!20//{blue},7/0/gray!20/{thick}/,10/1/gray!20//{blue}}{
                   \node[mynode,draw = \dc,circle, fill = \col, \style]  (\x X\y) at (.25*\x,.25*\y) {};
                }
                \foreach \x/\y in {
                    1/1,1/3,3/1,4/0,4/2,6/0,8/0,8/2,9/1,11/3,11/1}{
                   \node[inner sep = 0, outer sep = 0, minimum width = 0] (\x X\y) at (.25*\x,.25*\y) {};
                }
                \node[left = 4pt of 2X1](s1){}; \path (s1) edge[earrow, ->] (2X1);
                \node[right = 4pt of 10X1](t1){}; \path (10X1) edge[earrow, ->] (t1);
                \foreach \xp/\x/\y/\col in {3/4/1/blue}{
                    \node  (\x X\y) at (.25*\x+.1,.25*\y) {};
                    \path (\xp X\y) edge [opacity = 0] node[pos= .9, elabel, color = \col, inner sep = 0, outer sep = 0, minimum width = 0](\xp X\y Xf){}(\x X\y);
                }
                \foreach \xp/\x/\y/\col in {8/9/1/blue}{
                    \node (\xp X\y) at (.25*\xp-.1,.25*\y) {};
                    \path (\xp X\y) edge [opacity = 0] node[pos= .1, elabel, color = \col, inner sep = 0, outer sep = 0, minimum width = 0](\x X\y Xj){}(\x X\y);
                }
                \graph[use existing nodes, edges={color=black, pos = .5, earrow, line width = .5pt, 
                decorate, decoration={snake, segment length=2mm, amplitude=.2mm}
                },edge quotes={fill=white, inner sep=1pt,font= \scriptsize}]{
                {2X1} --[ color = blue] {3X1};
                {3X1} --[color = blue] {3X1Xf} -- [color = blue] {4X0, 4X2}; {4X2} --[ color = blue] {6X2};
                {4X0} ->[ color = blue] {5X0} --[ color = red!60] {6X0};
                {6X0} ->[ color = red!60] {7X0} --[ color = blue] {8X0}; {6X2} --[ color = blue] {8X2};
                {8X0,8X2} --[color = blue] {9X1Xj} -- [color = blue] {9X1};
                {9X1} ->[ color = blue] {10X1};
                };
            \end{tikzpicture} in this system because the rule ``$-1$'' does not exist.
    To avoid this problem, instead of one-way automata, we use \kl{2AFAs}.

\section{Conclusion} \label{section: conclusion}
We have presented \kl{derivatives} on \kl{structures}/\kl{graphs} and have shown a decomposition theorem of the \kl{derivatives} (\Cref{theorem: decomposition}).
Consequently, we have shown that the \kl{equational theory} of the positive calculus of relations with transitive closure (\kl[PCoR* terms]{PCoR*}) is decidable and EXPSPACE-complete (\Cref{corollary: PCoR* EXPSPACE-complete}).

\subsection*{Related and future work}
In the following, we present related and future work.

\paragraph*{Identity-free Kleene lattices}
Brunet and Pous have shown that the equational theory of \kl{identity-free Kleene lattices} $\set{\emp, \compo, \union, \intersection, \bl^{+}}$ is EXPSPACE-complete \cite{brunetPetriAutomataKleene2015,brunetPetriAutomata2017}.
They presented an algorithm for comparing downward closed sets of \kl{runs} (denoting acyclic and connected \kl{graphs}) with respect to \kl{graph homomorphisms} using \kl{Petri automata}.
However, their approach is problematic when terms have non-acyclic or non-connected \kl{graphs} (\eg, when the identity $\id$, the converse $\smile$, or the universality $\top$ occurs) \cite[Section 9.3]{brunetPetriAutomata2017}.
In our approach, based on the linearly bounded \kl{pathwidth} model property (\Cref{proposition: bounded pw property}),
we consider decomposing \kl{$\struc$-runs} (instead of \kl{runs}) on \kl{path decompositions}.
It would also be possible to connect our \kl{derivatives} to \kl{branching automata} \cite{lodayaSeriesparallelPosetsAlgebra1998,lodayaSeriesParallelLanguages2000,lodayaRationalityAlgebrasSeries2001} or \kl{Petri automata} \cite{brunetPetriAutomata2017}.

\paragraph*{PDL with intersection}
Propositional dynamic logic with intersection (IPDL) \cite{daneckiNondeterministicPropositionalDynamic1984} is propositional dynamic logic (PDL) of regular programs with intersection.
The theory of IPDL is decidable and 2EXPTIME-complete \cite{daneckiNondeterministicPropositionalDynamic1984,lange2ExpTimeLowerBounds2005,gollerPDLIntersectionConverse2009}.
The known algorithms \cite{daneckiNondeterministicPropositionalDynamic1984,lutzPDLIntersectionConverse2005,gollerPDLIntersectionConverse2009} are based on the treewidth at most $2$ model property of IPDL (\cf\ \Cref{proposition: tw 2 property}).
In \cite{gollerPDLIntersectionConverse2009}, G{\"o}eller, Lohrey, and Lutz presented a polynomial-time reduction from the theory of IPDL with converse (ICPDL) into that over $\omega$-regular trees and presented a reduction to the universality problem of two-way alternating parity tree automata.
To express relational intersection ($\intersection$),
their approach is based on the product construction of automata (the product construction naively fails, but by adding transitions so that there always exists the shortest run if a run exists between two \kl{vertices} on an input tree, the product construction works on trees \cite[p.\ 288]{gollerPDLIntersectionConverse2009}),
while our approach uses branches of \kl[sub-run]{sub} \kl{SP} \kl{runs}.
Recently by extending their approach, the elementary decidability result of ICPDL has been extended to CPDL+ \cite{figueiraPDLSteroidsExpressive2023}.
Note that we can naturally reduce the \kl{equational theory} of \kl{KL terms} (with respect to relations) into the theory of IPDL formulas because $\REL \models \term[1] = \term[2]$ iff the IPDL formula $\langle\term[1]\rangle p \leftrightarrow \langle\term[2]\rangle p$ is valid (as with \cite[p.\ 209]{fischerPropositionalDynamicLogic1979} for PDL) where $p$ is a propositional variable (disjoint from \kl{variables}) and $\langle - \rangle$ denotes the diamond modality.
Thus, the 2EXPTIME upper bound of the \kl{equational theory} of \kl{KL terms} can be obtained via this reduction.
To obtain the EXPSPACE upper bound using this approach, we need to give a reduction from the theory of \kl{KL terms} over linearly bounded \kl{pathwidth} \kl{structures} into that over words,
\cf\  the \kl{inclusion problem} is in $\mathrm{PSPACE}$ for (\emph{word}) \kl{2AFAs} (\Cref{proposition: 2AFA PSPACE}) whereas the \kl{inclusion problem} is $\mathrm{EXPTIME}$-hard for \emph{tree} (two-way alternating) automata \cite[Theorem 1.7.7]{comonTreeAutomataTechniques2007}.
This approach would be possible by modifying the encoding of tree decompositions but is not immediate,
because the reduction \cite{lutzPDLIntersectionConverse2005,gollerPDLIntersectionConverse2009} always generates trees (not words) even when a given \kl{tree decomposition} is a \kl{path decomposition}.
Conversely, it would also be possible to extend our decomposition approach to the extensions of PDL above, but we leave it as a future work.

\paragraph*{Regular queries}
Regular queries \cite{reutterRegularQueriesGraph2017} are binary non-recursive positive Datalog programs extended with \kl{Kleene star}.
The containment problem for regular queries is decidable and 2EXPSPACE-complete \cite{reutterRegularQueriesGraph2017}.
Because we can naturally translate \kl{KL terms} (with respect to relations) and also \kl{PCoR* terms} into regular queries, we can also reduce the \kl{equational theory} of \kl{PCoR* terms} into the containment problem for regular queries.

\paragraph*{The existential calculus of relations with transitive closure}
The existential calculus of relations with transitive closure (ECoR*) \cite{nakamuraExistentialCalculiRelations2023} is PCoR* with variable and constant complements (i.e., the complement operator only applies to variables or constants).
ECoR* has the same expressive power as $3$-variable existential first-order logic with variable-confined monodic transitive closure \cite{nakamuraExpressivePowerSuccinctness2020,nakamuraExpressivePowerSuccinctness2022}.
This extension is also related to boolean modal logic \cite{gargovNoteBooleanModal1990} and (I)PDL with negation of atomic programs \cite{lutzPDLNegationAtomic2005,gollerPDLIntersectionConverse2009}.
The \kl{equational theory} of PCoR* with variable complements is $\Pi^{0}_{1}$-complete \cite{nakamuraExistentialCalculiRelations2023} (an open problem left in the conference version \cite{nakamuraPartialDerivativesGraphs2017}) and
the \kl{equational theory} of PCoR* with constant complements (i.e., with the difference constant) is $\Pi^{0}_{1}$-complete \cite{nakamuraUndecidabilityPositiveCalculus2024}.
Nevertheless, the \kl{equational theory} is decidable for the intersection-free fragment of ECoR* \cite{nakamuraExistentialCalculiRelations2023}.

\paragraph*{On Kleene algebra with loop}
The (graph) loop operator $\term^{\lop}$ \cite{daneckiPropositionalDynamicLogic1984} is the operator defined by: $\term^{\lop} \defeq \term \intersection \id$.
We leave it open whether the \kl{equational theory} of Kleene algebras with loop (with respect to relations) is PSPACE-complete.
Related to this, PDL with loop (loop-PDL) \cite{daneckiPropositionalDynamicLogic1984,gollerPDLIntersectionConverse2009} is EXPTIME-complete, whereas PDL with intersection is 2EXPSPACE-complete.
For loop-PDL, the formulas can be translated into ICPDL formulas of intersection width at most $2$, so the EXPTIME upper bound of loop-PDL is a corollary of the fact that the bounded intersection width fragment of ICPDL is in EXPTIME \cite[Corollary 4.9]{gollerPDLIntersectionConverse2009}.
However, this translation essentially needs the diamond modality as a primitive, so \Cref{corollary: PCoR* with tests v iw fixed PSPACE-complete} does not imply the PSPACE decidability, immediately;
the closure size (with respect to the closure function of \Cref{definition: closure}) is still exponential because the number of \kl{labels} occurring in Kleene algebra terms with loop is not bounded.
A related problem is posed by Sedl{\'a}r \cite[p.\ 16]{sedlarKleeneAlgebraDynamic2023} for the complexity of the \kl{equational theory} of Kleene algebra with the domain operator (dKA) with respect to relations.
Here, the domain operator $\mathsf{d}(\term)$ is the operator $\mathsf{d}(\term) \defeq (\term \top) \intersection \id$.
Note that, because $\REL \models \mathsf{d}(\term) = (\term \top)^{\lop}$, we can reduce the \kl{equational theory} of dKA into that of Kleene algebras with loop and top.

\paragraph*{On axiomatization}
Another interesting question is the axiomatizability.
Doumane and Pous presented a finite axiomatization for the \kl{equational theory} of \kl{identity-free Kleene lattices} \cite{doumaneCompletenessIdentityfreeKleene2018}.
However, the finite axiomatizability is still an open problem for the \kl{equational theory} of \kl{Kleene lattices} (with additional operators) with respect to relations \cite[p.\ 15]{doumaneCompletenessIdentityfreeKleene2018}\cite{pousPositiveCalculusRelations2018}.

\section*{Acknowledgment}
We would like to thank anonymous referees for useful comments.
This work was supported by JSPS KAKENHI Grant Numbers JP16J08119, JP21K13828, JP25K14985.

\bibliographystyle{alphaurl}
\bibliography{main}

\appendix

\section{Proof of {\Cref{proposition: decomposition and 2AFA}}}\label{section: proposition: decomposition and 2AFA}
\begin{proof}
    Both directions are shown by easy induction on the derivative trees, respectively.

    ($\Longrightarrow$):
    By induction on the derivation tree.
    \begin{itemize}
        \item 
        Case \begin{prooftree}
            \hypo{\lterm[1] \longrightarrow^{(\struc_{j})_{\bullet}} \lterm[2]}
            \infer1[D]{\lterm[1]_{(j')} \mathrel{(\odot_{i = 1}^{n} \longrightarrow^{(\struc_i)_{\bullet}})} \lterm[2]_{(j')}}
        \end{prooftree}
        where $j'$ satisfies
        $\lterm[1]_{(j')} = \lterm[1]_{(j)}$ and
        $\lterm[2]_{(j')} = \lterm[2]_{(j)}$:
        Then,
        \[\begin{prooftree}
            \hypo{\mathstrut}
            \infer1[D]{\tuple{\tuple{\lterm[1], \lterm[2]}_{\checkmark}, j'} \in S^{{\automaton_{k}^{\term[3]}}}_{\triangleright \struc_1 \dots \struc_{n} \triangleleft}}
            \infer[double]1[$-1$, $+1$, $\checkmark$]{\tuple{\tuple{\lterm[1], \lterm[2]}_{\checkmark}, j} \in S^{{\automaton_{k}^{\term[3]}}}_{\triangleright \struc_1 \dots \struc_{n} \triangleleft}}
        \end{prooftree}.\]
        \item
        Case \begin{prooftree}
            \hypo{\lterm[1]_{(j')} \mathrel{(\odot_{i = 1}^{n} \longrightarrow^{(\struc_i)_{\bullet}})} \lterm[3]'_{(j')}}
            \hypo{\lterm[3]'_{(j')} \mathrel{(\odot_{i = 1}^{n} \longrightarrow^{(\struc_i)_{\bullet}})} \lterm[2]_{(j')}}
            \infer2[T]{\lterm[1]_{(j')} \mathrel{(\odot_{i = 1}^{n} \longrightarrow^{(\struc_i)_{\bullet}})} \lterm[2]_{(j')}}
        \end{prooftree} where $j'$ satisfies
        $\lterm[1]_{(j')} = \lterm[1]_{(j)}$ and
        $\lterm[2]_{(j')} = \lterm[2]_{(j)}$:
        Then we have:
        \[\begin{prooftree}
            \hypo{\mathstrut}
            \infer[double]1[(IH)]{\tuple{\tuple{\lterm[1], \lterm[3]'}_{\checkmark}, j'} \in S^{{\automaton_{k}^{\term[3]}}}_{\triangleright \struc_1 \dots \struc_{n} \triangleleft}}
            \infer1[$\checkmark$]{\tuple{\tuple{\lterm[1], \lterm[3]'}_{?}, j'} \in S^{{\automaton_{k}^{\term[3]}}}_{\triangleright \struc_1 \dots \struc_{n} \triangleleft}}
            \hypo{\mathstrut}
            \infer[double]1[(IH)]{\tuple{\tuple{\lterm[3]', \lterm[2]}_{\checkmark}, j'} \in S^{{\automaton_{k}^{\term[3]}}}_{\triangleright \struc_1 \dots \struc_{n} \triangleleft}}
            \infer1[$\checkmark$]{\tuple{\tuple{\lterm[3]', \lterm[2]}_{?}, j'} \in S^{{\automaton_{k}^{\term[3]}}}_{\triangleright \struc_1 \dots \struc_{n} \triangleleft}}
            \infer2[T]{\tuple{\tuple{\lterm[1], \lterm[2]}_{\checkmark}, j'} \in S^{{\automaton_{k}^{\term[3]}}}_{\triangleright \struc_1 \dots \struc_{n} \triangleleft}}
            \infer[double]1[$-1$, $+1$, $\checkmark$]{\tuple{\tuple{\lterm[1], \lterm[2]}_{\checkmark}, j} \in S^{{\automaton_{k}^{\term[3]}}}_{\triangleright \struc_1 \dots \struc_{n} \triangleleft}}
        \end{prooftree}.\]
        \item
        Case \begin{prooftree}
            \hypo{\lterm[1][\bullet/l]_{(j')} \mathrel{(\odot_{i = 1}^{n} \longrightarrow^{(\struc_i)_{\bullet}})} \lterm[2][\bullet/r]_{(j')}}
            \infer1[L]{\lterm[1][\lab/l]_{(j')} \mathrel{(\odot_{i = 1}^{n} \longrightarrow^{(\struc_i)_{\bullet}})} \lterm[2][\lab/r]_{(j')}}
        \end{prooftree} where $j'$ satisfies that
        $\lterm[1][\lab/l]_{(j')} = \lterm[1][\lab/l]_{(j)}$ and
        $\lterm[2][\lab/r]_{(j')} = \lterm[2][\lab/r]_{(j)}$:
        Then we have:
        \[\begin{prooftree}
            \hypo{\mathstrut}
            \infer[double]1[(IH)]{\tuple{\tuple{\lterm[1][\bullet/l], \lterm[2][\bullet/r]}_{\checkmark}, j'} \in S^{{\automaton_{k}^{\term[3]}}}_{\triangleright \struc_1 \dots \struc_{n} \triangleleft}}
            \infer1[$\checkmark$]{\tuple{\tuple{\lterm[1][\bullet/l], \lterm[2][\bullet/r]}_{?}, j'} \in S^{{\automaton_{k}^{\term[3]}}}_{\triangleright \struc_1 \dots \struc_{n} \triangleleft}}
            \infer1[L]{\tuple{\tuple{\lterm[1][\lab/l], \lterm[2][\lab/r]}_{\checkmark}, j'} \in S^{{\automaton_{k}^{\term[3]}}}_{\triangleright \struc_1 \dots \struc_{n} \triangleleft}}
            \infer[double]1[$-1$, $+1$, $\checkmark$]{\tuple{\tuple{\lterm[1][\lab/l], \lterm[2][\lab/r]}_{\checkmark}, j} \in S^{{\automaton_{k}^{\term[3]}}}_{\triangleright \struc_1 \dots \struc_{n} \triangleleft}}
        \end{prooftree}.\]
    \end{itemize}\noindent 

    ($\Longleftarrow$):
    By induction on the derivation tree.
    \begin{itemize}
        \item
        Case \begin{prooftree}
            \hypo{\mathstrut}
            \infer1[D]{\tuple{\tuple{\lterm[1], \lterm[2]}_{\checkmark}, j} \in S^{{\automaton_{k}^{\term[3]}}}_{\triangleright \struc_1 \dots \struc_{n} \triangleleft}}
        \end{prooftree} where $\lterm[1] \longrightarrow^{\struc} \lterm[2]$:
        Then, $\begin{prooftree}
                \hypo{\lterm[1] \longrightarrow^{(\struc_{j})_{\bullet}} \lterm[2]}
                \infer1[D]{\lterm[1]_{(j)} \mathrel{(\odot_{i = 1}^{n} \longrightarrow^{(\struc_i)_{\bullet}})} \lterm[2]_{(j)}}
            \end{prooftree}$.
        \vskip0.5\baselineskip 
        \item
        Case 
        \begin{prooftree}
            \hypo{\tuple{\tuple{\lterm[1], \lterm[2]}_{\checkmark}, j-1} \in S^{{\automaton_{k}^{\term[3]}}}_{\triangleright \struc_1 \dots \struc_{n} \triangleleft}}
            \infer1[$\checkmark$]{\tuple{\tuple{\lterm[1], \lterm[2]}_{?}, j-1} \in S^{{\automaton_{k}^{\term[3]}}}_{\triangleright \struc_1 \dots \struc_{n} \triangleleft}}
            \infer1[$-1$]{\tuple{\tuple{\lterm[1], \lterm[2]}_{\checkmark}, j} \in S^{{\automaton_{k}^{\term[3]}}}_{\triangleright \struc_1 \dots \struc_{n} \triangleleft}}
        \end{prooftree}:
        (Note that $\checkmark$ is the only rule that can apply to $\tuple{\bl,\bl}_{?}$.)
        We then have:
        \[\begin{prooftree}
            \hypo{\mathstrut}
            \infer[double]1[(IH)]{\lterm[1]_{(j-1)} \mathrel{(\odot_{i = 1}^{n} \longrightarrow^{(\struc_i)_{\bullet}})} \lterm[2]_{(j-1)}}
         \end{prooftree}.\]
        By
        $\overrightarrow{\mathsf{lab}}(\lterm[1]), \overrightarrow{\mathsf{lab}}(\lterm[2]) \in \univ{(\struc_{j-1})_{\bullet}}^{+} \cap \univ{(\struc_{j})_{\bullet}}^{+}$,
        we have $\lterm[1]_{(j-1)} = \lterm[1]_{(j)}$ and $\lterm[2]_{(j-1)} = \lterm[2]_{(j)}$, and thus this case has been proved.

        \item 
        Case
        \begin{prooftree}
            \hypo{\tuple{\tuple{\lterm[1], \lterm[2]}_{?}, j+1} \in S^{{\automaton_{k}^{\term[3]}}}_{\triangleright \struc_1 \dots \struc_{n} \triangleleft}}
            \infer1[$+1$]{\tuple{\tuple{\lterm[1], \lterm[2]}_{\checkmark}, j} \in S^{{\automaton_{k}^{\term[3]}}}_{\triangleright \struc_1 \dots \struc_{n} \triangleleft}}
        \end{prooftree}:
        In the same way as the case above.
        \vskip0.5\baselineskip 
        \item
        Case \begin{prooftree}
            \hypo{\tuple{\tuple{\lterm[1], \lterm[3]'}_{\checkmark}, j}  \in S^{{\automaton_{k}^{\term[3]}}}_{\triangleright \struc_1 \dots \struc_{n} \triangleleft}}
            \infer1[$\checkmark$]{\tuple{\tuple{\lterm[1], \lterm[3]'}_{?}, j}  \in S^{{\automaton_{k}^{\term[3]}}}_{\triangleright \struc_1 \dots \struc_{n} \triangleleft}}
            \hypo{\tuple{\tuple{\lterm[3]', \lterm[2]}_{\checkmark}, j}  \in S^{{\automaton_{k}^{\term[3]}}}_{\triangleright \struc_1 \dots \struc_{n} \triangleleft}}    
            \infer1[$\checkmark$]{\tuple{\tuple{\lterm[3]', \lterm[2]}_{?}, j}  \in S^{{\automaton_{k}^{\term[3]}}}_{\triangleright \struc_1 \dots \struc_{n} \triangleleft}}    
            \infer2[T]{\tuple{\tuple{\lterm[1], \lterm[2]}_{\checkmark}, j}  \in S^{{\automaton_{k}^{\term[3]}}}_{\triangleright \struc_1 \dots \struc_{n} \triangleleft}}
        \end{prooftree}:
        Then we have:
        \[\begin{prooftree}
            \hypo{\mathstrut}
            \infer[double]1[(IH)]{\lterm[1]_{(j)} \mathrel{(\odot_{i = 1}^{n} \longrightarrow^{(\struc_i)_{\bullet}})} \lterm[3]'_{(j)}}
            \hypo{\mathstrut}
            \infer[double]1[(IH)]{\lterm[3]'_{(j)} \mathrel{(\odot_{i = 1}^{n} \longrightarrow^{(\struc_i)_{\bullet}})} \lterm[2]_{(j)}}
            \infer2[T]{\lterm[1]_{(j)} \mathrel{(\odot_{i = 1}^{n} \longrightarrow^{(\struc_i)_{\bullet}})} \lterm[2]_{(j)}}
         \end{prooftree}.\]
         \vskip0.5\baselineskip 
        \item 
        Case 
        \begin{prooftree}
            \hypo{\tuple{\tuple{\lterm[1][\bullet/l], \lterm[2][\bullet/r]}_{\checkmark}, j} \in S^{{\automaton_{k}^{\term[3]}}}_{\triangleright \struc_1 \dots \struc_{n} \triangleleft}}
            \infer1[$\checkmark$]{\tuple{\tuple{\lterm[1][\bullet/l], \lterm[2][\bullet/r]}_{?}, j} \in S^{{\automaton_{k}^{\term[3]}}}_{\triangleright \struc_1 \dots \struc_{n} \triangleleft}}
            \infer1[L]{\tuple{\tuple{\lterm[1][x/l], \lterm[2][x/r]}_{\checkmark}, j} \in S^{{\automaton_{k}^{\term[3]}}}_{\triangleright \struc_1 \dots \struc_{n} \triangleleft}}
        \end{prooftree}:
        Then,
        $\begin{prooftree}
            \hypo{\mathstrut}
            \infer[double]1[(IH)]{\lterm[1][\bullet/l]_{(j)} \mathrel{(\odot_{i = 1}^{n} \longrightarrow^{(\struc_i)_{\bullet}})} \lterm[2][\bullet/r]_{(j)}}
            \infer1[L]{\lterm[1][x/l]_{(j)} \mathrel{(\odot_{i = 1}^{n} \longrightarrow^{(\struc_i)_{\bullet}})} \lterm[2][x/r]_{(j)}}
         \end{prooftree}$.
    \end{itemize}
    Hence, this completes the proof.
\end{proof}

\section{Proof of {\Cref{proposition: subSPR equiv}}}\label{section: proposition: subSPR equiv}
For a \kl{run} $\graph$, we let
\begin{align*}
    \subl(\graph) &\;\defeq\; \set{\graph[1]' \mid \mbox{there is a \kl{run} $\graph[2]$ such that $\graph[1] = \graph[2] \series \graph[1]'$}},\\
    \subr(\graph) &\;\defeq\; \set{\graph[1]' \mid \mbox{there is a \kl{run} $\graph[3]$ such that $\graph[1] = \graph[1]' \series \graph[3]$}},\\
    \sub(\graph) &\;\defeq\; \set{\graph[1]' \mid \mbox{there are \kl{runs} $\graph[2]$ and $\graph[3]$ such that $\graph[1] = \graph[2] \series \graph[1]' \series \graph[3]$}}.
\end{align*}
\begin{proof}[Proof of $\subSPRl = \subSPRl'$]
    ($\subSPRl' \subseteq \subSPRl$):
    By induction on the derivation tree of $\graph[1] \in \subSPRl'$,
    it suffices to show:
    for every $\graph[1] \in \subSPRl'$, there is some $\graph[1]' \in \subSPRr'$ such that $\graph[1]' \series \graph[1] \in \SPR$.
    We distinguish the following cases:
    \begin{itemize}
        \item Case $\graph[1] \in \SPR$:
        By letting $\graph[1]' = \id^{1}$, we have $\graph[1]' \series \graph[1] = \graph[1] \in \SPR$.
        \item Case $\graph[1] = (\graph[1]^{\mathrm{l}} \parallel \graph[2]^{\mathrm{l}}) \diamond \join^{1}_{1} \diamond \graph[3]$
        where $\graph[1]^{\mathrm{l}}, \graph[2]^{\mathrm{l}} \in \subSPRl'$ and $\graph[3] \in \SPR$:
        Let $\graph[1]^{\mathrm{r}}, \graph[2]^{\mathrm{r}} \in \subSPRr'$ be the ones obtained by IH with respect to $\graph[1]^{\mathrm{l}}$ and $\graph[2]^{\mathrm{l}}$, respectively.
        Then by letting $\graph[1]' = \fork^{1}_{1} \series (\graph[1]^{\mathrm{r}} \parallel \graph[2]^{\mathrm{r}})$,
        we have $\graph[1]' \series \graph[1] = (\fork^{1}_{1} \series ((\graph[1]^{\mathrm{r}} \series \graph[1]^{\mathrm{l}}) \parallel (\graph[2]^{\mathrm{r}} \series \graph[2]^{\mathrm{l}})) \series \join^{1}_{1}) \series \graph[3] \in \SPR$.
    \end{itemize}\noindent 
    Hence, this completes the proof.
    
    ($\subSPRl \subseteq \subSPRl'$):
    By induction on the derivation tree of $\graph_0 \in \SPR$,
    we show: for every $\graph_0 \in \SPR$ and $\graph \in \subl(\graph_0)$, we have $\graph \in \subSPRl'$.
    We distinguish the following cases:
    \begin{itemize}
        \item Case $\graph_0 = \id^{1}$, Case $\graph_0 = a^{1}_{1}$:
        Clear by $\subl(\graph_0) \subseteq \SPR \subseteq \subSPRl'$.
        \item Case $\graph_0 = \graph[2]_0 \series \graph[3]_0$ where $\graph[2]_0, \graph[3]_0 \in \SPR$:
         We distinguish the following cases:
         \begin{itemize}
            \item Case $\graph = \graph[2] \series \graph[3]_0$ where $\graph[2] \in \subl(\graph[2]_0)$:
            By IH, we have $\graph[2] \in \subSPRl'$.
            Then (1) $\graph[2] = \graph[3]'$ or (2) $\graph[2] = (\graph[1]^{\mathrm{l}} \parallel \graph[2]^{\mathrm{l}}) \series \join^{1}_{1} \series \graph[3]'$ holds where $\graph[1]^{\mathrm{l}}, \graph[2]^{\mathrm{l}} \in \subSPRl'$, $\graph[3]' \in \SPR$.
            For (1),
            we have $\graph[2] \series \graph[3]_0 = \graph[3]' \series \graph[3]_0 \in \SPR \subseteq \subSPRl'$.
            For (2),
            we have $\graph[2] \series \graph[3]_0 = (\graph[1]^{\mathrm{l}} \parallel \graph[2]^{\mathrm{l}}) \series \join^{1}_{1} \series (\graph[3]' \series \graph[3]_0) \in \subSPRl'$.
            Hence, in both cases, we have $\graph[2] \series \graph[3]_0 \in \subSPRl'$.
            \item Case $\graph \in \subl(\graph[3]_0)$: By IH.
        \end{itemize}

        \item Case $\graph_0 = \fork^{1}_{1} \series (\graph[2]_0 \parallel \graph[3]_0) \series \join^{1}_{1}$ where $\graph[2]_0, \graph[3]_0 \in \SPR$:
        We distinguish the following cases:
        \begin{itemize}
            \item Case $\graph = \graph_0$:
            Clear by $\graph_0 \in \SPR \subseteq \subSPRl'$.
            \item Case $\graph = (\graph[2] \parallel \graph[3]) \series \join^{1}_{1}$ where $\graph[2] \in \subl(\graph[2]_0)$ and $\graph[3] \in \subl(\graph[3]_0)$:
            By IH.
            \item Case $\graph = \id^{1}$:
            Clear.
        \end{itemize}
    \end{itemize}\noindent 
    Hence, this completes the proof.
\end{proof}

\begin{proof}[Proof of $\subSPRr = \subSPRr'$]
    Similar to the proof of $\subSPRl = \subSPRl'$.
\end{proof}

\begin{proof}[Proof of $\subSPR = \subSPR'$]
    ($\subSPR' \subseteq \subSPR$):
    By the proof of ($\subSPRl' \subseteq \subSPRl$), we have:
    \begin{itemize}
        \item for every $\graph[1] \in \subSPRl'$, there is some $\graph[1]' \in \subSPRr'$ such that $\graph[1]' \series \graph[1] \in \SPR$.
    \end{itemize}
    Similarly, we have:
    \begin{itemize}
        \item for every $\graph[1] \in \subSPRr'$, there is some $\graph[1]' \in \subSPRl'$ such that $\graph[1] \series \graph[1]' \in \SPR$.
    \end{itemize}
    By induction on the derivation tree of $\graph[1] \in \subSPR'$,
    we show:
    for every $\graph[1] \in \subSPR'$,
    there are some $\graph[2] \in \subSPRl'$ and $\graph[3] \in \subSPRr'$ such that $\graph[2] \series \graph[1] \series \graph[3] \in \SPR$.
    \begin{itemize}
        \item Case $\graph[1] = \graph[2]^{\mathrm{l}} \series \graph[3]^{\mathrm{r}}$
        where $\graph[2]^{\mathrm{l}} \in \subSPRl'$ and  $\graph[3]^{\mathrm{r}} \in \subSPRr'$:
        Let $\graph[2]^{\mathrm{r}} \in \subSPRr'$ be such that $\graph[2]^{\mathrm{r}} \series \graph[2]^{\mathrm{l}} \in \SPR$
        and let $\graph[3]^{\mathrm{l}} \in \subSPRl$ be such that $\graph[3]^{\mathrm{r}} \series \graph[3]^{\mathrm{l}} \in \SPR$.
        We then have $\graph[2]^{\mathrm{r}} \series (\graph[2]^{\mathrm{l}} \series \graph[3]^{\mathrm{r}}) \series \graph[3]^{\mathrm{l}} \in \SPR$.
        \item Case $\graph[1] = (\graph[1]_1 \parallel \graph[1]_2)$ where $\graph[1]_1, \graph[1]_2 \in \subSPR'$:
        By IH, there are some $\graph[2]_k \in \subSPRr', \graph[3]_k \in \subSPRl'$ such that $\graph[2]_k \series \graph[1]_k \series \graph[3]_k \in \SPR$ for each $k \in \range{1, 2}$.
        We then have
        \begin{align*}
            & (\fork^{1}_{1} \series (\graph[2]_1 \parallel \graph[2]_2))
            \series (\graph[1]_1 \parallel \graph[1]_2)
            \series ((\graph[3]_1 \parallel \graph[3]_2) \series \join^{1}_{1})\\
            &= 
            \fork^{1}_{1} \series (
                (\graph[2]_1 \series \graph[1]_1 \series \graph[3]_1)
                \parallel
                (\graph[2]_2 \series \graph[1]_2 \series \graph[3]_2)
            ) \series \join^{1}_{1} \in \SPR.
        \end{align*}
    \end{itemize}
    Hence, this completes the proof.

    ($\subSPR \subseteq \subSPR'$):
    By induction on the derivation tree of $\graph_0 \in \SPR$,
    we show:
    for every $\graph_0 \in \SPR$ and $\graph \in \sub(\graph_0)$, we have $\graph \in \subSPR'$.
    We distinguish the following cases:
    \begin{itemize}
        \item Case $\graph_0 = \id^{1}$, Case $\graph_0 = a^{1}_{1}$: Clear by $\sub(\graph_0) \subseteq \SPR \subseteq \subSPR'$.
        \item Case $\graph_0 = \graph[2]_0 \series \graph[3]_0$ where $\graph[2]_0, \graph[3]_0 \in \SPR$:
         We distinguish the following cases:
         \begin{itemize}
            \item Case $\graph \in \sub(\graph[2]_0)$ or $\graph \in \sub(\graph[3]_0)$: By IH.
            \item Case $\graph = \graph[2] \series \graph[3]$ where
            $\graph[2] \in \subl(\graph[2]_0)$ and $\graph[3] \in \subr(\graph[3]_0)$.
            From the above, we also have $\graph[2] \in \subSPRl', \graph[3] \in \subSPRr'$.
            Thus, $\graph[2] \series \graph[3] \in \subSPR'$.
        \end{itemize}

        \item Case $\graph_0 = \fork^{1}_{1} \series (\graph[2]_0 \parallel \graph[3]_0) \series \join^{1}_{1}$ where $\graph[2]_0, \graph[3]_0 \in \SPR$:
        We distinguish the following cases:
        \begin{itemize}
            \item Case $\graph = \graph_0$:
            Clear by $\graph_0 \in \SPR \subseteq \subSPR'$.
            \item Case $\graph = (\graph[2] \parallel \graph[3]) \series \join^{1}_{1}$ where $\graph[2] \in \subl(\graph[2]_0)$ and $\graph[3] \in \subl(\graph[3]_0)$:
            From the above, we also have $\graph[2] \in \subSPRl', \graph[3] \in \subSPRl'$.
            Thus, $\graph \in \subSPRl' \subseteq \subSPR'$.
   
            \item Case $\graph = \fork^{1}_{1} \series (\graph[2] \parallel \graph[3])$ where $\graph[2] \in \subr(\graph[2]_0)$ and $\graph[3] \in \subr(\graph[3]_0)$:
            Similarly.
            \item Case $\graph = (\graph[2] \parallel \graph[3])$ where $\graph[2] \in \sub(\graph[2]_0)$ and $\graph[3] \in \sub(\graph[3]_0)$:
            By IH, $\graph[2], \graph[3] \in \subSPR'$.
            We thus have $\graph \in \subSPR'$.
            \item Case $\graph = \id^{1}$:
            Clear.
        \end{itemize}
    \end{itemize}
    Hence, this completes the proof.
\end{proof}

\section{Proof of {\Cref{proposition: transitivity elimination}}}\label{section: proposition: transitivity elimination}

\begin{proof}
    ($\supseteq$):
    We show that each rule of $\longleadsto_{\trace}$ is admissible in $\longrightarrow_{\trace}$.
    \begin{itemize}
        \item Case  
        \begin{prooftree}
            \hypo{\tuple{\lab[1], \lab[2]} \in a^{\struc}}
            \infer[double]1{@\lab. a \longrightarrow^{\struc}_{a^{x, y}_{1}} @\lab[2].\id}
        \end{prooftree} \mbox{ for $a \in \vsig$},
        Case  
        \begin{prooftree}
            \hypo{@\lab.\term[1] \longrightarrow^{\struc}_{\trace[3]} \lterm[1]'}
            \infer[double]1{@\lab.\term[1] \union \term[2] \longrightarrow^{\struc}_{\trace[3]} \lterm[1]'}
        \end{prooftree},
        Case 
        \begin{prooftree}
            \hypo{@\lab.\term[2] \longrightarrow^{\struc}_{\trace[3]} \lterm[2]'}
            \infer[double]1{@\lab.\term[1] \union \term[2] \longrightarrow^{\struc}_{\trace[3]} \lterm[2]'}
        \end{prooftree},
        Case \begin{prooftree}
            \hypo{\mathstrut}
            \infer[double]1{\lterm[1] \longrightarrow^{\struc}_{\id^{\overrightarrow{\mathsf{lab}}(\lterm[1])}} \lterm[1]}
        \end{prooftree},
        Case  
        \begin{prooftree}
            \hypo{@\lab.\term[1] \longrightarrow^{\struc}_{\trace[3]} \lterm[1]'}
            \infer[double]1{@\lab.\term[1] \compo \term[2] \longrightarrow^{\struc}_{\trace[3]} \lterm[1]' \compo_{1} \term[2]}
        \end{prooftree},
        Case                  
        \begin{prooftree}
            \hypo{\lterm[1]  \longrightarrow^{\struc}_{\trace[3]} \lterm[2]}
            \infer[double]1{\lterm[1] \compo_{1} \term[3] \longrightarrow^{\struc}_{\trace[3]} \lterm[2] \compo_{1} \term[3]}
        \end{prooftree}:
        Trivial.

        \item Case 
        \begin{prooftree}
            \hypo{@\lab[1].\term[1] \longrightarrow^{\struc}_{\trace[3]} \lterm[1]'}
            \hypo{\EPS_{\lab[3]}(\lterm[1]')}
            \hypo{@\lab[3].\term[2] \longrightarrow^{\struc}_{\trace[3]'} \lterm[2]'}
            \infer[double]3{@\lab[1].\term[1] \compo \term[2] \longrightarrow^{\struc}_{\trace[3] \series \trace[3]'} \lterm[2]'}
        \end{prooftree}:
        By 
        \begin{prooftree}
            \hypo{@\lab[1].\term[1] \longrightarrow^{\struc}_{\trace[3]} \lterm[1]'}
            \hypo{\EPS_{\lab[3]}(\lterm[1]')}
            \hypo{@\lab[3].\term[2] \longrightarrow^{\struc}_{\trace[3]'} \lterm[2]'}
            \infer2{\lterm[1]' \longrightarrow^{\struc}_{\trace[3]'} \lterm[2]'}
            \infer2[T]{@\lab[1].\term[1] \compo \term[2] \longrightarrow^{\struc}_{\trace[3] \series \trace[3]'} \lterm[2]'}
        \end{prooftree}.
        \item Case 
        \begin{prooftree}
            \hypo{\lterm[1] \longrightarrow^{\struc}_{\trace[3]} \lterm[1]'}
            \hypo{\EPS_{\lab[3]}(\lterm[1]')}
            \hypo{@ \lab[3]. \term[2] \longrightarrow^{\struc}_{\trace[3]'} \lterm[2]}
            \infer[double]3{\lterm[1] \compo_{1} \term[2] \longrightarrow^{\struc}_{\trace[3] \series \trace[3]'} \lterm[2]}
        \end{prooftree}:
        Similar to the case above.

        \item Case  
        \begin{prooftree}[separation = .7em]
            \hypo{@\lab[1]_0.\term[1] \longrightarrow^{\struc}_{\trace[3]_0} \lterm[1]_1'}
            \hypo{\EPS_{\lab[1]_1}(\lterm[1]_1')}
            \hypo{\dots}
            \hypo{@\lab_{n-1}.\term[1] \longrightarrow^{\struc}_{\trace[3]_{n-1}} \lterm[1]_n'}
            \hypo{\EPS_{\lab[1]_n}(\lterm[1]_n')}
            \hypo{@\lab[1]_{n}.\term[1] \longrightarrow^{\struc}_{\trace[3]_{n}} \lterm[1]_{n+1}'}
            \infer[double]6{ @\lab[1]_0.\term[1]^* \longrightarrow^{\struc}_{\trace[3]_0 \series \dots \series \trace[3]_n} \lterm[1]_{n+1}' \compo_{1} \term[1]^*}
        \end{prooftree}:

        By the following derivation tree:

        \begin{prooftree}
            \hypo{@\lab[1]_0.\term[1] \longrightarrow^{\struc}_{\trace[3]_0} \lterm[1]_1'}
            \infer1{@\lab[1]_0.\term[1]^* \longrightarrow^{\struc}_{\trace[3]_0} \lterm[1]_1' \compo_1 \term[1]^*}
            \hypo{\dots}
            \hypo{\EPS_{\lab[1]_n}(\lterm[1]_n')}
            \hypo{@\lab[1]_{n}.\term[1] \longrightarrow^{\struc}_{\trace[3]_{n}} \lterm[1]_{n+1}'}
            \infer1{@\lab[1]_{n}.\term[1]^* \longrightarrow^{\struc}_{\trace[3]_{n}} \lterm[1]_{n+1}' \compo_1 \term[1]^*}
            \infer2{\lterm[1]_n' \compo_{1} \term[1]^* \longrightarrow^{\struc}_{\trace[3]_{n}} \lterm[1]_{n+1}' \compo_1 \term[1]^*}
            \infer[double]3[T]{ @\lab[1]_0.\term[1]^* \longrightarrow^{\struc}_{\trace[3]_0 \series \dots \series \trace[3]_n} \lterm[1]_{n+1}' \compo_{1} \term[1]^*}
        \end{prooftree}
        \vskip0.5\baselineskip 
        \item Case                          
        \begin{prooftree}
            \hypo{@\lab.\term[1] \longrightarrow^{\struc}_{\trace[3]} \lterm[1]' \; \EPS_{\lab[3]}(\lterm[1]')}
            \hypo{@\lab.\term[2] \longrightarrow^{\struc}_{\trace[3]'} \lterm[2]' \; \EPS_{\lab[3]}(\lterm[2]')}
            \infer[double]2{@\lab.\term[1] \intersection \term[2] \longrightarrow^{\struc}_{\fork^{x}_{1} \series (\trace[3] \parallel \trace[3]') \series \join^{\lab[3]}_{1}}  @\lab[3].\id}
        \end{prooftree}:
        We have:
        \begin{gather*}
            \begin{prooftree}
                \hypo{\mathstrut}
                \infer1{@\lab.\term[1] \intersection \term[2] \longrightarrow^{\struc}_{\fork^{x}_{1}} (@\lab.\term[1]) \intersection_{1} (@\lab. \term[2])}
            \end{prooftree},
            \hspace{3em}
            \begin{prooftree}
                \hypo{@\lab.\term[1] \longrightarrow^{\struc}_{\trace[3]} \lterm[1]'}
                \infer1{(@\lab.\term[1]) \intersection_{1} (@\lab. \term[2]) \longrightarrow^{\struc}_{\trace[3]} \lterm[1]' \intersection_{1} (@\lab. \term[2])}
            \end{prooftree},
            \\
            \begin{prooftree}
                \hypo{@\lab.\term[2] \longrightarrow^{\struc}_{\trace[3]'} \lterm[2]'}
                \infer1{\lterm[1]' \intersection_{1} (@\lab. \term[2]) \longrightarrow^{\struc}_{\trace[3]'} \lterm[1]' \intersection_{1} \lterm[2]'}
            \end{prooftree}, \mbox{ and }
            \hspace{3em}
            \begin{prooftree}
                \hypo{\EPS_{\lab[3]}(\lterm[1]')}
                \hypo{\EPS_{\lab[3]}(\lterm[2]')}
                \infer2{\lterm[1]' \intersection_{1} \lterm[2]' \longrightarrow^{\struc}_{\join^{\lab[3]}_{1}} @\lab[3]. \id}
            \end{prooftree}.
        \end{gather*}
        Hence by (T), we have obtained $@\lab.\term[1] \intersection \term[2] \longrightarrow^{\struc}_{\fork^{x}_{1} \series (\trace[3] \parallel \trace[3]') \series \join^{\lab[3]}_{1}}  @\lab[3].\id$.
        \vskip0.5\baselineskip 
        \item Case                 
        \begin{prooftree}
            \hypo{@\lab.\term[1] \longrightarrow^{\struc}_{\trace[3]} \lterm[1]'}
            \hypo{@\lab.\term[2] \longrightarrow^{\struc}_{\trace[3]'} \lterm[2]'}
            \infer[double]2{@\lab.\term[1] \intersection \term[2] \longrightarrow^{\struc}_{\fork^{x}_{1} \series (\trace[3] \parallel \trace[3]')} \lterm[1] \intersection_{1} \lterm[2]}
        \end{prooftree}, Case                          
        \begin{prooftree}
            \hypo{\lterm[1] \longrightarrow^{\struc}_{\trace[3]} \lterm[1]'}
            \hypo{\EPS_{\lab[3]}(\lterm[1]')}
            \hypo{\lterm[2] \longrightarrow^{\struc}_{\trace[3]'} \lterm[2]'}
            \hypo{\EPS_{\lab[3]}(\lterm[2]')}
            \infer[double]4{\lterm[1] \intersection_{1} \lterm[2] \longrightarrow^{\struc}_{(\trace[3] \parallel \trace[3]') \series \join^{\lab[3]}_{1}}  @\lab[3].\id}
        \end{prooftree},\\
        \vskip0.5\baselineskip 
        Case 
        \begin{prooftree}
            \hypo{\lterm[1] \longrightarrow^{\struc}_{\trace[3]} \lterm[1]'}
            \hypo{\lterm[2] \longrightarrow^{\struc}_{\trace[3]'} \lterm[2]'}
            \infer[double]2{\lterm[1] \intersection_{1} \lterm[2] \longrightarrow^{\struc}_{(\trace[3] \parallel \trace[3]')} \lterm[1]' \intersection_{1} \lterm[2]'}
        \end{prooftree}:
        Similar to the case above.
    \end{itemize}
    Hence, this direction has been proved.

    ($\subseteq$):
    We show that the transitivity rule $\begin{prooftree}
        \hypo{\lterm[1] \longleadsto^{\struc}_{\trace[1]} \lterm[3]}
        \hypo{\lterm[3] \longleadsto^{\struc}_{\trace[2]} \lterm[2]}
        \infer2[T]{\lterm[1] \longleadsto^{\struc}_{\trace[1] \series \trace[2]} \lterm[2]}
    \end{prooftree}$ is admissible,
    by induction on the sum of the sizes of the derivation trees of $\lterm[1] \longleadsto^{\struc}_{\trace[1]} \lterm[3]$ and $\lterm[3] \longleadsto^{\struc}_{\trace[2]} \lterm[2]$.
    \begin{itemize}
        \item Case $\begin{prooftree}
            \hypo{\lterm[1] \longleadsto^{\struc}_{\trace[3]} \lterm[2]}
            \hypo{\mathstrut}
            \infer1[R]{\lterm[2] \longleadsto^{\struc}_{\id^{\overrightarrow{\mathsf{lab}}(\lterm[2])}} \lterm[2]}
            \infer2[T]{\lterm[1] \longleadsto^{\struc}_{\trace[3]} \lterm[2]}
        \end{prooftree}$, Case $\begin{prooftree}
            \hypo{\mathstrut}
            \infer1[R]{\lterm[1] \longleadsto^{\struc}_{\id^{\overrightarrow{\mathsf{lab}}(\lterm[1])}} \lterm[1]}
            \hypo{\lterm[1] \longleadsto^{\struc}_{\trace[3]} \lterm[2]}
            \infer2[T]{\lterm[1] \longleadsto^{\struc}_{\trace[3]} \lterm[2]}
        \end{prooftree}$ (i.e., when the immediate above rule contains (R)):
        Clear.

        \item Case $\begin{prooftree}
            \hypo{\mathstrut}
            \infer1{@ \lab[1]. a \longleadsto^{\struc}_{\trace[3]} @ \lab[2]. \id}
            \hypo{@ \lab[2]. \id \longleadsto^{\struc}_{\trace[3]'} \lterm[2]}
            \infer2[T]{@ \lab[1]. a \longleadsto^{\struc}_{\trace[3] \series \trace[3]'} \lterm[2]}
        \end{prooftree}$:
        Then the right-hand side rule should be the rule (R); thus, this case has already been proved.

        \item Case $\begin{prooftree}
            \hypo{@\lab.\term[1] \longleadsto^{\struc}_{\trace[3]} \lterm[1]''}
            \infer1{@\lab.\term[1] \compo \term[2] \longleadsto^{\struc}_{\trace[3]} \lterm[1]'' \compo_{1} \term[2]}
            \hypo{\lterm[1]'' \compo_{1} \term[2] \longleadsto^{\struc}_{\trace[3]'} \lterm[3]}
            \infer2[T]{@\lab.\term[1] \compo \term[2] \longleadsto^{\struc}_{\trace[3] \series \trace[3]'} \lterm[3]}
        \end{prooftree}$:
        For the form of the above of $\lterm[1]'' \compo_{1} \term[2] \longleadsto^{\struc}_{\trace[3]'} \lterm[3]$,
        we distinguish the following cases (except for the rule (R)):
        \begin{itemize}
            \item Case $\begin{prooftree}
                \hypo{\lterm[1]'' \longleadsto^{\struc}_{\trace[3]'} \lterm[1]'}
                \infer1{\lterm[1]'' \compo_{1} \term[2] \longleadsto^{\struc}_{\trace[3]'} \lterm[1]' \compo_{1} \term[2]}
            \end{prooftree}$:
            Then we have:
            $\begin{prooftree}
                \hypo{@\lab.\term[1] \longleadsto^{\struc}_{\trace[3]} \lterm[1]''}
                \hypo{\lterm[1]'' \longleadsto^{\struc}_{\trace[3]'} \lterm[1]'}
                \infer2[T]{@\lab.\term[1] \longleadsto^{\struc}_{\trace[3] \series \trace[3]'} \lterm[1]'}
                \infer1{@\lab.\term[1] \compo \term[2] \longleadsto^{\struc}_{\trace[3] \series \trace[3]'} \lterm[1]' \compo_{1} \term[2]}
            \end{prooftree}$.
            By IH, this case has been proved.
            \item Case $\begin{prooftree}
                \hypo{\lterm[1]'' \longleadsto^{\struc}_{\trace[3]_1'} \lterm[1]'}
                \hypo{\EPS_{\lab[3]}(\lterm[1]')}
                \hypo{@\lab[3].\term[2] \longleadsto^{\struc}_{\trace[3]_2'} \lterm[3]}
                \infer3{\lterm[1]'' \compo_{1} \term[2] \longleadsto^{\struc}_{\trace[3]_1' \series \trace[3]_2'} \lterm[3]}
            \end{prooftree}$:
            Then we have:\\
            $\begin{prooftree}
                \hypo{@\lab.\term[1] \longleadsto^{\struc}_{\trace[3]} \lterm[1]''}
                \hypo{\lterm[1] \longleadsto^{\struc}_{\trace[3]_1'} \lterm[1]'}
                \infer2[T]{@\lab.\term[1] \longleadsto^{\struc}_{\trace[3] \series \trace[3]_1'} \lterm[1]'}
                \hypo{\EPS_{\lab[3]}(\lterm[1]')}
                \hypo{@\lab[3].\term[2] \longleadsto^{\struc}_{\trace[3]_2'} \lterm[3]}
                \infer3{@\lab.\term[1] \compo \term[2] \longleadsto^{\struc}_{(\trace[3] \series \trace[3]_1') \series \trace[3]_2'} \lterm[3]}
            \end{prooftree}$.
            By IH, this case has been proved.
            \vskip0.5\baselineskip 
        \end{itemize}
        \item Case $\begin{prooftree}
            \hypo{@\lab[1].\term[1] \longleadsto^{\struc}_{\trace[3]} \lterm[1]'}
            \hypo{\EPS_{\lab[3]}(\lterm[1]')}
            \hypo{@\lab[3].\term[2] \longleadsto^{\struc}_{\trace[3]'} \lterm[2]''}
            \infer3{@\lab[1].\term[1] \compo \term[2] \longleadsto^{\struc}_{\trace[3] \series \trace[3]'} \lterm[2]''}
            \hypo{\lterm[2]'' \longleadsto^{\struc}_{\trace[3]''} \lterm[2]'}
            \infer2[T]{@\lab[1].\term[1] \compo \term[2] \longleadsto^{\struc}_{(\trace[3] \series \trace[3]') \series \trace[3]''} \lterm[2]'}
        \end{prooftree}$:
        Then we have:
        
        $\begin{prooftree}[separation=.5em]
            \hypo{@\lab[1].\term[1] \longleadsto^{\struc}_{\trace[3]} \lterm[1]'}
            \hypo{\EPS_{\lab[3]}(\lterm[1]')}
            \hypo{@\lab[3].\term[2] \longleadsto^{\struc}_{\trace[3]'} \lterm[2]''}
            \hypo{\lterm[2]'' \longleadsto^{\struc}_{\trace[3]''} \lterm[2]'}
            \infer2[T]{@\lab[3].\term[2] \longleadsto^{\struc}_{\trace[3]' \series \trace[3]''} \lterm[2]'}
            \infer3{@\lab[1].\term[1] \compo \term[2] \longleadsto^{\struc}_{\trace[3] \compo (\trace[3]' \series \trace[3]'')} \lterm[2]'}
        \end{prooftree}$.
        By IH, this case has been proved.
        \vskip0.5\baselineskip 
        \item Case \begin{prooftree}[separation=.5em]
            \hypo{\lterm[1] \longleadsto^{\struc}_{\trace[3]} \lterm[1]'}
            \infer1{\lterm[1] \compo_{1} \term[2] \longleadsto^{\struc}_{\trace[3]} \lterm[1]' \compo_{1} \term[2]}
            \hypo{ \lterm[1]' \compo_{1} \term[2]  \longleadsto^{\struc}_{\trace[3]'} \lterm[2]'}
            \infer2[T]{\lterm[1] \longleadsto^{\struc}_{\trace[3] \series \trace[3]'} \lterm[2]'}
        \end{prooftree},
        Case \begin{prooftree}[separation=.5em]
            \hypo{\lterm[1] \longleadsto^{\struc}_{\trace[3]} \lterm[1]'}
            \hypo{\EPS_{\lab[3]}(\lterm[1]')}
            \hypo{@\lab[3].\term[2] \longleadsto^{\struc}_{\trace[3]'} \lterm[2]''}
            \infer3{\lterm[1] \compo_{1} \term[2] \longleadsto^{\struc}_{\trace[3] \series \trace[3]'} \lterm[2]''}
            \hypo{\lterm[2]'' \longleadsto^{\struc}_{\trace[3]''} \lterm[2]'}
            \infer2[T]{\lterm[1] \compo_{1} \term[2] \longleadsto^{\struc}_{(\trace[3] \series \trace[3]') \series \trace[3]''} \lterm[2]'}
        \end{prooftree}:
        Similar to the two cases above.
    
        \item Case\\
        \begin{prooftree}
            \hypo{@\lab[1]_0.\term[1] \longleadsto^{\struc}_{\trace[3]_0} \lterm[1]_1'}
            \hypo{\EPS_{\lab[1]_1}(\lterm[1]_1')}
            \hypo{\dots}
            \hypo{\EPS_{\lab[1]_n}(\lterm[1]_n')}
            \hypo{@\lab[1]_{n}.\term[1] \longleadsto^{\struc}_{\trace[3]_{n}'} \lterm[1]_{n+1}''}
            \infer5{ @\lab[1]_0.\term[1]^* \longleadsto^{\struc}_{\trace[3]_0 \series \dots \series \trace[3]_{n-1} \series \trace[3]_n'} \lterm[1]_{n+1}'' \compo_{1} \term[1]^*}
            \hypo{ \lterm[1]_{n+1}'' \compo_{1} \term[1]^* \longleadsto^{\struc}_{\trace[3]''} \lterm[3]}
            \infer2[T]{@\lab[1]_0.\term[1]^* \longleadsto^{\struc}_{(\trace[3]_0 \series \dots \series \trace[3]_{n-1} \series \trace[3]_n') \series \trace[3]''} \lterm[3]}
        \end{prooftree}:

        For the form of the above of $\lterm[1]_{n+1}' \compo_{1} \term[1]^* \longleadsto^{\struc}_{\trace[3]''} \lterm[3]$,
        we distinguish the following cases (except for the rule (R)):
        \begin{itemize}
            \item Case \begin{prooftree}
                \hypo{\lterm[1]_{n+1}'' \longleadsto^{\struc}_{\trace[3]''} \lterm[1]_{n+1}'}
                \infer1{\lterm[1]_{n+1}'' \compo_{1} \term[1]^* \longleadsto^{\struc}_{\trace[3]''} \lterm[1]_{n+1}' \compo_{1} \term[1]^*}
            \end{prooftree}:
            Then we have:

            \begin{prooftree}
                \hypo{@\lab[1]_0.\term[1] \longleadsto^{\struc}_{\trace[3]_0} \lterm[1]_1'}
                \hypo{\EPS_{\lab[1]_1}(\lterm[1]_1')}
                \hypo{\dots}
                \hypo{\EPS_{\lab[1]_n}(\lterm[1]_n')}
                \hypo{@\lab[1]_{n}.\term[1] \longleadsto^{\struc}_{\trace[3]_{n}'} \lterm[1]_{n+1}''}
                \hypo{ \lterm[1]_{n+1}'' \longleadsto^{\struc}_{\trace[3]''} \lterm[1]_{n+1}'}
                \infer2[T]{@\lab[1]_{n}.\term[1] \longleadsto^{\struc}_{\trace[3]_n' \series \trace[3]''} \lterm[1]_{n+1}'}
                \infer5{ @\lab[1]_0.\term[1]^* \longleadsto^{\struc}_{\trace[3]_0 \series \dots \series (\trace[3]_n' \series \trace[3]'')} \lterm[1]_{n+1}' \compo_{1} \term[1]^*}
            \end{prooftree}.

            Thus by IH, this case has been proved.

            \item Case \begin{prooftree}
                \hypo{\lterm[1]_{n+1}'' \longleadsto^{\struc}_{\trace[3]_1''} \lterm[1]_{n+1}'}
                \hypo{\EPS_{\lab[1]_{n+1}}(\lterm[1]_{n+1}')}
                \hypo{@\lab[1]_{n+1}. \term[1]^* \longleadsto^{\struc}_{\trace[3]_2''} \lterm[3]}
                \infer3{\lterm[1]_{n+1}' \compo_{1} \term[1]^* \longleadsto^{\struc}_{\trace[3]_1'' \series \trace[3]_2''} \lterm[3]}
            \end{prooftree}:
            For the form of the above of $@\lab[1]_{n+1}. \term[1]^* \longleadsto^{\struc}_{\trace[3]_2''} \lterm[3]$,
            we distinguish the following cases:
            \begin{itemize}
                \item Case \begin{prooftree}
                    \hypo{\mathstrut}
                    \infer1[R]{@\lab[1]_{n+1}. \term[1]^* \longleadsto^{\struc}_{\id^{\lab[1]_{n+1}}} @\lab[1]_{n+1}. \term[1]^*}
                \end{prooftree} (then, $\trace[3]'' = \trace[3]_1''$):
                Then we have:\\
                \begin{prooftree}
                    \hypo{@\lab[1]_{0}.\term[1] \longleadsto^{\struc}_{\trace[3]_{0}} \lterm[1]_{1}' \; \EPS_{\lab[1]_{1}}(\lterm[1]_{1}') \; \dots}
                    \hypo{@\lab[1]_{n}. \term[1] \longleadsto^{\struc}_{\trace[3]_{n}'} \lterm[1]_{n+1}''}
                    \hypo{\lterm[1]_{n+1}'' \longleadsto^{\struc}_{\trace[3]''} \lterm[1]_{n+1}'}
                    \infer2[T]{@\lab[1]_{n}. \term[1] \longleadsto^{\struc}_{\trace[3]_{n}' \series \trace[3]''} \lterm[1]_{n+1}'}
                    \hypo{\EPS_{\lab[1]_{n+1}}(\lterm[1]_{n+1}')}
                    \infer3{@\lab[1]_{0}. \term[1]^* \longleadsto^{\struc}_{\trace[3]_{0} \series \dots \series \trace[3]_{n-1} \series \trace[3]_{n}' \series \trace[3]''} @ \lab[1]_{n+1}. \term[1]^*}
                \end{prooftree}.\\
                Thus by IH, this case has been proved.

                \item Case\\
                \begin{prooftree}
                    \hypo{@\lab[1]_{n+1}.\term[1] \longleadsto^{\struc}_{\trace[3]_{n+1}} \lterm[1]_{n+2}'}
                    \hypo{\EPS_{\lab[1]_{n+2}}(\lterm[1]_{n+2}') \quad\dots\quad \EPS_{\lab[1]_{m+1}}(\lterm[1]_{m+1}')}
                    \hypo{@\lab[1]_{m+1}.\term[1] \longleadsto^{\struc}_{\trace[3]_{m+1}} \lterm[1]_{m+2}'}
                    \infer3{@\lab[1]_{n+1}. \term[1]^* \longleadsto^{\struc}_{\trace[3]_{n+1} \series \dots \series \trace[3]_{m+1}} \lterm[1]_{m+2}' \compo_{1} \term[1]^*}
                \end{prooftree}:

                Then by letting $\trace[3]_{n} = \trace[3]_{n}' \series \trace[3]_1''$, we have:\\
                \begin{prooftree}
                    \hypo{@\lab[1]_{0}.\term[1] \longleadsto^{\struc}_{\trace[3]_{0}} \lterm[1]_{1}'}
                    \hypo{\EPS_{\lab[1]_{1}}(\lterm[1]_{1}')}
                    \hypo{\dots\quad\dots}
                    \hypo{\EPS_{\lab[1]_{m+1}}(\lterm[1]_{m+1}')}
                    \hypo{@\lab[1]_{m+1}.\term[1] \longleadsto^{\struc}_{\trace[3]_{m+1}} \lterm[1]_{m+2}'}
                    \infer5{@\lab[1]_{0}. \term[1]^* \longleadsto^{\struc}_{\trace[3]_{0} \series \dots \series \trace[3]_{m+1}} \lterm[1]_{m+2}' \compo_{1} \term[1]^*}
                \end{prooftree}.\\
                Thus by IH, this case has been proved.
            \end{itemize}
        \end{itemize}
        \vskip0.5\baselineskip 
        \item Case
        \begin{prooftree}[separation=.5em]
            \hypo{@\lab[1]_0.\term[1] \longleadsto^{\struc}_{\trace[3]_0} \lterm[1]_1'}
            \hypo{\EPS_{\lab[1]_1}(\lterm[1]_1')}
            \hypo{\dots}
            \hypo{\EPS_{\lab[1]_n}(\lterm[1]_n')}
            \infer4{ @\lab[1]_0.\term[1]^* \longleadsto^{\struc}_{\trace[3]_0 \series \dots \series \trace[3]_{n-1}} @\lab[1]_{n} \term[1]^*}
            \hypo{ @\lab[1]_{n} \term[1]^* \compo_{1} \term[1]^* \longleadsto^{\struc}_{\trace[3]'} \lterm[3]}
            \infer2[T]{@\lab[1]_0.\term[1]^* \longleadsto^{\struc}_{\trace[3]_0 \series \dots \series \trace[3]_{n-1}} \lterm[3]}
        \end{prooftree}:
        Similar to the case above.

        \item Case 
        \begin{prooftree}
            \hypo{@\lab.\term[1] \longleadsto^{\struc}_{\trace[3]} \lterm[1]'}
            \infer1{@\lab.\term[1] \union \term[2] \longleadsto^{\struc}_{\trace[3]} \lterm[1]'}
            \hypo{\lterm[1]' \longleadsto^{\struc}_{\trace[3]'} \lterm[3]}
            \infer2[T]{@\lab[1].\term[1] \union \term[2] \longleadsto^{\struc}_{\trace[3] \series \trace[3]'} \lterm[3]}
        \end{prooftree}:
        We then have:
        \begin{prooftree}
            \hypo{@\lab.\term[1] \longleadsto^{\struc}_{\trace[3]} \lterm[1]'}
            \hypo{\lterm[1]' \longleadsto^{\struc}_{\trace[3]'} \lterm[3]}
            \infer2[T]{@\lab.\term[1] \longleadsto^{\struc}_{\trace[3] \series \trace[3]'} \lterm[3]}
            \infer1{@\lab[1].\term[1] \union \term[2] \longleadsto^{\struc}_{\trace[3] \series \trace[3]'} \lterm[3]}
        \end{prooftree}.
        Thus by IH, this case has been proved.

        \item Case 
        \begin{prooftree}
            \hypo{@\lab.\term[2] \longleadsto^{\struc}_{\trace[3]} \lterm[2]'}
            \infer1{@\lab.\term[1] \union \term[2] \longleadsto^{\struc}_{\trace[3]} \lterm[2]'}
            \hypo{\lterm[2]' \longleadsto^{\struc}_{\trace[3]'} \lterm[3]}
            \infer2[T]{@\lab[1].\term[1] \union \term[2] \longleadsto^{\struc}_{\trace[3] \series \trace[3]'} \lterm[3]}
        \end{prooftree}:
        Similar to the case above.
        \vskip0.5\baselineskip 
        \item Case
        \begin{prooftree}
            \hypo{@\lab.\term[1] \longleadsto^{\struc}_{\trace[3]_1} \lterm[1]'}
            \hypo{\EPS_{\lab[3]}(\lterm[1]')}
            \hypo{@\lab.\term[2] \longleadsto^{\struc}_{\trace[3]_2} \lterm[2]'}
            \hypo{\EPS_{\lab[3]}(\lterm[2]')}
            \infer4{@\lab.\term[1] \intersection \term[2] \longleadsto^{\struc}_{\fork^{x}_{1} \series (\trace[3]_1 \parallel \trace[3]_2) \series \join^{\lab[3]}_{1}}  @\lab[3].\id}
            \hypo{@\lab[3].\id \longleadsto^{\struc}_{\trace[3]'} \lterm[3]}
            \infer2[T]{@\lab.\term[1] \intersection \term[2] \longleadsto^{\struc}_{(\fork^{x}_{1} \series (\trace[3]_1 \parallel \trace[3]_2) \series \join^{\lab[3]}_{1}) \series \trace[3]'} \lterm[3]}
        \end{prooftree},\\
        \vskip0.5\baselineskip 
        Case 
        \begin{prooftree}
            \hypo{\lterm[1] \longleadsto^{\struc}_{\trace[3]_1} \lterm[1]'}
            \hypo{\EPS_{\lab[3]}(\lterm[1]')}
            \hypo{\lterm[2] \longleadsto^{\struc}_{\trace[3]_2} \lterm[2]'}
            \hypo{\EPS_{\lab[3]}(\lterm[2]')}
            \infer4{\lterm[1] \intersection_{1} \lterm[2] \longleadsto^{\struc}_{(\trace[3]_1 \parallel \trace[3]_2) \series \join^{\lab[3]}_{1}}  @\lab[3].\id}
            \hypo{@\lab[3].\id \longleadsto^{\struc}_{\trace[3]'} \lterm[3]}
            \infer2[T]{\lterm[1] \intersection_{1} \lterm[2] \longleadsto^{\struc}_{((\trace[3]_1 \parallel \trace[3]_2) \series \join^{\lab[3]}_{1}) \series \trace[3]'} \lterm[3]}
        \end{prooftree}:
        Then the right-hand side rule should be the rule (R); thus, this case has already been proved.
        \vskip0.5\baselineskip 
        \item Case 
        \begin{prooftree}
            \hypo{@\lab.\term[1] \longleadsto^{\struc}_{\trace[3]_1} \lterm[1]''}
            \hypo{@\lab.\term[2] \longleadsto^{\struc}_{\trace[3]_2} \lterm[2]''}
            \infer2{@\lab.\term[1] \intersection \term[2] \longleadsto^{\struc}_{\fork^{x}_{1} \series (\trace[3]_1 \parallel \trace[3]_2)} \lterm[1]'' \intersection_{1} \lterm[2]''}
            \hypo{\lterm[1]'' \intersection_{1} \lterm[2]'' \longleadsto^{\struc}_{\trace[3]'} \lterm[3]}
            \infer2[T]{@\lab.\term[1] \intersection \term[2] \longleadsto^{\struc}_{(\fork^{x}_{1} \series (\trace[3]_1 \parallel \trace[3]_2)) \series \trace[3]'} \lterm[3]}
        \end{prooftree}:
        For the form of the above of $\lterm[1]'' \intersection_{1} \lterm[2]'' \longleadsto^{\struc}_{\trace[3]'} \lterm[3]$,
        we distinguish the following cases (except for the rule (R)):
        \begin{itemize}
            \item Case
            \begin{prooftree}
                \hypo{\lterm[1]'' \longleadsto^{\struc}_{\trace[3]'_1} \lterm[1]'}
                \hypo{\lterm[2]'' \longleadsto^{\struc}_{\trace[3]'_2} \lterm[2]'}
                \infer2{\lterm[1]'' \intersection_{1} \lterm[2]'' \longleadsto^{\struc}_{\trace[3]'_1 \parallel \trace[3]'_2} \lterm[1]' \intersection_{1} \lterm[2]'}
            \end{prooftree}:
            Then we have:

            \begin{prooftree}
                \hypo{@\lab.\term[1] \longleadsto^{\struc}_{\trace[3]_1} \lterm[1]''}
                \hypo{\lterm[1]'' \longleadsto^{\struc}_{\trace[3]'_1} \lterm[1]'}
                \infer2[T]{@\lab.\term[1] \longleadsto^{\struc}_{\trace[3]_1 \series \trace[3]'_1} \lterm[1]'}
                \hypo{@\lab.\term[2] \longleadsto^{\struc}_{\trace[3]_1} \lterm[2]''}
                \hypo{\lterm[2]'' \longleadsto^{\struc}_{\trace[3]'_1} \lterm[2]'}
                \infer2[T]{@\lab.\term[2] \longleadsto^{\struc}_{\trace[3]_2 \series \trace[3]'_2} \lterm[2]'}
                \infer2{@\lab.\term[1] \intersection \term[2] \longleadsto^{\struc}_{\fork^{x}_{1} \series ((\trace[3]_1 \series \trace[3]'_1) \parallel (\trace[3]_2 \series \trace[3]'_2))} \lterm[1]' \intersection_{1} \lterm[2]'}
            \end{prooftree}.

            Note that $\fork^{x}_{1} \series ((\trace[3]_1 \series \trace[3]'_1) \parallel (\trace[3]_2 \series \trace[3]'_2)) = (\fork^{x}_{1} \series (\trace[3]_1 \parallel \trace[3]_2)) \series (\trace[3]'_1 \parallel \trace[3]'_2)$.
            Thus by IH, this case has been proved.

            \item Case
            \begin{prooftree}
                \hypo{\lterm[1]'' \longleadsto^{\struc}_{\trace[3]'_1} \lterm[1]'}
                \hypo{\EPS_{\lab[3]}(\lterm[1]')}
                \hypo{\lterm[2]'' \longleadsto^{\struc}_{\trace[3]'_2} \lterm[2]'}
                \hypo{\EPS_{\lab[3]}(\lterm[2]')}
                \infer4{\lterm[1]'' \intersection_{1} \lterm[2]'' \longleadsto^{\struc}_{(\trace[3]'_1 \parallel \trace[3]'_2) \series \join^{\lab[3]}_{1}} @\lab[3]. \id}
            \end{prooftree}:
            Then we have:

            \begin{prooftree}
                \hypo{@\lab.\term[1] \longleadsto^{\struc}_{\trace[3]_1} \lterm[1]''}
                \hypo{\lterm[1]'' \longleadsto^{\struc}_{\trace[3]'_1} \lterm[1]'}
                \infer2[T]{@\lab.\term[1] \longleadsto^{\struc}_{\trace[3]_1 \series \trace[3]'_1} \lterm[1]'}
                \hypo{\EPS_{\lab[3]}(\lterm[1]')}
                \hypo{@\lab.\term[2] \longleadsto^{\struc}_{\trace[3]_1} \lterm[2]''}
                \hypo{\lterm[2]'' \longleadsto^{\struc}_{\trace[3]'_1} \lterm[2]'}
                \infer2[T]{@\lab.\term[2] \longleadsto^{\struc}_{\trace[3]_2 \series \trace[3]'_2} \lterm[2]'}
                \hypo{\EPS_{\lab[3]}(\lterm[2]')}
                \infer4{@\lab.\term[1] \intersection \term[2] \longleadsto^{\struc}_{\fork^{x}_{1} \series ((\trace[3]_1 \series \trace[3]'_1) \parallel (\trace[3]_2 \series \trace[3]'_2)) \series \join^{\lab[3]}_{1}} \lterm[1]' \intersection_{1} \lterm[2]'}
            \end{prooftree}.

            Note that $\fork^{x}_{1} \series ((\trace[3]_1 \series \trace[3]'_1) \parallel (\trace[3]_2 \series \trace[3]'_2)) \series \join^{\lab[3]}_{1}
             = (\fork^{x}_{1} \series (\trace[3]_1 \parallel \trace[3]_2)) \series ((\trace[3]'_1 \parallel \trace[3]'_2) \series \join^{\lab[3]}_{1})$.
            Thus by IH, this case has been proved.
        \end{itemize}
        \vskip0.5\baselineskip 
        \item Case 
        \begin{prooftree}
            \hypo{\lterm[1] \longleadsto^{\struc}_{\trace[3]_1} \lterm[1]''}
            \hypo{\lterm[2] \longleadsto^{\struc}_{\trace[3]_2} \lterm[2]''}
            \infer2{\lterm[1] \intersection \term[2] \longleadsto^{\struc}_{\trace[3]_1 \parallel \trace[3]_2} \lterm[1]'' \intersection_{1} \lterm[2]''}
            \hypo{\lterm[1]'' \intersection_{1} \lterm[2]'' \longleadsto^{\struc}_{\trace[3]'} \lterm[3]}
            \infer2[T]{\lterm[1] \intersection_{1} \lterm[2] \longleadsto^{\struc}_{(\trace[3]_1 \parallel \trace[3]_2) \series \trace[3]'} \lterm[3]}
        \end{prooftree}:
        Similar to the case above.
    \end{itemize}

    Hence, this completes the proof.
\end{proof}
\end{document}